\documentclass[num-refs]{wiley-article}
\usepackage{color}
\usepackage{ulem}

\usepackage{amssymb}
\usepackage[colorlinks]{hyperref}
\usepackage{graphicx}
\usepackage{caption}
\usepackage{appendix}
\usepackage{mathtools} % for \bsmallmatrix
\usepackage{array} % to center columns of given width in tabular
\usepackage{float} % allows for specific placement of tables
\usepackage{multirow} % to center column labels over multiple rows in tables
\usepackage{multicol} % to create multiple tabular columns
\usepackage{booktabs} % for \cmidrule in tables
\usepackage{centernot} % this is to get a larger slash than with the \not command. We use it in the equation H0 versus H1.
\usepackage{url} % to make line breaks possible in URLs in the References section
\usepackage{colortbl} % to color cells in tables
\usepackage[table]{xcolor} % to color comments and cells in tables
\usepackage{soul} % for highlighting

\usepackage{tabularx}

\microtypesetup{expansion=false}

\makeatletter
\@ifundefined{remark}{}{%

}
\makeatother

\newtheorem{remark}{Remark}

\newcommand{\N}{\mathbb{N}}
\newcommand{\R}{\mathbb{R}}
\newcommand{\PP}{\mathsf{P}} % Russian style, do not change
\newcommand{\EE}{\mathsf{E}} % Russian style, do not change
\newcommand{\Var}{\mathsf{Var}} % Russian style, do not change
\newcommand{\bb}[1]{\boldsymbol{#1}}
\newcommand{\nvert}[0]{\, \vert \, }
\newcommand{\rd}{\mathrm{d}}
\newcommand{\Tn}[1]{\cellcolor{palegray}{#1}}
\newcommand{\LK}[1]{\cellcolor{verypalegray}{#1}}
\newcommand{\sqrtsmash}[1]{\sqrt{\smash[b]{#1}}}
\newcommand{\1}{\mathbf{1}}

\DeclareMathSymbol{\mhyphen}{\mathord}{AMSa}{"39}

\definecolor{palegray}{RGB}{140,140,140}
\definecolor{verypalegray}{RGB}{200,200,200}

\newcommand{\proofEPD}{S1{~}}

\newcommand{\proofuniform}{S12{~}}
\newcommand{\prooflocalgamma}{S13{~}}
\newcommand{\prooflocalEPD}{S14{~}}
\newcommand{\powercurves}{S15{~}}

\allowdisplaybreaks

\papertype{Original Article}
\paperfield{Journal Section}

\title{Omnibus goodness-of-fit tests for univariate continuous distributions based on trigonometric moments}

\author[1\authfn{1}]{Alain Desgagn\'e}
\author[2\authfn{1}]{Fr\'ed\'eric Ouimet}

\affil[1]{D\'epartement de math\'ematiques, Universit\'e du Qu\'ebec \`a Montr\'eal, Canada}
\affil[2]{D\'epartement de math\'ematiques et d'informatique, Universit\'e du Qu\'ebec \`a Trois-Rivi\`eres, Canada}

\corraddress{Alain Desgagn\'e \\ D\'epartement de math\'ematiques \\ Universit\'e du Qu\'ebec \`a Montr\'eal \\
201, Avenue du Pr\'{e}sident-Kennedy \\ Montr\'eal (Qu\'ebec) \\ Canada H2X 3Y7}
\corremail{desgagne.alain@uqam.ca}

\fundinginfo{F.\ Ouimet is supported by the start-up fund (1729971) from the Universit\'e du Qu\'ebec \`a Trois-Rivi\`eres.}

\runningauthor{Desgagn\'e \& Ouimet}

\begin{document}

\maketitle

%% ABSTRACT
\begin{abstract}
We propose a new omnibus goodness-of-fit test based on trigonometric moments of probability-integral-transformed data. The test builds on the framework of the LK test introduced by Langholz and Kronmal [J. Amer. Statist. Assoc. 86 (1991), 1077--1084], but fully exploits the covariance structure of the associated trigonometric statistics. As a result, our test statistic converges under the null hypothesis to a $\chi_2^2$ distribution, even in the presence of nuisance parameters, yielding a well-calibrated rejection region. We derive the exact asymptotic covariance matrix required for normalization and propose a unified approach to computing the LK normalizing scalar. The applicability of both the proposed test and the LK test is substantially expanded by providing implementation details for 11 families of continuous distributions, covering most commonly used parametric models. Simulation studies demonstrate accurate empirical size, close to the nominal level, and strong power properties, yielding fully plug-and-play procedures. Further insight is provided by an analysis under local alternatives. The methodology is illustrated using surface temperature forecast errors from a numerical weather prediction model.

% Please include a maximum of seven keywords
%\MSC[2020]{Primary: 62F03; Secondary: 60F05, 62E20, 62F12, 62H10, 62H12, 62H15}
\end{abstract}

\section{Introduction}\label{sec:introduction}

Parametric goodness-of-fit tests are essential statistical tools used to determine how well a given parametric distribution family fits a set of observations. Their importance spans various fields, including economics, biology, medicine, and engineering, as they ensure the accurate representation of real-world data. These tests are crucial for assessing the appropriateness of a chosen model, identifying its deficiencies, and verifying underlying data assumptions. Furthermore, they provide a quantitative basis for comparing multiple candidate distributions and selecting the one that offers the best predictive accuracy. A well-fitting model delivers valuable insights into the data's distribution and characteristics, which is fundamental for any subsequent analysis.

Within this framework, an omnibus goodness-of-fit test is designed to evaluate the overall fit of a model without targeting a specific alternative hypothesis \citep{doi:10.1201/9780203753064}. Unlike directional tests that focus on particular types of deviations from the null hypothesis (such as skewness or kurtosis), omnibus tests are engineered to detect a broad spectrum of discrepancies. As stated by \citet[p.~5]{MR1029526}, ``omnibus tests are intended to have moderate power against all alternatives.'' This versatility makes them powerful instruments for identifying any form of lack of fit, though they might be less sensitive to specific deviations compared to more specialized tests.

The literature on omnibus goodness-of-fit testing is extensive; see, e.g., \citet{doi:10.1201/9780203753064} for a book reference on this subject. Many classical tests are based on the empirical distribution function (EDF); after applying the probability integral transform to the data, they assess the uniformity of the transformed sample. This approach is broadly applicable to nearly any null hypothesis and includes well-known supremum statistics like the Kolmogorov--Smirnov and Kuiper tests, as well as quadratic statistics such as the Cram\'er--von Mises, Anderson--Darling, and Watson tests \citep[Section~4.2.2]{doi:10.1201/9780203753064}. When nuisance parameters are present, they are typically replaced by consistent estimators, a technique discussed by \citet{MR451521,MR653521,MR653522,MR885745,MR4411854,MR4906044}, among many others. For classical EDF-based tests, however, the effect of parameter estimation generally cannot be handled through a single, unified analytical adjustment. As a result, calibration relies either on ad hoc, distribution-specific corrections available for a few special cases, or on resampling-based methods.

Another powerful class of omnibus tests is based on orthogonal series expansions, where the density or log-density of the probability-integral-transformed data is expanded using an orthonormal basis; see, e.g., Section~8.11 of \citet{doi:10.1201/9780203753064}. The coefficients in these kinds of expansions are trivial under the null hypothesis of uniformity, so tests can be constructed by summing the squares of the estimated coefficients up to a cutoff. The choice of the orthonormal basis is critical and can influence the power of the test against different alternatives. Legendre polynomials have been a popular choice \citep[see, e.g.,][]{MR1294744,MR1482140,MR1893336,MR3010617,MR3536197,MR4718464,MR4874943}, leading to a test statistic that is effective but can be computationally intensive in the presence of nuisance parameters as the series cutoff increases, given that the statistic must be rescaled by an appropriate covariance matrix to maintain its asymptotic distribution.

The tests based on trigonometric moments proposed in this paper can be viewed as belonging to this latter framework, since trigonometric functions form the Fourier basis associated with orthogonal series expansions. This idea was originally put forward by \citet{MR1146353}, who proposed using the sample means of the first two nontrivial Fourier basis functions evaluated at the probability-integral-transformed data $U_i = F(X_i \nvert \boldsymbol{\theta})$. These quantities correspond to our $U$-statistic
$[C_n(\bb{\theta}),S_n(\bb{\theta})]^{\top}=n^{-1} \sum_{i=1}^n [\cos (2\pi U_i), \sin (2\pi U_i)]^{\top}$ as defined in \eqref{eq:Cn.Sn}.
Their resulting procedure, referred to as the LK test in this paper and defined by $K_1(\hat{\boldsymbol{\theta}}_n)=V(\hat{\boldsymbol{\theta}}_n)^{-1}\, 2n \bigl\{ C_n^2(\hat{\boldsymbol{\theta}}_n) + S_n^2(\hat{\boldsymbol{\theta}}_n) \bigr\}$, appears to possess many compelling attributes of an omnibus goodness-of-fit test, including:
(i) simplicity and ease of implementation, with no tuning parameters required;
(ii) interpretability of $C_n(\boldsymbol{\theta})$ as a measure of relative tail weight and central concentration, and of $S_n(\boldsymbol{\theta})$ as a measure of relative skewness;
(iii) genuinely omnibus behavior;
(iv) strong power properties;
(v) convergence under $\mathcal{H}_0$ to a simple chi-square limiting distribution despite the presence of nuisance parameters;
(vi) a remarkably accurate chi-square approximation even for small sample sizes, enabling a fully plug-and-play implementation; and
(vii) applicability to a broad range of continuous parametric distributions.

To the best of our knowledge, no other goodness-of-fit test combines such wide applicability with a fully plug-and-play implementation that allows critical values and $p$-values to be computed directly from chi-square quantiles, without resorting to simulations, even in the presence of nuisance parameters. However, \citet{MR1146353} provide full implementation details only for a limited number of specific cases, namely the normal, exponential (implicitly covering the Pareto distribution via a logarithmic transformation), Weibull, Laplace, and uniform distributions. Moreover, determining the normalizing scalar $V(\boldsymbol{\theta})$ of the LK test for a given distribution is far from trivial and requires substantial analytical effort, which in practice limits the full potential of this approach. Therefore, one of our main contributions is to substantially broaden the applicability of the test to a wide range of null distributions encompassing most standard parametric families.

In addition, \citet{MR1146353} claim that their $K_1(\hat{\boldsymbol{\theta}}_n)$ test statistics have asymptotic $\chi_2^2$ distributions under the null hypothesis and certain regularity conditions. Our theoretical results (see Proposition~\ref{prop:asymp.dist.ML}) and simulation studies (see Section~\ref{sec:empirical.size.analysis}) show instead that the LK statistic is not exactly asymptotically $\chi_2^2$-distributed, even though the chi-square remains a very good approximation. This discrepancy arises because the LK test does not fully exploit the covariance matrix $\Sigma(\boldsymbol{\theta})$ of the vector $\sqrt{n}[C_n(\boldsymbol{\theta}), S_n(\boldsymbol{\theta})]^{\top}$ when normalizing the statistic, relying only on its trace. Another main contribution of this paper is therefore to take advantage of this observation and to propose a new test statistic, $T_n$, which fully accounts for the underlying covariance structure. This improved construction is expected to yield, on average, higher power for the $T_n$ test relative to the LK test, a conjecture that is supported by our simulation results (see Sections~\ref{sec:empirical.power.analysis} and~\ref{sec:empirical.power.analysis.Laplace}), where systematic average power gains are observed.

Our work makes several contributions to the field of goodness-of-fit testing:
\begin{enumerate}

\item Building on the theoretical framework of \citet{MR451521} and \citet{MR653521}, as well as its recent extension by \citet{MR4906044}, we derive the exact  covariance matrix $\Sigma(\boldsymbol{\theta})$ required to properly rescale $\sqrt{n}[C_n(\hat{\boldsymbol{\theta}}_n), S_n(\hat{\boldsymbol{\theta}}_n)]^{\top}$ in order to obtain asymptotic normality, for an arbitrary null distribution under $\mathcal{H}_0$.

\item We propose a new goodness-of-fit test $T_n(\hat{\boldsymbol{\theta}}_n)$ which, like the LK test, is based on $\sqrt{n}[C_n(\hat{\boldsymbol{\theta}}_n), S_n(\hat{\boldsymbol{\theta}}_n)]^{\top}$ but fully exploits the covariance matrix $\Sigma(\boldsymbol{\theta})$, leading to a more efficient use of the available information and improved power. We show that the resulting quadratic-form statistic converges to a $\chi_2^2$ distribution under $\mathcal{H}_0$.

\item We propose an alternative and straightforward way to compute the normalizing scalar $V(\boldsymbol{\theta})$ appearing in the LK test, based directly on the covariance matrix $\Sigma(\boldsymbol{\theta})$.

\item We substantially expand the applicability of both the $T_n$ and LK tests beyond the five specific null distributions considered by \citet{MR1146353}. Our framework covers 11 parametric distribution families and accounts for all combinations of known and unknown parameters within each family, yielding 53 distinct testing configurations characterized by different covariance matrices $\Sigma(\boldsymbol{\theta})$ and normalizing scalars $V(\boldsymbol{\theta})$. Note that fixing a subset of parameters under $\mathcal{H}_0$ does not define a single test, but rather a collection of tests obtained by specializing those parameters to particular values. For example, in the case of the exponential power distribution (EPD), treating the parameter $\lambda$ as known leads to distinct procedures for each fixed $\lambda$; in particular, $\lambda=1$ and $\lambda=2$ recover Laplace and normal tests, respectively. Consequently, our $T_n$ and LK constructions yield plug-and-play procedures for most commonly used parametric models.

\item We demonstrate through simulation studies, across a wide range of distributions, that the chi-square $\chi_2^2$ approximation to the null distribution of the $T_n$ and LK statistics provides a very good fit even for small and moderate sample sizes. As a consequence, critical values and $p$-values for the $T_n$ and LK tests can be computed accurately using chi-square quantiles, without relying on Monte Carlo simulations or pre-tabulated values.

\item We demonstrate the strong power properties of the $T_n$ and LK tests through simulation studies. In a first analysis, we examine the normal, Student's $t_2$, and exponential null distributions by generating power curves against several families of alternatives and comparing the results with classical EDF-based competitors. In a second analysis, we revisit the comprehensive simulation study of goodness-of-fit tests for the Laplace distribution conducted by \citet{MR4512291}, which evaluated $40$ competing procedures (including the LK test) against a large collection of $400$ alternative distributions. We augment this study by adding the $T_n$ test, which emerges as the most powerful Laplace test among the procedures considered.

\item We analyze the asymptotic behavior of this class of tests under local alternatives, focusing on two representative examples. Building on recent theoretical advances by \citet{MR4906044}, these analyses provide a deeper understanding of the power properties of the proposed tests in neighborhoods of the null hypothesis. The performance of the $T_n$ test is then compared with that of a natural benchmark provided by Rao's score test (also known as the Lagrange multiplier test), as well as with its asymptotically equivalent counterpart, the generalized likelihood ratio test (GLRT), for which the exact limiting distribution is also derived under the null hypothesis and under sequences of local alternatives.

\end{enumerate}

The remainder of this paper is organized as follows. Section~\ref{sec:GOF.general} details the construction of our goodness-of-fit tests, including the treatment of nuisance parameters. Section~\ref{sec:GOF.specific} presents the components required to implement the $T_n$ and LK tests for $11$ families of null distributions. Section~\ref{sec:empirical.simulations} reports the results of our simulation studies, including the empirical size under $\mathcal{H}_0$, empirical power analyses for the normal, Student's $t_2$, and exponential distributions, as well as a comprehensive empirical power study for the Laplace distribution. Section~\ref{sec:asymp.local.alternatives} investigates the asymptotic behavior of the $T_n$ test under local alternatives. Section~\ref{sec:application} analyzes a dataset of surface temperature forecast errors from a numerical weather prediction model to illustrate the practical utility of the proposed tests. Section~\ref{sec:conclusion} summarizes our findings and provides future directions for research.

The proofs of all results stated in the paper are collected in the \ref{supp}. The appendices provide a link to the \textsf{R} code used to reproduce all numerical results (\ref{app:reproducibility}), as well as a table of estimators for each parameter across the 11 distribution families and a table of constants (\ref{app:table.constants}).

\section{Goodness-of-fit tests based on trigonometric moments}\label{sec:GOF.general}

We consider independent and identically distributed (i.i.d.) observations $X_1, \ldots, X_n$ and want to test the following composite hypotheses:
\[
\mathcal{H}_0 : X_1, \ldots, X_n \sim F(\cdot \nvert \bb{\theta}_0)
\qquad \text{vs} \qquad
\mathcal{H}_1 : X_1, \ldots, X_n \centernot{\sim} F(\cdot \nvert \bb{\theta}_0),
\]
where the cumulative distribution function (CDF) $F(\cdot \nvert \boldsymbol{\theta}_0)$ belongs to a family of continuous CDFs, $\{F(\cdot \nvert \boldsymbol{\theta}) : \boldsymbol{\theta} \in \Theta\}$, and the parameter space $\Theta$ is an open subset of $\mathbb{R}^p$, ensuring that the true parameter $\boldsymbol{\theta}_0$ is always an interior point. For a generic parameter vector $\boldsymbol{\theta} = [\theta_1, \ldots, \theta_p]^{\top} \in \Theta$, we denote the associated probability density function (PDF) by
$f(\cdot \nvert \boldsymbol{\theta})$.

Throughout this paper, all components of the $p$-dimensional nuisance parameter $\boldsymbol{\theta}_0$ are treated as unknown unless stated otherwise; see Remark~\ref{rem:theta.known.components} for the straightforward modifications required when any subset of parameters is assumed known.
Let
\[
\hat{\boldsymbol{\theta}}_n = \hat{\boldsymbol{\theta}}_n(X_1,\ldots,X_n)
\]
denote a consistent $\mathbb{R}^p$-valued estimator of $\boldsymbol{\theta}_0$ under $\mathcal{H}_0$, satisfying regularity conditions specified later in this section.

All tests introduced in this paper are primarily based on the maximum likelihood (ML) estimator, as the regularity conditions required for ML estimation are generally mild
and widely satisfied in practice. To illustrate the flexibility of the proposed approach, we also present the general framework based on an arbitrary estimator, and we provide explicit implementations based on ML and on the method-of-moments (MM) estimator for the EPD distribution.

Consider the following non-degenerate $U$-statistic of degree $1$ based on sample trigonometric moments:
\begin{equation}\label{eq:Cn.Sn}
\begin{bmatrix}
C_n(\bb{\theta}) \\[1mm]
S_n(\bb{\theta})
\end{bmatrix}=
\begin{bmatrix}
\frac{1}{n} \sum_{i=1}^n \cos\left\{2\pi F(X_i \nvert \bb{\theta})\right\} \\[1mm]
\frac{1}{n} \sum_{i=1}^n \sin\left\{2\pi F(X_i \nvert \bb{\theta})\right\}
\end{bmatrix}
= \frac{1}{n} \sum_{i=1}^n \bb{\tau}(X_i, \bb{\theta}),
\end{equation}
with the corresponding two-dimensional kernel given by
\begin{equation}\label{eq:tau}
 \bb{\tau}(x, \bb{\theta})
=
\begin{bmatrix}
\cos\left\{2\pi F(x \nvert \bb{\theta})\right\} \\[1mm]
\sin\left\{2\pi F(x \nvert \bb{\theta})\right\}
\end{bmatrix}.
\end{equation}

In \citet{MR653521}, the asymptotic normality of non-degenerate $U$-statistics with univariate kernels, univariate observations, and estimated multivariate nuisance
parameters was established under mild regularity conditions on the kernel, the null distribution, and the consistency of the nuisance parameter estimators.
More recently, Randles' results were extended by \citet{MR4906044} to cover the case of multivariate kernels and multivariate observations, as well as local alternatives,
under similar regularity conditions.

In the present setting, the kernel $\boldsymbol{\tau}$ defined in \eqref{eq:tau} is two-dimensional. The derivation of the results presented below builds directly on the framework developed in \citet{MR4906044}. Our first objective is to determine the asymptotic distribution of $\sqrt{n} \,\smash{[C_n(\hat{\bb{\theta}}_n),S_n(\hat{\bb{\theta}}_n)]^{\top}}$ under $\mathcal{H}_0$.

Consider the score vector $\boldsymbol{s}(x,\boldsymbol{\theta})$, the corresponding Fisher information matrix $I(\boldsymbol{\theta})$, and the score--kernel cross-moment matrix $G(\boldsymbol{\theta})$, defined by
\begin{equation}\label{eq:s.I.G}
\boldsymbol{s}(x,\boldsymbol{\theta})
= \partial_{\boldsymbol{\theta}} \ln f(x \nvert \boldsymbol{\theta}), \qquad
I(\boldsymbol{\theta})
= \EE\left\{\boldsymbol{s}(X,\boldsymbol{\theta})
\boldsymbol{s}(X,\boldsymbol{\theta})^{\top}\right\}, \qquad
G(\boldsymbol{\theta})
= \EE\left\{\boldsymbol{\tau}(X,\boldsymbol{\theta})
\boldsymbol{s}(X,\boldsymbol{\theta})^{\top}\right\}.
\end{equation}
Applying Theorem~1 of \citet{MR4906044}, we obtain the following result.

\begin{proposition}\label{prop:asymp.dist.ML}
If $\hat{\boldsymbol{\theta}}_n$ denotes the maximum likelihood estimator and the usual regularity conditions hold, then under $\mathcal{H}_0$ and as $n \to \infty$,
\begin{equation*}\label{eq:asymp.Cn.Sn.H0}
\sqrt{n}\, \Sigma(\hat{\boldsymbol{\theta}}_n)^{-1/2}
\begin{bmatrix}
C_n(\hat{\boldsymbol{\theta}}_n)\\[1mm]
S_n(\hat{\boldsymbol{\theta}}_n)
\end{bmatrix}
\;\rightsquigarrow\;
\mathcal{N}_2(\boldsymbol{0}_2, I_2),
\end{equation*}
where $\bigl[C_n(\boldsymbol{\theta}), S_n(\boldsymbol{\theta})\bigr]^{\top}$ is defined in \eqref{eq:Cn.Sn}, $I_2$ denotes the $2\times2$ identity matrix, and
\begin{equation}\label{eq:Sigma.ML}
\Sigma(\boldsymbol{\theta})
= \tfrac{1}{2} I_2 - G(\boldsymbol{\theta})\, I(\boldsymbol{\theta})^{-1}\, G(\boldsymbol{\theta})^{\top}.
\end{equation}
\end{proposition}

\begin{remark}
For the general case based on an arbitrary estimator $\hat{\boldsymbol{\theta}}_n$ not restricted to maximum likelihood estimation, Proposition~\ref{prop:asymp.dist.ML} continues to hold, with the asymptotic covariance matrix given by
\begin{equation}\label{eq:Sigma.general}
\Sigma(\boldsymbol{\theta})
= \tfrac{1}{2} I_{2}
- G(\boldsymbol{\theta}) R(\boldsymbol{\theta})^{-1} J(\boldsymbol{\theta})^{\top}
- J(\boldsymbol{\theta}) R(\boldsymbol{\theta})^{-1} G(\boldsymbol{\theta})^{\top}
+ G(\boldsymbol{\theta}) R(\boldsymbol{\theta})^{-1} G(\boldsymbol{\theta})^{\top},
\end{equation}
with
\vspace{-6mm}
\[
J(\bb{\theta}) = \EE\big\{\bb{\tau}(X, \bb{\theta}) \bb{r}(X, \bb{\theta})^{\top}\big\},
\]
provided that the estimator $\hat{\boldsymbol{\theta}}_n$ admits an asymptotic  representation based on a suitable function $\boldsymbol{r}(X,\boldsymbol{\theta})$. More precisely, there must exist an $\mathbb{R}^p$-valued measurable mapping $x \mapsto \boldsymbol{r}(x,\boldsymbol{\theta})$ such that, for
$X \sim F(\cdot \nvert \boldsymbol{\theta}_0)$, we have $\EE\left\{\boldsymbol{r}(X,\boldsymbol{\theta}_0)\right\} = \boldsymbol{0}_p$,
the covariance matrix
\vspace{-2mm}
\begin{equation}\label{eq:R}
R(\boldsymbol{\theta}_0)
= \EE\big\{\boldsymbol{r}(X,\boldsymbol{\theta}_0)
\boldsymbol{r}(X,\boldsymbol{\theta}_0)^{\top}\big\}
\end{equation}
is positive definite, and the following convergence in probability holds under
$\mathcal{H}_0$ as $n \to \infty$:
\begin{equation}\label{eq:hyp.est}
\sqrt{n} \, (\hat{\bb{\theta}}_n - \bb{\theta}_0) - R(\bb{\theta}_0)^{-1} \frac{1}{\sqrt{n}} \sum_{i=1}^n \bb{r}(X_i, \bb{\theta}_0) \stackrel{\PP}{\longrightarrow} \bb{0}_p.
\end{equation}
In particular, when ${\hat{\bb{\theta}}_n}$ is the maximum likelihood estimator,  condition~\eqref{eq:hyp.est} is satisfied under mild regularity conditions by the score function, by taking $\bb{r}(x,\bb{\theta})=\bb{s}(x, \bb{\theta})$ \citep[see, e.g.,][Chapter~5]{MR1652247}. In this case, $R(\bb{\theta})=I(\bb{\theta})$, $J(\bb{\theta}) = G(\bb{\theta})$, and the asymptotic covariance matrix $\Sigma(\bb{\theta})$ defined in \eqref{eq:Sigma.general} reduces to the form given in \eqref{eq:Sigma.ML}.
\end{remark}

Based on the results of Proposition~\ref{prop:asymp.dist.ML}, we now introduce the new
test statistic $T_n$, together with its asymptotic null distribution.

\begin{proposition}\label{prop:test.Tn}
The new test statistic, denoted by $T_n(\hat{\boldsymbol{\theta}}_n)$, is defined as
the quadratic form
\begin{equation*}\label{eq:test.Tn}
T_n(\hat{\bb{\theta}}_n) = n [C_n(\hat{\bb{\theta}}_n), S_n(\hat{\bb{\theta}}_n)] \, \Sigma(\hat{\bb{\theta}}_n)^{-1} \, [C_n(\hat{\bb{\theta}}_n), S_n(\hat{\bb{\theta}}_n)]^{\top},
\end{equation*}
where $\bigl[C_n(\boldsymbol{\theta}), S_n(\boldsymbol{\theta})\bigr]^{\top}$ is
defined in \eqref{eq:Cn.Sn} and $\Sigma(\boldsymbol{\theta})$ is defined in
\eqref{eq:Sigma.ML} for the ML-based case and in \eqref{eq:Sigma.general} for the general case.
Furthermore, under $\mathcal{H}_0$ and as $n \to \infty$,
\vspace{-2mm}
\begin{equation*}\label{eq:asymp.H0.Tn}
T_n(\hat{\boldsymbol{\theta}}_n) \;\rightsquigarrow\; \chi_2^2.
\end{equation*}
\end{proposition}

We now introduce an alternative expression for the normalizing scalar $V(\boldsymbol{\theta})$ appearing in the LK test statistic.

\begin{proposition}\label{prop:test.LK}
Consider the test statistic originally introduced by \citet{MR1146353}, referred to
in this paper as the LK test, and defined by
\vspace{-2mm}
\begin{equation*}\label{eq:stat.LK}
K_1(\hat{\boldsymbol{\theta}}_n)
=
V(\hat{\boldsymbol{\theta}}_n)^{-1}\,
2n \bigl\{ C_n^2(\hat{\boldsymbol{\theta}}_n) + S_n^2(\hat{\boldsymbol{\theta}}_n) \bigr\}.
\end{equation*}
Then an alternative way to compute the normalizing scalar $V(\boldsymbol{\theta})$ is given by
\begin{equation}\label{eq:V.LK}
V(\boldsymbol{\theta}) = \mathrm{tr}\{\Sigma(\boldsymbol{\theta})\},
\end{equation}
where $\Sigma(\boldsymbol{\theta})$ is defined in \eqref{eq:Sigma.ML} for the ML-based case and in \eqref{eq:Sigma.general} for the general case.
\end{proposition}

\begin{remark}\label{rem:approx.V}
The alternative way to compute the normalizing scalar $V(\boldsymbol{\theta})$ in~\eqref{eq:V.LK} stems from the fact that \citet{MR1146353} defined $V(\boldsymbol{\theta})$ as the asymptotic expectation of
$\smash{n\{C_n^2(\hat{\bb{\theta}}_n)+S_n^2(\hat{\bb{\theta}}_n)\}}$. Under $\mathcal{H}_0$ with $X_i\sim F(\cdot\nvert\bb{\theta}_0)$, write
$Z_n=\sqrt{n}\,[C_n(\hat{\bb{\theta}}_n),\,S_n(\hat{\bb{\theta}}_n)]^{\top}$.
By Proposition~\ref{prop:asymp.dist.ML} and Slutsky's theorem, $Z_n \rightsquigarrow \mathcal{N}_2(\mathbf{0},\Sigma(\bb{\theta}_0))$. Under a mild moment condition ensuring uniform integrability of $\|Z_n\|^2$, it follows that
\[
\EE\big[n\{C_n^2(\hat{\bb{\theta}}_n)+S_n^2(\hat{\bb{\theta}}_n)\}\big]
=\EE(\|Z_n\|^2)
=\mathrm{tr}\{\Var(Z_n)\}+\|\EE(Z_n)\|^2
\longrightarrow
\mathrm{tr}\bigl\{\Sigma(\bb{\theta}_0)\bigr\}.
\]

Moreover, Proposition~\ref{prop:test.Tn} clearly shows that the asymptotic distribution of $K_1(\hat{\boldsymbol{\theta}}_n)$ under $\mathcal{H}_0$ is generally not $\chi_2^2$, as claimed by \citet{MR1146353}. More specifically, letting $\smash{\hat{\sigma}_{ij} \equiv [\Sigma(\hat{\boldsymbol{\theta}}_n)]_{ij} \stackrel{\PP}{\longrightarrow} \sigma_{ij}}$ for $i,j\in \{1,2\}$, we see that $K_1(\hat{\boldsymbol{\theta}}_n)$ weakly converges to a weighted mixture of asymptotically dependent $\chi_1^2$ random variables ($\sigma_{11} \neq \sigma_{22}$ and $\sigma_{12}\neq 0$ in general):
\[
K_1(\hat{\boldsymbol{\theta}}_n) = \frac{2 \hat{\sigma}_{11}}{\hat{\sigma}_{11} + \hat{\sigma}_{22}} \left(\frac{n \, C_n^2(\hat{\boldsymbol{\theta}}_n)}{\hat{\sigma}_{11}}\right) + \frac{2 \hat{\sigma}_{22}}{\hat{\sigma}_{11} + \hat{\sigma}_{22}} \left(\frac{n \, S_n^2(\hat{\boldsymbol{\theta}}_n)}{\hat{\sigma}_{22}}\right) \;\rightsquigarrow\; \frac{2 \sigma_{11}}{\sigma_{11} + \sigma_{22}} \chi_1^2 + \frac{2 \sigma_{22}}{\sigma_{11} + \sigma_{22}} \chi_1^2 \stackrel{\mathrm{law}}{\neq} \chi_2^2.\vspace{-1mm}
\]
Nevertheless, the simulations in Section~\ref{sec:empirical.size.analysis} indicate that the distribution of $K_1(\hat{\boldsymbol{\theta}}_n)$ under $\mathcal{H}_0$ is well approximated by a $\chi_2^2$.

Note that the values denoted by $V^{-1}$ and reported in Table~1 of \citet{MR1146353} for a few specific tests are all correctly reproduced by $V(\boldsymbol{\theta})^{-1}$ as given in \eqref{eq:V.LK}, except for the MM-based Laplace test. In that case, the reported value $V^{-1} = 1.13$ is incorrect (as can be confirmed by simulations). Instead, the correct value is $V(\hat{\boldsymbol{\theta}}_n)^{-1} = 0.92751735$. Note also that, although the values of $V^{-1}$ reported by \citet{MR1146353} are constants, in general both $V(\hat{\boldsymbol{\theta}}_n)$ and the matrix $\Sigma(\hat{\boldsymbol{\theta}}_n)$ depend on some or all of the estimated parameters.
\end{remark}

\begin{remark}
The critical values and $p$-values for the $T_n$ and LK tests are derived from the limiting result. The null hypothesis is rejected for large values of the test statistic, specifically if the observed value of $T_n(\hat{\bb{\theta}}_n)$ or $K_1(\hat{\bb{\theta}}_n)$ exceeds the upper $\alpha$-quantile of the $\chi_2^2$ distribution. The $p$-value is calculated as the probability that a $\chi_2^2$ random variable is greater than the observed value of the statistic.
\end{remark}

\begin{remark}\label{rem:v}
The representation $V(\boldsymbol{\theta}) = \mathrm{tr}(\Sigma(\boldsymbol{\theta}))$ provides a useful geometric interpretation of the LK test.
While the statistic $T_n$ corresponds to the squared Mahalanobis norm of $\sqrt{n}[C_n(\hat{\boldsymbol{\theta}}_n), S_n(\hat{\boldsymbol{\theta}}_n)]^{\top}$,
and therefore fully exploits the covariance structure $\Sigma(\boldsymbol{\theta})$, the LK test is based on the squared Euclidean norm
$n\{C_n^2(\hat{\boldsymbol{\theta}}_n) + S_n^2(\hat{\boldsymbol{\theta}}_n)\}$, normalized by half the total variance, $\mathrm{tr}\{\Sigma(\hat{\boldsymbol{\theta}}_n)\}/2$.
As a result, the LK test can be interpreted as an isotropic scalar projection of the vector statistic, treating all directions equally.
This interpretation explains both the simplicity of the LK test and the typical gain in power achieved by $T_n$ (see Section~\ref{sec:empirical.simulations}) through the use of the full covariance structure.
\end{remark}

\begin{remark}\label{rem:theta.known.components}
To obtain the proper normalization when some of the parameters are known under $\mathcal{H}_0$, set the relevant components of the estimator $\smash{\hat{\bb{\theta}}_n}$ to their known values in $\bb{\theta}_0$. Then, to calculate $\Sigma(\bb{\theta})$ in \eqref{eq:Sigma.ML} or \eqref{eq:Sigma.general}, remove the columns of $G(\bb{\theta})$ and $J(\bb{\theta})$ corresponding to the positions of the known values in $\bb{\theta}_0$, and remove the corresponding rows and columns of $I(\bb{\theta})$ and $R(\bb{\theta})$. For example, if the second component of $\bb{\theta}_0$ is known under $\mathcal{H}_0$, then remove the second columns of $G(\bb{\theta})$ and $J(\bb{\theta})$, and remove the second rows and columns of $I(\bb{\theta})$ and $R(\bb{\theta})$. For this reason, the matrices $G(\bb{\theta})$ and $I(\bb{\theta})$ are reported in Table~\ref{table:G.I} for specific ML-based tests ($G(\bb{\theta})$, $J(\bb{\theta})$, and $R(\bb{\theta})$ for the MM-based test) instead of the explicit form of the covariance matrix $\Sigma(\bb{\theta})$, as each configuration of known and unknown parameters leads to a distinct $\Sigma(\bb{\theta})$ and $V(\bb{\theta})$. In particular, when all the components of $\bb{\theta}_0$ are known under $\mathcal{H}_0$, we have $\smash{\hat{\bb{\theta}}_n} = \bb{\theta}_0$ and $\Sigma(\hat{\bb{\theta}}_n) = \frac{1}{2} I_2$, in which case the test statistic reduces to $T_n(\hat{\bb{\theta}}_n) = n [C_n(\bb{\theta}_0), S_n(\bb{\theta}_0)] \, (\frac{1}{2} I_2)^{-1} \, [C_n(\bb{\theta}_0), S_n(\bb{\theta}_0)]^{\top} = 2n \{C_n^2(\bb{\theta}_0) + S_n^2(\bb{\theta}_0)\} = K_1(\bb{\theta}_0)$. Here, the scaling $\frac{1}{2} I_2$ is an easy consequence of the fact that $F(X_1\nvert\bb{\theta}_0),\ldots,F(X_n\nvert\bb{\theta}_0)$ are i.i.d.\ standard uniforms under $\mathcal{H}_0$ and $\{u\mapsto \sqrt{2} \cos(2\pi u),u\mapsto \sqrt{2} \sin(2\pi u)\}$ forms an orthonormal set of functions in $L^2([0,1])$.
\end{remark}

\begin{figure}[H]
\centering
\includegraphics[width=0.37\linewidth]{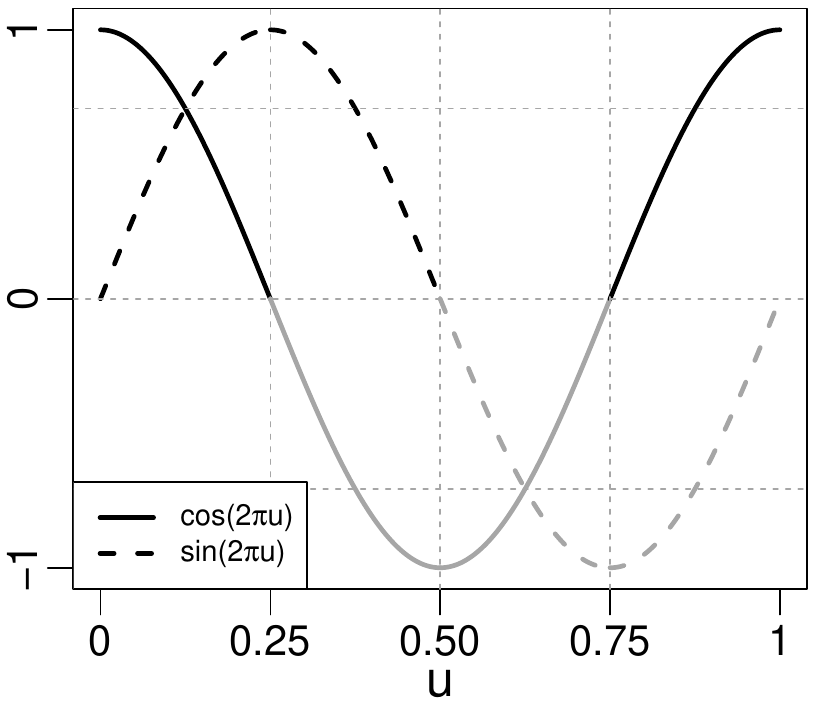}
\caption{The functions $u\mapsto \cos(2\pi u)$ (solid) and $u\mapsto \sin(2\pi u)$ (dashed) plotted on the interval $[0,1]$. The line color indicates the function's sign: black for positive values and gray for negative values.}
\label{fig:cos.sin}
\end{figure}

\begin{remark}\label{rem:interp}
The two components that form the test statistics $K_1(\hat{\boldsymbol{\theta}}_n)$ and $T_n(\hat{\bb{\theta}}_n)$, namely $\smash{\sqrt{n} C_n(\hat{\bb{\theta}}_n)}$ and $\smash{\sqrt{n} S_n(\hat{\bb{\theta}}_n)}$, which are respectively scaled sample means of $\smash{\cos(2\pi \hat{U}_i)}$ and $\smash{\sin(2\pi \hat{U}_i)}$ with $\smash{\hat{U}_i = F(X_i \nvert \hat{\bb{\theta}}_n)}$, can be interpreted in the asymptotic regime ($n \to \infty$) using Figure~\ref{fig:cos.sin}. Under the null hypothesis, $X_i \sim F(\cdot \nvert \bb{\theta}_0)$, and assuming the consistency of $\smash{\hat{\bb{\theta}}_n}$ and the continuity of $(x,\bb{\theta})\mapsto F(x \nvert \bb{\theta})$, each transformed observation $\smash{\hat{U}_i}$ weakly converges to a standard uniform random variable by the continuous mapping theorem \citep[p.~7]{MR1652247}. In particular, both $u\mapsto \cos(2\pi u)$ and $u\mapsto \sin(2\pi u)$ integrate to zero over $(0,1)$, so $\smash{\sqrt{n} C_n(\hat{\bb{\theta}}_n)}$ and $\smash{\sqrt{n} S_n(\hat{\bb{\theta}}_n)}$ converge to zero in expectation under $\mathcal{H}_0$, meaning that no systematic deviation from the model is detected.

Now consider the case where $X_i \centernot{\sim} F(\cdot \nvert \bb{\theta}_0)$. The component $\smash{\sqrt{n} S_n(\hat{\bb{\theta}}_n)}$ can be interpreted as a directional measure of skewness relative to the fitted model. Specifically, if more than 50\% of the observations fall below the fitted median $\smash{F^{-1}(0.5 \nvert \hat{\bb{\theta}}_n)}$, then most $\smash{\hat{U}_i}$ values will lie below $0.5$. Since $u\mapsto \sin(2\pi u)$ is positive on $(0, 0.5)$ and negative on $(0.5, 1)$, this typically leads to a positive value of $\smash{\sqrt{n} S_n(\hat{\bb{\theta}}_n)}$. In such cases, the dataset is more right-skewed (or less left-skewed) than the null model. Conversely, a negative value suggests that the dataset is more left-skewed (or less right-skewed).

The component $\smash{\sqrt{n} C_n(\hat{\bb{\theta}}_n)}$ can be interpreted as a measure of relative tail weight and central concentration. If fewer than 50\% of the observations fall in the central interval ($\smash{F^{-1}(0.25 \nvert \hat{\bb{\theta}}_n)},\smash{F^{-1}(0.75 \nvert \hat{\bb{\theta}}_n)})$ --- that is, between the first and third quartiles under the fitted model --- then most $\smash{\hat{U}_i}$ values will lie in the tail intervals $(0, 0.25)$ and $(0.75, 1)$. Since $u\mapsto \cos(2\pi u)$ is negative over the central region and positive in the tails, this typically leads to a positive value of $\smash{\sqrt{n} C_n(\hat{\bb{\theta}}_n)}$. In such cases, the dataset often has lighter tails than the null model, that is, more mass in the tails and less in the center. Conversely, a negative value of $\smash{\sqrt{n} C_n(\hat{\bb{\theta}}_n)}$ may indicate heavier tails, that is, more mass in the center and less in the tails. Indeed, it is generally understood that heavy-tailed distributions allocate more mass in the extreme tails (e.g., $u < 0.02$ and $u > 0.98$) while assigning comparatively less mass to the intermediate tail regions. As a result, the overall proportion of transformed data $\hat{U}_i$ in the broad tail intervals $(0, 0.25)$ and $(0.75, 1)$ may be smaller than under the null, leading to a negative value of $\smash{\sqrt{n} C_n(\hat{\bb{\theta}}_n)}$.
\end{remark}

\section{Trigonometric tests for specific distributions}\label{sec:GOF.specific}

In this section, detailed information on the trigonometric tests $T_n$ and LK is provided for 11 specific families of distributions, namely the exponential power
distribution (EPD; also known as the generalized normal or Subbotin distribution), half-EPD, skew normal (SN), generalized gamma (GG), logistic, Student's $t$,
Gompertz, Lomax (also known as the Pareto type II), inverse-Gaussian (IG), beta, and Kumaraswamy (Kum) distributions.

The PDF, CDF, support, and parameter domain are listed for each of these 11 families in Table~\ref{table:density.CDF}, to clearly identify the specific version of each distribution considered and because the CDF is used directly in the computation of the $T_n$ and LK test statistics.

\begin{table}[!htbp]
\small\def\arraystretch{1.2}% 1 is the default, change whatever you need
\begin{center}
\begin{tabular}{|p{1.1cm}|p{12.0cm}|}
\hline
$\vcenter{\hbox{\shortstack{EPD \\ $(\lambda,\mu,\sigma)$}}}$
 & $f(x) \!=\! \frac{1}{2\sigma \lambda^{1/\lambda-1} \Gamma(1/\lambda)} e^{-\frac{1}{\lambda}(|x-\mu|/\sigma)^\lambda}$,  $F(x) \!=\! \frac{1}{2}\! \left[1 + \mathrm{sign}(x-\mu) \, \Gamma_{1/\lambda, 1}\hspace{-.3mm}\big\{\frac{1}{\lambda}(|x-\mu|/\sigma)^\lambda\big\}\right]$,  $x\in\R$,  $\lambda>0$, $\mu\in\R$, $\sigma>0$\! \rule{0pt}{11pt} \vspace{1mm}\cr
\hline
$\vcenter{\hbox{\shortstack{Half-EPD \\ $(\lambda,\sigma)$}}}$  & $f(x) = \frac{1}{\sigma \lambda^{1/\lambda-1} \Gamma(1/\lambda)} e^{-\frac{1}{\lambda} (x/\sigma)^{\lambda}}$, ~ $F(x) = \Gamma_{1/\lambda, 1}\hspace{-.3mm}\big\{\frac{1}{\lambda} (x/\sigma)^{\lambda}\big\}$, ~ $x>0$,  $\lambda>0$, $\sigma>0$ \rule{0pt}{12pt} \vspace{1mm}\cr
\hline
$\mathrm{SN}(\lambda,\mu,\sigma)$ & $f(x)=\tfrac{2}{\sigma}\phi(y)\Phi(\lambda y)$, ~ $F(x) = \Phi(y)-2 T(y, \lambda)$, ~ $y=\frac{x-\mu}{\sigma}\in\R$,  $\lambda\in\R$, $\mu\in\R$, $\sigma>0$ \rule{0pt}{9pt} \vspace{0.5mm}\cr
\hline
$\mathrm{GG}(\lambda,\beta, \rho)$ & $f(x) = \frac{\rho}{\beta \Gamma(\lambda)} \big(\tfrac{x}{\beta}\big)^{\lambda\rho-1} e^{-(x/\beta)^{\rho}}$, ~ $F(x) = \Gamma_{\lambda, 1}\hspace{-.3mm} \left\{(x/\beta)^{\rho}\right\}$, ~ $x>0$, $\lambda>0$, $\beta>0$, $\rho>0$ \rule{0pt}{9pt} \vspace{0.5mm}\cr
\hline
$\vcenter{\hbox{\shortstack{Logistic \\ $(\mu,\sigma)$}}}$ & $f(x) =\frac{e^{-(x - \mu)/\sigma}}{\sigma \{1 + e^{-(x - \mu)/\sigma}\}^2}$, ~ $F(x) = \frac{1}{1 + e^{-(x - \mu)/\sigma}}$, ~ $x\in\R$, $\mu\in\R$, $\sigma>0$ \rule{0pt}{11pt} \vspace{1mm}\cr
\hline
$\vcenter{\hbox{\shortstack{Student's \\ $t(\lambda,\mu,\sigma)$}}}$ & $f(x) = \frac{\Gamma(\frac{\lambda+1}{2})}{\sigma\sqrt{\lambda\pi}\Gamma(\frac{\lambda}{2})} \big(1 + \frac{y^2}{\lambda}\big)^{-\frac{\lambda + 1}{2}}$, ~ $F(x) = \frac{1}{2} \Big[1 + \mathrm{sign}(y) I\Big(\frac{y^2}{y^2+\lambda};1/2,\lambda/2\Big)\Big]$, ~ $y=\frac{x-\mu}{\sigma}\in\R$,  $\lambda>0$, $\mu\in\R$, $\sigma>0$ \rule{0pt}{11pt} \vspace{0.5mm}\cr
\hline
$\vcenter{\hbox{\shortstack{Gompertz \\ $(\beta, \rho)$}}}$ &  $f(x) = \beta\rho\exp(\rho + \beta x - \rho e^{\beta x})$, ~ $F(x) = 1 - \exp\left\{-\rho( e^{\beta x} - 1)\right\}$, ~ $x>0$, $\beta>0$, $\rho>0$\rule{0pt}{11pt} \vspace{1mm}\cr
\hline
$\vcenter{\hbox{\shortstack{Lomax \\ $(\alpha, \sigma)$}}}$ & $f(x) = (\alpha/\sigma) (1 + x/\sigma)^{-(\alpha + 1)}$, ~ $F(x) = 1 - (1 + x/\sigma)^{-\alpha}$, ~ $x>0$, $\alpha>0$, $\sigma>0$ \rule{0pt}{11pt} \vspace{1mm}\cr
\hline
$\mathrm{IG}(\mu, \lambda)$ & $f(x) = \sqrt{\frac{\lambda}{2\pi x^3}} \exp\left\{-\frac{\lambda (x - \mu)^2}{2 \mu^2 x}\right\}$, ~
 $F(x) = \Phi\left\{\sqrt{\frac{\lambda}{x}} \left(\frac{x}{\mu} - 1\right)\right\} + \exp\left(\frac{2\lambda}{\mu}\right)\Phi\left\{-\sqrt{\frac{\lambda}{x}} \left(\frac{x}{\mu} + 1\right)\right\}$, ~ $x>0$,  $\mu>0$, $\lambda>0$ \rule{0pt}{12pt} \vspace{0.5mm}\cr
\hline
$\mathrm{Beta}(\alpha, \beta)$ &  $f(x) = \frac{1}{B(\alpha, \beta)} x^{\alpha - 1} (1 - x)^{\beta - 1}$, ~ $F(x ) = I(x; \alpha, \beta)$, ~ $x\in (0, 1)$, $\alpha>0$, $\beta>0$  \rule{0pt}{11pt} \vspace{0.5mm}\cr
\hline
$\mathrm{Kum}(\alpha, \beta)$ & $f(x) = \alpha\beta x^{\alpha - 1}(1 - x^{\alpha})^{\beta - 1}$, ~  $F(x) = 1 - (1 - x^{\alpha})^{\beta}$,  ~ $x\in (0, 1)$, $\alpha>0$, $\beta>0$\rule{0pt}{9pt} \vspace{0.5mm}\cr
\hline

\end{tabular}
\end{center}
\vspace{-2mm}
\caption{Density, CDF, support, and parameter domain for each family of distributions}
\label{table:density.CDF}
\end{table}

\vspace{-2mm}
\noindent
{\bf \normalsize Notation.} The following standard notations are used throughout the paper: $\phi(\cdot)$ and $\Phi(\cdot)$ denote the PDF and CDF of the standard normal distribution; $\Gamma(\cdot)$, $\psi(\cdot)$, and $\psi_1(\cdot)$ denote the gamma, digamma, and trigamma functions; $B(a,b)$, $B(x;a,b)$, and $I(x;a,b)=B(x;a,b)/B(a,b)$ denote the beta, incomplete beta, and regularized incomplete beta functions; $\gamma(a,x)$ denotes the lower incomplete gamma function, and $\Gamma_{a,b}(x)=\gamma(a,x/b)/\Gamma(a)$ denotes the CDF of a gamma distribution with shape parameter $a$ and scale parameter $b$. Finally, Owen's $T$ function is denoted by
\begin{equation*}\label{eq:OwenT}
T(y,\lambda)=\frac{1}{2\pi}\int_0^{\lambda}\frac{e^{-y^2(1+t^2)/2}}{1+t^2}\,\rd t.
\end{equation*}

Many special cases can be derived from the 11 distribution families, in particular from the rich EPD and GG families. Log-variant, inverse-variant, limiting-case,
and other types of equivalences can also be obtained, which greatly extend the applicability of the tests $T_n$ and LK. We provide below a list that is intended to be as exhaustive as possible:
\vspace{-5mm}
\begin{itemize}

\item \textbf{Laplace}$(\mu,\sigma)
= \mathrm{EPD}(\lambda = 1, \mu, \sigma)$.

\item \textbf{Normal}$(\mu,\sigma)
= \mathrm{EPD}(\lambda = 2, \mu, \sigma)
= \mathrm{SN}(\lambda = 0, \mu,\sigma)$.

\item \textbf{Weibull}$(\beta,\rho)
= \mathrm{GG}(\lambda = 1, \beta, \rho)$.

\item \textbf{Gamma}$(\lambda,\beta)
= \mathrm{GG}(\lambda, \beta, \rho = 1)$.

\item \textbf{Nakagami}$(\lambda,\omega)
= \mathrm{GG}\big(\lambda,\; \beta = \sqrt{\omega/\lambda},\; \rho = 2\big)$.

\item \textbf{Exponential}$(\beta)
= \mathrm{GG}(\lambda = 1, \beta, \rho = 1)
= \mathrm{half\mbox{-}EPD}(\lambda = 1, \sigma = \beta)$.

\item \textbf{Half-normal}$(\delta)
= \mathrm{GG}(\lambda = 1/2,\; \beta = \sqrt{2}\,\delta,\; \rho = 2)
= \mathrm{half\mbox{-}EPD}(\lambda = 2, \sigma = \delta)$.

\item \textbf{Rayleigh}$(\delta)
= \mathrm{GG}(\lambda = 1,\; \beta = \sqrt{2}\,\delta,\; \rho = 2)$.

\item \textbf{Maxwell--Boltzmann}$(\delta)
= \mathrm{GG}(\lambda = 3/2,\; \beta = \sqrt{2}\,\delta,\; \rho = 2)$.

\item \textbf{Chi-squared} $\chi^2(k)
= \mathrm{GG}(\lambda = k/2,\; \beta = 2,\; \rho = 1)$.

\item \textbf{Cauchy}$(\mu,\sigma)= t(\lambda=1,\mu,\sigma)$.

\item All \textbf{log-variant} families are handled through the transformation $X=\ln(Y)$, e.g., $Y \sim \mathbf{log\text{-}normal}(\mu,\sigma)
\iff X=\ln(Y)\sim\mathrm{normal}(\mu,\sigma)$, and similarly for the $\mathbf{log\text{-}EPD}$, $\mathbf{log\text{-}Laplace}$, $\mathbf{log\text{-}logistic}$, and other log-variant families.

\item All \textbf{inverse-variant} families are treated via the transformation $X = 1/Y$, e.g.,
$Y \sim \mathbf{inverse\text{-}gamma}(\lambda,\tau) \iff X=1/Y \sim \mathrm{gamma}(\lambda, \beta = 1/\tau)$, or $Y \sim \mathrm{\textbf{Fr\'echet}}(\tau,\rho) \equiv \mathbf{inverse\text{-}Weibull}(\tau,\rho) \iff X=1/Y \sim \mathrm{Weibull}(\beta = 1/\tau, \rho)$.

\item All \textbf{exponential-variant} families are handled through the transformation $X=e^Y$, e.g., $Y \sim \mathbf{exp\text{-}GG}(\lambda,\mu,\sigma)
\iff X=e^Y\sim\mathrm{GG}(\lambda,\beta=e^{\mu},\rho=1/\sigma)$, and similarly for the $\mathbf{exp\text{-}Weibull}$, $\mathbf{exp\text{-}gamma}$, and other exp-variant families.

\item The \textbf{Gumbel} distribution is handled via $Y \sim \mathrm{\textbf{Gumbel}}(\mu,\sigma) \iff X=e^{-Y} \sim \mathrm{Weibull}(\beta=e^{-\mu},\, \rho = 1/\sigma)$.

\item The \textbf{Pareto type I} distribution on a known support $(\sigma,\infty)$ is handled via $Y \sim \mathrm{\textbf{Pareto}}(\alpha,\sigma) \iff X=\ln(Y/\sigma)\sim \mathrm{exponential}(\beta = 1/\alpha)$. If the lower endpoint $\sigma$ is unknown, we use the plug-in $\hat{\sigma}_n=\min\{Y_1,\ldots,Y_n\}$ and set $X=\ln(Y/\hat{\sigma}_n)$.

\item The \textbf{Beta-prime} distribution is handled via $Y \sim \mathrm{\textbf{Beta-prime}}(\alpha,\beta) \iff X = Y/(Y+1) \sim \mathrm{Beta}(\alpha,\beta)$.

\item For any \textbf{completely specified} continuous distribution, $\Sigma = \tfrac12 I_2$ and $V=1$, resulting in identical test statistics for $T_n$ and the LK test.

\item The $\mathrm{\textbf{continuous uniform}}(a,b)$ distribution arises as the limiting case of the $\mathrm{EPD}(\lambda,\mu,\sigma)$ family when $\lambda \to \infty$, with $\mu=(a+b)/2$ and $\sigma=(b-a)/2$. In this limit, the ML-based test yields  $\hat a_n = x_{(1)}$ and $\hat b_n = x_{(n)}$, and the covariance matrix reduces to $\Sigma = \tfrac12 I_2$, whether $a$, $b$, or both parameters are unknown.

\end{itemize}

\begin{table}[!t]
\small\def\arraystretch{1.2}% 1 is the default, change whatever you need
\begin{center}
\begin{tabular}{|p{1.1cm}|p{12.1cm}|}
\hline
$\vcenter{\hbox{\shortstack{EPD \\ $(\lambda,\mu,\sigma)$}}}$
& \hspace{-1.5mm} $G(\lambda,0,1)=
\begin{bsmallmatrix}
\frac{h_{1}(\lambda)-h_{3}(\lambda)}{\lambda^2}& 0 & h_{1}(\lambda) \\[1mm]
0 & \frac{h_{2}(\lambda)}{\lambda^{1/\lambda-1} \Gamma(1/\lambda)} & 0
\end{bsmallmatrix}$,  $I(\lambda,0, 1)=
\begin{bsmallmatrix}
\frac{1}{\lambda^3}\big[(\frac{1}{\lambda}+1)\psi_1(\frac{1}{\lambda} + 1)+\big\{\psi(\frac{1}{\lambda} + 1)+\ln(\lambda)\big\}^2 -1\big] & 0 & \frac{-\psi(1/\lambda + 1)-\ln(\lambda)}{\lambda} \\[1mm]
0 & \frac{\lambda^{2 - 2/\lambda}\Gamma(2-1/\lambda)}{\Gamma(1/\lambda)} & 0 \\[1mm]
\frac{-\psi(1/\lambda + 1)-\ln(\lambda)}{\lambda} & 0 & \lambda
\end{bsmallmatrix}$  \rule{0pt}{22pt}\vspace{.5mm} \cr
\hline
$\vcenter{\hbox{\shortstack{$\mathrm{EPD}_{\lambda}(\mu,\sigma)$ \\ \footnotesize MM-based}}}$  & $G_{\lambda}(0, 1)=
\begin{bsmallmatrix}
0 ~&~ h_{1}(\lambda) \\[1mm]
\frac{h_{2}(\lambda)}{\lambda^{1/\lambda-1} \Gamma(1/\lambda)} ~&~ 0
\end{bsmallmatrix}$, ~ $J_{\lambda}(0,1)=
\begin{bsmallmatrix}
0 & 2D_{\lambda}h_{4}(\lambda)\\[1mm]
\frac{h_{5}(\lambda)\Gamma(2/\lambda)}{\lambda^{1/\lambda}\Gamma(3/\lambda)} & 0
\end{bsmallmatrix}$, ~ $R_{\lambda}(0,1)=\begin{bsmallmatrix}
\frac{\Gamma(1/\lambda)}{\lambda^{2/\lambda}\Gamma(3/\lambda)} & 0 \\[1mm]
0 & 4 D_{\lambda}
\end{bsmallmatrix}$, $D_{\lambda}=\frac{\Gamma^2(\frac{3}{\lambda})}{\Gamma(\frac{1}{\lambda})\Gamma(\frac{5}{\lambda})-\Gamma^2(\frac{3}{\lambda})}$ \rule{0pt}{13pt}\vspace{.5mm}\cr
\hline
$\vcenter{\hbox{\shortstack{Half-EPD \\ $(\lambda,\sigma)$}}}$
& $G(\lambda,1)
=
\begin{bsmallmatrix}
\frac{h_{6}(1/\lambda,1/\lambda+1,1)-h_{15}(\lambda)}{\lambda^2} & h_{6}(1/\lambda, 1 / \lambda + 1, 1) \\[1mm]
\frac{h_{7}(1/\lambda,1/\lambda+1,1)-h_{16}(\lambda)}{\lambda^2} & h_{7}(1/\lambda, 1 / \lambda + 1, 1)
\end{bsmallmatrix}$, ~  $I(\lambda,1) =
\begin{bsmallmatrix}
\frac{1}{\lambda^3}\big[(\frac{1}{\lambda}+1)\psi_1(\frac{1}{\lambda} + 1)+\big\{\psi(\frac{1}{\lambda} + 1)+\ln(\lambda)\big\}^2 -1\big] ~ &  ~ \frac{-\psi(1/\lambda + 1)-\ln(\lambda)}{\lambda} \\[1mm]
\frac{-\psi(1/\lambda + 1)-\ln(\lambda)}{\lambda} ~ & ~ \lambda
\end{bsmallmatrix}$\rule{0pt}{15pt}\vspace{.5mm} \cr
\hline
$\mathrm{SN}(\lambda,\mu,\sigma)$ & $G(\lambda, 0, 1)=
\begin{bsmallmatrix}
2 h_{39}(\lambda)~&~ 2h_{41}(\lambda) ~&~ 2 h_{43}(\lambda)\\[1mm]
2 h_{40}(\lambda)~&~ 2 h_{42}(\lambda)~&~ 2 h_{44}(\lambda)
\end{bsmallmatrix}$, ~ $I(\lambda,0,1)=
\begin{bsmallmatrix}
2  h_{38}(\lambda) & \frac{2}{\sqrt{2\pi}(1+ \lambda^2)^{3/2}} - 2\lambda  h_{37}(\lambda) & - 2\lambda h_{38}(\lambda) \\[1mm]
\frac{2}{\sqrt{2\pi}(1+ \lambda^2)^{3/2}} - 2\lambda  h_{37}(\lambda)  & 1 + 2 \lambda^2 h_{36}(\lambda) &  \frac{2\lambda(1+2\lambda^2)}{\sqrt{2\pi}(1+ \lambda^2)^{3/2}} + 2\lambda^2  h_{37}(\lambda) \\[1mm]
- 2\lambda h_{38}(\lambda) & \frac{2\lambda(1+2\lambda^2)}{\sqrt{2\pi}(1+ \lambda^2)^{3/2}} + 2\lambda^2  h_{37}(\lambda) & 2\{1 + \lambda ^ 2 h_{38}(\lambda)\}
\end{bsmallmatrix}$\rule{0pt}{22pt}\vspace{.5mm} \cr
\hline
$\mathrm{GG}(\lambda,\beta, \rho)$ &  $G(\lambda, 1, 1)
=
\begin{bsmallmatrix}
h_{10}(\lambda) ~ & ~ \lambda h_{6}(\lambda,\lambda+1,1) ~&~ -h_{8}(\lambda) \\[1mm]
h_{11}(\lambda) ~ & ~ \lambda h_{7}(\lambda,\lambda+1,1) ~&~ -h_{9}(\lambda)
\end{bsmallmatrix}$, ~ $I(\lambda, 1, 1)
=
\begin{bsmallmatrix}
 \psi_1(\lambda) ~ & ~ 1 ~ & ~ -\psi(\lambda) \\[1mm]
1 ~ & ~ \lambda ~&~ -\lambda \psi(\lambda) - 1 \\[1mm]
-\psi(\lambda) ~ & ~ -\lambda \psi(\lambda) - 1 ~&~ \lambda \psi^2(\lambda) + 2 \psi(\lambda) + \lambda \psi_1(\lambda) + 1
\end{bsmallmatrix}$\rule{0pt}{15pt}\vspace{.5mm}\cr
\hline
$\vcenter{\hbox{\shortstack{Logistic \\ $(\mu,\sigma)$}}}$
& $G(0,1)
=
\begin{bsmallmatrix}
0 ~&~ 0.698397593884459 \\[1mm]
-1/\pi ~&~ 0
\end{bsmallmatrix}$, ~
$I(0,1)
=
\begin{bsmallmatrix}
1/3 ~&~ 0 \\[1mm]
0 ~&~ (3 + \pi^2)/9
\end{bsmallmatrix}$ \rule{0pt}{12pt}\vspace{1mm} \cr
\hline
$\vcenter{\hbox{\shortstack{Student's \\ $t(\lambda,\mu,\sigma)$}}}$
 & $G(\lambda,0,1)
=\begin{bsmallmatrix}
\frac{1}{2} h_{14}(\lambda) & 0 & h_{12}(\lambda) \\[1mm]
0 & \frac{2\Gamma(\frac{\lambda+1}{2})}{\sqrt{\lambda\pi}\Gamma(\lambda/2)} h_{13}(\lambda) & 0
\end{bsmallmatrix}$, ~ $I(\lambda,0,1)=
\begin{bsmallmatrix}
\frac{1}{4}\left\{\psi_1\left(\frac{\lambda}{2}\right)-\psi_1\left(\frac{\lambda+1}{2}\right)-\frac{2(\lambda+5)}{\lambda(\lambda+1)(\lambda+3)}\right\} & 0 & \frac{-2}{(\lambda+1)(\lambda+3)} \\[1mm]
0 & \frac{\lambda + 1}{\lambda + 3} & 0 \\[1mm]
\frac{-2}{(\lambda+1)(\lambda+3)} & 0 & \frac{2\lambda}{\lambda + 3}
\end{bsmallmatrix}$\rule{0pt}{19pt}\vspace{.5mm} \cr
\hline
$\vcenter{\hbox{\shortstack{Gompertz \\ $(\beta, \rho)$}}}$
 & $G(1,\rho)
=
\rho e^{\rho}
\begin{bsmallmatrix}
h_{19}(\rho) ~&~ -h_{21}(\rho) \\[1mm]
h_{20}(\rho) ~&~ -h_{22}(\rho)
\end{bsmallmatrix},
~
I(1,\rho)
=
\begin{bsmallmatrix}
1 + \rho^2 e^{\rho} h_{17}(\rho) ~&~ \rho e^{\rho} h_{18}(\rho) \\[1mm]
\rho e^{\rho} h_{18}(\rho) ~&~ \frac{1}{\rho^2}
\end{bsmallmatrix}$ \rule{0pt}{13pt}\vspace{.5mm} \cr
\hline
$\vcenter{\hbox{\shortstack{Lomax \\ $(\alpha, \sigma)$}}}$
 & $G(\alpha, 1)
=
\begin{bsmallmatrix}
\frac{-h_{6}(1,2,1)}{\alpha} ~&~ -\alpha h_{6}(1,1,\alpha/(\alpha+1)) \\[1mm]
\frac{-h_{7}(1,2,1)}{\alpha} ~&~ -\alpha h_{7}(1,1,\alpha/(\alpha+1))
\end{bsmallmatrix},
~
I(\alpha, 1)
=
\begin{bsmallmatrix}
\frac{1}{\alpha^2} ~&~ \frac{-1}{\alpha + 1} \\[1mm]
\frac{-1}{\alpha + 1} ~&~ \frac{\alpha}{\alpha + 2}
\end{bsmallmatrix}$\rule{0pt}{14pt}\vspace{.5mm} \cr
\hline
$\mathrm{IG}(\mu, \lambda)$
 & $G(\mu, \lambda)
=
\begin{bsmallmatrix}
\frac{\lambda}{\mu^3} h_{27}(\mu, \lambda) ~&~ - \frac{1}{2\mu^2} h_{29}(\mu, \lambda) \\[1mm]
\frac{\lambda}{\mu^3} h_{28}(\mu, \lambda) ~&~ - \frac{1}{2\mu^2} h_{30}(\mu, \lambda)
\end{bsmallmatrix}, ~
I(\mu, \lambda)
=
\begin{bsmallmatrix}
\frac{\lambda}{\mu^3} ~&~ 0 \\[1mm]
0 ~&~ \frac{1}{2\lambda^2}
\end{bsmallmatrix}$\rule{0pt}{15pt}\vspace{.5mm}\cr
\hline
$\mathrm{Beta}(\alpha, \beta)$ & $G(\alpha, \beta)
=
\begin{bsmallmatrix}
h_{23}(\alpha, \beta) ~&~ h_{25}(\alpha, \beta) \\[1mm]
h_{24}(\alpha, \beta) ~&~ h_{26}(\alpha, \beta)
\end{bsmallmatrix},
~
I(\alpha, \beta)
=
\begin{bsmallmatrix}
\psi_1(\alpha) - \psi_1(\alpha + \beta) ~&~ - \psi_1(\alpha + \beta) \\[1mm]
- \psi_1(\alpha + \beta) ~&~ \psi_1(\beta) - \psi_1(\alpha + \beta)
\end{bsmallmatrix}$\rule{0pt}{11pt}\vspace{.5mm} \cr
\hline
$\mathrm{Kum}(\alpha, \beta)$ & $G(1, \beta) = \beta
\begin{bsmallmatrix}
h_{31}(\beta) & h_{33}(\beta) \\[1mm]
h_{32}(\beta) & h_{34}(\beta)
\end{bsmallmatrix}$,  ~
 $I(1, \beta) =
\begin{bsmallmatrix}
1 + \frac{\beta}{\beta - 2} \left\{(\psi(\beta) -\psi(1)-1)^2 - \psi_1(\beta) + \pi^2/6 -1\right\}~ & \frac{\psi(\beta) -\psi(1) - 1 + 1 / \beta}{1 - \beta}\\[1mm]
\frac{\psi(\beta) -\psi(1) - 1 + 1 / \beta}{1 - \beta} & \frac{1}{\beta^2}
\end{bsmallmatrix}$, \rule{0pt}{15pt}\vspace{.5mm}\cr
&  ~ with  $[I(1, \beta=1)]_{1, 2} = 1 - \pi^2/6$  and $[I(1, \beta=2)]_{1, 1} = 1.80822761263836$ \cr
\hline
\end{tabular}
\end{center}
\vspace{-2mm}
\caption{Matrices used to compute the covariance matrix $\Sigma(\boldsymbol{\theta})$,
as defined in \eqref{eq:Sigma.ML} or \eqref{eq:Sigma.general}, and the normalizing scalar $V(\boldsymbol{\theta})$
defined in \eqref{eq:V.LK}, assuming that the parameter vector $\boldsymbol{\theta}$ is unknown.
If a given component $k$ of $\boldsymbol{\theta}$ is assumed known, it suffices to omit
column $k$ of $G$ (and of $J$ in the MM case), as well as row and column $k$ of $I$
(or of $R$). The constants $h_{i}(\cdot)$ are listed in Table~\ref{table:constants} in \ref{app:table.constants}.}
\label{table:G.I}
\end{table}

The first step in computing the $T_n$ and LK test statistics is to estimate the unknown nuisance parameters. The score equations or, where available, explicit estimators for each parameter of the 11 distribution families are given in Table~\ref{table:estimator} in \ref{app:table.constants}, assuming that all parameters are unknown, for the purpose of maximum likelihood estimation. If one or more parameters
are known, the corresponding score functions or estimators are simply omitted and the known values are used instead. In addition, method-of-moments (MM) estimators are provided for the EPD family as an alternative to maximum likelihood estimation, with the restriction that, in contrast to the ML case, the parameter $\lambda$ is always assumed to be known, as indicated by the notation $\mathrm{EPD}_{\lambda}(\mu,\sigma)$.

Next, the covariance matrix $\Sigma(\boldsymbol{\theta})$ for the $T_n$ test statistic must be computed, as defined in \eqref{eq:Sigma.ML} or \eqref{eq:Sigma.general}, as well as the normalizing scalar $V(\boldsymbol{\theta})$ for the LK test, defined in \eqref{eq:V.LK}. To this end, Table~\ref{table:G.I} reports the matrices $G(\boldsymbol{\theta})$ and
$I(\boldsymbol{\theta})$ for the ML-based tests, and the matrices $G(\boldsymbol{\theta})$, $J(\boldsymbol{\theta})$, and $R(\boldsymbol{\theta})$
for the MM-based case, assuming that all parameters are unknown. See Sections~\proofEPD to \proofuniform of the \ref{supp} for the derivation of all matrices reported in Table~\ref{table:G.I}.

If a given component $k$ of $\boldsymbol{\theta}$ is assumed known, it suffices to omit column $k$ of $G$ (and of $J$ in the MM case), as well as row and column $k$ of $I$
(or of $R$). This yields a different covariance matrix $\Sigma(\boldsymbol{\theta})$ for each configuration of known and unknown parameters, and consequently a different test. For most distribution families, the resulting matrix $\Sigma(\boldsymbol{\theta})$ depends only on a subset of the parameters. Therefore, to simplify the expressions,
the remaining parameters are set to standard values. For example, in the EPD case, we set $\mu=0$ and $\sigma=1$ in $G(\lambda,0,1)$ and $I(\lambda,0,1)$, reflecting the fact that $\Sigma(\boldsymbol{\theta})$ depends solely on $\lambda$.

To illustrate how the matrices $\Sigma(\boldsymbol{\theta})$ and $V(\boldsymbol{\theta})$
are computed, we now consider an example. Suppose that we wish to test
\vspace{-2mm}
\[
\mathcal{H}_0 : X_i \sim \mathrm{normal}(\mu_0,\sigma_0)
= \mathrm{EPD}(\lambda=2,\mu_0,\sigma_0),
\]
with $\boldsymbol{\theta}_0=[\mu_0,~\sigma_0]^{\top}$ unknown, using the ML-based $T_n$ and LK
tests. In this case, we omit column~1 of $G(\lambda,0,1)$ and row and column~1 of
$I(\lambda,0,1)$ provided in the first row of Table~\ref{table:G.I}, and set
$\lambda=2$, which yields
\[
G(\boldsymbol{\theta}) = [G(2,0,1)]_{1:2,2:3} =
\begin{bmatrix}
0 & h_{1}(2) \\[1mm]
h_{2}(2)\sqrt{2/\pi} & 0
\end{bmatrix}
\approx
\begin{bmatrix}
0 & 0.7421518289 \\[1mm]
-0.4638466719 & 0
\end{bmatrix},
\]
and $I(\boldsymbol{\theta}) = [I(2,0,1)]_{2:3,2:3} = \mathrm{diag}(1, 2)$.
The constants $h_{i}(\cdot)$ are listed in Table~\ref{table:constants} in \ref{app:table.constants} and are computed numerically using the \textsf{R} software (for example, using the function \texttt{cubature::adaptIntegrate}). Using equations \eqref{eq:Sigma.ML} and \eqref{eq:V.LK}, we then obtain
\[
\Sigma(\boldsymbol{\theta})
= \frac{1}{2} I_2
- G(\boldsymbol{\theta}) I(\boldsymbol{\theta})^{-1} G(\boldsymbol{\theta})^{\top}
\approx
\begin{bmatrix}
0.2246053314 & 0 \\[1mm]
0 & 0.284846265
\end{bmatrix},
\qquad
V^{-1}(\boldsymbol{\theta})
= \frac{1}{\mathrm{tr}\{\Sigma(\boldsymbol{\theta})\}}
\approx 1.962895017.
\]
To perform the same normality test while now assuming that $\mu$ is known and $\sigma$ is unknown, it suffices to omit columns~1 and~2 of $G(\lambda,0,1)$ and rows and columns~1 and~2 of $I(\lambda,0,1)$ provided in the first row of Table~\ref{table:G.I}, and set $\lambda=2$. This yields
\[
G(\boldsymbol{\theta})
\approx
\begin{bmatrix}
0.7421518289 \\[1mm]
0
\end{bmatrix},
\qquad
I(\boldsymbol{\theta}) = 2,
\qquad
\Sigma(\boldsymbol{\theta})
\approx
\begin{bmatrix}
0.2246053314 & 0 \\[1mm]
0 & 0.5
\end{bmatrix},
\qquad
V^{-1}(\boldsymbol{\theta})
\approx 1.38006161.
\]
A third possibility is to perform the normality test while assuming that $\mu$ is unknown and $\sigma$ is known. In this case, it suffices to omit columns~1 and~3 of $G(\lambda,0,1)$ and rows and columns~1 and~3 of $I(\lambda,0,1)$, and to set $\lambda = 2$. This yields
\[
G(\boldsymbol{\theta})
\approx
\begin{bmatrix}
0 \\[1mm]
-0.4638466719
\end{bmatrix},
\qquad
I(\boldsymbol{\theta}) = 1,
\qquad
\Sigma(\boldsymbol{\theta})
\approx
\begin{bmatrix}
0.5 & 0 \\[1mm]
0 & 0.284846265
\end{bmatrix},
\qquad
V^{-1}(\boldsymbol{\theta})
\approx 1.274134878.
\]
Note that the three values of $V^{-1}(\boldsymbol{\theta})$ computed for the three versions of the normality test coincide with the values of $V^{-1}$ reported in Table~1 of \citet{MR1146353}.

\section{Empirical size and power}\label{sec:empirical.simulations}

\subsection{Empirical size under \texorpdfstring{$\mathcal{H}_0$}{H0}}\label{sec:empirical.size.analysis}

In this section, we show in Tables~\ref{table:empirical.size.1} and~\ref{table:empirical.size.2} that the chi-square $\chi_2^2$ approximation to the distributions of the $T_n$ and LK statistics is remarkably accurate even for small sample sizes such as $n=30$, with the accuracy naturally improving as the sample size increases (e.g., $n=100$). This level of accuracy was studied across a wide range of distributions, including the normal (also covering the log-normal), logistic (also covering the log-logistic), Student's $t$, exponential (also covering the Rayleigh and Pareto), gamma (also covering the inverse-gamma and Nakagami), and Weibull (also covering the Fr\'echet and Gumbel) distributions. As a result, critical values and $p$-values for the $T_n$ and LK tests can be computed with high precision without relying on Monte Carlo simulations or tabulated values.

\begin{table}[!htbp]
\centering
\footnotesize
\setlength{\tabcolsep}{3pt}

\begin{tabular}{l *{18}{c}}
\toprule
 & \multicolumn{6}{c}{Normal ($\mu_0$ and $\sigma_0$ unknown)}
 & \multicolumn{6}{c}{Logistic ($\mu_0$ and $\sigma_0$ unknown)}
 & \multicolumn{6}{c}{Student's $t$ ($\lambda_0=4$ known, $\mu_0$ and $\sigma_0$ unknown)} \\
\cmidrule(lr){2-7}\cmidrule(lr){8-13}\cmidrule(lr){14-19}

 & \multicolumn{3}{c}{$T_n$} & \multicolumn{3}{c}{LK}
 & \multicolumn{3}{c}{$T_n$} & \multicolumn{3}{c}{LK}
 & \multicolumn{3}{c}{$T_n$} & \multicolumn{3}{c}{LK} \\
\cmidrule(lr){2-4}\cmidrule(lr){5-7}
\cmidrule(lr){8-10}\cmidrule(lr){11-13}
\cmidrule(lr){14-16}\cmidrule(lr){17-19}

$n$
 & 30 & 100 & $\infty$ & 30 & 100 & $\infty$
 & 30 & 100 & $\infty$ & 30 & 100 & $\infty$
 & 30 & 100 & $\infty$ & 30 & 100 & $\infty$ \\
\midrule

10\%
 & 10.2 & 10.1 & 10.0 & 10.3 & 10.1 & 10.0
 & 10.1 & 10.1 & 10.0 & 10.3 & 10.1 & 10.0
 & 10.1 & 10.1 & 10.0 & 10.2 & 10.2 & 10.1 \\

5\%
 & 5.0 & 5.0 & 5.0 & 5.1 & 5.1 & 5.1
 & 5.0 & 5.0 & 5.0 & 5.1 & 5.0 & 5.0
 & 5.0 & 5.0 & 5.0 & 5.2 & 5.1 & 5.1 \\

1\%
 & 0.9 & 1.0 & 1.0 & 1.0 & 1.0 & 1.0
 & 0.9 & 1.0 & 1.0 & 1.0 & 1.0 & 1.0
 & 0.9 & 1.0 & 1.0 & 1.0 & 1.1 & 1.1 \\

\bottomrule
\end{tabular}

\caption{Empirical rejection frequencies (\%) for the normal (also covering the log-normal), logistic (also covering the log-logistic), and Student's $t(\lambda_0=4$ known) distributions.}
\label{table:empirical.size.1}
\end{table}

\begin{table}[!htbp]
\centering
\footnotesize
\setlength{\tabcolsep}{3pt}

\begin{tabular}{l *{18}{c}}
\toprule
 & \multicolumn{6}{c}{Exponential ($\beta_0$ unknown)}
 & \multicolumn{6}{c}{Gamma ($\lambda_0$ and $\beta_0$ unknown)}
 & \multicolumn{6}{c}{Weibull ($\beta_0$ and $\rho_0$ unknown)} \\
\cmidrule(lr){2-7}\cmidrule{8-13}\cmidrule{14-19}

 & \multicolumn{3}{c}{$T_n$} & \multicolumn{3}{c}{LK}
 & \multicolumn{3}{c}{$T_n$} & \multicolumn{3}{c}{LK}
 & \multicolumn{3}{c}{$T_n$} & \multicolumn{3}{c}{LK} \\
\cmidrule(lr){2-4}\cmidrule(lr){5-7}
\cmidrule(lr){8-10}\cmidrule(lr){11-13}
\cmidrule(lr){14-16}\cmidrule(lr){17-19}

$n$
 & 30 & 100 & $\infty$ & 30 & 100 & $\infty$
 & 30 & 100 & $\infty$ & 30 & 100 & $\infty$
 & 30 & 100 & $\infty$ & 30 & 100 & $\infty$ \\
\midrule

10\%
 & 10.2 & 9.9 & 10.0 & 10.4 & 10.0 & 10.1
 & 10.1 & 10.1 & 10.0 & 10.3 & 10.0 & 10.1
 & 9.9 & 9.9 & 10.0 & 10.0 & 10.0 & 10.1 \\

5\%
 & 5.0 & 4.9 & 5.0 & 5.3 & 5.2 & 5.2
 & 4.8 & 5.0 & 5.0 & 5.1 & 5.1 & 5.1
 & 4.9 & 4.9 & 5.0 & 5.0 & 5.0 & 5.1 \\

1\%
 & 0.8 & 1.0 & 1.0 & 1.1 & 1.1 & 1.2
 & 0.8 & 1.0 & 1.0 & 1.0 & 1.1 & 1.2
 & 0.8 & 1.0 & 1.0 & 1.0 & 1.1 & 1.1 \\

\bottomrule
\end{tabular}

\caption{Empirical rejection frequencies (\%) for the exponential (also covering the Rayleigh and Pareto), gamma (also covering the inverse-gamma and Nakagami), and Weibull (also covering the Fr\'echet and Gumbel) distributions.}
\label{table:empirical.size.2}
\end{table}

The empirical rejection probabilities under $\mathcal{H}_0$ were obtained by simulating samples of sizes $30$ and $100$ and computing the test statistics of the $T_n$
and LK tests. For each test, we estimated the empirical sizes at the $1\%$, $5\%$, and $10\%$ nominal levels by comparing the test statistics with the
corresponding $\chi_2^2$ critical values. These Monte Carlo rejection rates, based on $100{,}000$ replications, are reported in Tables~\ref{table:empirical.size.1} and~\ref{table:empirical.size.2}. The empirical infinite-sample sizes are also reported for the $T_n$ test (which equal $1\%$, $5\%$, and $10\%$ due to its $\chi_2^2$ limiting distribution) and for the LK test, the latter being computed by simulation since, under $\mathcal{H}_0$, $\sqrt{n}[C_n(\hat{\bb{\theta}}_n), S_n(\hat{\bb{\theta}}_n)]^{\top}\rightsquigarrow \mathcal{N}_2(\bb{0}_2, \Sigma(\bb{\theta}_0))$. We observe that although the LK statistic is not exactly asymptotically $\chi_2^2$-distributed, the chi-square approximation remains very accurate.

In Table~\ref{table:empirical.size.1}, the specific values of the location--scale parameters $\mu_0$ and $\sigma_0$ (normal/log-normal, logistic/log-logistic, Student's $t$) used to generate the data under $\mathcal{H}_0$ are irrelevant, since the test statistics $T_n$ and LK are location--scale invariant.
In Table~\ref{table:empirical.size.2}, the generalized gamma distribution $\mathrm{GG}(\lambda_0,\beta_0,\rho_0)$ is examined under three specific cases:
(1) $\lambda_0=1$ known, $\beta_0$ unknown, and $\rho_0$ known, which covers the exponential, Rayleigh, and Pareto families; (2) $\lambda_0$ unknown, $\beta_0$ unknown, and $\rho_0$ known, which covers the gamma, inverse-gamma, and Nakagami families; and (3) $\lambda_0=1$ known, $\beta_0$ unknown, and $\rho_0$ unknown, which covers the Weibull, Fr\'echet, and Gumbel families. The particular choices of the scale and power parameters $\beta_0$ and $\rho_0$ used in the simulations under $\mathcal{H}_0$ are likewise irrelevant, because the statistics $T_n$ and LK are scale- and power-invariant for the GG test.

\subsection{Empirical power under \texorpdfstring{$\mathcal{H}_1$}{H1}}\label{sec:empirical.power.analysis}

Because the $T_n$ and LK goodness-of-fit tests introduced herein apply to 11 distribution families (EPD, GG, etc.) including many specific cases (normal, Laplace, exponential, etc.), yielding a total of 53 distinct tests once all configurations of known and unknown nuisance parameters are taken into account, a fully exhaustive power study is beyond the scope of this paper. Moreover, even for a single test, a meaningful power analysis requires examining a wide spectrum of carefully selected alternatives that
capture as fully as possible the range of plausible deviations from the null, as well as comparing the procedure with all relevant competing tests available
in the literature. Nevertheless, evaluating power against a limited set of alternatives and a small number of competitors remains informative for
illustrating that the $T_n$ and LK tests behave as expected, while keeping in mind that definitive conclusions require a more comprehensive investigation.

In this section, we conduct a simulation study to compare the power of the $T_n$ and LK tests with classical goodness-of-fit procedures based on the
probability integral transform $F(X_i \nvert \hat{\boldsymbol{\theta}}_n)$, namely the Anderson--Darling (AD), Cram\'{e}r--von Mises (CvM), Kuiper (Ku), and Watson (Wa) test statistics, given respectively by
\begin{align*}
\mathrm{AD}
& = -n - \frac{1}{n} \sum_{i=1}^{n}
\left[
(2i-1)\,\ln F(X_{(i)} \nvert \hat{\boldsymbol{\theta}}_n)
+ (2n+1-2i)\,\ln\!\left\{1 - F(X_{(i)} \nvert \hat{\boldsymbol{\theta}}_n)\right\}
\right],\\
\mathrm{CvM}
&= \frac{1}{12n}
  + \sum_{i=1}^{n} \left\{
      \frac{2i-1}{2n} - F(X_{(i)} \nvert \hat{\boldsymbol{\theta}}_n)
    \right\}^2,\\
\mathrm{Ku}
&= \max_{1 \le i \le n} \left\{
      \frac{i}{n} - F(X_{(i)} \nvert \hat{\boldsymbol{\theta}}_n)
    \right\}
  + \max_{1 \le i \le n} \left\{
      F(X_{(i)} \nvert \hat{\boldsymbol{\theta}}_n) - \frac{i-1}{n}
    \right\},\\
\mathrm{Wa}
&= \mathrm{CvM}
   - n\left\{
       \frac{1}{n} \sum_{i=1}^n F(X_i \nvert \hat{\boldsymbol{\theta}}_n)
       - \frac{1}{2}
     \right\}^2 .
\end{align*}

The study is conducted with a sample size of $n=50$ and a significance level of $\alpha=0.05$. We evaluate empirical powers across three families of null
distributions: (1) $\mathcal{H}_0 : X_i \sim \mathrm{normal}(\mu_0, \sigma_0)$, (2) $\mathcal{H}_0 : X_i \sim \mathrm{Student's}\, t(\lambda_0 = 2, \mu_0, \sigma_0)$, and (3) $\mathcal{H}_0 : X_i \sim \mathrm{exponential}(\beta_0)$, with all parameters treated as unknown except for the Student's $t$ distribution, where $\lambda_0 = 2$ is assumed known. The unknown nuisance parameters are replaced by their maximum likelihood estimators under the null hypothesis; see Table~\ref{table:estimator}. Critical values for each of the three families were computed using 1{,}000{,}000 simulations for each of the six goodness-of-fit tests, for $n=50$ and $\alpha = 0.05$. They are reported in Table~\ref{table:crit50}. We observe that the simulated critical values for the LK test are very close to the chi-square quantile $\chi^2_{2,\,0.95} = 5.99146$, and those of the $T_n$ test are even closer. This again illustrates that the chi-square approximation is excellent for both $T_n$ and LK, even for relatively small samples such as $n = 50$.

\begin{table}[H]
\centering
\footnotesize
\setlength{\tabcolsep}{6pt}

\begin{tabular}{lcccccc}
\toprule
Null distribution & $T_n$ & LK & AD & CvM & Ku & Wa \\
\midrule
Normal$(\mu_0,\sigma_0)$
 & 5.98199 & 6.00978 & 0.74615 & 0.12539 & 0.20554 & 0.11611 \\

Student's $t_2(\mu_0,\sigma_0)$
 & 5.97522 & 6.04111 & 0.81509 & 0.09768 & 0.17408 & 0.06979 \\

Exponential$(\beta_0)$
 & 5.97462 & 6.08453 & 1.30988 & 0.22017 & 0.22793 & 0.15833 \\
\bottomrule
\end{tabular}

\caption{Critical values at $\alpha = 0.05$ for $n=50$, obtained from 1{,}000{,}000 simulations for each goodness-of-fit test.}
\label{table:crit50}
\end{table}

Regarding the choice of alternatives, we followed the approach used in \citet{MR4512291}, selecting families of alternatives rather than a single and
somewhat arbitrary alternative. For example, when testing normality, we considered the family of EPD$(\lambda,\mu,\sigma)$ distributions, which offers
a wide range of tail behaviors through variation of $\lambda$, and which includes the normal distribution as the special case $\lambda = 2$. Instead of
choosing a single member of this family, for instance the Laplace distribution with $\lambda = 1$ as a heavy-tailed alternative, we examined regularly spaced
values of $\lambda$ over the interval $0.4 \leq \lambda \leq 2$ in order to generate a full power curve. Power was computed using 10,000 simulations for each grid
value (in this case $\lambda = 0.4, 0.5, 0.6, \ldots, 1.9, 2$).
We also considered short-tail alternatives within the EPD family by varying $\lambda$ from $2$ to $25$, as well as asymmetric alternatives by varying the
asymmetry parameter $\alpha$ of the asymmetric power distribution (APD; defined in \eqref{def:APD} to follow)  from $0.5$ to $1$. The same methodology was applied to the Student's $t$ case, using the skew normal family as asymmetric alternatives, and to the exponential case, using log-normal and inverse-Gaussian families as alternatives; see Table~\ref{table:density.CDF}.

We then summarized the overall behavior by averaging the resulting powers, i.e., by
computing the area under each power curve divided by the length of the parameter range, and these averages are reported in Tables~\ref{table:power-normal},~\ref{table:power-student}~and~\ref{table:power-exponential}. The full power
curves themselves are provided in Section~\powercurves of the \ref{supp}. To aid interpretation, the tables report the ``gap with best,'' defined as the difference between a given test's average power and the highest average power achieved among all $6$ tests.
The overall summary of empirical powers across the three families of null distributions is presented in Table~\ref{table:power-summary}. For each test,
the average power over the ten families of alternatives was computed to obtain a single overall measure. The $T_n$ test achieves the highest mean power at 59.9\%, followed by three close competitors: the LK test at 56.9\% (a gap of 3.0\%), the Watson test at 56.6\% (a gap of 3.3\%), and the AD test at 55.9\% (a gap of 4.0\%). As mentioned above, these results illustrate that the $T_n$ and LK tests perform well, while recognizing that definitive conclusions would require a more extensive study.

When examining the tests individually, the two highest average powers for normality are obtained by AD (42.4\%) and $T_n$ (40.7\%).
For the Student's $t_2$ test, the best performance is achieved by $T_n$ (60.6\%), followed by LK (55.8\%) and Watson (54.0\%). For the exponential test, AD (74.6\%) and $T_n$ (73.8\%) again yield the strongest results. Interestingly, the gain in power obtained by using the $T_n$ test over the LK test averages 3.0\% across all settings, with respective improvements of 1.4\%, 4.8\%, and 2.9\% for the normality, Student's $t_2$, and exponential tests.

\begin{table}[!htbp]
\centering
\footnotesize
\setlength{\tabcolsep}{5pt}

\begin{tabular}{lcccccc}
\toprule
Alternative $\mathcal{H}_1$ & $T_n$ & LK & AD & CvM & Kuiper & Watson \\
\midrule
$\mathrm{EPD}(0.4  \leq \lambda\leq 2,\,\mu,\sigma)$ & 43.8 & 42.6 & 42.6 & 42.1 & 40.5 & 42.4 \\
$\mathrm{EPD}(2 \leq \lambda \leq 25,\,\mu,\sigma)$ & 48.4 & 43.6 & 45.7 & 36.8 & 36.4 & 40.7 \\
$\mathrm{APD}_{\lambda=2}(0.5 \leq \alpha < 1,\rho=2,\mu,\sigma)$ & 29.8 & 31.5 & 38.9 & 34.7 & 27.2 & 30.4 \\
\textbf{Average power} & \textbf{40.7} & \textbf{39.3} & \textbf{42.4} & \textbf{37.9} & \textbf{34.7} & \textbf{37.8} \\
\textbf{Gap with best} & \textbf{1.7} & \textbf{3.1} & \textbf{0.0} & \textbf{4.5} & \textbf{7.7} & \textbf{4.6} \\
\bottomrule
\end{tabular}
\caption{Average empirical powers (\%) for testing $\mathcal{H}_0: X_i \sim \mathrm{normal}(\mu_0,\sigma_0\ \text{unknown})$ under EPD and APD (defined in \eqref{def:APD}) alternatives, for a sample size of $n=50$ and a significance level of $0.05$. The power results are invariant to the specific values of $\mu$ and $\sigma$ under $\mathcal{H}_1$.}
\label{table:power-normal}
\end{table}

\begin{table}[!htbp]
\centering
\footnotesize
\setlength{\tabcolsep}{5pt}

\begin{tabular}{lcccccc}
\toprule
Alternative $\mathcal{H}_1$ & $T_n$ & LK & AD & CvM & Kuiper & Watson \\
\midrule
$\mathrm{EPD}(0.25 \leq \lambda \leq 0.8,\mu,\sigma)$ & 44.5 & 40.2 & 34.0 & 37.7 & 41.3 & 45.1 \\
$\mathrm{EPD}(1.2 \leq \lambda \leq 10,\mu,\sigma)$ & 72.3 & 65.2 & 43.1 & 42.1 & 49.3 & 62.7 \\
$\mathrm{SN}(0 \leq \lambda \leq 16,\mu,\sigma)$ & 64.9 & 62.0 & 56.7 & 55.8 & 45.7 & 54.1 \\
\textbf{Average power} & \textbf{60.6} & \textbf{55.8} & \textbf{44.6} & \textbf{45.2} & \textbf{45.5} & \textbf{54.0} \\
\textbf{Gap with best} & \textbf{0.0} & \textbf{4.8} & \textbf{16.0} & \textbf{15.4} & \textbf{15.1} & \textbf{6.6} \\
\bottomrule
\end{tabular}
\caption{Average empirical powers (\%) for testing $\mathcal{H}_0: X_i \sim \mathrm{Student's}\, t(\lambda_0=2\,\text{known},\mu_0,\sigma_0\ \text{unknown})$ under EPD and skew normal (SN) alternatives, for a sample size of $n=50$ and a significance level of $0.05$. The power results are invariant to the specific values of $\mu$ and $\sigma$ under $\mathcal{H}_1$.}
\label{table:power-student}
\end{table}

\begin{table}[!htbp]
\centering
\footnotesize
\setlength{\tabcolsep}{5pt}

\begin{tabular}{lcccccc}
\toprule
Alternative $\mathcal{H}_1$ & $T_n$ & LK & AD & CvM & Kuiper & Watson \\
\midrule
$\mathrm{log}\mhyphen\mathrm{normal}(\mu,0.6 \leq \sigma \leq 1)$  & 75.8 &  74.0  & 75.9  & 70.4  & 71.7  &  75.2   \\
$\mathrm{log}\mhyphen\mathrm{normal}(\mu,1 \leq \sigma \leq 1.8)$ & 71.7 &  70.2 &  76.0 &  76.3 &  68.4 & 70.8  \\
$\mathrm{inverse}\mhyphen\mathrm{Gaussian}(\mu=1,0.1 \leq \lambda \leq 0.7)$  & 74.3 & 69.2 & 70.0 & 68.6 & 67.8 & 70.3 \\
$\mathrm{inverse}\mhyphen\mathrm{Gaussian}(\mu=1,0.7 \leq \lambda \leq 1.5)$  & 73.3 & 70.2 & 76.5 & 65.7 & 70.7 & 74.0 \\
\textbf{Average power} & \textbf{73.8} & \textbf{70.9} & \textbf{74.6} & \textbf{70.3} & \textbf{69.7} & \textbf{72.6} \\
\textbf{Gap with best} & \textbf{0.8} & \textbf{3.7} & \textbf{0.0} & \textbf{4.3} & \textbf{4.9} & \textbf{2.0} \\
\bottomrule
\end{tabular}
\caption{Average empirical powers (\%) for testing $\mathcal{H}_0: X_i \sim \mathrm{exponential}(\beta_0\,\text{unknown})$ under log-normal and inverse-Gaussian alternatives, for a sample size of $n=50$ and a significance level of $0.05$. For the log-normal alternatives, the power results are invariant to the specific value of $\mu$ under $\mathcal{H}_1$.}
\label{table:power-exponential}
\end{table}

\begin{table}[!htbp]
\centering
\footnotesize
\setlength{\tabcolsep}{5pt}

\begin{tabular}{lcccccc}
\toprule
 & $T_n$ & LK & AD & CvM & Kuiper & Watson \\
\midrule
\textbf{Overall average power} & \textbf{59.9} & \textbf{56.9} & \textbf{55.9} & \textbf{53.0} & \textbf{51.9} & \textbf{56.6} \\
\textbf{Gap with best} & \textbf{0.0} & \textbf{3.0} & \textbf{4.0} & \textbf{6.9} & \textbf{8.0} & \textbf{3.3} \\
\bottomrule
\end{tabular}

\caption{Overall summary of empirical powers across the three families of null distributions.}
\label{table:power-summary}
\end{table}

\subsection{Empirical power for the Laplace distribution}\label{sec:empirical.power.analysis.Laplace}

In this section, we revisit the comprehensive simulation study of goodness-of-fit tests
for the Laplace distribution presented by \citet{MR4512291}, assuming that both $\mu$ and
$\sigma$ are unknown. This Monte Carlo study evaluated $40$ competing goodness-of-fit
tests against a broad set of $400$ alternative distributions, including both symmetric
and asymmetric cases with light and heavy tails, for sample sizes
$n \in \{20, 50, 100, 200\}$ and a significance level $\alpha = 0.05$.
In particular, it included the MM-based LK test proposed by \citet{MR1146353}.
Our new MM-based Laplace test $T_n$ was added to this set of contenders, increasing the
total number of tests considered to $41$. We deliberately chose the MM-based version of
$T_n$, rather than its ML-based counterpart, in order to ensure consistency with the
MM-based LK test already included in the original study.

Table~\ref{table:power.Laplace} summarizes the results, listing the $12$ most powerful tests for each sample size. The results clearly demonstrate the strength of our new procedure. Our test $T_n$ (highlighted in dark gray) ranks 12th, 4th, 1st, and 1st out of 41 contenders for the sample sizes $n = 20$, 50, 100, and 200, respectively. The last two rows of Table~\ref{table:power.Laplace} summarize the gaps across the four sample sizes considered for each of the leading tests, reporting both the maximum and the average of these four gaps. Overall, when all sample sizes are considered, our test $T_n$ emerges as the most powerful test on average among all 41 competitors, with an average power gap of $1.2\%$. Furthermore, $T_n$ improves average power by margins of 1.5\%, 3.2\%, 3.1\%, and 1.8\% compared to the LK test (in light gray), which ranks 16th, 9th, 3rd, and 4th, respectively, for $n = 20$, 50, 100, and 200.

% centered, fixed-width column
\newcolumntype{C}{>{\centering\arraybackslash}p{0.94cm}}

\smallskip
\begin{table}[!htbp]
\scriptsize\def\arraystretch{1.1}% 1 is the default, change whatever you need
\begin{center}
\setlength\tabcolsep{1.7pt} % default value: 6pt
\begin{tabular}{C|C|C|C|C|C|C|C|C|C|C|C|C}
\hline
$n \!=\! 20$ & $\hspace{-0.2mm}\!\text{AP}_y$ & $\text{AP}_e$ & $\text{AP}_z$ & $\text{AP}_v$ & $\text{CK}_v$ & $\hspace{-0.2mm}\!\text{AP}_y^{\normalfont \text{(MLE)}}$ & $\text{Me}_2^{(1)}$ & $\text{AP}_a$ & Wa & $X^{\mathrm{APD}}_{1}$ & $\text{Me}_{0.5}^{\scriptscriptstyle (2)}$ & \Tn{\bf $T_n$} \\
{Power} & 47.0 & 46.4 & 46.2 & 45.3 & 45.2 & 44.9 & 44.7 & 44.5 & 44.2 & 44.0 & 43.7 & \Tn{\bf 43.5} \\
{Gap} & 0 & 0.6 & 0.8 & 1.7 & 1.8 & 2.1 & 2.3 & 2.5 & 2.8 & 3.0 & 3.3 & \Tn{\bf 3.5} \\
\hline
$n \!=\! 50$ & $\text{AP}_v$ & $\text{AP}_e$ & $\hspace{-0.2mm}\!\text{AP}_y^{\scriptscriptstyle (\text{MLE})}$ & \Tn{\bf $T_n$} & $X^{\mathrm{APD}}_{1}$ & $\hspace{-0.2mm}\!\text{AP}_y$ & $\text{Me}_{0.5}^{\scriptscriptstyle (2)}$ & $\text{Me}_2^{\scriptscriptstyle (1)}$ & \LK{LK} & Wa & $\text{CK}_v$ & $\text{CK}_c$ \\
{Power} & 63.7 & 62.6 & 62.3 & \Tn{\bf 62.2} & 60.8 & 60.4 & 60.4 & 60.1 & \LK{\bf 59.0} & 58.9 & 58.8 & 58.2\\
{Gap} & 0 & 1.1 & 1.4 & \Tn{\bf 1.5} & 2.9 & 3.3 & 3.3 & 3.6 & \LK{\bf 4.7} & 4.8 & 4.9& 5.5 \\
\hline
$n \!=\! 100$ & \Tn{\bf $T_n$} & $X^{\mathrm{APD}}_{1}$ & \LK{LK} & Wa & $\text{AP}_v$ & $\text{Me}_{0.5}^{\scriptscriptstyle (2)}$ & $\text{Me}_2^{\scriptscriptstyle (1)}$ & $\hspace{-0.2mm}\!\text{AP}_y^{\normalfont \text{(MLE)}}$ & $\text{AP}_e$ & $\text{AB}_{\text{He}}$ & Ku & AJ \\
{Power} & \Tn{\bf 75.4} & 74.3 & \LK{\bf 72.3} & 71.8 & 71.7 & 71.7 & 71.4 & 69.0 & 68.6 & 68.1 & 67.8 & 67.8\\
{Gap} & \Tn{\bf 0} & 1.1 & \LK{\bf 3.1} & 3.6 & 3.7 & 3.7 & 4.0 & 6.4 & 6.8 & 7.3 & 7.6 & 7.6 \\
\hline
$n \!=\! 200$ & \Tn{\bf $T_n$} & $X^{\mathrm{APD}}_{1}$ & Wa & \LK{LK} & $\text{Me}_2^{\scriptscriptstyle (1)}$ & $\text{Me}_{0.5}^{\scriptscriptstyle (2)}$ & Ku & $\text{AB}_{\text{He}}$ & BS & $\text{AB}_{\text{Je}}$ & $\text{AP}_v$ & $Z(K_1^{\text{net}})$ \\
{Power} & \Tn{\bf 82.1} & 81.8 & 80.7 & \LK{\bf 80.3} & 79.5 & 79.1 & 78.0 & 77.3 & 77.1 & 76.8 & 76.8 & 76.3 \\
{Gap} & \Tn{\bf 0} & 0.3 & 1.4 & \LK{\bf 1.8} & 2.6 & 3.0 & 4.1 & 4.8 & 5.0 & 5.3 & 5.3 & 5.8 \\
\hline
\hline
{Max} & $X^{\mathrm{APD}}_{1}$ & \Tn{\bf $T_n$} & $\text{Me}_{0.5}^{(2)}$ & $\text{Me}_2^{(1)}$ & Wa & \LK{LK} & $\text{AP}_v$ & $\text{AB}_{\text{He}}$ & $\hspace{-0.2mm}\!\text{AP}_y^{\normalfont \text{(MLE)}}$ & $\text{Z}_A$ & $\text{AP}_e$ & Ku \\
{Gap} & 3.0 & \Tn{\bf 3.5} & 3.7 & 4.0 & 4.7 & \LK{\bf 5.0} & 5.3 & 7.3 & 7.9 & 8.0 & 8.2 & 8.3 \\
\hline
{Average} & \Tn{\bf $T_n$} & $X^{\mathrm{APD}}_{1}$ & $\text{AP}_v$ & $\text{Me}_2^{(1)}$ & Wa & $\text{Me}_{0.5}^{(2)}$ & \LK{LK} & $\text{AP}_e$ & $\hspace{-0.2mm}\!\text{AP}_y^{\normalfont \text{(MLE)}}$ & $\hspace{-0.2mm}\!\text{AP}_y$ & Ku & $\text{AB}_{\text{He}}$ \\
{Gap} & \Tn{\bf 1.2} & 1.8 & 2.7 & 3.1 & 3.1 & 3.3 & \LK{\bf 3.6} & 4.2 & 4.5 & 6.1 & 6.2 & 6.6 \\
\hline
\end{tabular}
\end{center}
\caption{The average power \% for the 12 best performing Laplace tests (among 41) against a selection of $400$ alternatives, as a function of sample size ($\alpha=0.05$). The labels, formulas, and references for each test are all detailed by \citet{MR4512291}.}
\label{table:power.Laplace}
\end{table}

\vspace{-3mm}
\section{Asymptotic power under sequences of local alternatives}\label{sec:asymp.local.alternatives}

This section examines the asymptotic power of the test statistic $T_n(\hat{\bb{\theta}}_n)$ under sequences of local alternatives. Asymptotic power curves are provided for the gamma and exponential power distribution (EPD) tests with unknown nuisance parameters. In each case, a sequence of local alternatives is constructed by embedding the null distribution within a more general parametric family, which allows us to apply the theoretical framework of \citet[Theorem~2]{MR4906044} and derive the noncentrality parameter of the limiting chi-square distribution of the test statistic. The performance of the test $T_n(\hat{\bb{\theta}}_n)$ is compared with that of a natural benchmark provided by Rao's score test (also known as the Lagrange multiplier test), as well as with its asymptotically equivalent counterpart, the generalized likelihood ratio test (GLRT), for which the limiting distribution is also derived under the null hypothesis and under sequences of local alternatives.

Given that the null distribution $F(\cdot \nvert\bb{\theta})$ is embedded within a more general parametric family
$H(\cdot \nvert\bb{\phi},\bb{\theta})$ (with density denoted by $h$), the null hypothesis can be written as
\[
\mathcal{H}_0 : X_1, \ldots, X_n \sim F(\cdot \nvert\bb{\theta}_0)=H(\cdot \nvert\bb{\phi}_0,\bb{\theta}_0),
\]
where $\bb{\phi}_0$ is a fixed and known parameter vector chosen to correspond to a
specific case of the family $H(\bb{\phi},\bb{\theta})$ and $\bb{\theta}_0$ denotes the unknown parameter vector of the null distribution. The sequence of local alternative hypotheses is defined by
\vspace{-1mm}
\[
\mathcal{H}_{1,n}(\bb{\delta}) : X_1, \ldots, X_n \sim
H(\cdot \nvert\bb{\phi}_n,\bb{\theta}_0), ~\text{ with } ~ \bb{\phi}_n=\bb{\phi}_0 + \frac{\bb{\delta}}{\sqrt{n}}\{1 + o(1)\},
\]
where $n\in \N$ and $\bb{\delta}$ is a fixed, nonzero real vector of the same dimension as $\bb{\phi}_0$.

Consider the score vectors and their corresponding Fisher information blocks
\begin{align*}
&\bb{s}_{\bb{\phi}}(x, \bb{\phi}_0,\bb{\theta})
  = \partial_{\bb{\phi}} \ln h(x \nvert \bb{\phi},\bb{\theta})|_{\bb{\phi}=\bb{\phi}_0}, \qquad
 \bb{s}(x,\bb{\theta})\equiv \bb{s}_{\bb{\theta}}(x,\bb{\phi}_0,\bb{\theta})= \partial_{\bb{\theta}} \ln h(x \nvert \bb{\phi}_0,\bb{\theta}),\\[2mm]
& I_{\phi,\phi}(\bb{\phi}_0,\bb{\theta}) = \EE\{\bb{s}_{\bb{\phi}}(X,\bb{\phi}_0,\bb{\theta})\,\bb{s}_{\bb{\phi}}(X,\bb{\phi}_0,\bb{\theta})^{\top}\}, \qquad
I(\bb{\theta})\equiv I_{\theta,\theta}(\bb{\phi}_0,\bb{\theta}) = \EE\{\bb{s}_{\bb{\theta}}(X,\bb{\phi}_0,\bb{\theta})\,\bb{s}_{\bb{\theta}}(X,\bb{\phi}_0,\bb{\theta})^{\top}\},\\[2mm]
& I_{\phi,\theta}(\bb{\phi}_0,\bb{\theta}) = \EE\{\bb{s}_{\bb{\phi}}(X,\bb{\phi}_0,\bb{\theta})\,\bb{s}_{\bb{\theta}}(X,\bb{\phi}_0,\bb{\theta})^{\top}\},
\end{align*}
where the expectation is taken under \(X \sim F(\cdot \nvert\bb{\theta})=H(\cdot \nvert\bb{\phi}_0,\bb{\theta})\). The ML estimators of the restricted and enlarged models are
\[
\tilde{\bb{\theta}}_n
    = \arg\max_{\bb{\theta}} \sum_{i=1}^n \ln f(x_i \nvert \bb{\theta}) = \arg\max_{\bb{\theta}} \sum_{i=1}^n \ln h(x_i \nvert \bb{\phi}_0,\bb{\theta}),
\qquad
(\overline{\bb{\phi}}_n, \overline{\bb{\theta}}_n)^{\top}
    = \arg\max_{\bb{\phi},\bb{\theta}} \sum_{i=1}^n \ln h(x_i \nvert \bb{\phi},\bb{\theta}).
\]
Applying the theoretical framework of \citet[Theorem~1]{MR4906044}, as we did in Section~\ref{sec:GOF.general}, we derive Rao's score test statistic with unknown nuisance parameters $\bb{\theta}_0$ estimated by ML. Consider the kernel of dimension $\dim(\bb{\phi})$ and the corresponding score--kernel cross-moment matrix given by
\begin{equation*}\label{eq:tau.Rao}
\bb{\tau}_{\mathcal{R}}(x,\bb{\phi}_0,\bb{\theta})
=
\bb{s}_{\bb{\phi}}(x,\bb{\phi}_0,\bb{\theta}),
\qquad
G_{\mathcal{R}}(\bb{\phi}_0,\bb{\theta})
= \EE\{\bb{s}_{\bb{\phi}}(X,\bb{\phi}_0,\bb{\theta})\,\bb{s}_{\bb{\theta}}(X,\bb{\phi}_0,\bb{\theta})^{\top}\}
= I_{\phi,\theta}(\bb{\phi}_0,\bb{\theta}).
\end{equation*}
Rao's score test statistic, denoted by $\mathcal{R}_n$, is then given by
\[
\mathcal{R}_n
=
n\,\bb{\bar{s}}_{\bb{\phi}}(\bb{\phi}_0,\tilde{\bb{\theta}}_n)^{\top}
\Sigma_{\mathcal{R}}(\bb{\phi}_0,\tilde{\bb{\theta}}_n)^{-1}
\bb{\bar{s}}_{\bb{\phi}}(\bb{\phi}_0,\tilde{\bb{\theta}}_n),
\qquad
\text{where}
\qquad
\bb{\bar{s}}_{\bb{\phi}}(\bb{\phi}_0,\bb{\theta})
=
\frac{1}{n}\sum_{i=1}^{n}
\bb{s}_{\bb{\phi}}(x_i,\bb{\phi}_0,\bb{\theta}),
\]
with the corresponding covariance matrix given by
\[
\Sigma_{\mathcal{R}}(\bb{\phi}_0,\bb{\theta})
=
I_{\phi,\phi}(\bb{\phi}_0,\bb{\theta})
-
I_{\phi,\theta}(\bb{\phi}_0,\bb{\theta})
\,I_{\theta,\theta}(\bb{\phi}_0,\bb{\theta})^{-1}
\,I_{\phi,\theta}(\bb{\phi}_0,\bb{\theta})^{\top}.
\]
 The test statistic of the generalized likelihood ratio test (GLRT), denoted by $\Lambda_n$, is given by
\[
\Lambda_n=2 \sum_{i=1}^{n}\ln h(x_i \nvert \overline{\bb{\phi}}_n, \overline{\bb{\theta}}_n)-2  \sum_{i=1}^{n}\ln f(x_i \nvert \tilde{\bb{\theta}}_n).
\]
Under $\mathcal{H}_0$ and as $n \to \infty$, we obtain
$
\mathcal{R}_n , \Lambda_n\rightsquigarrow \chi_{\dim(\bb{\phi}_0)}^2.
$

Applying now the theoretical framework of \citet[Theorem~2]{MR4906044}, we derive the asymptotic distribution of the tests $T_n(\hat{\bb{\theta}}_n)$, $\mathcal{R}_n$ and $\Lambda_n$ under sequences of local alternatives.
For the ML-based $T_n(\hat{\bb{\theta}}_n)$, consider the matrix $M(\bb{\phi}_0,\bb{\theta}_0)$ defined by
\begin{align}\label{eq:M.ML}
M(\bb{\phi}_0,\bb{\theta}_0)
    = G_{\bb{\phi}}(\bb{\phi}_0,\bb{\theta}_0)
      - G(\bb{\theta}_0)\,I(\bb{\theta}_0)^{-1}
        I_{\phi,\theta}(\bb{\phi}_0,\bb{\theta}_0)^{\top},
        ~~ \text{where }~~ G_{\bb{\phi}}(\bb{\phi_0},\bb{\theta}_0)
  = \EE\left\{\bb{\tau}(X,\bb{\theta}_0)\,\bb{s}_{\bb{\phi}}(X,\bb{\phi}_0,\bb{\theta}_0)^{\top}\right\}.
\end{align}
The vector $\bb{\tau}$ denotes the kernel defined in~\eqref{eq:tau}, and $G(\bb{\theta}_0)$ is the score--kernel cross-moment matrix defined in~\eqref{eq:s.I.G}. Similarly for the ML-based tests $\mathcal{R}_n$ and $\Lambda_n$, by using instead the kernel $\bb{\tau}_{\mathcal{R}}(x,\bb{\phi}_0,\bb{\theta}_0)$, we obtain
\[
M_{\mathcal{R}}(\bb{\phi}_0,\bb{\theta}_0)=I_{\phi,\phi}(\bb{\phi}_0,\bb{\theta}_0)
-
I_{\phi,\theta}(\bb{\phi}_0,\bb{\theta}_0)
\,I_{\theta,\theta}(\bb{\phi}_0,\bb{\theta}_0)^{-1}
\,I_{\phi,\theta}(\bb{\phi}_0,\bb{\theta}_0)^{\top}
=
\Sigma_{\mathcal{R}}(\bb{\phi}_0,\bb{\theta}_0).
\]

For the general case where $\smash{\hat{\bb{\theta}}_n}$ is not restricted to the ML estimator, we define
\begin{align}\label{eq:M.general}
M(\bb{\phi}_0,\bb{\theta}_0)=G_{\bb{\phi}}(\bb{\phi}_0,\bb{\theta}_0)-G(\bb{\theta}_0)R(\bb{\theta}_0)^{-1}S_{\phi,\theta}(\bb{\phi}_0,\bb{\theta}_0)^{\top},
\end{align}
where $R(\bb{\theta}_0)$ is defined in \eqref{eq:R} and
\[
S_{\phi,\theta}(\bb{\phi}_0,\bb{\theta}_0)
  = \EE\big\{\bb{s}_{\bb{\phi}}(X,\bb{\phi}_0,\bb{\theta}_0)\,\bb{r}(X,\bb{\theta}_0)^{\top}\big\}.
\]
Note that for the maximum likelihood case, that is,  $\hat{\bb{\theta}}_n=\tilde{\bb{\theta}}_n$, we have $\bb{r}(x,\bb{\theta}_0) = \bb{s}(x,\bb{\theta}_0)\equiv \bb{s}_{\theta}(x,\bb{\phi}_0,\bb{\theta}_0)$.  In this case, $R(\bb{\theta}_0) = I(\bb{\theta}_0)\equiv I_{\theta,\theta}(\bb{\phi}_0,\bb{\theta}_0)$, $S_{\phi,\theta}(\bb{\phi}_0,\bb{\theta}_0) = I_{\phi,\theta}(\bb{\phi}_0,\bb{\theta}_0)$ and $M(\bb{\phi}_0,\bb{\theta}_0)$ defined in \eqref{eq:M.general} is equivalent to the form in \eqref{eq:M.ML}. By applying Theorem~2 of \citet{MR4906044}, we obtain, under $\mathcal{H}_{1,n}(\bb{\delta})$ and as $n\to\infty$,
\begin{equation*}\label{eq:asymp.H1}
T_n(\hat{\bb{\theta}}_n) \rightsquigarrow \chi_2^2\big(\bb{\delta}^{\top} V(\bb{\phi}_0,\bb{\theta}_0)\bb{\delta}\big),\qquad
\mathcal{R}_n , \Lambda_n\rightsquigarrow \chi_{\dim(\phi_0)}^2\big(\bb{\delta}^{\top} V_{\mathcal{R}}(\bb{\phi}_0,\bb{\theta}_0)\bb{\delta}\big),
\end{equation*}
where
\[
V(\bb{\phi}_0,\bb{\theta}_0)=M(\bb{\phi}_0,\bb{\theta}_0)^{\top} \Sigma(\bb{\theta}_0)^{-1} \, M(\bb{\phi_0},\bb{\theta}_0),\quad
V_{\mathcal{R}}(\bb{\phi}_0,\bb{\theta}_0)=M_{\mathcal{R}}(\bb{\phi}_0,\bb{\theta}_0)^{\top} \Sigma_{\mathcal{R}}(\bb{\phi}_0,\bb{\theta}_0)^{-1} \, M_{\mathcal{R}}(\bb{\phi_0},\bb{\theta}_0)=\Sigma_{\mathcal{R}}(\bb{\phi}_0,\bb{\theta}_0),
\]
the covariance matrix $\Sigma(\bb{\theta}_0)$ is defined in \eqref{eq:Sigma.ML} for the ML case and in \eqref{eq:Sigma.general} for the general case, and
$\chi_b^2(a)$ denotes a chi-square distribution with $b$ degrees of freedom and noncentrality parameter $a$. In particular, when $\bb{\delta}$ is zero, the situation reduces to $\mathcal{H}_0$, and the noncentrality parameter equals zero.

It should be noted that comparing the omnibus test $T_n(\hat{\bb{\theta}}_n)$ with $\mathcal{R}_n$ or  $\Lambda_n$, both tailored to this particular alternative, is not entirely fair, because their construction inherently favors them in that specific direction. Rao's score test and the GLRT, which are locally most powerful in that direction, should therefore be regarded as benchmarks.

Within this framework, the quantities required to compute $V(\bb{\phi}_0,\bb{\theta}_0)$ and $V_{\mathcal{R}}(\bb{\phi}_0,\bb{\theta}_0)=\Sigma_{\mathcal{R}}(\bb{\phi}_0,\bb{\theta}_0)$, and ultimately the noncentrality parameters, are derived for the gamma and EPD tests under specific sequences of local alternatives in order to study the asymptotic
power of the test statistic $T_n(\hat{\bb{\theta}}_n)$ and the benchmark tests $\Lambda_n$ and
$\mathcal{R}_n$.

\begin{proposition}\label{prop:gamma.local.alternative}
Consider the gamma test under the following sequence of local alternatives:
\[
\mathcal{H}_0 : X_i \sim \mathrm{gamma}(\lambda_0, \beta_0) \equiv \mathrm{GG}(\lambda_0, \beta_0, \rho_0=1)\quad \text{ vs. }\quad
\mathcal{H}_{1,n}(\delta) : X_i \sim \mathrm{GG}\Big(\lambda_0, \beta_0, \rho_n=1 + \frac{\delta}{\sqrt{n}} \{1 + o(1)\}\Big),
\]
where $\delta\in\R\setminus\{0\}$ is fixed. The fixed parameter is $\phi_0 = \rho_0 = 1$, while the unknown parameter vector is $\bb{\theta}_0 = [\lambda_0, \beta_0]^{\top}$. Consider the ML-based test $T_n(\hat{\lambda}_n,\hat{\beta}_n)$, Rao's score test $\mathcal{R}_n$, and the GLRT $\Lambda_n$. Then, we have
\begin{align*}
& G(\bb{\theta}_0) =
\begin{bmatrix}
h_{10}(\lambda_0) ~ & ~ \lambda_0\, h_{6}(\lambda_0,\lambda_0+1,1) \\[1mm]
h_{11}(\lambda_0) ~ & ~ \lambda_0\, h_{7}(\lambda_0,\lambda_0+1,1)
\end{bmatrix},
\qquad
G_{\bb{\phi}}(\bb{\phi_0},\bb{\theta}_0)
  =\begin{bmatrix}
 -h_{8}(\lambda_0) \\[1mm]
 -h_{9}(\lambda_0)
\end{bmatrix},
\qquad
  I(\bb{\theta}_0) =\begin{bmatrix}
 \psi_1(\lambda_0) ~ & ~ 1  \\[1mm]
1 ~ & ~ \lambda_0
\end{bmatrix},\\[2mm]
& I_{\phi,\phi}(\bb{\phi}_0,\bb{\theta}_0) =
\lambda_0 \psi^2(\lambda_0) + 2 \psi(\lambda_0) + \lambda_0 \psi_1(\lambda_0) + 1,
\qquad
I_{\phi,\theta}(\bb{\phi}_0,\bb{\theta}_0)=\begin{bmatrix}
 -\psi(\lambda_0) ~ & ~ -\lambda_0 \psi(\lambda_0) - 1
\end{bmatrix},\\
&\Sigma_{\mathcal{R}}(\bb{\phi}_0,\bb{\theta}_0)= \frac{\lambda_0^{2}\,\psi_1(\lambda_0)^{2} - \psi_1(\lambda_0) - 1}{\lambda_0\,\psi_1(\lambda_0) - 1},
\end{align*}
where $\psi(\cdot)$ and $\psi_1(\cdot)$ denote the digamma and trigamma functions, respectively, and the functions $h_i(\cdot)$ are defined in Table~\ref{table:constants}. Note that the resulting quantities $V(\bb{\phi}_0,\bb{\theta}_0)$ and $\Sigma_{\mathcal{R}}(\bb{\phi}_0,\bb{\theta}_0)$ depend only on the value of $\lambda_0$. Therefore, to simplify the expressions, we set $\beta_0=1$.

\end{proposition}
\begin{proof}
This essentially consists of decomposing the matrices $G(\lambda_0,1,1)$ and $I(\lambda_0,1,1)$ of the GG distribution from Table~\ref{table:G.I} according to the fixed and unknown parameters; see Section~\prooflocalgamma of the \ref{supp} for full details.
\end{proof}

Consider now the asymmetric power distribution (APD) introduced by \citet{MR1379475} and further studied by \citet{AyeboKozubowski2003,MR2395888,MR4547729}, among others, as a generalization of the exponential power distribution (EPD). The parametrization of \citet{MR4547729} is adopted in this paper.
\begin{definition}
Given an asymmetry parameter $\alpha\in(0,1)$, a tail-decay parameter $\rho\in(0,\infty)$, a location parameter $\mu\in\R$, a scale parameter $\sigma\in(0,\infty)$, and a fixed, user-specified parameter $\lambda\in(0,\infty)$, define the density of the $\mathrm{APD}_{\lambda}(\alpha,\rho,\mu,\sigma)$ at any $x \in \R$ by
\begin{equation}\label{def:APD}
f_{\lambda}(x \nvert \alpha,\rho,\mu,\sigma)
= \frac{\rho(\delta_{\alpha,\rho}/\lambda)^{1/\rho}}{\sigma\Gamma(1/\rho)}
\exp\left\{-\frac{(\delta_{\alpha,\rho}/\lambda)|y|^{\rho}} {\alpha^{\rho}\, \1_{\{y \le 0\}} + (1 - \alpha)^{\rho}\,\1_{\{y > 0\}}} \right\},
\end{equation}
where
\[
y = \frac{x - \mu}{\sigma} \in \R, \qquad \delta_{\alpha,\rho} = \frac{2 \alpha^{\rho}
(1 - \alpha)^{\rho}}{\alpha^{\rho} + (1 - \alpha)^{\rho}}\in (0,1).
\]
\end{definition}
Note that the parameter $\lambda$, fixed by the user, was added to the original version of \citet{MR2395888} in order to recover the Laplace and normal distributions as special cases when $\lambda = 1$ and $\lambda = 2$, respectively, together with $\alpha = 1/2$ and $\rho = \lambda$, so as to retrieve the $\mathrm{EPD}_{\lambda}(\mu, \sigma)$.

\begin{proposition}\label{prop:EPD.local.alternative}
For a fixed, user-specified value of $\lambda_0$, consider the EPD test under the following sequence of local alternatives:
\begin{align*}
&\mathcal{H}_0 : X_i \sim \mathrm{EPD}_{\lambda_0}(\mu_0, \sigma_0) \equiv \mathrm{APD}_{\lambda_0}(\alpha_0=1/2, \rho_0 =\lambda_0,\mu_0, \sigma_0)
\qquad \text{vs.}\\
&\mathcal{H}_{1,n}(\bb{\delta}) : X_i \sim \mathrm{APD}_{\lambda_0}(\alpha_n,\rho_n,\mu_0, \sigma_0), \quad \text{with } \begin{bmatrix}
\alpha_n \\
\rho_n
\end{bmatrix}
=
\begin{bmatrix}
1/2 \\
\lambda_0
\end{bmatrix}
+
\frac{1}{\sqrt{n}}
\begin{bmatrix}
\delta_1 \\
\delta_2
\end{bmatrix}
\{1 + o(1)\},
\end{align*}
where $\bb{\delta}=[\delta_1,\delta_2]^{\top}\in \R^2 \setminus \{[0,0]^{\top}\}$ is fixed. The fixed parameter vector is $\phi_0 = [\alpha_0, \rho_0]^{\top} = [1/2, \lambda_0]^{\top}$, while the unknown parameter vector is $\bb{\theta}_0 = [\mu_0, \sigma_0]^{\top}$. Consider the ML-based and MM-based tests $T_n(\hat{\mu}_n,\hat{\sigma}_n)$, Rao's score test $\mathcal{R}_n$, and the GLRT $\Lambda_n$. Then, we have $G(\bb{\theta}_0)= [G(\lambda_0,0,1)]_{1:2,2:3}$, $I(\bb{\theta}_0)=[I(\lambda_0,0,1)]_{2:3,2:3}$, and $R(\bb{\theta}_0)=R_{\lambda_0}(0,1)$, using the matrices of the EPD distribution (ML and MM cases) from Table~\ref{table:G.I}, and
\begin{align*}
&  G_{\bb{\phi}}(\bb{\phi_0},\bb{\theta}_0)
  =\begin{bmatrix}
0 & -\frac{1}{\lambda_0^2}h_{3}(\lambda_0)\\[1mm]
-2 h_{35}(\lambda_0) & 0
\end{bmatrix},\,\,  I_{\phi,\phi}(\bb{\phi}_0,\bb{\theta}_0) =\begin{bmatrix}
4 (\lambda_0 + 1) ~&~ 0 \\[1mm]
0 ~&~ \frac{1}{\lambda_0^3}\left\{\Big(\frac{1}{\lambda_0}+1\Big) \psi_1\Big(\frac{1}{\lambda_0}+1\Big) - 1 + (C_{\lambda_0} + 1)^ 2\right\}
\end{bmatrix},\\[2mm]
& I_{\phi,\theta}(\bb{\phi}_0,\bb{\theta}_0)=\begin{bmatrix}
\frac{-2\lambda_0^{2-1/\lambda_0}}{\Gamma(1/\lambda_0)} & 0\\[1mm]
0 & -\frac{1}{\lambda_0}(C_{\lambda_0}+1)
\end{bmatrix},
\,\,
\Sigma_{\mathcal{R}}(\bb{\phi}_0,\bb{\theta}_0)=
\begin{bmatrix}
4(\lambda_0 + 1) - \displaystyle \frac{4\lambda_0^{2}}{\Gamma(2 - 1/\lambda_0)\,\Gamma(1/\lambda_0)} & 0 \\[1mm]
0 &  \frac{\big(\frac{1}{\lambda_0}+1\big)\psi_{1}\big(\tfrac{1}{\lambda_0}+1\big) - 1} {\lambda_0^{3}}
\end{bmatrix}, \\[1mm]
& S_{\phi,\theta}(\bb{\phi}_0,\bb{\theta}_0)
=
\begin{bmatrix}
\frac{-4\Gamma(2/\lambda_0)}{\lambda_0^{1/\lambda_0}\Gamma(3/\lambda_0)} & 0\\[1mm]
0 & -\frac{2D_{\lambda_0}}{\lambda_0^2}\left\{2\ln(\lambda_0)+ 3\psi(3/\lambda_0) - \psi(1/\lambda_0)\right\}
\end{bmatrix},
\end{align*}
where $C_{\lambda_0}=\psi(1/\lambda_0 + 1)+\ln(\lambda_0)$, $\Gamma(\cdot)$, $\psi(\cdot)$ and $\psi_1(\cdot)$ denote the gamma, digamma and trigamma functions, respectively, and the functions $h_i(\cdot)$ are defined in Table~\ref{table:constants}.
Note that the resulting matrices $V(\bb{\phi}_0,\bb{\theta}_0)$ (for both the ML and MM cases) and $\Sigma_{\mathcal{R}}(\bb{\phi}_0,\bb{\theta}_0)$ depend only on the value of $\lambda_0$. Therefore, to simplify the expressions, we set $\mu_0=0,\sigma_0=1$.
\end{proposition}
\begin{proof}
See Section~\prooflocalEPD of the \ref{supp}.
\end{proof}

Propositions~\ref{prop:gamma.local.alternative} and~\ref{prop:EPD.local.alternative} allow us to derive the asymptotic power curves of the
gamma test statistic $T_n(\hat{\lambda}_n,\hat{\beta}_n)$ and the EPD test statistic
$T_n(\hat{\mu}_n,\hat{\sigma}_n)$ under the local alternatives $\mathcal{H}_{1,n}(\bb{\delta})$, for a nominal
significance level of $\alpha = 0.05$. These curves are shown in Figure~\ref{fig:gamma.EPD} along with the
GLRT/score-test benchmark. The left panel displays the gamma test for $\lambda_0 = 1$ (the specific value
of $\beta_0$ is irrelevant due to the scale invariance of the tests). The central and right panels present
the EPD test for an auxiliary parameter value of $\lambda_0 = 1.5$ (the specific values of $\mu_0$ and
$\sigma_0$ are irrelevant due to the location and scale invariance of the tests). Both the ML- and
MM-based versions of the $T_n(\hat{\mu}_n,\hat{\sigma}_n)$ test are shown.

\begin{figure}[!ht]
\centering
\includegraphics[width=0.33\textwidth]{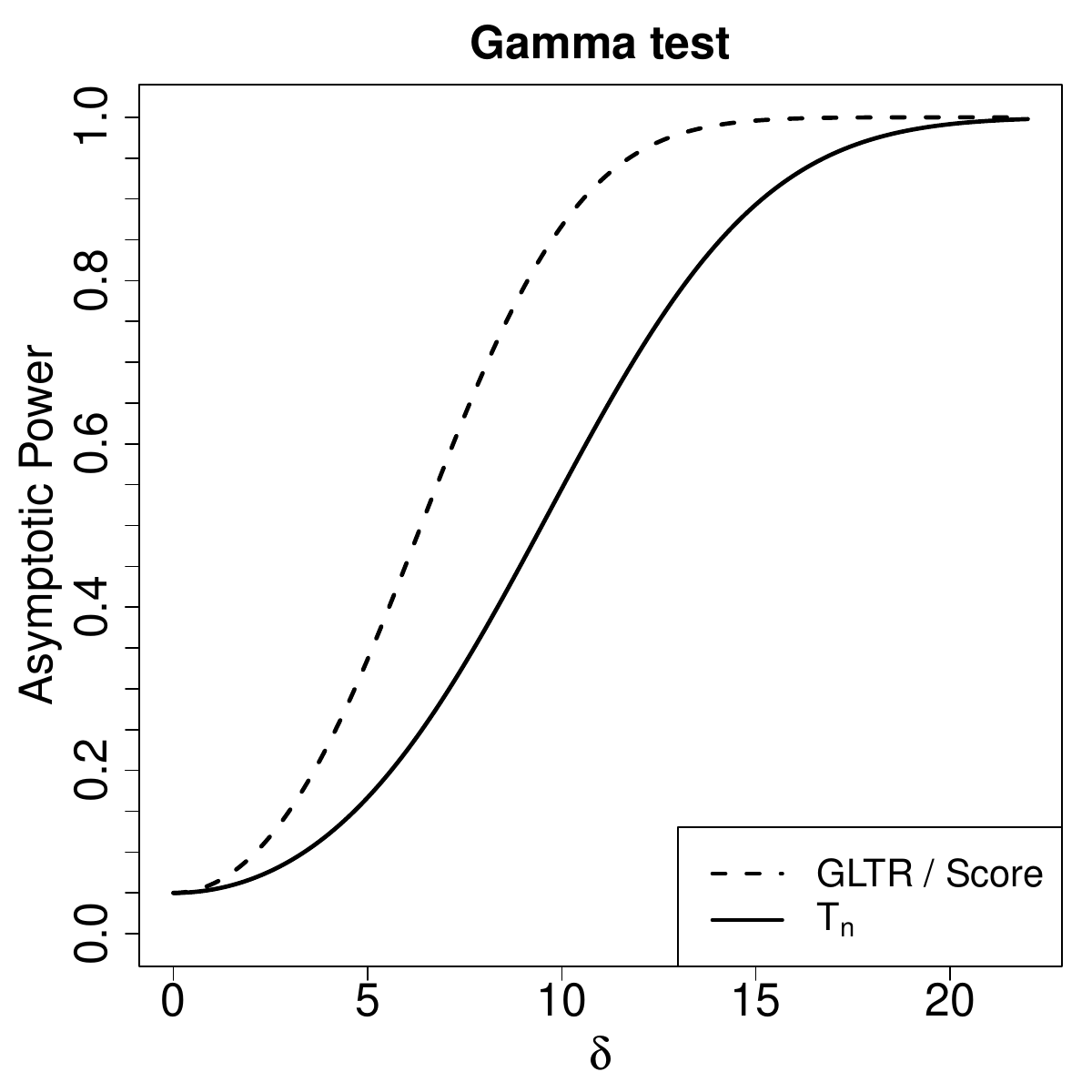}
\includegraphics[width=0.33\textwidth]{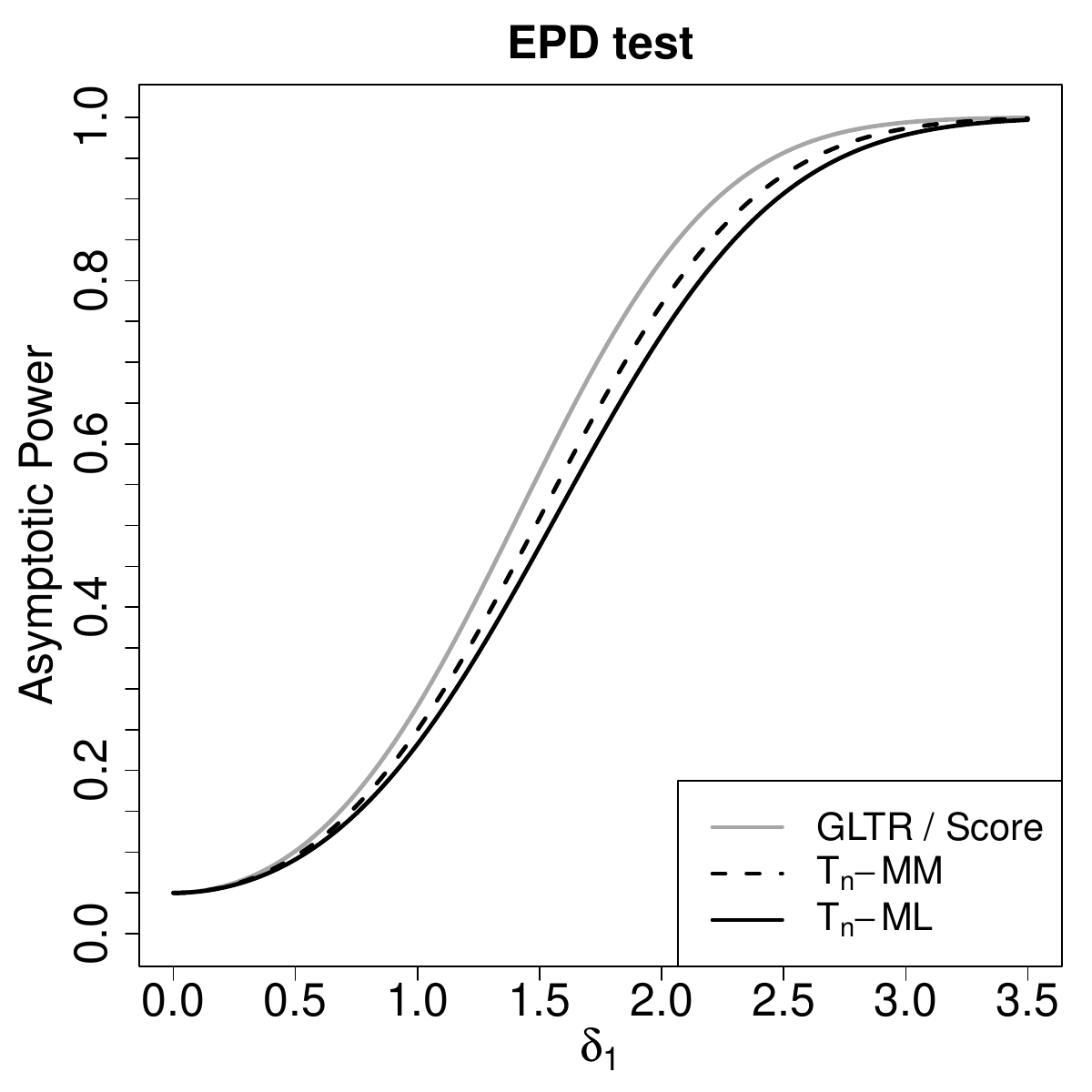}
\includegraphics[width=0.33\textwidth]{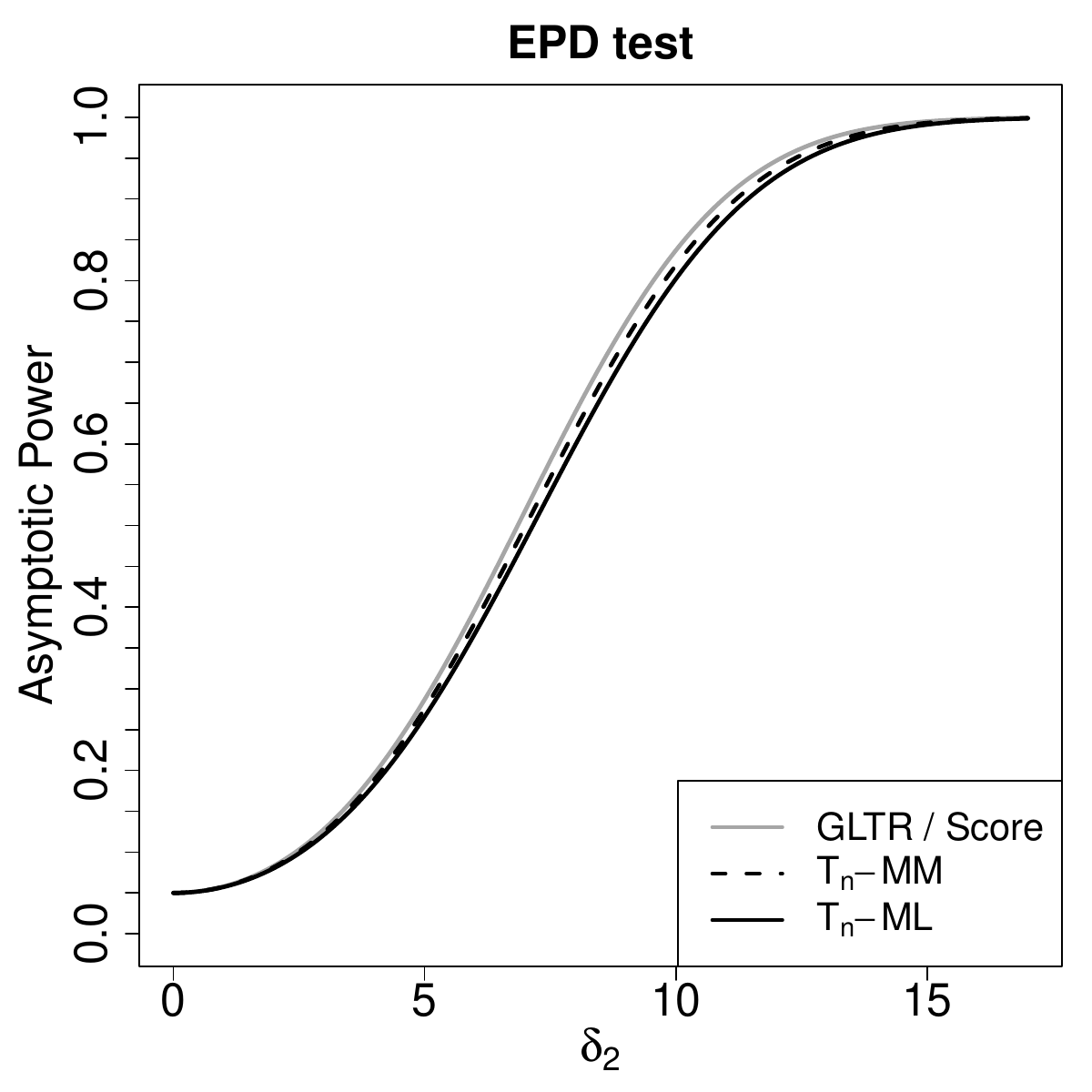}
\caption{Asymptotic power of our test $T_n(\hat{\bb{\theta}}_n)$ and of the GLRT/score-test benchmark for testing the gamma distribution (with $\lambda_0 = 1$) and the EPD distribution (with $\lambda_0 = 1.5$) under local alternatives. The nominal significance level is $0.05$.}\label{fig:gamma.EPD}
\end{figure}

The curves are displayed as a function of $\delta$ (left), of $\delta_1$ with $\delta_2 = 0$ (center),
and of $\delta_2$ with $\delta_1 = 0$ (right). In all cases, the asymptotic power reaches its minimum
at $\delta_i = 0$, where it is equal to the significance level $0.05$, and increases monotonically as
$\delta_i$ moves away from zero, as expected. Note that the noncentrality parameters
depend only on $\delta^2$, $\delta_1^{2}$, and $\delta_2^{2}$, respectively, so the corresponding
asymptotic power curves are symmetric in $\delta$, $\delta_1$, and $\delta_2$. Consequently, the
asymptotic power for $\delta \in [-22,0]$, $\delta_1 \in [-3.5,0]$, and $\delta_2 \in [-17,0]$ is simply
the mirror image of that displayed in Figure~\ref{fig:gamma.EPD}.

Note that Rao's score test for the EPD model coincides with the goodness-of-fit test based on the
$\lambda$-th power skewness and kurtosis introduced by \citet{MR4547729}. Since the family of local
alternatives is the APD, which accommodates a wide range of asymmetry and tail thickness, this test can
be used as an omnibus procedure even though it is specifically designed for this APD alternative. This
helps explain why the power curves of the $T_n$ tests and of the GLRT/score-test benchmark are much
closer for the EPD case than for the gamma case, for which the family of alternatives is more restrictive.

\section{Illustrative application to temperature forecast errors}\label{sec:application}

We illustrate the use of the $T_n$ and LK goodness-of-fit tests based on trigonometric moments with a real dataset containing 48-hour-ahead surface temperature forecast errors (biases) from the MM5 numerical weather prediction model \citep{doi:10.18637/jss.v028.i03}. Here, MM5 refers to the fifth-generation Pennsylvania State University -- National Center for Atmospheric Research Mesoscale Model. The data consist of bias values, each representing the difference between the forecasted and observed surface air temperatures on January 3, 2000, at $96$ different land-based meteorological stations located across the US Pacific Northwest. A positive value indicates an overestimation by the model, while a negative value indicates an underestimation. This dataset, originally compiled by the University of Washington's Department of Atmospheric Sciences, is available in the \texttt{lawstat} package in $\textsf{R}$ under the name \texttt{bias}. The histogram of the $n = 96$ temperature forecast errors is shown in Figure~\ref{fig:forecast.error}.

\begin{figure}[ht]
\centering
\includegraphics[width=0.60\textwidth]{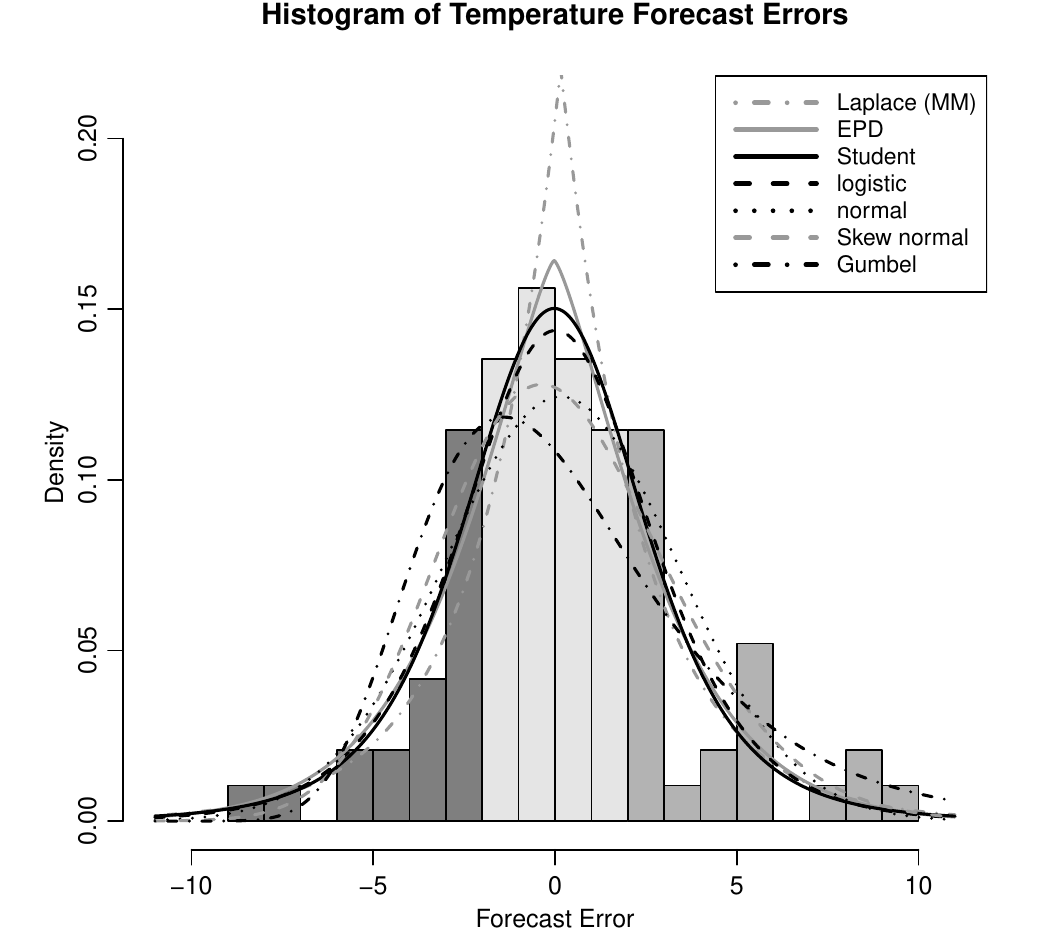}
\caption{Histogram of $n=96$ temperature forecast errors with seven fitted models.}\label{fig:forecast.error}
\end{figure}

Recall that the tests used in this article assume that the observations are i.i.d. The present dataset has a spatial component, which inevitably induces spatial dependence in the temperature field and may also create correlation among the forecast biases. Since detailed spatial information for the 96 locations is not available, we proceed under the working assumption that any such dependence is sufficiently weak not to invalidate the analysis.

Six null distributions supported on the real line, namely the EPD, normal, skew normal, Gumbel, logistic, and Student's $t$ distributions, were fitted to the data using maximum likelihood estimation. A seventh distribution, the Laplace distribution, was fitted using a method-of-moments estimator instead of maximum likelihood, given its gold-standard power performance in the comprehensive study of Section~\ref{sec:empirical.power.analysis.Laplace}. The goodness-of-fit was then assessed using the $T_n$ and LK tests, assuming that all parameters are unknown. For each fitted model, we report: (i) the parameter estimates; (ii) $-2\ell$ (twice the negative log-likelihood); (iii) the component $Z$-scores
\[
Z(C_n)=\sqrt{n}\,\frac{C_n(\hat{\bb{\theta}}_n)}{\sqrt{[\Sigma(\hat{\bb{\theta}}_n)]_{_{1,1}}}},
\qquad
Z(S_n)=\sqrt{n}\,\frac{S_n(\hat{\bb{\theta}}_n)}{\sqrt{[\Sigma(\hat{\bb{\theta}}_n)]_{_{2,2}}}},
\]
which provide measures of relative tail weight and central concentration, and of relative skewness, respectively; (iv) the test statistics $T_n$ and LK; (v) the corresponding $p$-values for $T_n$ and LK.

The dataset's near-symmetry, evident in the histogram, makes the normal distribution a natural first candidate. However, it is rejected at the 5\% significance level, with $p$-values of $0.027$ ($T_n$) and $0.031$ (LK). The result for the $T_n$ test can be interpreted through Figure~\ref{fig:C.n.vs.S.n} (left panel), where the point $\smash{[\sqrt{n}C_n(\hat{\bb{\theta}}_n),\sqrt{n}S_n(\hat{\bb{\theta}}_n)]^{\top}}$ is plotted along with the $95\%$-confidence ellipse for the normal model. The point lies just outside the $95\%$-confidence ellipse, reflecting the fact that the $p$-value is smaller than $0.05$.

\begin{figure}[!ht]
\centering
\includegraphics[width=0.325\textwidth]{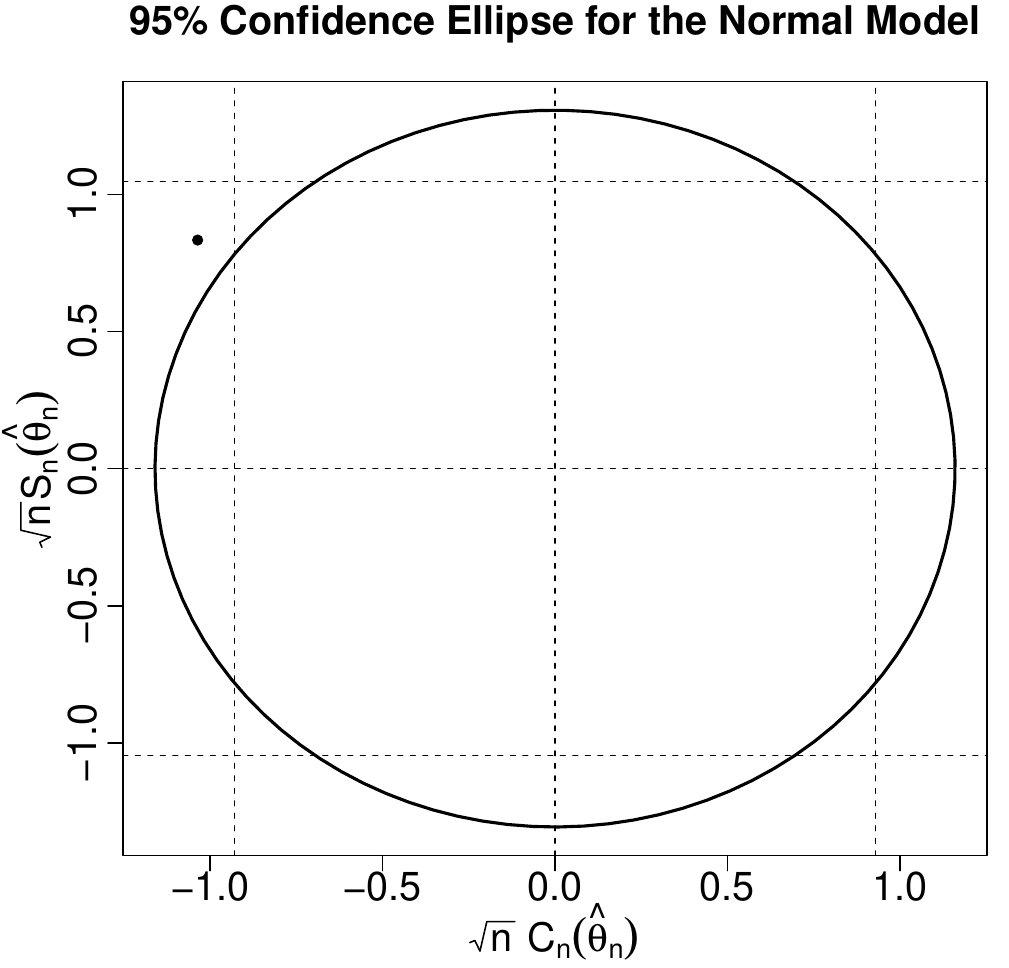}
\includegraphics[width=0.325\textwidth]{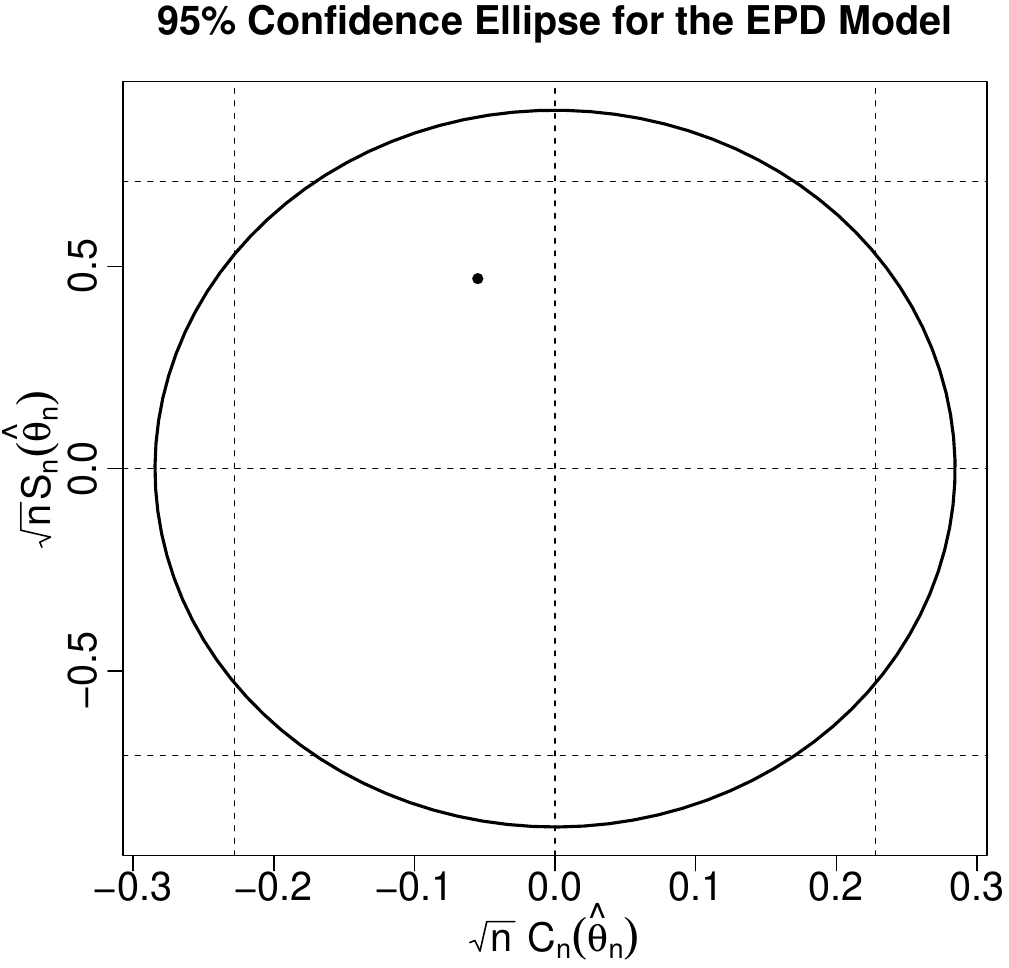}
\includegraphics[width=0.325\textwidth]{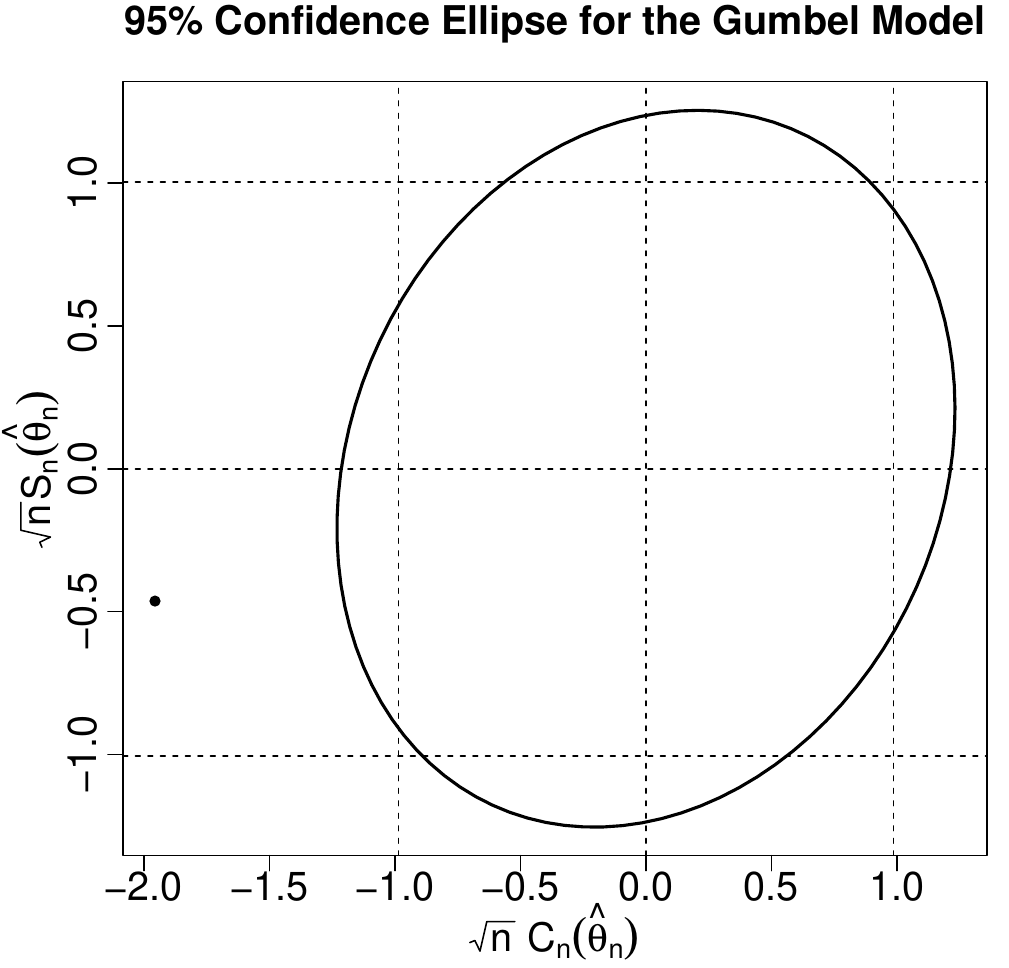}
\caption{The observed values of $\smash{[\sqrt{n} C_n(\hat{\bb{\theta}}_n), \sqrt{n} S_n(\hat{\bb{\theta}}_n)]^{\top}}$ with $95\%$-confidence ellipses for the normal, EPD, and Gumbel models. For each panel, the outer vertical dotted lines mark $\pm Z_{0.025} \sqrtsmash{[\Sigma(\hat{\bb{\theta}}_n)]_{_{1,1}}}$, while the horizontal lines mark $\pm Z_{0.025} \sqrtsmash{[\Sigma(\hat{\bb{\theta}}_n)]_{_{2,2}}}$, visually reproducing the $95\%$ univariate $Z$-score thresholds.}\label{fig:C.n.vs.S.n}
\end{figure}

A univariate analysis of each component shows that the $Z$-scores are $Z(C_n) = -2.19$ and $Z(S_n) = 1.56$, indicating that the dataset has significantly heavier tails than the normal distribution, since $Z(C_n) = -2.19$ is negative and more extreme than the left-tail critical value of $Z_{0.025}=-1.96$ under the standard normal distribution. Specifically, the number of observations in the central 50\% (the light gray bars in the histogram of Figure~\ref{fig:forecast.error}) is greater than expected under the normal model, while the number of observations in the outer quartiles (the bars with the two darkest shades of gray) is smaller. The value $Z(S_n) = 1.56$, on the other hand, is positive but smaller than $Z_{0.975}=1.96$, suggesting that the distribution is slightly more right-skewed than the normal model, but not significantly so. In summary, the rejection of the normal model is driven by a combination of heavier tails and mild right skewness in the dataset, with the dominant factor being the heavier tails. We observe that the SN model slightly improves the fit relative to the normal distribution by capturing the mild right skewness, with $\hat{\lambda}_n = 1.539 > 0$, but it does not improve the fit of the tail
thickness. As a result, the corresponding $p$-values, $0.089$ ($T_n$) and $0.045$ (LK), lie near conventional significance thresholds, leading to an ambiguous conclusion regarding the null hypothesis.

Given the need for heavier tails, we fit three additional candidates: $\mathrm{EPD}(\lambda,\mu,\sigma)$, $\mathrm{logistic}(\mu,\sigma)$, and Student's $t(\lambda,\mu,\sigma)$. Table~\ref{table:fits} shows that all three achieve nearly identical likelihoods ($-2\ell = 491.9,492.0,492.0$) and yield $p$-values ranging from $0.213$ to $0.509$, clearly leading to non-rejection of the null hypothesis. Because the logistic and Student's $t$ fits are very close to the EPD fit, Figure~\ref{fig:C.n.vs.S.n} (central panel) displays only the EPD confidence ellipse. For the EPD model, the $p$-value of $0.385$ for the $T_n$ test means the point $\smash{[\sqrt{n}C_n(\hat{\bb{\theta}}_n),\sqrt{n}S_n(\hat{\bb{\theta}}_n)]^{\top}}$ lies comfortably inside the $95\%$-confidence ellipse. The point also lies within both univariate confidence intervals; equivalently, the absolute values of the $Z$-scores, $Z(C_n)=-0.47$ and $Z(S_n)=1.30$, are less than $1.96$, indicating that the dataset exhibits slightly heavier tails than the EPD model and is slightly more right-skewed, although neither deviation is statistically significant.

\begin{table}[!ht]
\centering
\small
\setlength{\tabcolsep}{5pt}
\renewcommand{\arraystretch}{1}
\begin{tabular}{llccccccc}
 & & & & & \multicolumn{2}{c}{Test statistics} & \multicolumn{2}{c}{$p$-values} \\
\cmidrule(lr){8-9}
Distribution & Parameter estimates & $-2\ell$ & $Z(C_n)$ & $Z(S_n)$ & $T_n$ & LK & $T_n$ & LK \\
\midrule
$\mathrm{EPD}(\lambda,\mu, \sigma)$ & $(1.323, -0.024, 2.676)$ & 491.9 & -0.47 & 1.30 & 1.91 & 3.09 & 0.385 & 0.213\\
 $\mathrm{Laplace}(\mu, \sigma)$ & $(0.158, 2.269)$ & 495.2 & 1.46 & 0.99 & 3.12 & 2.90 & 0.210 & 0.235 \\
$\mathrm{Normal}(\mu, \sigma)$ & $(0.158, 3.209)$ & 496.3 & -2.19 & 1.56 & 7.22 & 6.94 & 0.027 & 0.031 \\
 $\mathrm{SN}(\lambda,\mu, \sigma)$ & $(1.539, -2.699,  4.296)$ &  493.8 & -2.01 &  1.09 &  4.85 & 6.21 &  0.089 & 0.045   \\
$\mathrm{Gumbel}(\mu, \sigma)$ & $(-1.395, 3.108)$ & 505.7 & -3.89 & -0.90 & 15.19 & 15.67 & 0.0005 & 0.0004\\
 $\mathrm{Logistic}(\mu, \sigma)$ & $(0.020, 1.739)$ & 492.0 & -0.70 & 1.24 & 2.03 & 2.14 & 0.362 & 0.343\\
 $\mathrm{Student's}\,t(\lambda,\mu, \sigma)$ & $(4.772, -0.020, 2.521)$ & 492.0 & -0.22 & 1.14 & 1.35 & 1.95 & 0.509 & 0.377 \\
\hline
\end{tabular}
\vspace{3mm}
\caption{Parameter estimates, $-2\ell$ (twice the minimized negative log-likelihood), the component $Z$-scores, observed test statistics $T_n$ and LK, and $p$-values for $T_n$ and LK.}
\label{table:fits}
\end{table}

The $\mathrm{Laplace}(\mu,\sigma)$ distribution, which is a special case of the EPD with $\lambda = 1$, is also considered because of its heavier tails compared to the normal distribution. Naturally, its fit cannot be as good as that of the EPD, which yields an estimated value of $\hat{\lambda}_n = 1.323$, albeit not far from $1$. Nonetheless, the Laplace distribution remains an acceptable model, with $p$-values of $0.210$ ($T_n$) and $0.235$ (LK), and has the advantage of being simpler. As can be seen in Figure~\ref{fig:forecast.error}, the Laplace model provides a less accurate fit near zero than the EPD, logistic, and Student's $t$ models.

The $\mathrm{Gumbel}(\mu,\sigma)$ distribution was also fitted and is decisively rejected, with $p$-values less than or equal~to~$0.0005$. The right panel in Figure~\ref{fig:C.n.vs.S.n} illustrates the Gumbel case: the point $\smash{[\sqrt{n}C_n(\hat{\bb{\theta}}_n),\sqrt{n}S_n(\hat{\bb{\theta}}_n)]^{\top}}$ lies well outside the $95\%$-confidence ellipse, in line with its very small $p$-values. The failure is driven by much heavier tails in the data than the Gumbel model allows; specifically, $Z(C_n) = -3.89 < -1.96 < 0$ indicates excess mass in the center of the dataset (light bars) and a deficit in the tails (dark bars).

\section{Summary and outlook}\label{sec:conclusion}

In this paper, we introduced a new omnibus goodness-of-fit test, denoted by $T_n$, based on trigonometric moments. The proposed test was designed as an alternative to the LK test of \citet{MR1146353}, fully exploiting the underlying covariance structure and ensuring a proper $\chi_2^2$ limiting distribution under the null hypothesis in the presence of nuisance parameters. We also proposed an alternative and straightforward method for computing the normalizing scalar appearing in the LK test.

To facilitate widespread adoption, we substantially expanded the scope of both the $T_n$ and LK tests by providing the necessary implementation details for $11$ distribution families, yielding $53$ distinct tests that covered most commonly used parametric distributions. We showed that the chi-square $\chi_2^2$ approximation employed to compute critical values and $p$-values remained very accurate in practice, even for relatively small sample sizes such as $n=30$, for both the $T_n$ and LK tests. As a result, these procedures offer broad applicability, comparable to that of classical EDF-based tests, combined with fully plug-and-play implementation that does not rely on Monte Carlo simulations or pre-tabulated values. To the best of our knowledge, this combination of features is unique. The strong power properties and practical utility of the proposed procedures were illustrated through two simulation studies, a theoretical analysis under local alternatives, and a real-data application involving meteorological forecast errors.

Several promising avenues for future research remain. A natural and important extension would be the development of a multivariate version of the proposed methodology, which could be based, for example, on a small number of functions from a tensor-product multivariate Fourier basis. Another valuable direction would be to adapt the approach to discrete or censored data, which are common in many applied fields but pose challenges for goodness-of-fit tests based on the probability integral transform. Finally, while incorporating higher-order trigonometric moments may yield increased power against more complex alternatives, such extensions would need to balance potential gains against the increased computational and analytical complexity associated with deriving the corresponding asymptotic covariance structures.

\section*{Conflict of interest}

The authors declare that they have no conflicts of interest.

\renewcommand{\refname}{References}
\setlength{\bibsep}{13pt plus 5ex}

\appendix

\renewcommand{\thetable}{\arabic{table}}

\begin{appendices}

\renewcommand{\thesection}{Appendix \Alph{section}}

\section{Reproducibility}\label{app:reproducibility}

\renewcommand{\thesection}{\Alph{section}}

The GitHub repository \citep{DesgagneOuimet2026github} provides the complete \textsf{R} implementation of the $T_n$ and LK goodness-of-fit tests for the 11 families of distributions listed in Tables~\ref{table:density.CDF} and~\ref{table:G.I}. It also includes the scripts used to generate all figures, perform the real-data analysis, and produce a numerical validation of all results reported in the paper.

\renewcommand{\thesection}{Appendix \Alph{section}}

\section{Estimators and constants required for the \texorpdfstring{$T_n$}{Tn} and LK test statistics}\label{app:table.constants}

\renewcommand{\thesection}{\Alph{section}}

Table~\ref{table:estimator} provides the score equations or, where available, explicit estimators for the maximum likelihood estimation of each parameter in the 11 distribution families, together with method-of-moments estimators for the EPD family.
Table~\ref{table:constants} provides the list of constants $h_{i}(\cdot)$ that must be computed numerically using the \textsf{R} software (for example, using the function \texttt{cubature::adaptIntegrate}) in order to evaluate the $T_n$ and LK test statistics. Here, $f_{\mathrm{ga}}(\cdot \nvert a,b)$ and $\Gamma_{a,b}(\cdot)$ denote the PDF and CDF of a gamma distribution with shape parameter $a$ and scale parameter $b$; $f_{\mathrm{be}}$ and $F_{\mathrm{be}}$ denote the PDF and CDF of a beta distribution, as defined in Table~\ref{table:density.CDF}; $f_{\mathrm{IG}}$ and $F_{\mathrm{IG}}$ denote the PDF and CDF of an inverse-Gaussian distribution, as defined in Table~\ref{table:density.CDF}; $\phi$ and $\Phi$ denote the PDF and CDF of the standard normal distribution; and $F_{\mathrm{SN}}$ denotes the CDF of the skew normal distribution, as defined in Table~\ref{table:density.CDF}.

\begin{table}[!b]
\small\def\arraystretch{2.3}% 1 is the default, change whatever you need
\begin{center}
\begin{tabular}{|p{1.1cm}|p{12.0cm}|}
\hline
$\vcenter{\hbox{\shortstack{EPD \\ $(\lambda,\mu,\sigma)$}}}$   & $\hat{\lambda}_n: \psi(1/\hat{\lambda}_n + 1)+\ln(\hat{\lambda}_n) - \frac{1}{n} \sum_{i=1}^n\big|\frac{x_i-\hat{\mu}_n}{\hat{\sigma}_n}\big|^{\hat{\lambda}_n} \ln\big(\big|\frac{x_i-\hat{\mu}_n}{\hat{\sigma}_n}\big|^{\hat{\lambda}_n}\big)+\Big(\frac{1}{n} \sum_{i=1}^n\big|\frac{x_i-\hat{\mu}_n}{\hat{\sigma}_n}\big|^{\hat{\lambda}_n}-1\Big)=0$,\rule{0pt}{11pt}\vspace{.5mm}\cr
\footnotesize ML-based &  $\hat{\mu}_n:\sum_{i=1}^n |x_i - \hat{\mu}_n|^{\hat{\lambda}_n - 1} \mathrm{sign}(x_i - \hat{\mu}_n) = 0$, ~ $\hat{\sigma}_n = \big(\frac{1}{n} \sum_{i=1}^n |x_i - \hat{\mu}_n|^{\hat{\lambda}_n}\big)^{1/\hat{\lambda}_n}$\vspace{.5mm} \cr
\hline
$\vcenter{\hbox{\shortstack{$\mathrm{EPD}_{\lambda}(\mu,\sigma)$ \\ \footnotesize MM-based}}}$  &  $\hat{\lambda}_n =\lambda_0$ known, ~ $\hat{\mu}_n = \frac{1}{n} \sum_{i=1}^n x_i$, ~ $\hat{\sigma}_n = \left\{\frac{\Gamma(1/\lambda)}{\lambda^{2/\lambda}\Gamma(3/\lambda)}\,\frac{1}{n} \sum_{i=1}^n (x_i - \hat{\mu}_n)^2\right\}^{1/2}$\rule{0pt}{11pt}\vspace{1mm} \cr
\hline
$\vcenter{\hbox{\shortstack{Half-EPD \\ $(\lambda,\sigma)$}}}$  &  $\hat{\lambda}_n:\psi(1/\hat{\lambda}_n + 1)+\ln(\hat{\lambda}_n) - \frac{1}{n} \sum_{i=1}^n(x_i/\hat{\sigma}_n)^{\hat{\lambda}_n} \ln\big\{(x_i/\hat{\sigma}_n)^{\hat{\lambda}_n}\big\}+\frac{1}{n} \sum_{i=1}^n(x_i/\hat{\sigma}_n)^{\hat{\lambda}_n}-1=0 $,
~ $\hat{\sigma}_n = \big(\frac{1}{n} \sum_{i=1}^n x_i ^{\hat{\lambda}_n}\big)^{1/\hat{\lambda}_n}$\rule{0pt}{11pt}\vspace{1mm} \cr
\hline
$\mathrm{SN}(\lambda,\mu,\sigma)$ & $\hat{\lambda}_n: \sum_{i=1}^n \hat{y}_i H(\hat{\lambda}_n \hat{y}_i) = 0$, ~$\hat{\mu}_n: \sum_{i=1}^n \hat{y}_i - \hat{\lambda}_n \sum_{i=1}^n H(\hat{\lambda}_n \hat{y}_i) =0$, ~ $\hat{\sigma}_n: -n + \sum_{i=1}^n \hat{y}_i^2 - \hat{\lambda}_n\sum_{i=1}^n \hat{y}_i H(\hat{\lambda}_n \hat{y}_i) = 0$, \vspace{.5mm}\cr
 & with $\hat{\sigma}_n=\sqrt{\frac{1}{n}\sum_{i=1}^n (x_i-\hat{\mu}_n)^2}$ if $\lambda_0$ is jointly estimated, where $\hat{y}_i=(x_i-\hat{\mu}_n)/\hat{\sigma}_n$ and  $H(t)=\phi(t)/\Phi(t)$ \rule{0pt}{11pt} \vspace{0.5mm}\cr
\hline
$\mathrm{GG}(\lambda,\beta, \rho)$  & $\hat{\lambda}_n:\frac{1}{n}\sum_{i=1}^{n}\ln(x_i^{\hat{\rho}_n})-\ln(\hat{\beta}_n^{\hat{\rho}_n}) -\psi(\hat{\lambda}_n) =0$, ~  $\hat{\beta}_n = \left(\frac{1}{\hat{\lambda}_n}\frac{1}{n}\sum_{i=1}^n x_i^{\hat{\rho}_n}\right)^{1/\hat{\rho}_n}$, \rule{0pt}{11pt}\vspace{.5mm}\cr
& $\hat{\rho}_n: \frac{\sum_{i=1}^{n}x_i^{\hat{\rho}_n}\ln(x_i)}{n \hat{\beta}_n^{\hat{\rho}_n}} - \hat{\lambda}_n\frac{1}{n}\sum_{i=1}^{n}\ln(x_i)- 1/\hat{\rho}_n-\ln(\hat{\beta}_n)\left(\frac{\sum_{i=1}^{n}x_i^{\hat{\rho}_n}}{n\hat{\beta}_n^{\hat{\rho}_n}} - \hat{\lambda}_n\right) =0$ \rule{0pt}{12pt}\vspace{.5mm}\cr
\hline
$\vcenter{\hbox{\shortstack{Logistic \\ $(\mu,\sigma)$}}}$ & $\hat{\mu}_n:\frac{1}{n} \sum_{i=1}^n \!\frac{2}{1 + e^{\hat{y}_i}} = 1$,  ~  $\hat{\sigma}_n:\frac{1}{n} \sum_{i=1}^n\! \hat{y}_i -\frac{1}{n} \sum_{i=1}^n \!\frac{2\hat{y}_i}{1+e^{\hat{y}_i}}=1$, ~ with $\hat{y}_i = \frac{x_i - \hat{\mu}_n}{\hat{\sigma}_n}$ \rule{0pt}{12pt}\vspace{1mm} \cr
\hline
$\vcenter{\hbox{\shortstack{Student's \\ $t(\lambda,\mu,\sigma)$}}}$  &  $\hat{\lambda}_n:\psi\left( \frac{\hat{\lambda}_n + 1}{2} \right)
- \psi\left( \frac{\hat{\lambda}_n}{2} \right)
- \frac{1}{n}\sum_{i=1}^n\ln\left(1 + \frac{\hat{y}_i^2}{\hat{\lambda}_n} \right)
+ \frac{1}{\hat{\lambda}_n}\left(\frac{\hat{\lambda}_n + 1}{\hat{\lambda}_n}\frac{1}{n}\sum_{i=1}^n\frac{\hat{y}_i^2}{1 + \hat{y}_i^2/\hat{\lambda}_n}-1\right)=0$,\rule{0pt}{11pt}\vspace{.5mm} \cr
& $\hat{\mu}_n:\sum_{i=1}^n \frac{\hat{y}_i}{1 + \hat{y}_i^2/\hat{\lambda}_n} = 0, \quad \hat{\sigma}_n: \frac{\hat{\lambda}_n + 1}{\hat{\lambda}_n} \, \frac{1}{n} \sum_{i=1}^n \frac{\hat{y}_i^2}{1 + \hat{y}_i^2/\hat{\lambda}_n} = 1$, ~  with $\hat{y}_i = \frac{x_i - \hat{\mu}_n}{\hat{\sigma}_n}$ \rule{0pt}{11pt}\vspace{.5mm} \cr
\hline
$\vcenter{\hbox{\shortstack{Gompertz \\ $(\beta, \rho)$}}}$ &  $\hat{\beta}_n : 1 / \hat{\beta}_n + \frac{1}{n}\sum_{i=1}^n x_i - \hat{\rho}_n\frac{1}{n}\sum_{i=1}^n x_i e^{\hat{\beta}_n x_i}=0$, ~ $\hat{\rho}_n = \left(\frac{1}{n}\sum_{i=1}^n e^{\hat{\beta}_n x_i} - 1\right)^{-1}$\rule{0pt}{11pt} \vspace{1mm} \cr
\hline
$\vcenter{\hbox{\shortstack{Lomax \\ $(\alpha, \sigma)$}}}$ & $\hat{\alpha}_n = \left\{\frac{1}{n} \sum_{i=1}^n \ln(1 + x_i/\hat{\sigma}_n)\right\}^{-1}$, ~ $\hat{\sigma}_n: \left(1 + \hat{\alpha}_n\right) \left(\frac{1}{n} \sum_{i=1}^n \frac{x_i}{1 + x_i/\hat{\sigma}_n}\right) - \hat{\sigma}_n = 0$\rule{0pt}{11pt} \vspace{.5mm} \cr
\hline
$\mathrm{IG}(\mu, \lambda)$ &  $\hat{\mu}_n= \frac{1}{n} \sum_{i=1}^n x_i$, ~ $\hat{\lambda}_n = \left\{\frac{1}{n} \sum_{i=1}^n \frac{1}{x_i} -\frac{1}{\hat{\mu}_n}\big(2-\frac{1}{\hat{\mu}_n}\frac{1}{n} \sum_{i=1}^n x_i\big)\right\}^{-1}$ \rule{0pt}{11pt}\vspace{.5mm}\cr
\hline
$\mathrm{Beta}(\alpha, \beta)$ &  $\hat{\alpha}_n: \psi(\hat{\alpha}_n + \hat{\beta}_n) - \psi(\hat{\alpha}_n) + \frac{1}{n} \sum_{i=1}^n \ln (x_i) = 0$, ~
$\hat{\beta}_n: \psi(\hat{\alpha}_n + \hat{\beta}_n) - \psi(\hat{\beta}_n) + \frac{1}{n} \sum_{i=1}^n \ln (1 - x_i) = 0$\vspace{0.5mm} \cr
\hline
$\mathrm{Kum}(\alpha, \beta)$ & $\hat{\alpha}_n: 1 - \hat{\beta}_n + \frac{1 + \frac{1}{n}\sum_{i=1}^n \ln (x_i^{\hat{\alpha}_n})}{\frac{1}{n}\sum_{i=1}^n x_i^{\hat{\alpha}_n}(1 - x_i^{\hat{\alpha}_n})^{-1} \ln (x_i^{\hat{\alpha}_n})} = 0$, ~  $\hat{\beta}_n = \frac{-n}{\sum_{i=1}^n \ln(1 - x_i^{\hat{\alpha}_n})}$ \rule{0pt}{11pt}\vspace{.5mm} \cr
\hline

\end{tabular}
\end{center}
\vspace{-2mm}
\caption{Score equations, or (where available) explicit parameter estimators, for maximum likelihood estimation of each parameter in the 11 distribution families, together with method-of-moments estimators for the EPD family.}
\label{table:estimator}
\end{table}

\begin{table}[!htbp]
\centering
\footnotesize\def\arraystretch{2.0}% 1 is the default, change whatever you need
\begin{multicols}{2}
\begin{tabular}{|l|}
\hline
$h_{1}(\lambda) = \int_0^{\infty} \cos\left[\pi \left\{1 + \Gamma_{1/\lambda, 1}(v)\right\}\right] f_{\mathrm{ga}}(v \nvert 1/\lambda+1,1) \rd v$ \cr
\hline
$h_{2}(\lambda) = \int_0^{\infty} \sin\left[\pi \left\{1 + \Gamma_{1/\lambda, 1}(v)\right\}\right] f_{\mathrm{ga}}(v \nvert 1,1) \rd v$ \cr
\hline
$h_{3}(\lambda) = \int_0^{\infty} \cos\left[\pi \left\{1 + \Gamma_{1/\lambda, 1}(v)\right\}\right]\ln(\lambda v) f_{\mathrm{ga}}(v \nvert 1/\lambda+1, 1) \rd v$ \cr
\hline
$h_{4}(\lambda) = \int_{0}^{\infty} \cos\left[\pi \left\{1 + \Gamma_{1/\lambda, 1}(v)\right\}\right] f_{\mathrm{ga}}(v \nvert 3/\lambda,1) \rd v$ \cr
\hline
$h_{5}(\lambda) = \int_{0}^{\infty} \sin\left[\pi \left\{1 + \Gamma_{1/\lambda, 1}(v)\right\}\right] f_{\mathrm{ga}}(v \nvert 2/\lambda,1) \rd v$ \cr
\hline
$h_{6}(a,b,c) = \int_0^{\infty} \cos\left\{2\pi \Gamma_{a, 1}(v)\right\} f_{\mathrm{ga}}(v \nvert b,c) \rd v$ \cr
\hline
$h_{7}(a,b,c) = \int_0^{\infty} \sin\left\{2\pi \Gamma_{a, 1}(v)\right\} f_{\mathrm{ga}}(v \nvert b,c) \rd v$ \cr
\hline
$h_{8}(\lambda) = \int_0^{\infty} (v - \lambda) \ln (v) \cos\left\{2\pi \Gamma_{\lambda, 1}(v)\right\} f_{\mathrm{ga}}(v \nvert \lambda, 1) \rd v$ \cr
\hline
$h_{9}(\lambda) = \int_0^{\infty} (v - \lambda) \ln (v) \sin\left\{2\pi \Gamma_{\lambda, 1}(v)\right\} f_{\mathrm{ga}}(v \nvert \lambda, 1) \rd v$ \cr
\hline
$h_{10}(\alpha) = \int_0^{\infty} \ln (v) \cos\left\{2\pi \Gamma_{\alpha, 1}(v)\right\} f_{\mathrm{ga}}(v \nvert \alpha, 1) \rd v$ \cr
\hline
$h_{11}(\alpha) = \int_0^{\infty} \ln (v) \sin\left\{2\pi \Gamma_{\alpha, 1}(v)\right\} f_{\mathrm{ga}}(v \nvert \alpha, 1) \rd v$ \cr
\hline
$h_{12}(\lambda) = \int_0^1 \cos\left[\pi \left\{2-F_{\mathrm{be}}(v \nvert \lambda/2, 1/2)\right\}\right] f_{\mathrm{be}}(v\nvert \lambda/2, 3/2) \rd v$ \cr
\hline
$h_{13}(\lambda) = \int_0^1 \sin\left[\pi \left\{2-F_{\mathrm{be}}(v \nvert \lambda/2, 1/2)\right\}\right] f_{\mathrm{be}}\big(v\nvert \frac{\lambda+1}{2}, 1\big) \rd v$ \cr
\hline
$h_{14}(\lambda) = \int_0^{1} \cos\left[\pi \left\{2 - F_{\mathrm{be}}(v\nvert\lambda/2, 1/2) \right\}\right]$\cr
\hspace{15mm}$\big\{\ln(v) + \frac{\lambda + 1}{\lambda}(1-v)\big\} f_{\mathrm{be}}(v\nvert \lambda/2, 1) \rd v$ \cr
\hline
$h_{15}(\lambda) = \int_0^{\infty} \cos\left[2\pi \Gamma_{1/\lambda, 1}(v)\right]\ln(\lambda v) f_{\mathrm{ga}}(v \nvert 1/\lambda+1, 1) \rd v$ \cr
\hline
$h_{16}(\lambda) = \int_0^{\infty} \sin\left[2\pi \Gamma_{1/\lambda, 1}(v)\right]\ln(\lambda v) f_{\mathrm{ga}}(v \nvert 1/\lambda+1, 1) \rd v$ \cr
$h_{17}(\rho) = \int_1^{\infty} \{\ln (v)\}^2 v e^{-\rho v} \rd v$ \cr
\hline
$h_{18}(\rho) = \int_1^{\infty} \ln (v) v e^{-\rho v} \rd v$ \cr
\hline
$h_{19}(\rho) = \int_1^{\infty} \cos\big\{2\pi (1 - e^{-\rho(v-1)})\big\} \ln (v) (1 - \rho v) e^{-\rho v} \rd v$ \cr
\hline
$h_{20}(\rho) = \int_1^{\infty} \sin\big\{2\pi (1 - e^{-\rho(v-1)})\big\} \ln (v) (1 - \rho v) e^{-\rho v} \rd v$ \cr
\hline
$h_{21}(\rho) = \int_1^{\infty} \cos\big\{2\pi (1 - e^{-\rho(v-1)})\big\} v e^{-\rho v} \rd v$ \cr
\hline
$h_{22}(\rho) = \int_1^{\infty} \sin\big\{2\pi (1 - e^{-\rho(v-1)})\big\} v e^{-\rho v} \rd v$ \cr
\hline
\end{tabular}

\hspace{-8mm}
\begin{tabular}{|l|}
\hline
$h_{23}(\alpha, \beta) = \int_0^1 \ln (v)\cos\{2\pi F_{\mathrm{be}}(v \nvert \alpha, \beta)\}f_{\mathrm{be}}(v \nvert \alpha, \beta) \rd v$ \cr
\hline
$h_{24}(\alpha, \beta) = \int_0^1 \ln (v)\sin\{2\pi F_{\mathrm{be}}(v \nvert \alpha, \beta)\} f_{\mathrm{be}}(v \nvert \alpha, \beta) \rd v$ \cr
\hline
$h_{25}(\alpha, \beta) = \int_0^1 \ln (1 - v)\cos\{2\pi F_{\mathrm{be}}(v \nvert \alpha, \beta)\} f_{\mathrm{be}}(v \nvert \alpha, \beta) \rd v$ \cr
\hline
$h_{26}(\alpha, \beta) = \int_0^1 \ln (1 - v)\sin\{2\pi F_{\mathrm{be}}(v \nvert \alpha, \beta)\}f_{\mathrm{be}}(v \nvert \alpha, \beta) \rd v$ \cr
\hline
$h_{27}(\mu, \lambda) = \int_0^{\infty} v \cos\left\{2\pi F_{\mathrm{IG}}(v \nvert \mu, \lambda)\right\} f_{\mathrm{IG}}(v \nvert \mu, \lambda) \rd v$ \cr
\hline
$h_{28}(\mu, \lambda) = \int_0^{\infty} v \sin\left\{2\pi F_{\mathrm{IG}}(v \nvert \mu, \lambda)\right\} f_{\mathrm{IG}}(v \nvert \mu, \lambda) \rd v$ \cr
\hline
$h_{29}(\mu, \lambda) = \int_0^{\infty} \frac{(v^2 + \mu^2)}{v} \cos\left\{2\pi F_{\mathrm{IG}}(v \nvert \mu, \lambda)\right\} f_{\mathrm{IG}}(v \nvert \mu, \lambda) \rd v$ \cr
\hline
$h_{30}(\mu, \lambda) = \int_0^{\infty} \frac{(v^2 + \mu^2)}{v} \sin\left\{2\pi F_{\mathrm{IG}}(v \nvert \mu, \lambda)\right\} f_{\mathrm{IG}}(v \nvert \mu, \lambda) \rd v$ \cr
\hline
$h_{31}(\beta) = \int_0^1 \cos\big[2\pi \big\{1 - (1 - v)^{\beta}\big\}\big] \ln (v)(1 - v)^{\beta-2}(1-\beta v)\rd v$ \cr
\hline
$h_{32}(\beta) = \int_0^1 \sin\big[2\pi \big\{1 - (1 - v)^{\beta}\big\}\big] \ln (v)(1 - v)^{\beta-2}(1-\beta v)\rd v$ \cr
\hline
$h_{33}(\beta) = \int_0^1 \cos\big[2\pi \big\{1 - (1 - v)^{\beta}\big\}\big] \ln(1 - v) (1 - v)^{\beta - 1} \rd v$ \cr
\hline
$h_{34}(\beta) = \int_0^1 \sin\big[2\pi \big\{1 - (1 - v)^{\beta}\big\}\big] \ln(1 - v) (1 - v)^{\beta - 1} \rd v$ \cr
\hline
$h_{35}(\lambda) = \int_0^{\infty} \sin\left[\pi \left\{1 + \Gamma_{1/\lambda, 1}(v)\right\}\right]f_{\mathrm{ga}}(v \nvert 1/\lambda+1, 1) \rd v$ \cr
\hline
$h_{36}(\lambda) = \int_{-\infty}^{\infty}\phi(v)\phi^2(\lambda v)\{\Phi(\lambda v)\}^{-1} \rd v$ \cr
\hline
$h_{37}(\lambda) = \int_{-\infty}^{\infty}v\phi(v)\phi^2(\lambda v)\{\Phi(\lambda v)\}^{-1} \rd v$ \cr
\hline
$h_{38}(\lambda) = \int_{-\infty}^{\infty}v^2\phi(v)\phi^2(\lambda v)\{\Phi(\lambda v)\}^{-1} \rd v$ \cr
\hline
$h_{39}(\lambda)=\int_{-\infty}^{\infty} v \cos\left\{2\pi F_{\mathrm{SN}}(v \nvert \lambda, 0, 1)\right\}\phi(\lambda v)\phi(v)\rd v$ \cr
\hline
$h_{40}(\lambda)=\int_{-\infty}^{\infty} v \sin\left\{2\pi F_{\mathrm{SN}}(v \nvert \lambda, 0, 1)\right\}\phi(\lambda v)\phi(v)\rd v$ \cr
\hline
$h_{41}(\lambda)= \int_{-\infty}^{\infty} \cos\left\{2\pi F_{\mathrm{SN}}(v \nvert \lambda, 0, 1)\right\} \left\{v \Phi(\lambda v)-\lambda \phi(\lambda v)\right\}\phi(v)\rd v$ \cr
\hline
$h_{42}(\lambda)=\int_{-\infty}^{\infty} \sin\left\{2\pi F_{\mathrm{SN}}(v \nvert \lambda, 0, 1)\right\} \left\{v \Phi(\lambda v)-\lambda \phi(\lambda v)\right\}\phi(v)\rd v$ \cr
\hline
$h_{43}(\lambda)= \int_{-\infty}^{\infty} v\cos\left\{2\pi F_{\mathrm{SN}}(v \nvert \lambda, 0, 1)\right\} \left\{v \Phi(\lambda v)-\lambda \phi(\lambda v)\right\}\phi(v)\rd v$ \cr
\hline
$h_{44}(\lambda)=\int_{-\infty}^{\infty} v\sin\left\{2\pi F_{\mathrm{SN}}(v \nvert \lambda, 0, 1)\right\} \left\{v \Phi(\lambda v)-\lambda \phi(\lambda v)\right\}\phi(v)\rd v$ \cr
\hline
\end{tabular}
\end{multicols}
\caption{List of constants to be computed numerically.}
\label{table:constants}
\end{table}

\end{appendices}

\newpage
\begin{center}
\bf \LARGE Supplementary material
\end{center}

\phantomsection
\makeatletter
\def\@currentlabel{Supplementary material}
\makeatother
\label{supp}

The supplement \citep{supp} collects the proofs of the results stated herein, namely: (i) those supporting the implementation details reported in Tables~\ref{table:G.I} and \ref{table:estimator}; and (ii) the proofs of Propositions~\ref{prop:gamma.local.alternative} and~\ref{prop:EPD.local.alternative}, which establish the asymptotic behavior of the gamma and EPD tests under local alternatives in Section~\ref{sec:asymp.local.alternatives}, as well as the corresponding power curves.

\section{Proof for the \texorpdfstring{$\mathrm{EPD}(\lambda,\mu,\sigma)$}{EPD} tests in the ML- and MM-based cases}\label{proof:EPD}

Letting $y = (x-\mu)/\sigma$, straightforward calculations yield
\[
\bb{s}(x, \lambda, \mu, \sigma)
=
\begin{bmatrix}
s_1(x, \lambda, \mu, \sigma) \\[1mm]
s_2(x, \lambda, \mu, \sigma) \\[1mm]
s_3(x, \lambda, \mu, \sigma)
\end{bmatrix}
=
\begin{bmatrix}
\partial_{\lambda} \ln\{f(x \nvert \lambda, \mu, \sigma)\} \\[1mm]
\partial_{\mu} \ln\{f(x \nvert \lambda, \mu, \sigma)\} \\[1mm]
\partial_{\sigma} \ln\{f(x \nvert \lambda, \mu, \sigma)\}
\end{bmatrix}
=
\begin{bmatrix}
 \frac{1}{\lambda^2}\left\{|y|^{\lambda}- |y|^{\lambda} \ln(|y|^\lambda) + C_{\lambda} - 1 \right\}\\[1mm]
\frac{1}{\sigma}\{|y|^{\lambda - 1} \mathrm{sign}(y)\} \\[1mm]
\frac{1}{\sigma}\{|y|^{\lambda} - 1\}
\end{bmatrix},
\]
where
\[
C_{\lambda} = \psi(1/\lambda+1) + \ln(\lambda).
\]
The ML estimators for $\lambda_0$, $\mu_0$ and/or $\sigma_0$ come from the equation $\sum_{i=1}^{n}\bb{s}(x_i, \lambda_0, \mu_0, \sigma_0)=0$.

Here are some equations that will be useful in the following sections. One has, for $\lambda\in (0, \infty)$,
\begin{align*}
X\sim \mathrm{EPD}(\lambda,\mu, \sigma)
~~ & \Leftrightarrow ~~ Y = \frac{X - \mu}{\sigma} \sim \mathrm{EPD}(\lambda, 0, 1)\\
\Leftrightarrow |Y| \sim \mathrm{half\mhyphen EPD}(\lambda, 1) ~~ & \Leftrightarrow ~~ V = \frac{1}{\lambda}|Y|^{\lambda} \sim \mathrm{gamma}(1/\lambda, 1).
\end{align*}

If $V\sim \mathrm{gamma}(1/\lambda, 1)$ and $k\in (-1/\lambda, \infty)$, it can also be verified in \texttt{Mathematica} that
\[
\EE\left\{V^k(\ln V)^j\right\}
= \frac{\Gamma(1/\lambda + k)}{\Gamma(1/\lambda)} \times
\begin{cases}
1, &\mbox{if } j = 0, \\
\psi(1/\lambda + k), &\mbox{if } j = 1, \\
\psi_1(1/\lambda + k) + \psi^2(1/\lambda + k), &\mbox{if } j = 2.
\end{cases}
\]

It follows that, for $k\in (-1, \infty)$,
\begin{align}
\EE\left\{|Y|^k(\ln |Y|)^j\right\}
&= \frac{\lambda^{k/\lambda-j}\Gamma(1/\lambda + k/\lambda)}{\Gamma(1/\lambda)} \nonumber\\
&\times
\begin{cases}
1, &\mbox{if } j = 0, \\
\psi(1/\lambda + k/\lambda)+\ln(\lambda), &\mbox{if } j = 1, \\
\psi_1(1/\lambda + k/\lambda) + \{\psi(1/\lambda + k/\lambda)+\ln(\lambda)\}^2, &\mbox{if } j = 2,
\end{cases}
\end{align}\label{eqn.EPD}
regardless of whether the expectation is taken with respect to $X\sim \mathrm{EPD}(\lambda,\mu, \sigma)$,
$Y \sim \mathrm{EPD}(\lambda, 0, 1)$ or $|Y| \sim \mathrm{half\mhyphen EPD}(\lambda, 1)$ if $Y= (X - \mu)/\sigma$.

For $z\in (0, \infty)$, it is well-known that
\[
\Gamma(1 + z) = z\Gamma(z), \quad
\psi(1 + z) = \psi(z) + 1/z, \quad
\psi_1(1 + z) = \psi_1(z) - 1/z^2.
\]

Next, one determines the Fisher information matrix $I(\lambda, \mu, \sigma)=\EE\big\{\bb{s}(X, \lambda, \mu, \sigma) \bb{s}(X, \lambda, \mu, \sigma)^{\top}\big\}$. One has, given that $\EE\{\bb{s}(X, \lambda, \mu, \sigma)\}= \bb{0}_{3}$,
\begin{align*}
\EE\left\{s_1^2(X, \lambda, \mu, \sigma)\right\}
&= \frac{1}{\lambda^4}\Big((C_{\lambda} - 1)^2 -2 (C_{\lambda} - 1)\EE\left\{|Y|^{\lambda}\ln(|Y|^\lambda)\right\}+2 (C_{\lambda} - 1)\EE\left\{|Y|^{\lambda}\right\}
\\
& + \lambda^2\EE\{|Y|^{2\lambda}(\ln(|Y|))^2\}+ \EE\{|Y|^{2\lambda}\}-2\lambda\EE\{|Y|^{2\lambda}\ln(|Y|)\}\Big)\\
& = \frac{1}{\lambda^4}\Big(-(C_{\lambda} - 1)^2 + \lambda^2 \frac{\Gamma(1/\lambda + 2)}{\Gamma(1/\lambda)}
 [\psi_1(1/\lambda + 2) + \{\psi(1/\lambda + 2)+\ln(\lambda)\}^2]\\
& + \frac{\lambda^2\Gamma(1/\lambda + 2)}{\Gamma(1/\lambda)} -2\lambda\frac{\lambda\Gamma(1/\lambda + 2)}{\Gamma(1/\lambda)}
\{\psi(1/\lambda + 2)+\ln(\lambda)\}\Big)\\
& = \frac{1}{\lambda^4}\Big(-(C_{\lambda} - 1)^2 + (\lambda+1)
 [\psi_1(1/\lambda + 1) - (1/\lambda+1)^{-2}+ \{\psi(1/\lambda + 1) \\
 & +\ln(\lambda)+ (1/\lambda+1)^{-1} \}^2] + \lambda+1 -2(\lambda+1)\{\psi(1/\lambda + 1)+\ln(\lambda)+ (1/\lambda+1)^{-1}\}\Big)\\
 &=\frac{1}{\lambda^3}\big\{(1/\lambda+1)\psi_1(1/\lambda + 1)+C_{\lambda}^2 -1\big\},
\end{align*}
\[
\EE\left\{s_2^2(X, \lambda, \mu, \sigma)\right\}
= \frac{1}{\sigma^2} \EE\left(|Y|^{2\lambda - 2}\right)
= \frac{\lambda^{2 - 2/\lambda}\Gamma(2-1/\lambda)}{\sigma^2\Gamma(1/\lambda)},
\]
\[
\EE\left\{s_3^2(X, \lambda, \mu, \sigma)\right\}
= \frac{1}{\sigma^2} \left\{\EE\left(|Y|^{2\lambda}\right) - 2\EE\left(|Y|^{\lambda}\right) + 1\right\}
= \frac{1}{\sigma^2} \left\{\frac{\lambda^2\Gamma(1/\lambda+2)}{\Gamma(1/\lambda)} - 2 + 1\right\}
= \frac{\lambda}{\sigma^2},
\]
\begin{align*}
 \EE\{s_1(X, \lambda, &\mu, \sigma) s_3(X, \lambda, \mu, \sigma)\} =
 \frac{-1}{\sigma\lambda^2}\EE\left\{\left( |Y|^{\lambda}\ln(|Y|^\lambda) - |Y|^{\lambda} \right) (|Y|^{\lambda} - 1)\right\}\\
 &=\frac{-1}{\sigma\lambda^2}\Big(\lambda\EE\left\{|Y|^{2\lambda}\ln(|Y|)\right\}-\lambda\EE\left\{|Y|^{\lambda}\ln(|Y|)\right\}
 -\EE\left\{|Y|^{2\lambda}\right\} + \EE\left\{|Y|^{\lambda}\right\}\Big)\\
 &=\frac{-1}{\sigma\lambda^2}\Big(\lambda\frac{\lambda\Gamma(1/\lambda + 2)}{\Gamma(1/\lambda)}\{\psi(1/\lambda + 2)+\ln(\lambda)\}
 -\lambda\frac{\Gamma(1/\lambda + 1)}{\Gamma(1/\lambda)}\{\psi(1/\lambda + 1)+\ln(\lambda)\}\\
 &-\frac{\lambda^2\Gamma(1/\lambda + 2)}{\Gamma(1/\lambda)}
 +\frac{\lambda\Gamma(1/\lambda + 1)}{\Gamma(1/\lambda)} \Big)\\
 &=\frac{-1}{\sigma\lambda^2}\Big((\lambda + 1)\{\psi(1/\lambda + 1)+(1/\lambda+1)^{-1}+\ln(\lambda)\}-\psi(1/\lambda + 1)-\ln(\lambda)-\lambda - 1 +1 \Big)\\
 &= \frac{-1}{\sigma\lambda}\{\psi(1/\lambda+1) + \ln(\lambda)\}= \frac{-C_{\lambda}}{\sigma\lambda}
\end{align*}
and
\begin{align*}
\EE\left\{s_2(X, \lambda, \mu, \sigma) s_3(X, \lambda, \mu, \sigma)\right\}& = \frac{1}{\sigma^2} \EE\left\{s_2(Y, \lambda, 0, 1) s_3(Y, \lambda, 0, 1)\right\} = 0,\\
\EE\left\{s_1(X, \lambda, \mu, \sigma) s_2(X, \lambda, \mu, \sigma)\right\}& = \frac{1}{\sigma} \EE\left\{s_1(Y, \lambda, 0, 1) s_2(Y, \lambda, 0, 1)\right\} = 0
\end{align*}
since $s_2(y, \lambda, 0, 1)$ is an odd function and $s_1(y, \lambda, 0, 1)$, $s_3(y, \lambda, 0, 1)$, $f(y \nvert \lambda, 0, 1)$ are even functions.

One has
\[
I(\lambda, \mu, \sigma) =
\begin{bmatrix}
\frac{1}{\lambda^3}\big\{(1/\lambda+1)\psi_1(1/\lambda + 1)+C_{\lambda}^2 -1\big\} & 0 & \frac{-C_{\lambda}}{\sigma\lambda} \\[1mm]
0 & \frac{\lambda^{2 - 2/\lambda}\Gamma(2-1/\lambda)}{\sigma^2\Gamma(1/\lambda)} & 0 \\[1mm]
\frac{-C_{\lambda}}{\sigma\lambda} & 0 & \frac{\lambda}{\sigma^2}
\end{bmatrix}.
\]

Next, one determines the matrix $G(\lambda, \mu, \sigma)=\EE\big\{\bb{\tau}(X, \lambda, \mu, \sigma) \bb{s}(X, \lambda, \mu, \sigma)^{\top}\big\}$. One has, if $y = (x - \mu)/\sigma$,
\begin{align*}
\bb{\tau}(x, \lambda, \mu, \sigma)
&=
\begin{bmatrix}
\tau_1(x, \lambda, \mu, \sigma) \\[1mm]
\tau_2(x, \lambda, \mu, \sigma)
\end{bmatrix}
=
\begin{bmatrix}
\cos\left\{2\pi F(x \nvert \lambda, \mu, \sigma)\right\} \\[1mm]
\sin\left\{2\pi F(x \nvert \lambda, \mu, \sigma)\right\}
\end{bmatrix}
=
\begin{bmatrix}
\cos\left[\pi \left\{1 + \mathrm{sign}(y) \Gamma_{1/\lambda, 1}(|y|^{\lambda}/\lambda)\right\}\right] \\[1mm]
\sin\left[\pi \left\{1 + \mathrm{sign}(y) \Gamma_{1/\lambda, 1}(|y|^{\lambda}/\lambda)\right\}\right]
\end{bmatrix}\\
&=\begin{bmatrix}
\cos\left[\pi \left\{1 + \Gamma_{1/\lambda, 1}(|y|^{\lambda}/\lambda)\right\}\right] \\[1mm]
\mathrm{sign}(y)\sin\left[\pi \left\{1 + \Gamma_{1/\lambda, 1}(|y|^{\lambda}/\lambda)\right\}\right]
\end{bmatrix},
\end{align*}
since $\cos\{\pi(1-z)\}=\cos\{\pi(1+z)\}$ and $\sin\{\pi(1-z)\}=-\sin\{\pi(1+z)\}$ for all $0\le z \le 1$.

Now, using that $\EE\{\bb{\tau}(X, \lambda, \mu, \sigma)\} = \bb{0}_{2}$,
% and using the change of variable $v = y^{\lambda}/\lambda\Leftrightarrow y = (\lambda v)^{1/\lambda}$ with $\rd v = y^{\lambda-1} \rd y \Leftrightarrow \rd y = (\lambda v)^{1/\lambda-1}\rd v$,
one has
\begin{align*}
 \lambda^2\EE\left\{\tau_1(X, \lambda, \mu, \sigma) s_1(X, \lambda, \mu, \sigma)\right\}
& =\EE\left(\cos\left[\pi \left\{1 + \Gamma_{1/\lambda, 1}(|Y|^{\lambda}/\lambda)\right\}\right]\left\{|Y|^{\lambda} \left(1 - \ln(|Y|^\lambda) \right) + C_{\lambda} - 1 \right\}\right)\\
& =\EE\left(\cos\left[\pi \left\{1 + \Gamma_{1/\lambda, 1}(|Y|^{\lambda}/\lambda)\right\}\right]|Y|^{\lambda} \left\{ 1 - \ln(|Y|^\lambda) \right\}\right)\\
& =\lambda\EE\left(\cos\left[\pi \left\{1 + \Gamma_{1/\lambda, 1}(V)\right\}\right]V\left\{1-\ln(\lambda V) \right\}\right)\\
&= \lambda\int_0^{\infty} \cos\left[\pi \left\{1 + \Gamma_{1/\lambda, 1}(v)\right\}\right]v \left\{ 1-\ln(\lambda v) \right\} f_{\mathrm{ga}}(v \nvert 1/\lambda, 1) \rd v\\
&= \int_0^{\infty} \cos\left[\pi \left\{1 + \Gamma_{1/\lambda, 1}(v)\right\}\right]\left\{1- \ln(\lambda v)\right\} f_{\mathrm{ga}}(v \nvert 1/\lambda+1, 1) \rd v\\
&= h_{1}(\lambda)-h_{3}(\lambda).
\end{align*}
\begin{align*}
\sigma \EE\left\{\tau_2(X, \lambda, \mu, \sigma) s_2(X, \lambda, \mu, \sigma)\right\}
& = \EE\left\{\mathrm{sign}(Y)\sin\left[\pi \left\{1 + \Gamma_{1/\lambda, 1}(|Y|^{\lambda}/\lambda)\right\}\right] |Y|^{\lambda - 1} \mathrm{sign}(Y)\right\}\\
&= \lambda^{1-1/\lambda}\EE\left\{\sin\left[\pi \left\{1 + \Gamma_{1/\lambda, 1}(V)\right\}\right] V^{1-1/\lambda}\right\}\\
&=\lambda^{1-1/\lambda} \int_0^{\infty} \sin\left[\pi \left\{1 + \Gamma_{1/\lambda, 1}(v)\right\}\right] v^{1-1/\lambda} f_{\mathrm{ga}}(v \nvert 1/\lambda, 1) \rd v\\
&= \frac{1}{\lambda^{1/\lambda-1} \Gamma(1/\lambda)} \int_0^{\infty} \sin\left[\pi \left\{1 + \Gamma_{1/\lambda, 1}(v)\right\}\right] e^{-v} \rd v
= \frac{h_{2}(\lambda)}{\lambda^{1/\lambda-1} \Gamma(1/\lambda)}.
\end{align*}
\begin{align*}
 \sigma\EE\left\{\tau_1(X, \lambda, \mu, \sigma) s_3(X, \lambda, \mu, \sigma)\right\}
& = \EE\left\{\cos\left[\pi \left\{1 + \Gamma_{1/\lambda, 1}(|Y|^{\lambda}/\lambda)\right\}\right] (|Y|^{\lambda}-1)\right\}\\
&= \EE\left\{\cos\left[\pi \left\{1 + \Gamma_{1/\lambda, 1}(|Y|^{\lambda}/\lambda)\right\}\right] |Y|^{\lambda}\right\}\\
&= \lambda\EE\left\{\cos\left[\pi \left\{1 + \Gamma_{1/\lambda, 1}(V)\right\}\right] V\right\}\\
&= \lambda\int_0^{\infty} \cos\left[\pi \left\{1 + \Gamma_{1/\lambda, 1}(v)\right\}\right] v f_{\mathrm{ga}}(v \nvert 1/\lambda, 1) \rd v\\
&= \int_0^{\infty} \cos\left[\pi \left\{1 + \Gamma_{1/\lambda, 1}(v)\right\}\right] f_{\mathrm{ga}}(v \nvert 1/\lambda+1, 1) \rd v
= h_{1}(\lambda).
%&********\\
%&= \int_{-\infty}^{\infty} \cos\left[\pi \left\{1 + \mathrm{sign}(y) \Gamma_{1/\lambda, 1}(|y|^{\lambda}/\lambda)\right\}\right] |y|^{\lambda} \frac{1}{2\lambda^{1/\lambda-1} \Gamma(1/\lambda)} \exp(-|y|^{\lambda}/\lambda) \rd y \\
%&= \frac{1}{\lambda^{1/\lambda-1} \Gamma(1/\lambda)} \int_0^{\infty} \cos\left[\pi \left\{1 + \Gamma_{1/\lambda, 1}(y^{\lambda}/\lambda)\right\}\right] y^{\lambda} \exp(-y^{\lambda}/\lambda) \rd y \\
%&= \lambda \frac{\Gamma(1/\lambda+1)}{\Gamma(1/\lambda)}\int_0^{\infty} \cos\left[\pi \left\{1 + \Gamma_{1/\lambda, 1}(v)\right\}\right] \frac{v^{1/\lambda+1-1} e^{-v}}{\Gamma(1/\lambda+1)} \rd v \\
%&= \int_0^{\infty} \cos\left[\pi \left\{1 + \Gamma_{1/\lambda, 1}(v)\right\}\right] f_{\mathrm{ga}}(v \nvert 1/\lambda+1, 1) \rd v
%= h_{1}(\lambda).
\end{align*}

Given that $\tau_2(y, \lambda, 0, 1)$ and $s_2(Y, \lambda, 0, 1)$ are odd functions, that $\tau_1(y, \lambda, 0, 1)$, $s_1(Y, \lambda, 0, 1)$, $s_3(Y, \lambda, 0, 1)$ and $f(y \nvert \lambda, 0, 1)$ are even functions, one has
\[
\EE\left\{\tau_2(X, \lambda, \mu, \sigma) s_1(X, \lambda, \mu, \sigma)\right\} = \EE\left\{\tau_2(Y, \lambda, 0, 1) s_1(Y, \lambda, 0, 1)\right\} = 0,
\]
\[
\EE\left\{\tau_1(X, \lambda, \mu, \sigma) s_2(X, \lambda, \mu, \sigma)\right\} = \frac{1}{\sigma} \EE\left\{\tau_1(Y, \lambda, 0, 1) s_2(Y, \lambda, 0, 1)\right\} = 0
\]
and
\[
\EE\left\{\tau_2(X, \lambda, \mu, \sigma) s_3(X, \lambda, \mu, \sigma)\right\} = \frac{1}{\sigma} \EE\left\{\tau_2(Y, \lambda, 0, 1) s_3(Y, \lambda, 0, 1)\right\} = 0
\]
and one obtains
\[
G(\lambda, \mu, \sigma)
=
\begin{bmatrix}
\frac{h_{1}(\lambda)-h_{3}(\lambda)}{\lambda^2}& 0 & \frac{h_{1}(\lambda)}{\sigma} \\[1mm]
0 & \frac{h_{2}(\lambda)}{\sigma\lambda^{1/\lambda-1} \Gamma(1/\lambda)} & 0
\end{bmatrix}.
\]

Consider now the MM estimators of $\mu$ and $\sigma$, $\lambda$ considered as known. We define
\[
C_{2,\lambda}=\frac{\Gamma(1/\lambda)}{\lambda^{2/\lambda}\Gamma(3/\lambda)},\qquad D_{\lambda}=\frac{\Gamma^2(3/\lambda)}{\Gamma(1/\lambda)\Gamma(5/\lambda)-\Gamma^2(3/\lambda)}.
\]
Using the symmetry of $f_{\lambda}(x \nvert \mu, \sigma)$ with respect to $\mu$ and \eqref{eqn.EPD}, one has
\begin{equation}\label{eq.EPD.moments}
\EE(X)=\mu ~~\text{ and }~~ \Var(X)=\EE\{(X-\mu)^2\}=\frac{\lambda^{2/\lambda}\Gamma(3/\lambda)}{\Gamma(1/\lambda)}\sigma^2=\frac{\sigma^2}{C_{2,\lambda}}.
\end{equation}
The joint MM estimators of $\mu$ and $\sigma$, matching the two first moments, are then given by
\[
\hat{\mu}_n =\frac{1}{n} \sum_{i=1}^n x_i \text{ (or }\hat{\mu}_n=\mu_0\text{ if }\mu_0\text{ is known})
~~\text{ and }~~
\hat{\sigma}_n = \left\{C_{2,\lambda}~\frac{1}{n} \sum_{i=1}^n (x_i - \hat{\mu}_n)^2\right\}^{1/2}.
\]

The next step is to find a vector $\bb{g}(x,\lambda,\mu,\sigma)=[g_1(x,\lambda,\mu,\sigma),g_2(x,\lambda,\mu,\sigma)]^\top$ such that
\[
\sqrt{n}
\begin{bmatrix}
\hat{\mu}_n-\mu_0 \\[1mm]
\hat{\sigma}_n -\sigma_0
\end{bmatrix}= \frac{1}{\sqrt{n}}\sum_{i=1}^{n}\bb{g}(X_i,\lambda_0,\mu_0,\sigma_0)+o_{\hspace{0.3mm}\PP_{\!\mathcal{H}_0}}(1) \bb{1}_2.
\]
From the delta method, we know that if $\hat{\sigma}^2_n\stackrel{\PP}{\rightarrow}\sigma_0^2$ and $\sqrt{n}(\hat{\sigma}^2_n -\sigma_0^2)=O_{\hspace{0.3mm}\PP_{\!\mathcal{H}_0}}(1)$, then $\sqrt{n}\{h(\hat{\sigma}^2_n) -h(\sigma_0^2)\}=\sqrt{n}(\hat{\sigma}^2_n -\sigma_0^2)h'(\sigma_0^2)+o_{\hspace{0.3mm}\PP_{\!\mathcal{H}_0}}(1)$ for any function $h$ such that $h'(\sigma_0^2)\neq 0$ exists. Using the CLT and WLLN and
\[
\sqrt{n}\frac{\hat{\sigma}^2_n}{C_{2,\lambda}}=\frac{1}{\sqrt{n}}\sum_{i=1}^n (x_i - \hat{\mu}_n)^2= \frac{1}{\sqrt{n}}\sum_{i=1}^n (x_i - \mu_0)^2-\sqrt{n}(\hat{\mu}_n-\mu_0)^2= \frac{1}{\sqrt{n}}\sum_{i=1}^n (x_i - \mu_0)^2+o_{\hspace{0.3mm}\PP_{\!\mathcal{H}_0}}(1),
\]
it is easy to see that the two conditions are satisfied. Then, setting $h(a)=\sqrt{a}$ with $h'(a)=\frac{1}{2\sqrt{a}}$, one has $h(\hat{\sigma}^2_n)=\hat{\sigma}_n$, $h(\sigma_0^2)=\sigma_0$ with $h'(\sigma_0^2)=\frac{1}{2\sigma_0}$ and
\[
\sqrt{n}(\hat{\sigma}_n -\sigma_0)=\sqrt{n}\frac{1}{2\sigma_0}(\hat{\sigma}^2_n -\sigma_0^2)+o_{\hspace{0.3mm}\PP_{\!\mathcal{H}_0}}(1)
= \frac{1}{\sqrt{n}}\sum_{i=1}^n\frac{1}{2\sigma_0}\left\{C_{2,\lambda}(x_i - \mu_0)^2 -\sigma_0^2\right\}+o_{\hspace{0.3mm}\PP_{\!\mathcal{H}_0}}(1).
\]
Therefore, we set
\[
\bb{g}(x,\lambda,\mu,\sigma) =
\begin{bmatrix}
x - \mu\\[1mm]
\frac{1}{2\sigma}\left\{C_{2,\lambda}(x - \mu)^2 - \sigma^2\right\}
\end{bmatrix}.
\]
Using \eqref{eq.EPD.moments}, one can verify that $\EE\{\bb{g}(X,\lambda,\mu,\sigma)\}=\bb{0}_2$ if $X\sim \mathrm{EPD}(\lambda,\mu, \sigma)$.
The next step is to find the matrix $R(\lambda,\mu,\sigma)^{-1}=\EE\left\{\bb{g}(X, \lambda,\mu,\sigma) \bb{g}(X, \lambda,\mu,\sigma)^{\top}\right\}$. One can show that $\EE\{g_1(X,\lambda,\mu,\sigma)g_2(X,\lambda,\mu,\sigma)\}=0$ using even and odd functions and using \eqref{eq.EPD.moments}, one has $\EE\{g_1(X,\lambda,\mu,\sigma)^2\}=\EE\{(X - \mu)^2\}=\frac{\sigma^2}{C_{2,\lambda}}$. We have
\begin{align*}
\EE\{g_2(X,\lambda,\mu,\sigma)^2\} & =\frac{1}{4\sigma^2}\EE\left[\left\{C_{2,\lambda}(X - \mu)^2 - \sigma^2\right\}^2\right] \\
&=\frac{C^2_{2,\lambda}}{4\sigma^2}\EE\{(X-\mu)^4\}+\frac{\sigma^2}{4}-\frac{C_{2,\lambda}}{2}\EE\{(X-\mu)^2\}\\
&=\frac{C^2_{2,\lambda}}{4\sigma^2}\frac{\lambda^{4/\lambda}\Gamma(5/\lambda)\sigma^4}{\Gamma(1/\lambda)}+\frac{\sigma^2}{4}-\frac{\sigma^2}{2}
=\frac{\Gamma(1/\lambda)\Gamma(5/\lambda)}{\Gamma^2(3/\lambda)}\frac{\sigma^2}{4}-\frac{\sigma^2}{4}\\
&= \frac{\Gamma(1/\lambda)\Gamma(5/\lambda)-\Gamma^2(3/\lambda)}{\Gamma^2(3/\lambda)}\cdot\frac{\sigma^2}{4}=\frac{\sigma^2}{4D_{\lambda}}.
\end{align*}
We obtain
\[
R(\lambda,\mu,\sigma)^{-1}=\sigma^2\begin{bmatrix}
\frac{1}{C_{2,\lambda}} & 0 \\[1mm]
0 & \frac{1}{4D_{\lambda}}
\end{bmatrix}~\text{and }~
R(\lambda,\mu,\sigma)=\frac{1}{\sigma^2}\begin{bmatrix}
C_{2,\lambda} & 0 \\[1mm]
0 & 4D_{\lambda}
\end{bmatrix}
\]
and letting $y = (x-\mu)/\sigma$,
\begin{equation}\label{eq.r.EPD}
\bb{r}(x,\lambda,\mu,\sigma)
= R(\lambda,\mu,\sigma)\bb{g}(x, \lambda,\mu,\sigma)=\frac{1}{\sigma}
\begin{bmatrix}
C_{2,\lambda}y \\[1mm]
2D_{\lambda}\left(C_{2,\lambda}y^2 - 1\right)
\end{bmatrix}.
\end{equation}

Finally, one determines the matrix $J(\lambda,\mu,\sigma)=\EE\big\{\bb{\tau}(X, \lambda,\mu,\sigma) \bb{r}(X, \lambda,\mu,\sigma)^{\top}\big\}$.
Given that $\tau_2(y, \lambda, 0, 1)$ and $r_1(Y, \lambda, 0, 1)$ are odd functions, that $\tau_1(y, \lambda, 0, 1)$, $r_2(Y, \lambda, 0, 1)$ and $f(y \nvert \lambda, 0, 1)$ are even functions, one has
\[
\EE\left\{\tau_1(X, \lambda,\mu,\sigma) r_1(X, \lambda,\mu,\sigma)\right\} = \frac{1}{\sigma} \EE\left\{\tau_1(Y, \lambda, 0, 1) r_1(Y, \lambda, 0, 1)\right\} = 0
\]
and
\[
\EE\left\{\tau_2(X, \lambda,\mu,\sigma) r_2(X, \lambda,\mu,\sigma)\right\} = \frac{1}{\sigma} \EE\left\{\tau_2(Y, \lambda, 0, 1) r_2(Y, \lambda, 0, 1)\right\} = 0.
\]
We have
\begin{align*}
\sigma \EE\big\{\tau_1(X, \lambda,\mu,\sigma)& r_2(X, \lambda,\mu,\sigma)\big\}=\EE\big\{\tau_1(Y, \lambda, 0,1) 2C_{2,\lambda}D_{\lambda}Y^2\big\}\\
&=2C_{2,\lambda}D_{\lambda}\EE\big\{ \cos\left[\pi \left\{1 + \Gamma_{1/\lambda, 1}(|Y|^{\lambda}/\lambda)\right\}\right]Y^2\big\}\\
&=2C_{2,\lambda}D_{\lambda}\lambda^{2/\lambda}\EE\big\{ \cos\left[\pi \left\{1 + \Gamma_{1/\lambda, 1}(V)\right\}\right]V^{2/\lambda}\big\}\\
&=2C_{2,\lambda}D_{\lambda}\lambda^{2/\lambda}\int_{0}^{\infty} \cos\left[\pi \left\{1 + \Gamma_{1/\lambda, 1}(v)\right\}\right]v^{2/\lambda}
f_{\mathrm{ga}}(v \nvert 1/\lambda, 1)\rd v \\
&=\frac{2 C_{2,\lambda}D_{\lambda}\lambda^{2/\lambda}\Gamma(3/\lambda)}{\Gamma(1/\lambda)}\int_{0}^{\infty} \cos\left[\pi \left\{1 + \Gamma_{1/\lambda, 1}(v)\right\}\right] f_{\mathrm{ga}}(v \nvert 3/\lambda, 1) \rd v\\
&=\frac{2 C_{2,\lambda}D_{\lambda}\lambda^{2/\lambda}\Gamma(3/\lambda)h_{4}(\lambda)}{\Gamma(1/\lambda)}=
\frac{\Gamma(1/\lambda)}{\lambda^{2/\lambda}\Gamma(3/\lambda)}\frac{2\lambda^{2/\lambda}\Gamma(3/\lambda)}{\Gamma(1/\lambda)}D_{\lambda}h_{4}(\lambda)
=2D_{\lambda}h_{4}(\lambda)
\end{align*}
and
\begin{align*}
\sigma \EE\big\{\tau_2(X, \lambda,\mu,\sigma)& r_1(X, \lambda,\mu,\sigma)\big\}=C_{2,\lambda}\EE\big\{\tau_2(Y, \lambda, 0,1) Y\big\}=C_{2,\lambda}\EE\big\{\sin\left[\pi \left\{1 + \Gamma_{1/\lambda, 1}(|Y|^{\lambda}/\lambda)\right\}\right]|Y|\big\}\\
&=C_{2,\lambda}\lambda^{1/\lambda}\EE\big\{\sin\left[\pi \left\{1 + \Gamma_{1/\lambda, 1}(V)\right\}\right]V^{1/\lambda}\big\}\\
&=C_{2,\lambda}\lambda^{1/\lambda}\int_{0}^{\infty} \sin\left[\pi \left\{1 + \Gamma_{1/\lambda, 1}(v)\right\}\right]v^{1/\lambda}
f_{\mathrm{ga}}(v \nvert 1/\lambda, 1)\rd v \\
&=\frac{C_{2,\lambda}\lambda^{1/\lambda}\Gamma(2/\lambda)}{\Gamma(1/\lambda)} \int_{0}^{\infty} \sin\left[\pi \left\{1 + \Gamma_{1/\lambda, 1}(v)\right\}\right]
f_{\mathrm{ga}}(v \nvert 2/\lambda, 1)\rd v =\frac{h_{5}(\lambda)\Gamma(2/\lambda)}{\lambda^{1/\lambda}\Gamma(3/\lambda)}.
\end{align*}
Therefore
\[
J(\lambda,\mu,\sigma)=\frac{1}{\sigma}
\begin{bmatrix}
0 & 2D_{\lambda}h_{4}(\lambda)\\[1mm]
\frac{h_{5}(\lambda)\Gamma(2/\lambda)}{\lambda^{1/\lambda}\Gamma(3/\lambda)} & 0
\end{bmatrix}.
\]

\section{Proof for the \texorpdfstring{$\mathrm{half\mhyphen EPD}(\lambda,\sigma)$}{half-EPD} test}\label{proof:half-EPD}

Letting $y = x/\sigma$, straightforward calculations yield
\[
\bb{s}(x, \lambda, \sigma)
=
\begin{bmatrix}
s_1(x, \lambda, \sigma) \\[1mm]
s_2(x, \lambda, \sigma)
\end{bmatrix}
=
\begin{bmatrix}
\partial_{\lambda} \ln f(x \nvert \lambda, \sigma) \\[1mm]
\partial_{\sigma} \ln f(x \nvert \lambda, \sigma)
\end{bmatrix}
=
\begin{bmatrix}
 \frac{1}{\lambda^2}\left\{y^{\lambda}- y^{\lambda} \ln(y^\lambda) + C_{\lambda} - 1 \right\}\\[1mm]
\frac{1}{\sigma}(y^{\lambda} - 1)
\end{bmatrix},
\]
where
\[
C_{\lambda} = \psi(1/\lambda+1) + \ln(\lambda).
\]
The ML estimators for $\lambda_0$ and/or $\sigma_0$ come from the equation $\sum_{i=1}^{n}\bb{s}(x_i, \lambda_0, \sigma_0)=0$.
Furthermore, one has, for $\lambda\in (0, \infty)$,
\[
X\sim \mathrm{half\mhyphen EPD}(\lambda,\sigma)
~~ \Leftrightarrow ~~ Y = \frac{X}{\sigma}\sim \mathrm{half\mhyphen EPD}(\lambda,1)
~~ \Leftrightarrow ~~ V = \frac{1}{\lambda}Y^{\lambda} \sim \mathrm{gamma}(1/\lambda, 1).
\]
From \eqref{eqn.EPD}, it follows that, for $k\in (-1, \infty)$,
\begin{align}
\EE\left\{Y^k(\ln Y)^j\right\}
&= \frac{\lambda^{k/\lambda-j}\Gamma(1/\lambda + k/\lambda)}{\Gamma(1/\lambda)} \nonumber\\
&\times
\begin{cases}
1, &\mbox{if } j = 0, \\
\psi(1/\lambda + k/\lambda)+\ln(\lambda), &\mbox{if } j = 1, \\
\psi_1(1/\lambda + k/\lambda) + \{\psi(1/\lambda + k/\lambda)+\ln(\lambda)\}^2, &\mbox{if } j = 2,
\end{cases}
\end{align}\label{eq:halfEPD}
regardless of whether the expectation is taken with respect to $X\sim \mathrm{half\mhyphen EPD}(\lambda,\sigma)$ or $Y\sim \mathrm{half\mhyphen EPD}(\lambda, 1)$ if $Y= X/\sigma$.

Next, one determines the Fisher information matrix $I(\lambda,\sigma)=\EE\big\{\bb{s}(X, \lambda, \sigma) \bb{s}(X, \lambda, \sigma)^{\top}\big\}$. One has, given that $\EE(\bb{s}(X, \lambda, \sigma))= \bb{0}_{2}$,
\begin{align*}
\EE\left\{s_1^2(X, \lambda, \sigma)\right\}
&= \frac{1}{\lambda^4}\Big((C_{\lambda} - 1)^2 -2 (C_{\lambda} - 1)\EE\left\{Y^{\lambda}\ln(Y^\lambda)\right\}+2 (C_{\lambda} - 1)\EE(Y^{\lambda}) \\
&+ \lambda^2\EE\{Y^{2\lambda}(\ln(Y))^2\}+ \EE(Y^{2\lambda})- 2\lambda\EE\{Y^{2\lambda}\ln(Y)\}\Big)\\
&= \frac{1}{\lambda^4}\Big(-(C_{\lambda} - 1)^2 + \lambda^2 \frac{\Gamma(1/\lambda + 2)}{\Gamma(1/\lambda)} [\psi_1(1/\lambda + 2) + \{\psi(1/\lambda + 2)+\ln(\lambda)\}^2]\\
&+ \frac{\lambda^2\Gamma(1/\lambda + 2)}{\Gamma(1/\lambda)} - 2\lambda\frac{\lambda\Gamma(1/\lambda + 2)}{\Gamma(1/\lambda)} \{\psi(1/\lambda + 2) + \ln(\lambda)\}\Big)\\
&= \frac{1}{\lambda^4}\Big(-(C_{\lambda} - 1)^2 + (\lambda+1) [\psi_1(1/\lambda + 1) - (1/\lambda+1)^{-2}+ \{\psi(1/\lambda + 1) \\
&+\ln(\lambda)+ (1/\lambda+1)^{-1} \}^2] + \lambda+1 -2(\lambda+1)\{\psi(1/\lambda + 1)+\ln(\lambda)+ (1/\lambda+1)^{-1}\}\Big)\\
&=\frac{1}{\lambda^3}\big\{(1/\lambda+1)\psi_1(1/\lambda + 1)+C_{\lambda}^2 -1\big\},
\end{align*}
\[
\EE\left\{s_2^2(X, \lambda,\sigma)\right\}
= \frac{1}{\sigma^2} \left\{\EE\left(Y^{2\lambda}\right) - 2\EE\left(Y^{\lambda}\right) + 1\right\}
= \frac{1}{\sigma^2} \left\{\frac{\lambda^2\Gamma(2 + 1/\lambda)}{\Gamma(1/\lambda)} - 2 + 1\right\}
= \frac{\lambda}{\sigma^2}.
\]
and
\begin{align*}
\EE\{s_1(X, \lambda, & \sigma) s_2(X, \lambda, \sigma)\} =
\frac{-1}{\sigma\lambda^2}\EE\left\{\left( Y^{\lambda}\ln(Y^\lambda) - Y^{\lambda} \right) (Y^{\lambda} - 1)\right\}\\
&=\frac{-1}{\sigma\lambda^2}\Big(\lambda\EE\left\{Y^{2\lambda}\ln(Y)\right\}-\lambda\EE\left\{Y^{\lambda}\ln(Y)\right\} -\EE\left\{Y^{2\lambda}\right\} + \EE\left\{Y^{\lambda}\right\}\Big)\\
&=\frac{-1}{\sigma\lambda^2}\Big(\lambda\frac{\lambda\Gamma(1/\lambda + 2)}{\Gamma(1/\lambda)}\{\psi(1/\lambda + 2)+\ln(\lambda)\} -\lambda\frac{\Gamma(1/\lambda + 1)}{\Gamma(1/\lambda)}\{\psi(1/\lambda + 1)+\ln(\lambda)\}\\
&-\frac{\lambda^2\Gamma(1/\lambda + 2)}{\Gamma(1/\lambda)} +\frac{\lambda\Gamma(1/\lambda + 1)}{\Gamma(1/\lambda)} \Big)\\
&=\frac{-1}{\sigma\lambda^2}\Big((\lambda + 1)\{\psi(1/\lambda + 1)+(1/\lambda+1)^{-1}+\ln(\lambda)\}-\psi(1/\lambda + 1)-\ln(\lambda)-\lambda - 1 +1 \Big)\\
&= \frac{-1}{\sigma\lambda}\{\psi(1/\lambda+1) + \ln(\lambda)\}= \frac{-C_{\lambda}}{\sigma\lambda}.
\end{align*}

One has
\[
I(\lambda,\sigma) =
\begin{bmatrix}
\frac{1}{\lambda^3}\big\{(1/\lambda+1)\psi_1(1/\lambda + 1)+C^2_{1, \lambda} -1\big\} & \frac{-C_{\lambda}}{\sigma\lambda} \\[1mm]
\frac{-C_{\lambda}}{\sigma\lambda} & \frac{\lambda}{\sigma^2}
\end{bmatrix}.
\]

Next, one determines the matrix $G(\lambda,\sigma)=\EE\big\{\bb{\tau}(X, \lambda,\sigma) \bb{s}(X, \lambda,\sigma)^{\top}\big\}$. One has, if $y = x/\sigma$,
\[
\bb{\tau}(x, \lambda,\sigma)
=
\begin{bmatrix}
\tau_1(x, \lambda,\sigma) \\[1mm]
\tau_2(x, \lambda,\sigma)
\end{bmatrix}
=
\begin{bmatrix}
\cos\left\{2\pi F(x \nvert \lambda,\sigma)\right\} \\[1mm]
\sin\left\{2\pi F(x \nvert \lambda,\sigma)\right\}
\end{bmatrix}
=
\begin{bmatrix}
\cos\left\{2\pi \Gamma_{1/\lambda, 1}(y^{\lambda}/\lambda)\right\} \\[1mm]
\sin\left\{2\pi \Gamma_{1/\lambda, 1}(y^{\lambda}/\lambda)\right\}
\end{bmatrix}.
\]
Note that $\bb{\tau}(x, \lambda,\sigma) = \bb{\tau}(y, \lambda,1)$.
Now, using that $\EE\{\bb{\tau}(X, \lambda,\sigma)\} = \bb{0}_2$, one has
\begin{align*}
 \lambda^2\EE\left\{\tau_1(X, \lambda, \sigma) s_1(X, \lambda, \sigma)\right\}
& =\EE\left(\cos\left[2\pi \Gamma_{1/\lambda, 1}(Y^{\lambda}/\lambda)\right]\left\{ Y^{\lambda} \left( 1- \ln(Y^\lambda) \right) + C_{\lambda} - 1\right\}\right)\\
& =\EE\left(\cos\left[2\pi \Gamma_{1/\lambda, 1}(Y^{\lambda}/\lambda)\right] Y^{\lambda} \left\{ 1- \ln(Y^\lambda) \right\}\right)\\
& =\lambda\EE\left(\cos\left[2\pi \Gamma_{1/\lambda, 1}(V)\right] V \left\{ 1- \ln(\lambda V) \right\}\right)\\
&= \lambda\int_0^{\infty} \cos\left[2\pi \Gamma_{1/\lambda, 1}(v)\right]v \left\{ 1-\ln(\lambda v) \right\} f_{\mathrm{ga}}(v \nvert 1/\lambda, 1) \rd v\\
&= \int_0^{\infty} \cos\left[2\pi \Gamma_{1/\lambda, 1}(v)\right]\left\{1- \ln(\lambda v)\right\} f_{\mathrm{ga}}(v \nvert 1/\lambda+1, 1) \rd v\\
&= h_{6}(1/\lambda,1/\lambda+1,1)-h_{15}(\lambda).
\end{align*}
Similarly, one obtains
\[
\lambda^2\EE\left\{\tau_2(X, \lambda, \sigma) s_1(X, \lambda, \sigma)\right\}= h_{7}(1/\lambda,1/\lambda+1,1)-h_{16}(\lambda).
\]
One has
\begin{align*}
 \sigma\EE\left\{\tau_1(X, \lambda, \sigma) s_2(X, \lambda, \sigma)\right\}
& = \EE\left\{\cos\left[2\pi \Gamma_{1/\lambda, 1}(Y^{\lambda}/\lambda)\right] (Y^{\lambda}-1)\right\}
= \EE\left\{\cos\left[2\pi \Gamma_{1/\lambda, 1}(Y^{\lambda}/\lambda)\right] Y^{\lambda}\right\}\\
&= \lambda\EE\left\{\cos\left[2\pi \Gamma_{1/\lambda, 1}(V)\right] V\right\}
= \lambda\int_0^{\infty} \cos\left[2\pi \Gamma_{1/\lambda, 1}(v)\right] v f_{\mathrm{ga}}(v \nvert 1/\lambda, 1) \rd v\\
&= \int_0^{\infty} \cos\left\{2\pi \Gamma_{1/\lambda, 1}(v)\right\} f_{\mathrm{ga}}(v \nvert 1/\lambda+1,1) \rd v
= h_{6}(1/\lambda,1 / \lambda+1,1).
\end{align*}
Similarly, one obtains
\[
\sigma \EE\left\{\tau_2(X, \lambda, \sigma) s_2(X, \lambda, \sigma)\right\}= h_{7}(1/\lambda, 1 / \lambda +1, 1)
\]
and\[
G(\lambda,\sigma)
=
\begin{bmatrix}
\frac{h_{6}(1/\lambda,1/\lambda+1,1)-h_{15}(\lambda)}{\lambda^2} & \frac{h_{6}(1/\lambda, 1 / \lambda +1, 1)}{\sigma} \\[1mm]
\frac{h_{7}(1/\lambda,1/\lambda+1,1)-h_{16}(\lambda)}{\lambda^2} & \frac{h_{7}(1/\lambda, 1 / \lambda +1, 1)}{\sigma}
\end{bmatrix}.
\]

\section{Proof for the \texorpdfstring{$\mathrm{SN}(\lambda, \mu, \sigma)$}{SN} test}\label{proof:SN}

Letting $y = (x-\mu)/\sigma$, straightforward calculations yield
\[
\bb{s}(x, \lambda, \mu, \sigma)
=
\begin{bmatrix}
s_1(x, \lambda, \mu, \sigma) \\[1mm]
s_2(x, \lambda, \mu, \sigma) \\[1mm]
s_3(x, \lambda, \mu, \sigma)
\end{bmatrix}
=
\begin{bmatrix}
\partial_{\lambda} \ln\{f(x \nvert \lambda, \mu, \sigma)\} \\[1mm]
\partial_{\mu} \ln\{f(x \nvert \lambda, \mu, \sigma)\} \\[1mm]
\partial_{\sigma} \ln\{f(x \nvert \lambda, \mu, \sigma)\}
\end{bmatrix}
=
\begin{bmatrix}
y H(\lambda y)\\[1mm]
\frac{1}{\sigma}\{y-\lambda H(\lambda y)\} \\[1mm]
\frac{1}{\sigma}\{-1+y^2-\lambda y H(\lambda y)\}
\end{bmatrix},
\]
where
\[
H(t)=\phi(t)/\Phi(t).
\]
The ML estimators for $\lambda_0$, $\mu_0$ and/or $\sigma_0$ come from the equation $\sum_{i=1}^{n}\bb{s}(x_i, \lambda_0, \mu_0, \sigma_0)=0$.

Here are some equations that will be useful. For all $\lambda\in \R, \mu \in\R, \sigma \in (0,\infty)$,
$$
X\sim \mathrm{SN}(\lambda,\mu,\sigma)\Leftrightarrow Y=\frac{X-\mu}{\sigma}\sim \mathrm{SN}(\lambda,0, 1),
$$
$$
\EE(Y)=\frac{2 \lambda}{\sqrt{2\pi(1+ \lambda^2)}}, ~ \EE(Y^2)=1,
$$
$$
\EE\{Y H(\lambda Y)\}=0, ~ \EE\{Y^3 H(\lambda Y)\} =0,
$$
$$
\EE\{H(\lambda Y)\}=2\int_{-\infty}^{\infty}  \phi(y)\phi(\lambda y) \rd y = \frac{2}{\sqrt{2\pi(1+ \lambda^2)}},
$$
$$
\EE\{Y^2 H(\lambda Y)\}=2\int_{-\infty}^{\infty}  y^2\phi(y)\phi(\lambda y) \rd y = \frac{2}{\sqrt{2\pi}(1+ \lambda^2)^{3/2}},
$$
$$
\EE\{H(\lambda Y)^2)\} = 2 \int_{-\infty}^{\infty}\frac{\phi(v)\phi^2(\lambda v)}{\Phi(\lambda v)} \rd v=  2 h_{36}(\lambda),
$$
$$
\EE\{Y H(\lambda Y)^2)\} = 2 \int_{-\infty}^{\infty}\frac{v\phi(v)\phi^2(\lambda v)}{\Phi(\lambda v)} \rd v=  2 h_{37}(\lambda),
$$
$$
\EE\{Y^2 H^2(\lambda Y)\}=2\int_{-\infty}^{\infty}\frac{v^2\phi(v)\phi^2(\lambda v)}{\Phi(\lambda v)} \rd v = 2  h_{38}(\lambda).
$$

Next, one determines the Fisher information matrix $I(\lambda, \mu, \sigma)=\EE\big\{\bb{s}(X, \lambda, \mu, \sigma) \bb{s}(X, \lambda, \mu, \sigma)^{\top}\big\}$ or using $I(\lambda, \mu, \sigma)=-\EE\big\{\partial^2_{\bb{\theta}} \ln\{f(X \nvert \lambda, \mu, \sigma)\}\big\}$. One has
\[
\EE\big[-\partial^2_{\lambda} \ln\{f(X \nvert \lambda, \mu, \sigma)\}\big]=\lambda \EE\{Y^3 H(\lambda Y)\} + \EE\{Y^2 H(\lambda Y)^2\}= 2  h_{38}(\lambda),
\]
\begin{align*}
\sigma^2\EE\big[-\partial^2_{\mu} \ln\{f(X \nvert \lambda, \mu, \sigma)\}\big] & =1+ \lambda^3 \EE\{Y H(\lambda Y)\} + \lambda^2 \EE\{H(\lambda Y)^2)\}\\
& = 1 + \lambda^2 \EE\{H(\lambda Y)^2)\} =  1 + 2 \lambda^2 h_{36}(\lambda),
\end{align*}
\begin{align*}
\sigma^2\EE\big[-\partial^2_{\sigma} \ln\{f(X \nvert \lambda, \mu, \sigma)\}\big] & =-1 + 3 \EE(Y^2) - 2 \lambda \EE\{Y  H(\lambda Y)\} + \lambda ^ 3  \EE\{Y ^ 3  H(\lambda Y)\}  + \lambda ^ 2  \EE\{Y ^ 2  H(\lambda Y) ^ 2\} \\
&= 2  + \lambda ^ 2  \EE\{Y ^ 2  H(\lambda Y) ^ 2\} = 2  + 2\lambda ^ 2 h_{38}(\lambda)=2\{1 + \lambda ^ 2 h_{38}(\lambda)\},
\end{align*}
\begin{align*}
\sigma\EE\big[-\partial^2_{\lambda\mu} \ln\{f(X \nvert \lambda, \mu, \sigma)\}\big] & =\EE\{H(\lambda Y)\} - \lambda^2\EE\{Y^2  H(\lambda Y)\} - \lambda \EE\{Y H(\lambda Y)^2\}\\
&= \frac{2}{\sqrt{2\pi(1+ \lambda^2)}} -\lambda^2 \frac{2}{\sqrt{2\pi}(1+ \lambda^2)^{3/2}} - 2\lambda  h_{37}(\lambda)\\
&= \frac{2}{\sqrt{2\pi}(1+ \lambda^2)^{3/2}} - 2\lambda  h_{37}(\lambda)= 2\left\{\frac{1}{\sqrt{2\pi}(1+ \lambda^2)^{3/2}} - \lambda  h_{37}(\lambda)\right\},
\end{align*}
\[
\sigma\EE\big[-\partial^2_{\lambda\sigma} \ln\{f(X \nvert \lambda, \mu, \sigma)\}\big]  =\EE\{Y H(\lambda Y)\}-\lambda^2\EE\{Y^3 H(\lambda Y)\}-\lambda\EE\{Y^2 H^2(\lambda Y)\}= - 2\lambda h_{38}(\lambda),
\]
\begin{align*}
\sigma^2\EE\big[-\partial^2_{\mu\sigma} \ln\{f(X \nvert \lambda, \mu, \sigma)\}\big] & =2\EE(Y)-\lambda\EE\{H(\lambda Y)\} + \lambda^3\EE\{Y^2  H(\lambda Y)\} + \lambda^2 \EE\{Y H(\lambda Y)^2\}\\
&= \frac{4\lambda}{\sqrt{2\pi(1+ \lambda^2)}} - \frac{2\lambda}{\sqrt{2\pi(1+ \lambda^2)}}  + \frac{2\lambda^3}{\sqrt{2\pi}(1+ \lambda^2)^{3/2}} + 2\lambda^2  h_{37}(\lambda)\\
&= 2\lambda\left\{\frac{1}{\sqrt{2\pi(1+ \lambda^2)}}  + \frac{\lambda^2}{\sqrt{2\pi}(1+ \lambda^2)^{3/2}} + \lambda  h_{37}(\lambda)\right\}\\
&= 2\lambda\left\{\frac{1+2\lambda^2}{\sqrt{2\pi}(1+ \lambda^2)^{3/2}} + \lambda  h_{37}(\lambda)\right\}.
\end{align*}

One has
\[
I(\lambda, \mu, \sigma) =2
\begin{bmatrix}
 h_{38}(\lambda) & \frac{1}{\sigma} \Big\{\frac{1}{\sqrt{2\pi}(1+ \lambda^2)^{3/2}} - \lambda  h_{37}(\lambda)\Big\} & \frac{- \lambda h_{38}(\lambda)}{\sigma} \\[1mm]
\frac{1}{\sigma} \Big\{\frac{1}{\sqrt{2\pi}(1+ \lambda^2)^{3/2}} - \lambda  h_{37}(\lambda)\Big\}  & \frac{1/2 +  \lambda^2 h_{36}(\lambda)}{\sigma^2} & \frac{\lambda}{\sigma^2} \Big\{\frac{1+2\lambda^2}{\sqrt{2\pi}(1+ \lambda^2)^{3/2}} + \lambda  h_{37}(\lambda)\Big\} \\[1mm]
\frac{- \lambda h_{38}(\lambda)}{\sigma} & \frac{\lambda}{\sigma^2} \Big\{\frac{1+2\lambda^2}{\sqrt{2\pi}(1+ \lambda^2)^{3/2}} + \lambda  h_{37}(\lambda)\Big\} & \frac{1}{\sigma^2}\{1 + \lambda ^ 2 h_{38}(\lambda)\}
\end{bmatrix}.
\]

Next, one determines the matrix $G(\lambda, \mu, \sigma)=\EE\big\{\bb{\tau}(X, \lambda, \mu, \sigma) \bb{s}(X, \lambda, \mu, \sigma)^{\top}\big\}$. One has, if $y = (x - \mu)/\sigma$,
\begin{align*}
\bb{\tau}(x, \lambda, \mu, \sigma)
&=
\begin{bmatrix}
\tau_1(x, \lambda, \mu, \sigma) \\[1mm]
\tau_2(x, \lambda, \mu, \sigma)
\end{bmatrix}
=
\begin{bmatrix}
\cos\left\{2\pi F_{SN}(x \nvert \lambda, \mu, \sigma)\right\} \\[1mm]
\sin\left\{2\pi F_{SN}(x \nvert \lambda, \mu, \sigma)\right\}
\end{bmatrix}
=
\begin{bmatrix}
\cos\left\{2\pi F_{SN}(y \nvert \lambda, 0, 1)\right\} \\[1mm]
\sin\left\{2\pi F_{SN}(y \nvert \lambda, 0, 1)\right\}
\end{bmatrix}.
\end{align*}

One has
\begin{align*}
\EE\left\{\tau_1(X, \lambda, \mu, \sigma) s_1(X, \lambda, \mu, \sigma)\right\}
& =\EE\left[\cos\left\{2\pi F_{SN}(Y \nvert \lambda, 0, 1)\right\}Y H(\lambda Y)\right]\\
&=2\int_{-\infty}^{\infty} v \cos\left\{2\pi F_{SN}(v \nvert \lambda, 0, 1)\right\}\phi(\lambda v)\phi(v)\rd v=2 h_{39}(\lambda).
\end{align*}
\begin{align*}
\EE\left\{\tau_2(X, \lambda, \mu, \sigma) s_1(X, \lambda, \mu, \sigma)\right\}
& =\EE\left[\sin\left\{2\pi F_{SN}(Y \nvert \lambda, 0, 1)\right\}Y H(\lambda Y)\right]\\
&=2\int_{-\infty}^{\infty} v \sin\left\{2\pi F_{SN}(v \nvert \lambda, 0, 1)\right\}\phi(\lambda v)\phi(v)\rd v=2 h_{40}(\lambda).
\end{align*}
\begin{align*}
\sigma\EE\left\{\tau_1(X, \lambda, \mu, \sigma) s_2(X, \lambda, \mu, \sigma)\right\}
& =\EE\left[\cos\left\{2\pi F_{SN}(Y \nvert \lambda, 0, 1)\right\}Y \right]-\lambda\EE\left[\cos\left\{2\pi F_{SN}(Y \nvert \lambda, 0, 1)\right\}H(\lambda Y)\right]\\
&=2\int_{-\infty}^{\infty} v \cos\left\{2\pi F_{SN}(v \nvert \lambda, 0, 1)\right\}\Phi(\lambda v)\phi(v)\rd v \\
& - 2\lambda\int_{-\infty}^{\infty} \cos\left\{2\pi F_{SN}(v \nvert \lambda, 0, 1)\right\}\phi(\lambda v)\phi(v)\rd v\\
&= 2\int_{-\infty}^{\infty} \cos\left\{2\pi F_{SN}(v \nvert \lambda, 0, 1)\right\} \left\{v \Phi(\lambda v)-\lambda \phi(\lambda v)\right\}\phi(v)\rd v \\
&=2 h_{41}(\lambda).
\end{align*}
\begin{align*}
\sigma\EE\left\{\tau_2(X, \lambda, \mu, \sigma) s_2(X, \lambda, \mu, \sigma)\right\}
& =\EE\left[\sin\left\{2\pi F_{SN}(Y \nvert \lambda, 0, 1)\right\}Y \right]-\lambda\EE\left[\sin\left\{2\pi F_{SN}(Y \nvert \lambda, 0, 1)\right\}H(\lambda Y)\right]\\
&=2\int_{-\infty}^{\infty} v \sin\left\{2\pi F_{SN}(v \nvert \lambda, 0, 1)\right\}\Phi(\lambda v)\phi(v)\rd v \\
& - 2\lambda\int_{-\infty}^{\infty} \sin\left\{2\pi F_{SN}(v \nvert \lambda, 0, 1)\right\}\phi(\lambda v)\phi(v)\rd v\\
&= 2\int_{-\infty}^{\infty} \sin\left\{2\pi F_{SN}(v \nvert \lambda, 0, 1)\right\} \left\{v \Phi(\lambda v)-\lambda \phi(\lambda v)\right\}\phi(v)\rd v \\
&=2 h_{42}(\lambda).
\end{align*}
\begin{align*}
\sigma\EE\left\{\tau_1(X, \lambda, \mu, \sigma) s_3(X, \lambda, \mu, \sigma)\right\}
& =-\EE\left[\cos\left\{2\pi F_{SN}(Y \nvert \lambda, 0, 1)\right\}\right]\\
&+\EE\left[\cos\left\{2\pi F_{SN}(Y \nvert \lambda, 0, 1)\right\}Y^2 \right]-\lambda\EE\left[\cos\left\{2\pi F_{SN}(Y \nvert \lambda, 0, 1)\right\}Y H(\lambda Y)\right]\\
&= 0 + 2\int_{-\infty}^{\infty} v^2 \cos\left\{2\pi F_{SN}(v \nvert \lambda, 0, 1)\right\}\Phi(\lambda v)\phi(v)\rd v \\
& - 2\lambda \int_{-\infty}^{\infty} v \cos\left\{2\pi F_{SN}(v \nvert \lambda, 0, 1)\right\}\phi(\lambda v)\phi(v)\rd v \\
&= 2\int_{-\infty}^{\infty}v \cos\left\{2\pi F_{SN}(v \nvert \lambda, 0, 1)\right\} \left\{v \Phi(\lambda v)-\lambda \phi(\lambda v)\right\}\phi(v)\rd v \\
& =2 h_{43}(\lambda).
\end{align*}
\begin{align*}
\sigma\EE\left\{\tau_2(X, \lambda, \mu, \sigma) s_3(X, \lambda, \mu, \sigma)\right\}
& =-\EE\left[\sin\left\{2\pi F_{SN}(Y \nvert \lambda, 0, 1)\right\}\right]\\
&+\EE\left[\sin\left\{2\pi F_{SN}(Y \nvert \lambda, 0, 1)\right\}Y^2 \right]-\lambda\EE\left[\sin\left\{2\pi F_{SN}(Y \nvert \lambda, 0, 1)\right\}Y H(\lambda Y)\right]\\
&= 0 + 2\int_{-\infty}^{\infty} v^2 \sin\left\{2\pi F_{SN}(v \nvert \lambda, 0, 1)\right\}\Phi(\lambda v)\phi(v)\rd v  \\
& - 2\lambda \int_{-\infty}^{\infty} v \sin\left\{2\pi F_{SN}(v \nvert \lambda, 0, 1)\right\}\phi(\lambda v)\phi(v)\rd v \\
&= 2\int_{-\infty}^{\infty}v \sin\left\{2\pi F_{SN}(v \nvert \lambda, 0, 1)\right\} \left\{v \Phi(\lambda v)-\lambda \phi(\lambda v)\right\}\phi(v)\rd v \\
& =2 h_{44}(\lambda).
\end{align*}

and one obtains
\[
G(\lambda, \mu, \sigma)
= 2
\begin{bmatrix}
h_{39}(\lambda)& \frac{1}{\sigma}h_{41}(\lambda) & \frac{1}{\sigma} h_{43}(\lambda)\\[1mm]
h_{40}(\lambda)& \frac{1}{\sigma} h_{42}(\lambda) & \frac{1}{\sigma} h_{44}(\lambda)
\end{bmatrix}.
\]

\section{Proof for the \texorpdfstring{$\mathrm{GG}(\lambda,\beta, \rho)$}{GG} test}\label{proof:GG}

The proof is first carried out for the exponential variant of the $\mathrm{GG}(\lambda,\beta,\rho)$ distribution, namely the \\ exp-GG$(\lambda,\mu,\sigma)$ distribution, whose density is given by
\[
f(x \mid \lambda,\mu,\sigma)
= \{\sigma \Gamma(\lambda)\}^{-1}
\exp\!\left[\lambda(x-\mu)/\sigma
- \exp\!\left\{(x-\mu)/\sigma\right\}\right],
\qquad x \in \mathbb{R},
\]
with $\beta = e^{\mu}$ and $\rho = 1/\sigma$.

Letting $y = (x-\mu)/\sigma$, straightforward calculations yield
\[
\bb{s}(x, \lambda, \mu, \sigma)
=
\begin{bmatrix}
s_1(x, \lambda, \mu, \sigma) \\[1mm]
s_2(x, \lambda, \mu, \sigma) \\[1mm]
s_3(x, \lambda, \mu, \sigma)
\end{bmatrix}
=
\begin{bmatrix}
\partial_{\lambda} \ln f(x \nvert \lambda,\mu, \sigma) \\[1mm]
\partial_{\mu} \ln f(x \nvert \lambda,\mu, \sigma) \\[1mm]
\partial_{\sigma} \ln f(x \nvert \lambda,\mu, \sigma)
\end{bmatrix}
=
\begin{bmatrix}
y - \psi(\lambda) \\[1mm]
\frac{e^y - \lambda}{\sigma} \\[1mm]
\frac{y (e^y - \lambda) - 1}{\sigma}
\end{bmatrix}.
\]

The ML estimators for $\lambda_0$, $\mu_0$ and/or $\sigma_0$ come from the equation $\sum_{i=1}^{n}\bb{s}(x_i, \lambda_0, \mu_0, \sigma_0)=0$.
Explicitly, one has
\begin{equation*}
 \sum_{i=1}^{n}s_1(x_i, \lambda_0, \mu_0, \sigma_0)=0
 \Leftrightarrow \frac{1}{n}\sum_{i=1}^{n}x_i/\sigma_0 -\mu_0/\sigma_0 - \psi(\lambda_0)=0,
\end{equation*}
\begin{equation*}
 \sum_{i=1}^{n}s_2(x_i, \lambda_0,\mu_0, \sigma_0)=0\Leftrightarrow \frac{1}{n}\sum_{i=1}^{n}\exp\{(x_i-\mu_0)/\sigma_0\}=\lambda_0 \Leftrightarrow
\mu_0 = \sigma_0\ln \left\{\frac{1}{\lambda_0}\frac{1}{n}\sum_{i=1}^n \exp(x_i/\sigma_0)\right\},
\end{equation*}
\begin{align*}
\sum_{i=1}^{n}&s_3(x_i, \lambda_0,\mu_0, \sigma_0)=0\Leftrightarrow \frac{1}{n}\sum_{i=1}^{n}(x_i-\mu_0) [\exp\{(x_i-\mu_0)/\sigma_0\} - \lambda_0] - \sigma_0 =0\\
&\Leftrightarrow \frac{\sum_{i=1}^{n}x_i\exp(x_i/\sigma_0)}{n\exp(\mu_0/\sigma_0)} - \lambda_0\frac{1}{n}\sum_{i=1}^{n}x_i- \sigma_0-\mu_0\left(\frac{\sum_{i=1}^{n}\exp(x_i/\sigma_0)}{n\exp(\mu_0/\sigma_0)} - \lambda_0\right) =0.
\end{align*}

Here are some equations that will be useful in the following sections. One has, for $\lambda\in(0,\infty)$,
\[
X\sim \mathrm{exp\mhyphen GG}(\lambda,\mu, \sigma)
~~ \Leftrightarrow ~~ Y = \frac{X - \mu}{\sigma}\sim \mathrm{exp\mhyphen GG}(\lambda,0, 1)
~~ \Leftrightarrow ~~ V = e^Y \sim \mathrm{gamma}(\lambda, 1).
\]
For $z\in (0, \infty)$, it is well-known that
\[
\Gamma(1 + z) = z\Gamma(z), \quad
\psi(1 + z) = \psi(z) + 1/z, \quad
\psi_1(1 + z) = \psi_1(z) - 1/z^2.
\]
If $V\sim \mathrm{gamma}(\lambda, 1)$ and $k\in (-\lambda, \infty)$, it can also be verified in \texttt{Mathematica} that
\[
\EE\left(Y^j e^{kY}\right)=\EE\left\{V^k(\ln V)^j\right\}
= \frac{\Gamma(\lambda + k)}{\Gamma(\lambda)} \times
\begin{cases}
1, &\mbox{if } j = 0, \\
\psi(\lambda + k), &\mbox{if } j = 1, \\
\psi_1(\lambda + k) + \psi^2(\lambda + k), &\mbox{if } j = 2,
\end{cases}
\]
regardless of whether the expectation is taken with respect to $X\sim \mathrm{exp\mhyphen GG}(\lambda,\mu, \sigma)$ if $Y= (X - \mu)/\sigma$,
$Y \sim \mathrm{exp\mhyphen GG}(\lambda,0, 1)$ or $V \sim \mathrm{gamma}(\lambda, 1)$ if and $V=e^Y$.

Next, one determines the Fisher information matrix $I(\lambda,\mu, \sigma)=\EE\big\{\bb{s}(X,\lambda, \mu, \sigma) \bb{s}(X,\lambda, \mu, \sigma)^{\top}\big\}$. One has
\begin{align*}
\EE\left\{s_1^2(X, \lambda,\mu, \sigma)\right\}&=\EE [\{Y-\psi(\lambda)\}^2]=\EE(Y^2)-2\psi(\lambda)\EE(Y)+\psi^2(\lambda)\\
&=\psi_1(\lambda) + \psi^2(\lambda)-2\psi^2(\lambda)+\psi^2(\lambda)=\psi_1(\lambda),
\end{align*}
\[
\EE\left\{s_2^2(X, \lambda,\mu, \sigma)\right\}= \sigma^{-2} \EE\left\{(e^Y - \lambda)^2\right\} =\sigma^{-2} \EE\left\{(V - \lambda)^2\right\} = \sigma^{-2}\Var(V) = \lambda/\sigma^2,
\]
and
\begin{align*}
\sigma^2 \EE\left\{s_3^2(X, \lambda,\mu, \sigma)\right\}
&= \EE\left\{(Y e^Y - \lambda Y - 1)^2\right\}
= \EE\left\{(V\ln V - \lambda\ln V - 1)^2\right\} \\
&= \EE\left\{V^2(\ln V)^2\right\} + \lambda^2 \EE\left\{(\ln V)^2\right\} + 1 - 2 \EE(V\ln V) + 2\lambda \EE(\ln V) - 2\lambda \EE\left\{V(\ln V)^2\right\} \\
&= \lambda(\lambda + 1) \psi_1(\lambda + 2) + \lambda(\lambda + 1) \psi^2(\lambda + 2) + \lambda^2\psi_1(\lambda) + \lambda^2 \psi^2(\lambda) + 1 \\
&\hspace{4mm} - 2\lambda\psi(\lambda + 1) + 2\lambda\psi(\lambda) - 2\lambda^2\psi_1(\lambda + 1) - 2\lambda^2\psi^2(\lambda + 1) \\
&= \lambda \psi^2(\lambda) + 2 \psi(\lambda) + \lambda \psi_1(\lambda) + 1.
\end{align*}
The last equality can be shown using the properties of the digamma and trigamma functions. As it is rather technical, details are omitted.
One has
\begin{align*}
&\sigma \EE\left\{s_1(X, \lambda,\mu, \sigma) s_2(X, \lambda,\mu, \sigma)\right\} = \EE\left\{s_1(Y, \lambda,0, 1) (e^Y-\lambda)\right\} = \EE\left\{s_1(Y, \lambda,0, 1) e^Y\right\} = \EE(Y e^Y)-\psi(\lambda)\EE(e^Y)\\
&=\frac{\Gamma(\lambda + 1)}{\Gamma(\lambda)}\psi(\lambda+1)-\psi(\lambda)\frac{\Gamma(\lambda + 1)}{\Gamma(\lambda)}=\lambda(\psi(\lambda)+1/\lambda-\psi(\lambda))=1,
\end{align*}
\begin{align*}
&\sigma \EE\left\{s_1(X, \lambda,\mu, \sigma) s_3(X, \lambda,\mu, \sigma)\right\} = \EE\left\{s_1(Y, \lambda,0, 1) (Y e^Y - \lambda Y - 1)\right\} = \EE\left\{(Y-\psi(\lambda)) (Y e^Y - \lambda Y)\right\} \\
&= \EE(Y^2 e^Y)-\psi(\lambda)\EE(Y e^Y)-\lambda\EE(Y^2)+\lambda\psi(\lambda)\EE(Y)\\
&= \lambda\{\psi_1(\lambda + 1) + \psi^2(\lambda + 1)\} - \lambda\psi(\lambda)\psi(\lambda+1)-\lambda \{\psi_1(\lambda) + \psi^2(\lambda)\}+\lambda\psi^2(\lambda)\\
&= \lambda \left\{\psi_1(\lambda + 1) + \psi^2(\lambda + 1) - (\psi(\lambda+1)-1/\lambda)\psi(\lambda+1)-\psi_1(\lambda+1)-1/\lambda^2\right\}\\
&= \lambda \left\{\psi(\lambda+1)/\lambda-1/\lambda^2\right\}=\psi(\lambda+1)-1/\lambda = \psi(\lambda),
\end{align*}
\begin{align*}
&\sigma^2 \EE\left\{s_2(X, \lambda,\mu, \sigma) s_3(X, \lambda,\mu, \sigma)\right\} = \EE\left\{s_2(Y, \lambda,0, 1) s_3(Y, \lambda,0, 1)\right\} = \EE\left\{s_2(Y, \lambda,0, 1) (-1 + Y s_2(Y, \lambda, 0, 1))\right\} \\
&\quad= \EE\left\{Ys_2^2(Y, \lambda, 0, 1)\right\} - \EE\left\{s_2(Y, \lambda, 0, 1)\right\} = \EE\left\{Y(e^Y - \lambda)^2\right\} - \EE(e^Y - \lambda) \\
&\quad= \EE\left\{(V - \lambda)^2\ln V\right\} - \EE(V - \lambda) = \EE(V^2\ln V) - 2\lambda \EE(V\ln V) + \lambda^2 \EE(\ln V) - 0 \\
&\quad= \lambda(\lambda + 1) \psi(\lambda + 2) - 2\lambda^2 \psi(\lambda + 1) + \lambda^2 \psi(\lambda) \\
&\quad= \lambda(\lambda + 1) \left\{\psi(\lambda) + 1/(\lambda + 1) + 1/\lambda\right\} - 2\lambda^2 \left\{\psi(\lambda) + 1/\lambda\right\} + \lambda^2 \psi(\lambda) \\
&\quad= \lambda \psi(\lambda) + 1.
\end{align*}

The Fisher information matrix is then given by
\[
I(\lambda,\mu, \sigma)
=
\begin{bmatrix}
\psi_1(\lambda) & \frac{1}{\sigma} & \frac{\psi(\lambda)}{\sigma} \\[1mm]
\frac{1}{\sigma} & \frac{\lambda}{\sigma^2} & \frac{\lambda\psi(\lambda) + 1}{\sigma^2} \\[1mm]
\frac{\psi(\lambda)}{\sigma} & \frac{\lambda\psi(\lambda) + 1}{\sigma^2} &\frac{\lambda\psi^2(\lambda) + 2\psi(\lambda) + \lambda\psi_1(\lambda) + 1}{\sigma^2}
\end{bmatrix}.
\]

Next, one determines the matrix $G(\lambda,\mu, \sigma)=\EE\big\{\bb{\tau}(X,\lambda, \mu, \sigma) \bb{s}(X, \lambda,\mu, \sigma)^{\top}\big\}$. One has, if $y = (x - \mu)/\sigma$,
\[
\bb{\tau}(x, \lambda,\mu, \sigma)
=
\begin{bmatrix}
\tau_1(x,\lambda, \mu, \sigma) \\[1mm]
\tau_2(x,\lambda, \mu, \sigma)
\end{bmatrix}
=
\begin{bmatrix}
\cos\left\{2\pi F(x \nvert\lambda, \mu, \sigma)\right\} \\[1mm]
\sin\left\{2\pi F(x \nvert \lambda,\mu, \sigma)\right\}
\end{bmatrix}
=
\begin{bmatrix}
\cos\left\{2\pi \Gamma_{\lambda, 1}(e^y)\right\} \\[1mm]
\sin\left\{2\pi \Gamma_{\lambda, 1}(e^y)\right\}
\end{bmatrix}.
\]
Note that $\bb{\tau}(x, \lambda,\mu, \sigma) = \bb{\tau}(y, \lambda,0, 1)$. Since $\EE(\tau_1(Y, \lambda,0, 1)) = \EE(\tau_2(Y, \lambda,0, 1)) = 0$, one has
\begin{align*}
\EE(\tau_1(X, \lambda,\mu, \sigma) s_1(X, \lambda,\mu, \sigma))
&= \EE\left\{\tau_1(Y,\lambda, 0, 1) (Y - \psi(\lambda))\right\}= \EE\left\{\cos\left(2\pi \Gamma_{\lambda, 1}(e^Y)\right)Y\right\}\\
& = \EE\left\{\ln(V)\cos\left(2\pi \Gamma_{\lambda, 1}(V)\right)\right\}\\
&= \int_0^{\infty} \ln(v)\cos\left\{2\pi \Gamma_{\lambda, 1}(v)\right\} f_{\mathrm{ga}}(v \nvert \lambda, 1) \rd v = h_{10}(\lambda).
\end{align*}
Similarly, one shows that
\[
\EE(\tau_2(X, \lambda,\mu, \sigma) s_1(X, \lambda,\mu, \sigma))=h_{11}(\lambda).
\]

\begin{align*}
\sigma \EE(\tau_1(X, \lambda,\mu, \sigma) s_2(X, \lambda,\mu, \sigma))
&= \EE\left\{\tau_1(Y,\lambda, 0, 1) (e^Y - \lambda)\right\}
= \EE\left\{\cos\left(2\pi \Gamma_{\lambda, 1}(e^Y)\right)e^Y\right\} \\
&= \EE\left\{V\cos\left(2\pi \Gamma_{\lambda, 1}(V)\right)\right\}
= \int_0^{\infty} v\cos\left\{2\pi \Gamma_{\lambda, 1}(v)\right\} f_{\mathrm{ga}}(v \nvert \lambda, 1) \rd v \\
&= \lambda\int_0^{\infty} \cos\left\{2\pi \Gamma_{\lambda, 1}(v)\right\} f_{\mathrm{ga}}(v \nvert \lambda+1, 1) \rd v = \lambda h_{6}(\lambda, \lambda + 1, 1).
\end{align*}

Similarly, one shows that
\[
\sigma \EE\left\{\tau_2(X, \lambda,\mu, \sigma) s_2(X, \lambda,\mu, \sigma)\right\} = \lambda h_{7}(\lambda, \lambda + 1, 1).
\]
Finally, one has
\begin{align*}
\sigma \EE\left\{\tau_1(X, \lambda,\mu, \sigma) s_3(X, \lambda,\mu, \sigma)\right\}
&= \EE\left\{\tau_1(Y,\lambda, 0, 1) (Y e^Y - \lambda Y - 1)\right\}
= \EE\left[\cos\left\{2\pi \Gamma_{\lambda, 1}(e^Y)\right\}(Y e^Y - \lambda Y)\right] \\
&= \EE\left[(V - \lambda) \ln(V) \cos\left\{2\pi \Gamma_{\lambda, 1}(V)\right\}\right] \\
&= \int_0^{\infty} (v - \lambda) \ln (v) \cos\left\{2\pi \Gamma_{\lambda, 1}(v)\right\} f_{\mathrm{ga}}(v \nvert \lambda, 1) \rd v = h_{8}(\lambda).
\end{align*}
Similarly, one shows that
\[
\sigma \EE\left\{\tau_2(X, \lambda,\mu, \sigma) s_3(X,\lambda, \mu, \sigma)\right\} = h_{9}(\lambda)
\]
and one obtains
\[
G(\lambda,\mu, \sigma)=
\begin{bmatrix}
h_{10}(\lambda) & \frac{\lambda h_{6}(\lambda, \lambda + 1, 1)}{\sigma} & \frac{h_{8}(\lambda)}{\sigma} \\[1mm]
h_{11}(\lambda) & \frac{\lambda h_{7}(\lambda, \lambda + 1, 1)}{\sigma} & \frac{h_{9}(\lambda)}{\sigma}
\end{bmatrix}.
\]

The proof is now carried out for the $\mathrm{GG}(\lambda,\beta,\rho)$ distribution.
Since
\[
 X \sim \mathrm{GG}(\lambda_0,\beta_0, \rho_0) \Leftrightarrow X'=\ln X\sim \mathrm{exp\mhyphen GG}(\lambda_0,\mu_0, \sigma_0),
\]
with $\mu_0=\ln \beta_0$ and $\sigma_0=1/\rho_0$, the ML estimators of $\beta_0$ and $\rho_0$ are found from the ML of $\mu_0$ and $\sigma_0$  by setting $\hat{\mu}_n=\ln(\hat{\beta}_n)$, $\hat{\sigma}_n=1/\hat{\rho}_n$ and $x'_i = \ln (x_i)$, where $x_i$ and $x'_i$ are the $\mathrm{GG}(\lambda_0,\beta_0, \rho_0)$ and $\mathrm{exp\mhyphen GG}(\lambda_0,\mu_0, \sigma_0)$ observations respectively.

Letting $y = \ln \left\{(x/\beta)^{\rho}\right\}$, straightforward calculations yield
\begin{align*}
\bb{s}(x, \lambda, \beta, \rho)
&=
\begin{bmatrix}
s_1(x, \lambda,\beta, \rho) \\[1mm]
s_2(x, \lambda,\beta, \rho) \\[1mm]
s_3(x, \lambda,\beta, \rho)
\end{bmatrix}
=
\begin{bmatrix}
\partial_{\lambda} \ln f(x \nvert \lambda, \beta, \rho) \\[1mm]
\partial_{\beta} \ln f(x \nvert \lambda, \beta, \rho) \\[1mm]
\partial_{\rho} \ln f(x \nvert \lambda, \beta, \rho)
\end{bmatrix}
=
\begin{bmatrix}
\ln \left\{(x/\beta)^{\rho}\right\} - \psi(\lambda) \\[1mm]
\frac{\rho}{\beta}\{(x/\beta)^{\rho} - \lambda\} \\[1mm]
\frac{-1}{\rho}\left[\ln \left\{(x/\beta)^{\rho}\right\}\{ (x/\beta)^{\rho} - \lambda\} - 1\right]
\end{bmatrix}\\
&=
\begin{bmatrix}
y - \psi(\lambda) \\[1mm]
\frac{\rho}{\beta}(e^y - \lambda) \\[1mm]
\frac{-1}{\rho}(y e^y - \lambda y - 1)
\end{bmatrix}.
\end{align*}

Furthermore, one has
\[
X\sim \mathrm{GG}(\lambda, \beta, \rho)
~~ \Leftrightarrow ~~ (X/\beta)^{\rho}\sim \mathrm{GG}(\lambda, 1, 1)
~~ \Leftrightarrow ~~ Y = \ln \left\{(X/\beta)^{\rho}\right\}\sim \mathrm{exp\mhyphen GG}(\lambda, 0, 1).
\]

Next, one determines the Fisher information matrix $I(\lambda,\beta,\rho)=\EE\big\{\bb{s}(X, \lambda, \beta,\rho) \bb{s}(X, \lambda,\beta,\rho)^{\top}\big\}$. Using results above, one has
\[
\EE\left\{s_1^2(X, \lambda,\beta, \rho)\right\} = \EE\left[\{Y - \psi(\lambda)\}^2\right] = \psi_1(\lambda),
\]
\[
\EE\left\{s_2^2(X, \lambda,\beta, \rho)\right\} = \frac{\rho^2}{\beta^2} \EE\left\{(e^Y - \lambda)^2\right\} = \frac{\rho^2\lambda}{\beta^2},
\]
\[
\EE\left\{s_3^2(X, \lambda, \beta, \rho)\right\}
= \frac{1}{\rho^2} \EE\left\{(Y e^Y - \lambda Y - 1)^2\right\}
= \frac{1}{\rho^2} \left\{\lambda \psi^2(\lambda) + 2 \psi(\lambda) + \lambda \psi_1(\lambda) + 1\right\},
\]
\[
\EE\left\{s_1(X, \lambda, \beta, \rho) s_2(X, \lambda, \beta, \rho)\right\}
= \frac{\rho}{\beta}\EE\left[ \{Y - \psi(\lambda)\}(e^Y - \lambda) \right]= \frac{\rho}{\beta},
\]
\[
\EE\left\{s_1(X, \lambda, \beta, \rho) s_3(X, \lambda, \beta, \rho)\right\}
= \frac{-1}{\rho}\EE\left[ \{Y - \psi(\lambda)\}(Y e^Y - \lambda Y - 1) \right]= -\frac{\psi(\lambda)}{\rho},
\]
\[
\EE\left\{s_2(X, \lambda, \beta, \rho) s_3(X, \lambda, \beta, \rho)\right\}
= \frac{-1}{\beta} \EE\left\{(e^Y - \lambda) (Y e^Y - \lambda Y - 1)\right\}
= \frac{-1}{\beta} \left\{\lambda \psi(\lambda) + 1\right\}.
\]

The Fisher information matrix is then given by
\[
I(\lambda,\beta,\rho)
=
\begin{bmatrix}
 \psi_1(\lambda) & \frac{\rho}{\beta} & -\frac{\psi(\lambda)}{\rho} \\
\frac{\rho}{\beta} & \frac{\rho^2\lambda}{\beta^2} & \frac{-1}{\beta} \left\{\lambda \psi(\lambda) + 1\right\} \\
-\frac{\psi(\lambda)}{\rho} & \frac{-1}{\beta} \left\{\lambda\psi(\lambda) + 1\right\} & \frac{1}{\rho^2} \left\{\lambda \psi^2(\lambda) + 2 \psi(\lambda) + \lambda \psi_1(\lambda) + 1\right\}
\end{bmatrix}.
\]

Next, one determines the matrix $G(\lambda, \beta, \rho)=\EE\big\{\bb{\tau}(X, \lambda, \beta, \rho) \bb{s}(X, \lambda, \beta, \rho)^{\top}\big\}$. One has, if $y = \ln\{(x/\beta)^{\rho}\}$,
\[
\bb{\tau}(x, \lambda, \beta, \rho)
=
\begin{bmatrix}
\tau_1(x, \lambda, \beta, \rho) \\[1mm]
\tau_2(x, \lambda, \beta, \rho)
\end{bmatrix}
=
\begin{bmatrix}
\cos\left\{2\pi F(x \nvert \lambda, \beta, \rho)\right\} \\[1mm]
\sin\left\{2\pi F(x \nvert \lambda, \beta, \rho)\right\}
\end{bmatrix}
=
\begin{bmatrix}
\cos\left[2\pi \Gamma_{\lambda, 1} \left\{(x/\beta)^{\rho}\right\}\right] \\[1mm]
\sin\left[2\pi \Gamma_{\lambda, 1} \left\{(x/\beta)^{\rho}\right\}\right]
\end{bmatrix}
=
\begin{bmatrix}
\cos\left\{2\pi \Gamma_{\lambda, 1}(e^y)\right\} \\[1mm]
\sin\left\{2\pi \Gamma_{\lambda, 1}(e^y)\right\}
\end{bmatrix}.
\]
Then using results found above, one has
\begin{align*}
\EE\left\{\tau_1(X, \lambda, \beta, \rho) s_1(X, \lambda, \beta, \rho)\right\}
&= \EE\left[\cos\left\{2\pi \Gamma_{\lambda, 1}(e^Y)\right\} \{Y - \psi(\lambda)\}\right]
= h_{10}(\lambda), \\
\EE\left\{\tau_2(X, \lambda, \beta, \rho) s_1(X, \lambda, \beta, \rho)\right\} &= \EE\left[\sin\left\{2\pi \Gamma_{\lambda, 1}(e^Y)\right\} \{Y - \psi(\lambda)\}\right]
= h_{11}(\lambda), \\
\EE\left\{\tau_1(X, \lambda, \beta, \rho) s_2(X, \lambda, \beta, \rho)\right\}
&= \frac{\rho}{\beta} \EE\left[\cos\left\{2\pi \Gamma_{\lambda, 1}(e^Y)\right\} (e^Y - \lambda)\right]
= \frac{\rho\lambda}{\beta} h_{6}(\lambda,\lambda+1,1), \\
\EE\left\{\tau_2(X, \lambda, \beta, \rho) s_2(X, \lambda, \beta, \rho)\right\}
&= \frac{\rho}{\beta} \EE\left[\sin\left\{2\pi \Gamma_{\lambda, 1}(e^Y)\right\} (e^Y - \lambda)\right]
= \frac{\rho\lambda}{\beta} h_{7}(\lambda,\lambda+1,1), \\
\EE\left\{\tau_1(X, \lambda, \beta, \rho) s_3(X, \lambda, \beta, \rho)\right\}
&= \frac{-1}{\rho} \EE\left[\cos\left\{2\pi \Gamma_{\lambda, 1}(e^Y)\right\} (Y e^Y - \lambda Y - 1)\right]
= \frac{-1}{\rho} h_{8}(\lambda), \\
\EE\left\{\tau_2(X, \lambda, \beta, \rho) s_3(X, \lambda, \beta, \rho)\right\}
&= \frac{-1}{\rho} \EE\left[\sin\left\{2\pi \Gamma_{\lambda, 1}(e^Y)\right\} (Y e^Y - \lambda Y - 1)\right]
= \frac{-1}{\rho} h_{9}(\lambda).
\end{align*}
One obtains
\[
G(\lambda, \beta, \rho)=
\begin{bmatrix}
h_{10}(\lambda) & \frac{\rho\lambda}{\beta} h_{6}(\lambda,\lambda+1,1) & \frac{-1}{\rho} h_{8}(\lambda) \\[1mm]
h_{11}(\lambda) & \frac{\rho\lambda}{\beta} h_{7}(\lambda,\lambda+1,1) & \frac{-1}{\rho} h_{9}(\lambda)
\end{bmatrix}.
\]

\section{Proof for the \texorpdfstring{$\mathrm{logistic}(\mu, \sigma)$}{logistic} test}\label{proof:logistic}

Letting $y = (x - \mu)/\sigma$, straightforward calculations yield
\[
\bb{s}(x, \mu, \sigma)
=
\begin{bmatrix}
s_1(x, \mu, \sigma) \\[1mm]
s_2(x, \mu, \sigma)
\end{bmatrix}
=
\begin{bmatrix}
\partial_{\mu} \ln f(x \nvert \mu, \sigma) \\[1mm]
\partial_{\sigma} \ln f(x \nvert \mu, \sigma)
\end{bmatrix}
=\frac{1}{\sigma}
\begin{bmatrix}
1 - \frac{2}{1 + e^{y}} \\[1mm]
y - \frac{2y}{1 + e^{y}} - 1
\end{bmatrix}
=
\frac{1}{\sigma}
\begin{bmatrix}
\frac{1 - e^{-y}}{1 + e^{-y}} \\[1mm]
\frac{y(1 - e^{-y})}{1 + e^{-y}} - 1
\end{bmatrix}.
\]
The ML estimators for $\mu_0$ and/or $\sigma_0$ come directly from the equation $\sum_{i=1}^{n}\bb{s}(x_i, \mu_0, \sigma_0)=0$.

Next, one determines the Fisher information matrix $I(\mu, \sigma)=\EE\big\{\bb{s}(X, \mu, \sigma) \bb{s}(X, \mu, \sigma)^{\top}\big\}$.
Given that $s_1(y, 0, 1)$ is an odd function, $s_1^2(y, 0, 1)$ and $f(y \nvert 0, 1)$ are even functions, and setting $Y=(X-\mu)/\sigma\sim \mathrm{logistic}(0,1)$ one has
\begin{align*}
\EE\left\{s_1^2(X, \mu, \sigma)\right\}
&= \frac{1}{\sigma^2} \EE\left\{s_1^2(Y, 0, 1)\right\}
= \frac{2}{\sigma^2} \int_0^{\infty} \frac{(e^{-y} - 1)^2 e^{-y}}{(1 + e^{-y})^4} \rd y
= \frac{2}{\sigma^2} \int_1^2 \frac{(v - 2)^2}{v^4} \rd v \\
&= \frac{1}{\sigma^2} \left(\int_1^2 2v^{-2} \rd v - \int_1^2 8v^{-3} \rd v + \int_1^2 8v^{-4} \rd v\right)
= \frac{1}{\sigma^2}(1 - 3 + 7/3) = \frac{1}{3\sigma^2},
\end{align*}
using the change of variable $v = 1 + e^{-y}$. Given that $s_2(y, 0, 1)$, $s_2^2(y, 0, 1)$ and $f(y \nvert 0, 1)$ are even functions, one has
\[
\EE\left\{s_2^2(X, \mu, \sigma)\right\}
= \frac{1}{\sigma^2} \EE\left\{s_2^2(Y, 0, 1)\right\}
= \frac{2}{\sigma^2} \int_0^{\infty} \left(\frac{y (1 - e^{-y})}{1 + e^{-y}} - 1\right)^2 \frac{e^{-y}}{(1 + e^{-y})^2} \rd y
= \frac{3 + \pi^2}{9\sigma^2},
\]
using \texttt{Mathematica} for the second equality. Given that $s_1(y, 0, 1)$ is an odd function, $s_2(y, 0, 1)$ and $f(y \nvert 0, 1)$ are even functions, one has
\[
\EE\left\{s_1(X, \mu, \sigma) s_2(X, \mu, \sigma)\right\}
= \sigma^{-2} \EE\left\{s_1(Y, 0, 1) s_2(Y, 0, 1)\right\}
= 0.
\]
The Fisher information matrix is then given by
\[
I(\mu, \sigma)
=
\frac{1}{\sigma^2}
\begin{bmatrix}
1/3 & 0 \\[1mm]
0 & (3 + \pi^2)/9
\end{bmatrix}.
\]

Next, one determines the matrix $G(\mu, \sigma)=\EE\left\{\bb{\tau}(X, \mu, \sigma) \bb{s}(X, \mu, \sigma)^{\top}\right\}$. One has, if $y = (x - \mu)/\sigma$,
\[
\bb{\tau}(x, \mu, \sigma)
=
\begin{bmatrix}
\tau_1(x, \mu, \sigma) \\[1mm]
\tau_2(x, \mu, \sigma)
\end{bmatrix}
=
\begin{bmatrix}
\cos\left\{2\pi F(x \nvert \mu, \sigma)\right\} \\[1mm]
\sin\left\{2\pi F(x \nvert \mu, \sigma)\right\}
\end{bmatrix}
=
\begin{bmatrix}
\cos\left\{2\pi F(y \nvert 0, 1)\right\} \\[1mm]
\sin\left\{2\pi F(y \nvert 0, 1)\right\}
\end{bmatrix}
=
\begin{bmatrix}
\cos\left\{2\pi (1 + e^{-y})^{-1}\right\} \\[1mm]
\sin\left\{2\pi (1 + e^{-y})^{-1}\right\}
\end{bmatrix}.
\]
Note that $\bb{\tau}(x, \mu, \sigma) = \bb{\tau}(y, 0, 1)$. Now, since $\tau_2(y, 0, 1)$ and $s_1(Y, 0, 1)$ are odd functions, their product is an even function, as is $f(y \nvert 0, 1)$. Then one has
\begin{align*}
\EE\left\{\tau_2(X, \mu, \sigma) s_1(X, \mu, \sigma)\right\}
&= \frac{1}{\sigma} \EE\left\{\tau_2(Y, 0, 1) s_1(Y, 0, 1)\right\}
= \frac{2}{\sigma} \int_0^{\infty} \sin\left(\frac{2\pi}{1 + e^{-y}}\right) \frac{e^{-y}(1 - e^{-y})}{(1 + e^{-y})^3} \rd y \\
&= \frac{2}{\sigma} \int_0^1 \sin\left(\frac{2\pi}{1 + v}\right) \frac{1 - v}{(1 + v)^3} \rd v
= \frac{-1}{\sigma\pi},
\end{align*}
using the change of variable $v = e^{-y}$ and \texttt{Mathematica}.

Given that $\EE\{\tau_1(Y, 0, 1)\} = 0$, $\tau_1(y, 0, 1)$ and $s_2(Y, 0, 1)$ are even functions, as is $f(y \nvert 0, 1)$, one has
\begin{align*}
\EE\left\{\tau_1(X, \mu, \sigma) s_2(X, \mu, \sigma)\right\}
&= \frac{1}{\sigma} \EE\left\{\tau_1(Y, 0, 1) s_2(Y, 0, 1)\right\}
= \frac{1}{\sigma} \EE\left[\tau_1(Y, 0, 1) \left\{\frac{Y (1 - e^{-Y})}{1 + e^{-Y}} - 1\right\}\right] \\
&= \frac{2}{\sigma} \int_0^{\infty} \cos\left(\frac{2\pi}{1 + e^{-y}}\right) \left\{\frac{y (1 - e^{-y})}{1 + e^{-y}}\right\} \frac{e^{-y}}{(1 + e^{-y})^2} \rd y \\
&= \frac{2}{\sigma} \int_0^{\infty} \cos\left(\frac{2\pi}{1 + e^{-y}}\right) \frac{y e^{-y}(1 - e^{-y})}{(1 + e^{-y})^3} \rd y \\
&= -\frac{2}{\sigma} \int_0^1 \cos\left(\frac{2\pi}{1 + v}\right) \frac{(1 - v) \ln v}{(1 + v)^3} \rd v
= \frac{0.698397593884459}{\sigma},
\end{align*}
using the change of variable $v = e^{-y}$ and numerical calculation for the last integral (with an accuracy of 15 significant digits). Given that $\tau_2(y, 0, 1)$ and $s_1(Y, 0, 1)$ are odd functions, that $\tau_1(y, 0, 1)$, $s_2(Y, 0, 1)$ and $f(y \nvert 0, 1)$ are even functions, one has
\begin{align*}
\EE\left\{\tau_1(X, \mu, \sigma) s_1(X, \mu, \sigma)\right\}
&= \frac{1}{\sigma} \EE\left\{\tau_1(Y, 0, 1) s_1(Y, 0, 1)\right\}
= 0, \\
\EE\left\{\tau_2(X, \mu, \sigma) s_2(X, \mu, \sigma)\right\}
&= \frac{1}{\sigma} \EE\left\{\tau_2(Y, 0, 1) s_2(Y, 0, 1)\right\}
= 0,
\end{align*}
and one obtains
\[
G(\mu, \sigma)
=
\frac{1}{\sigma}
\begin{bmatrix}
0 & 0.698397593884459 \\[1mm]
-1/\pi & 0
\end{bmatrix}.
\]

\section{Proof for the \texorpdfstring{$\mathrm{Student's}~t(\lambda, \mu, \sigma)$}{Student's t} test}\label{proof:Student}

Letting $y = (x - \mu)/\sigma$, straightforward calculations yield
\begin{align*}
\bb{s}(x, \lambda,\mu, \sigma)
&=
\begin{bmatrix}
s_1(x, \lambda,\mu, \sigma) \\[1mm]
s_2(x, \lambda,\mu, \sigma) \\[1mm]
s_3(x, \lambda,\mu, \sigma)
\end{bmatrix}
=
\begin{bmatrix}
\partial_{\lambda} \ln\{f(x \nvert \lambda, \mu, \sigma)\} \\[1mm]
\partial_{\mu} \ln\{f(x \nvert \lambda, \mu, \sigma)\} \\[1mm]
\partial_{\sigma} \ln\{f(x \nvert \lambda, \mu, \sigma)\}
\end{bmatrix}\\
& =
\begin{bmatrix}
\frac{1}{2}\left\{\psi\left( \frac{\lambda + 1}{2} \right)
- \psi\left( \frac{\lambda}{2} \right)
- \ln\left(1 + \frac{y^2}{\lambda} \right)
+ \frac{1}{\lambda}\left(\frac{\lambda + 1}{\lambda}\frac{y^2}{1 + y^2/\lambda}-1\right)\right\}\\[1mm]
\frac{1}{\sigma}\frac{\lambda + 1}{\lambda} \frac{y}{1 + y^2/\lambda} \\[1mm]
\frac{1}{\sigma}\left(\frac{\lambda + 1}{\lambda} \frac{y^2}{1 + y^2/\lambda} - 1\right)
\end{bmatrix}.
\end{align*}
The ML estimators for $\lambda_0$, $\mu_0$ and/or $\sigma_0$ come directly from the equation $\sum_{i=1}^{n}\bb{s}(x_i, \lambda_0, \mu_0, \sigma_0)=0$. Furthermore, one has
\[
 X\sim t(\lambda,\mu, \sigma)\Leftrightarrow Y=\frac{X-\mu}{\sigma}\sim t(\lambda,0, 1)\Leftrightarrow V=(1+Y^2/\lambda)^{-1}\sim \mathrm{Beta}\left(\frac{\lambda}{2},\frac{1}{2}\right).
\]

Next, one determines the Fisher information matrix $I(\lambda,\mu, \sigma)=\EE\big\{\bb{s}(X, \lambda,\mu, \sigma) \bb{s}(X, \lambda,\mu, \sigma)^{\top}\big\}$.
\begin{lemma}\label{lem:Student}
If $a\in (0, \infty)$, $b\in (0, \infty)$, then one has
\[
\EE\left\{\frac{|Y|^{2a}}{(1 + Y^2/\lambda)^b}\right\}=\frac{\lambda^a B(\lambda/2 + b - a,1/2+a)}{B(\lambda/2,1/2)}.
\]
\end{lemma}
\begin{proof}
Using $V=(1+Y^2/\lambda)^{-1}\Leftrightarrow Y^2=\lambda(1-V)/V$ and $V\sim \mathrm{Beta}\left(\frac{\lambda}{2},\frac{1}{2}\right)$, one has
\[
\EE\left\{\frac{|Y|^{2a}}{(1 + Y^2/\lambda)^b}\right\}=\lambda^a\EE\left\{V^{b-a}(1-V)^a\right\}=\frac{\lambda^a B(\lambda/2 + b - a,1/2+a)}{B(\lambda/2,1/2)}.
\]
\end{proof}

In particular, one has
\begin{align*}
\EE\left\{\frac{Y^2}{(1 + Y^2/\lambda)^2}\right\}
&=\frac{\lambda B(\lambda/2+1,3/2)}{B(\lambda/2,1/2)}= \frac{\lambda^2}{(\lambda + 1) (\lambda + 3)}, \\
\EE\left\{\frac{Y^2}{1 + Y^2/\lambda}\right\}
&=\frac{\lambda B(\lambda/2,3/2)}{B(\lambda/2,1/2)}= \frac{\lambda}{\lambda + 1}, \\
\EE\left\{\frac{Y^4}{(1 + Y^2/\lambda)^2}\right\}
& =\frac{\lambda^2 B(\lambda/2,5/2)}{B(\lambda/2,1/2)}= \frac{3\lambda^2}{(\lambda + 1) (\lambda + 3)}.
\end{align*}

It can also be verified in \texttt{Mathematica} that
\[
\EE\left\{\ln\left(1 + \frac{Y^2}{\lambda}\right)\right\}= \psi\left(\frac{\lambda+1}{2}\right) - \psi\left(\frac{\lambda}{2}\right),
\]
\[
\EE\left[\left\{\ln\left(1 + \frac{Y^2}{\lambda}\right)\right\}^2\right]=\psi_1\left(\frac{\lambda}{2}\right)-\psi_1\left(\frac{\lambda+1}{2}\right) + \left\{\psi\left(\frac{\lambda+1}{2}\right) - \psi\left(\frac{\lambda}{2}\right)\right\}^2,
\]
\[
\EE\left\{\frac{Y^2}{1+Y^2/\lambda}\ln\left(1 + \frac{Y^2}{\lambda}\right)\right\}=\frac{\lambda}{\lambda+1} \left\{\psi\left(\frac{\lambda+1}{2}\right) - \psi\left(\frac{\lambda}{2}\right)+\frac{2}{\lambda+1}\right\}.
\]

Now, letting $K_{\lambda}=\psi\left( \frac{\lambda + 1}{2} \right)
- \psi\left( \frac{\lambda}{2}\right)$, one has
\begin{align*}
4\EE&\left\{s_1^2(X, \lambda,\mu, \sigma)\right\}
=\EE\left[s_1(Y, \lambda,0,1)\left\{- \ln\left(1 + \frac{Y^2}{\lambda} \right) + \frac{\lambda + 1}{\lambda^2}\frac{Y^2}{1 + Y^2/\lambda}\right\}\right]\\
&=\EE\left[\left\{K_{\lambda} - \frac{1}{\lambda} - \ln\left(1 + \frac{Y^2}{\lambda} \right) + \frac{\lambda + 1}{\lambda^2}\frac{Y^2}{1 + Y^2/\lambda}\right\}\left\{- \ln\left(1 + \frac{Y^2}{\lambda} \right) + \frac{\lambda + 1}{\lambda^2}\frac{Y^2}{1 + Y^2/\lambda}\right\}\right]\\
&=\left(K_{\lambda} - \frac{1}{\lambda}\right)\frac{\lambda + 1}{\lambda^2}\EE \left\{\frac{Y^2}{1 + Y^2/\lambda}\right\}
- \left(K_{\lambda} - \frac{1}{\lambda}\right)\EE \left\{\ln\left(1 + \frac{Y^2}{\lambda} \right)\right\}+\EE\left[\left\{\ln\left(1 + \frac{Y^2}{\lambda}\right)\right\}^2\right]\\
&-2\frac{\lambda + 1}{\lambda^2}\EE\left\{\frac{Y^2}{1+Y^2/\lambda}\ln\left(1 + \frac{Y^2}{\lambda}\right)\right\}+\frac{(\lambda + 1)^2}{\lambda^4}\EE\left\{\frac{Y^4}{(1 + Y^2/\lambda)^2}\right\}\\
&=\left(K_{\lambda} - \frac{1}{\lambda}\right)\frac{\lambda + 1}{\lambda^2}\frac{\lambda}{\lambda + 1}
- \left(K_{\lambda} - \frac{1}{\lambda}\right)K_{\lambda}+\psi_1\left(\frac{\lambda}{2}\right)-\psi_1\left(\frac{\lambda+1}{2}\right) + K_{\lambda}^2\\
&-2\frac{\lambda + 1}{\lambda^2}\frac{\lambda}{\lambda+1} \left(K_{\lambda}+\frac{2}{\lambda+1}\right)+\frac{(\lambda + 1)^2}{\lambda^4}\frac{3\lambda^2}{(\lambda + 1) (\lambda + 3)}\\
&=\frac{K_{\lambda}}{\lambda} -\frac{1}{\lambda^2}
- K_{\lambda}^2 + \frac{K_{\lambda}}{\lambda} +\psi_1\left(\frac{\lambda}{2}\right)-\psi_1\left(\frac{\lambda+1}{2}\right) + K_{\lambda}^2-\frac{2 K_{\lambda}}{\lambda}-\frac{4}{\lambda(\lambda+1)}+\frac{3(\lambda + 1)}{\lambda^2(\lambda + 3)}\\
&= \psi_1\left(\frac{\lambda}{2}\right)-\psi_1\left(\frac{\lambda+1}{2}\right)-\frac{1}{\lambda^2}-\frac{4}{\lambda(\lambda+1)}+\frac{3(\lambda + 1)}{\lambda^2(\lambda + 3)}
= \psi_1\left(\frac{\lambda}{2}\right)-\psi_1\left(\frac{\lambda+1}{2}\right)-\frac{2(\lambda+5)}{\lambda(\lambda+1)(\lambda+3)},
\end{align*}

\begin{align*}
\EE\left\{s_2^2(X, \lambda,\mu, \sigma)\right\}
&= \frac{1}{\sigma^2} \EE\left\{s_2^2(Y, \lambda, 0, 1)\right\}
= \frac{(\lambda + 1)^2}{\sigma^2\lambda^2} \EE\left\{\frac{Y^2}{(1 + Y^2/\lambda)^2}\right\}
= \frac{(\lambda + 1)^2}{\lambda^2\sigma^2} \frac{\lambda^2}{(\lambda + 1) (\lambda + 3)}
= \frac{1}{\sigma^2} \frac{\lambda + 1}{\lambda + 3}, \\
\EE\left\{s_3^2(X, \lambda,\mu, \sigma)\right\}
&= \frac{1}{\sigma^2} \EE\left\{s_3^2(Y, \lambda, 0, 1)\right\}
= \frac{1}{\sigma^2} \EE\left[\left\{\frac{(\lambda + 1)}{\lambda} \frac{Y^2}{1 + Y^2/\lambda} - 1\right\}^2\right] \\
&= \frac{1}{\sigma^2} \left[\frac{(\lambda + 1)^2}{\lambda^2} \EE\left\{\frac{Y^4}{(1 + Y^2/\lambda)^2}\right\} - \frac{2(\lambda + 1)}{\lambda} \EE\left(\frac{Y^2}{1 + Y^2/\lambda}\right) + 1\right] \\
&= \frac{1}{\sigma^2} \left\{\frac{(\lambda + 1)^2}{\lambda^2} \frac{3\lambda^2}{(\lambda + 1) (\lambda + 3)} - \frac{2(\lambda + 1)}{\lambda} \frac{\lambda}{\lambda + 1} + 1\right\}
= \frac{1}{\sigma^2} \left\{\frac{3(\lambda + 1)}{(\lambda + 3)} - 1\right\}
= \frac{1}{\sigma^2} \frac{2\lambda}{\lambda + 3},
\end{align*}
and
\begin{align*}
2\sigma\EE&\left\{s_1(X, \lambda,\mu, \sigma)s_3(X, \lambda,\mu, \sigma)\right\}=2\EE\left\{s_1(Y, \lambda,0,1)s_3(Y, \lambda,0,1)\right\}\\
&=\EE\left[\left\{- \ln\left(1 + \frac{Y^2}{\lambda} \right)+ \frac{\lambda + 1}{\lambda^2}\frac{Y^2}{1 + Y^2/\lambda}\right\}s_3(Y, \lambda,0,1)\right]\\
&=\EE\left[\left\{- \ln\left(1 + \frac{Y^2}{\lambda} \right)+ \frac{\lambda + 1}{\lambda^2}\frac{Y^2}{1 + Y^2/\lambda}\right\}\left(\frac{\lambda + 1}{\lambda} \frac{Y^2}{1 + Y^2/\lambda} - 1\right)\right]\\
&=-\frac{\lambda + 1}{\lambda}\EE\left\{\ln\left(1 + \frac{Y^2}{\lambda} \right)\frac{Y^2}{1 + Y^2/\lambda}\right\}+\EE\left\{\ln\left(1 + \frac{Y^2}{\lambda} \right)\right\}\\
&+\frac{(\lambda + 1)^2}{\lambda^3}\EE\left\{\frac{Y^4}{(1 + Y^2/\lambda)^2}\right\}-\frac{\lambda + 1}{\lambda^2}\EE\left\{\frac{Y^2}{1 + Y^2/\lambda}\right\}\\
&=-\frac{\lambda + 1}{\lambda}\frac{\lambda}{\lambda+1} \left\{K_{\lambda}+\frac{2}{\lambda+1}\right\}+K_{\lambda}+
\frac{(\lambda + 1)^2}{\lambda^3}\frac{3\lambda^2}{(\lambda + 1) (\lambda + 3)}-\frac{\lambda + 1}{\lambda^2}\frac{\lambda}{\lambda + 1}\\
&=-\frac{2}{\lambda+1}+\frac{3(\lambda + 1)}{\lambda(\lambda + 3)}-\frac{1}{\lambda}=\frac{-4}{(\lambda+1)(\lambda+3)},
\end{align*}

\[
\EE\left\{s_1(X, \lambda, \mu, \sigma) s_2(X, \lambda, \mu, \sigma)\right\}
= \sigma^{-1} \EE\left\{s_1(Y, \lambda, 0, 1) s_2(Y, \lambda, 0, 1)\right\}
= 0,
\]
\[
\EE\left\{s_2(X, \lambda, \mu, \sigma) s_3(X, \lambda, \mu, \sigma)\right\}
= \sigma^{-2} \EE\left\{s_2(Y, \lambda, 0, 1) s_3(Y, \lambda, 0, 1)\right\}
= 0,
\]
given that $s_2(y, \lambda, 0, 1)$ is an odd function, $s_1(y, \lambda, 0, 1)$, $s_3(y, \lambda, 0, 1)$ and $f(y \nvert \lambda, 0, 1)$ are even functions. The Fisher information matrix is then given by
\[
I(\lambda, \mu, \sigma)
=
\begin{bmatrix}
\frac{1}{4}\left\{\psi_1\left(\frac{\lambda}{2}\right)-\psi_1\left(\frac{\lambda+1}{2}\right)-\frac{2(\lambda+5)}{\lambda(\lambda+1)(\lambda+3)}\right\} & 0 & \frac{-2}{\sigma(\lambda+1)(\lambda+3)} \\[1mm]
0 & \frac{\lambda + 1}{\sigma^2(\lambda + 3)} & 0 \\[1mm]
\frac{-2}{\sigma(\lambda+1)(\lambda+3)} & 0 & \frac{2\lambda}{\sigma^2(\lambda + 3)}
\end{bmatrix}.
\]

Next, one determines the matrix $G(\lambda,\mu, \sigma)=\EE\left\{\bb{\tau}(X, \lambda,\mu, \sigma) \bb{s}(X, \lambda,\mu, \sigma)^{\top}\right\}$. We define
\[
C_{4,\lambda} =\frac{1}{\sqrt{\lambda} B\left(\frac{\lambda}{2}, \frac{1}{2}\right)}= \frac{\Gamma(\frac{\lambda+1}{2})}{\sqrt{\lambda\pi}\Gamma(\frac{\lambda}{2})}.
\]

One has, if $y = (x - \mu)/\sigma$,
\[
\bb{\tau}(x, \lambda,\mu, \sigma)
=
\begin{bmatrix}
\tau_1(x, \lambda,\mu, \sigma) \\[1mm]
\tau_2(x, \lambda,\mu, \sigma)
\end{bmatrix}
=
\begin{bmatrix}
\cos\left\{2\pi F(x \nvert \lambda,\mu, \sigma)\right\} \\[1mm]
\sin\left\{2\pi F(x \nvert \lambda,\mu, \sigma)\right\}
\end{bmatrix}
=
\begin{bmatrix}
\cos\left\{2\pi F(y \nvert \lambda,0, 1)\right\} \\[1mm]
\sin\left\{2\pi F(y \nvert \lambda,0, 1)\right\}
\end{bmatrix}.
\]
Note that $\bb{\tau}(x, \lambda,\mu, \sigma) = \bb{\tau}(y, \lambda,0, 1)$. Now, using that $\EE\left\{\tau_1(Y \nvert \lambda,0, 1)\right\} = 0$, the fact that $\tau_1(y, \lambda,0, 1)$ and $s_3(Y, \lambda,0, 1)$ are even functions, as is $f(y \nvert \lambda,0, 1)$, one has
\begin{align*}
\EE&\left\{\tau_1(X, \lambda, \mu, \sigma) s_3(X, \lambda,\mu, \sigma)\right\}
= \sigma^{-1} \EE\left\{\tau_1(Y \nvert \lambda,0, 1) s_3(Y \nvert \lambda,0, 1)\right\} \\
&= \sigma^{-1} \EE\left[\cos\left\{2\pi F(Y \nvert \lambda,0, 1)\right\} \left(\frac{\lambda + 1}{\lambda} \frac{Y^2}{1 + Y^2/\lambda} - 1\right)\right] \\
&= \frac{\lambda + 1}{\sigma\lambda} \EE\left[\cos\left\{2\pi F(Y \nvert \lambda,0, 1)\right\} \frac{Y^2}{1 + Y^2/\lambda}\right] \\
&= \frac{2(\lambda + 1)}{\sigma\lambda} \int_0^{\infty} \cos\left[\pi \left\{2 - F_{\mathrm{be}}((1 + y^2/\lambda)^{-1}\nvert\lambda/2, 1/2)\right\}\right] \frac{y^2}{1 + y^2/\lambda} f(y \nvert \lambda,0, 1) \rd y \\
&= \frac{2(\lambda + 1)}{\sigma\lambda} \frac{\lambda^{3/2}C_{4,\lambda}}{2} \int_0^1 \cos\left[\pi \left\{2 - F_{\mathrm{be}}(v\nvert\lambda/2, 1/2)\right\}\right] v^{\lambda/2-1}(1 - v)^{1/2} \rd v \\
&= \frac{(\lambda + 1)}{\sigma} \frac{B(\lambda/2, 3/2)}{B(\lambda/2, 1/2)} \int_0^1 \cos\left[\pi \left\{2 - F_{\mathrm{be}}(v\nvert\lambda/2, 1/2)\right\}\right] f_{\mathrm{be}}(v\nvert\lambda/2, 3/2) \rd v \\
&= \frac{1}{\sigma}h_{12}(\lambda),
\end{align*}
using the change of variables $v = (1 + y^2/\lambda)^{-1} \Leftrightarrow y = \sqrt{\lambda(1 - v)/v}$ and $\rd y = \sqrt{\lambda} v^{-3/2}(1 - v)^{-1/2}/2$.

Using the fact that $\tau_2(y, \lambda,0, 1)$ and $s_2(Y, \lambda,0, 1)$ are odd functions and that $f(y \nvert \lambda,0, 1)$ is an even function, one has
\begin{align*}
\EE\left\{\tau_2(X, \lambda,\mu, \sigma) s_2(X, \lambda,\mu, \sigma)\right\}
&= \sigma^{-1} \EE\left\{\tau_2(Y, \lambda,0, 1) s_2(Y, \lambda,0, 1)\right\} \\
&= \frac{2(\lambda + 1)}{\sigma\lambda} \int_0^{\infty} \sin\left\{2\pi F(y \nvert \lambda, 0, 1)\right\} \frac{y}{1 + y^2/\lambda} f(y \nvert \lambda, 0, 1) \rd y \\
&= \frac{2(\lambda + 1)}{\sigma\lambda} \frac{\lambda C_{4,\lambda}}{2} \int_0^1 \sin\left[\pi \left\{2 - F_{\mathrm{be}}(v\nvert\lambda/2, 1/2) \right\}\right] v^{(\lambda-1)/2} \rd v \\
&= \frac{(\lambda + 1)C_{4,\lambda}B(\frac{\lambda+1}{2}, 1)}{\sigma} \int_0^1 \sin\left[\pi \left\{2 - F_{\mathrm{be}}(v\nvert\lambda/2, 1/2) \right\}\right] f_{\mathrm{be}}(v\nvert(\lambda+1)/2, 1) \rd v \\
&= \frac{2C_{4,\lambda}}{\sigma} h_{13}(\lambda).
\end{align*}
\begin{align*}
\EE&\left\{\tau_1(X, \lambda,\mu, \sigma) s_1(X, \lambda,\mu, \sigma)\right\}
= \EE\left\{\tau_1(Y, \lambda,0, 1) s_1(Y, \lambda,0, 1)\right\} \\
&=\frac{1}{2}\EE\left[\tau_1(Y, \lambda,0, 1)\left\{- \ln\left(1 + \frac{Y^2}{\lambda} \right)
+ \frac{\lambda + 1}{\lambda^2}\frac{Y^2}{1 + Y^2/\lambda}\right\}\right] \\
&= \frac{1}{2}\int_0^{\infty} \cos\left\{2\pi F(y \nvert \lambda, 0, 1)\right\} \left\{- \ln\left(1 + \frac{y^2}{\lambda} \right)
+ \frac{\lambda + 1}{\lambda^2}\frac{y^2}{1 + y^2/\lambda}\right\} 2 f(y \nvert \lambda, 0, 1) \rd y \\
&=\frac{1}{2}\ \int_0^{1} \cos\left[\pi \left\{2 - F_{\mathrm{be}}(v\nvert\lambda/2, 1/2) \right\}\right] \left\{\ln(v)
+ \frac{\lambda + 1}{\lambda}(1-v)\right\}f_{\mathrm{be}}(v\nvert \lambda/2, 1) \rd v\\
& = \frac{1}{2} h_{14}(\lambda).
\end{align*}
Given that $\tau_2(y, \lambda,0, 1)$ and $s_2(Y, \lambda,0, 1)$ are odd functions, that $\tau_1(y, \lambda, 0, 1)$, $s_1(Y, \lambda,0, 1)$, $s_3(Y, \lambda,0, 1)$ and $f(y \nvert \lambda,0, 1)$ are even functions, one has
\begin{align*}
\EE\left\{\tau_2(X, \lambda,\mu, \sigma) s_1(X, \lambda, \mu, \sigma)\right\}
&= \EE\left\{\tau_2(Y, \lambda, 0, 1) s_1(Y, \lambda, 0, 1)\right\}
= 0, \\
\EE\left\{\tau_1(X, \lambda,\mu, \sigma) s_2(X, \lambda, \mu, \sigma)\right\}
&= \frac{1}{\sigma} \EE\left\{\tau_1(Y, \lambda, 0, 1) s_2(Y, \lambda, 0, 1)\right\}
= 0, \\
\EE\left\{\tau_2(X, \lambda, \mu, \sigma) s_3(X, \lambda, \mu, \sigma)\right\}
&= \frac{1}{\sigma} \EE\left\{\tau_2(Y, \lambda, 0, 1) s_3(Y, \lambda, 0, 1)\right\}
= 0
\end{align*}
and one obtains
\[
G(\lambda, \mu, \sigma)
=
\begin{bmatrix}
\frac{1}{2} h_{14}(\lambda) & 0 & \frac{h_{12}(\lambda)}{\sigma} \\[1mm]
0 & \frac{2 C_{4,\lambda} h_{13}(\lambda)}{\sigma} & 0
\end{bmatrix}.
\]

\section{Proof for the \texorpdfstring{$\mathrm{Gompertz}(\beta, \rho)$}{Gompertz} test}\label{proof:Gompertz}

Straightforward calculations yield
\[
\bb{s}(x, \beta, \rho)
=
\begin{bmatrix}
s_1(x, \beta, \rho) \\
s_2(x, \beta, \rho)
\end{bmatrix}
=
\begin{bmatrix}
\partial_{\beta} \ln f(x \nvert \beta, \rho) \\
\partial_{\rho} \ln f(x \nvert \beta, \rho)
\end{bmatrix}
=
\begin{bmatrix}
\frac{1}{\beta} + x - \rho x e^{\beta x} \\
\frac{1}{\rho} + 1 - e^{\beta x}
\end{bmatrix}.
\]
The ML estimators for $\beta_0$ and/or $\rho_0$ come directly from the equation $\sum_{i=1}^{n}\bb{s}(x_i, \beta_0, \rho_0)=0$.

Next, one determines the Fisher information matrix $I(\beta, \rho)=\EE\big\{\bb{s}(X, \beta, \rho) \bb{s}(X, \beta, \rho)^{\top}\big\}$.
One has, using the change of variables $v = \exp(\beta x)$ in the fourth equality,
\begin{align*}
- \EE\left\{\partial_{\beta} s_1(X, \beta, \rho)\right\}
&= \frac{1}{\beta^2} + \rho \EE\left(X^2 e^{\beta X}\right)
= \frac{1}{\beta^2} + \rho \int_0^{\infty} x^2 e^{\beta x}\beta\rho\exp(\rho + \beta x - \rho e^{\beta x}) \rd x \\
&= \frac{1}{\beta^2} \left\{1 + \rho^2 e^{\rho} \int_1^{\infty} (\ln v)^2 v e^{-\rho v} \rd v\right\}
= \frac{1 + \rho^2 e^{\rho} h_{17}(\rho)}{\beta^2}.
\end{align*}
One has
\[
- \EE\left\{\partial_{\rho} s_2(X, \beta, \rho)\right\}
= \frac{1}{\rho^2}.
\]
One has, using the change of variables $v = \exp(\beta x)$ in the fourth equality,
\begin{align*}
- \EE\left\{\partial_{\rho} s_1(X, \beta, \rho)\right\}
&= \EE\left(X e^{\beta X}\right)
= \int_0^{\infty} x e^{\beta x}\beta\rho\exp(\rho + \beta x - \rho e^{\beta x}) \rd x \\
&= \frac{\rho e^{\rho}}{\beta} \int_1^{\infty} (\ln v) v e^{-\rho v} \rd v
= \frac{\rho e^{\rho} h_{18}(\rho)}{\beta}.
\end{align*}

Given that the usual regularity conditions are satisfied, the Fisher information matrix is given by
\[
I(\beta, \rho)
=
\begin{bmatrix}
\frac{1 + \rho^2 e^{\rho} h_{17}(\rho)}{\beta^2} & \frac{\rho e^{\rho} h_{18}(\rho)}{\beta} \\[1mm]
\frac{\rho e^{\rho} h_{18}(\rho)}{\beta} & \frac{1}{\rho^2}
\end{bmatrix}.
\]

Next, one determines the matrix $G(\beta, \rho)=\EE\left\{\bb{\tau}(X, \beta, \rho) \bb{s}(X, \beta, \rho)^{\top}\right\}$. One has
\[
\bb{\tau}(x, \beta, \rho)
=
\begin{bmatrix}
\tau_1(x, \beta, \rho) \\[1mm]
\tau_2(x, \beta, \rho)
\end{bmatrix}
=
\begin{bmatrix}
\cos\left\{2\pi F(x \nvert \beta, \rho)\right\} \\[1mm]
\sin\left\{2\pi F(x \nvert \beta, \rho)\right\}
\end{bmatrix}
=
\begin{bmatrix}
\cos\left(2\pi \left[1 - \exp\left\{-\rho( e^{\beta x} - 1)\right\}\right]\right) \\[1mm]
\sin\left(2\pi \left[1 - \exp\left\{-\rho( e^{\beta x} - 1)\right\}\right]\right)
\end{bmatrix}.
\]

Given that $\EE\left[\cos\left\{2\pi F(X \nvert \beta, \rho)\right\}\right] = 0$ and using the change of variables $v = \exp(\beta x)$, one has
\begin{align*}
\EE\left\{\tau_1(X, \beta, \rho) s_1(X, \beta, \rho)\right\}
&= \EE\left[\cos\left\{2\pi F(X \nvert \beta, \rho)\right\} (1/\beta + X - \rho X e^{\beta X})\right] \\
&= \EE\left[\cos\left\{2\pi F(X \nvert \beta, \rho)\right\} X (1 - \rho e^{\beta X})\right] \\
&= \int_0^{\infty} \cos\left(2\pi \left[1 - \exp\left\{-\rho( e^{\beta x} - 1)\right\}\right]\right) x (1 - \rho e^{\beta x}) \beta\rho \exp(\rho + \beta x - \rho e^{\beta x}) \rd x \\
&= \frac{\rho e^{\rho}}{\beta} \int_1^{\infty} \cos\left(2\pi \left[1 - \exp\left\{-\rho(v - 1)\right\}\right]\right) (\ln v) (1 - \rho v) e^{-\rho v} \rd v
= \frac{\rho e^{\rho} h_{19}(\rho)}{\beta}.
\end{align*}
Similarly, one has
\begin{align*}
\EE\left\{\tau_2(X, \beta, \rho) s_1(X, \beta, \rho)\right\}
&= \EE\left[\sin\left\{2\pi F(X \nvert \beta, \rho)\right\} (1/\beta + X - \rho X e^{\beta X})\right] \\
&= \frac{\rho e^{\rho}}{\beta} \int_1^{\infty} \sin\left(2\pi \left[1 - \exp\left\{-\rho(v - 1)\right\}\right]\right) (\ln v) (1 - \rho v) e^{-\rho v} \rd v
= \frac{\rho e^{\rho} h_{20}(\rho)}{\beta}.
\end{align*}
Given that $\EE\left[\cos\left\{2\pi F(X \nvert \beta, \rho)\right\}\right] = 0$ and using the change of variables $v = \exp(\beta x)$ in the fourth equality,
\begin{align*}
\EE\left\{\tau_1(X, \beta, \rho) s_2(X, \beta, \rho)\right\}
&= \EE\left[\cos\left\{2\pi F(X \nvert \beta, \rho)\right\} (1/\rho + 1 - e^{\beta X})\right] \\
&= -\int_0^{\infty} \cos\left(2\pi \left[1 - \exp\left\{-\rho( e^{\beta x} - 1)\right\}\right]\right) e^{\beta x} \beta \rho \exp(\rho + \beta x - \rho e^{\beta x}) \rd x \\
&= -\rho e^{\rho} \int_1^{\infty} \cos\left[2\pi \left\{1 - e^{-\rho(v - 1)}\right\}\right] v e^{-\rho v} \rd v
= -\rho e^{\rho} h_{21}(\rho).
\end{align*}
Similarly, one has
\begin{align*}
\EE\left\{\tau_2(X, \beta, \rho) s_2(X, \beta, \rho)\right\}
&= \EE\left[\sin\left\{2\pi F(X \nvert \beta, \rho)\right\} (1/\rho + 1 - e^{\beta X})\right] \\
&= -\rho e^{\rho} \int_1^{\infty} \sin\left[2\pi \left\{1 - e^{-\rho(v - 1)}\right\}\right] v e^{-\rho v} \rd v
= -\rho e^{\rho} h_{22}(\rho).
\end{align*}

Therefore, one has
\[
G(\beta, \rho)
=
\rho e^{\rho}
\begin{bmatrix}
\frac{ h_{19}(\rho)}{\beta} & -h_{21}(\rho) \\[1mm]
\frac{h_{20}(\rho)}{\beta} & -h_{22}(\rho)
\end{bmatrix}.
\]

\section{Proof for the \texorpdfstring{$\mathrm{Lomax}(\alpha, \sigma)$}{Lomax} test}\label{proof:Lomax}

Straightforward calculations yield
\[
\bb{s}(x, \alpha, \sigma)
=
\begin{bmatrix}
s_1(x, \alpha, \sigma) \\[1mm]
s_2(x, \alpha, \sigma)
\end{bmatrix}
=
\begin{bmatrix}
\partial_{\alpha} \ln f(x \nvert \alpha, \sigma) \\[1mm]
\partial_{\sigma} \ln f(x \nvert \alpha, \sigma)
\end{bmatrix}
=
\begin{bmatrix}
\frac{1}{\alpha} - \ln(1 + x / \sigma) \\[1mm]
- \frac{1}{\sigma} + \frac{(\alpha + 1)x}{\sigma^2(1 + x/\sigma)}
\end{bmatrix}.
\]
The ML estimators for $\alpha_0$ and/or $\sigma_0$ come directly from the equation $\sum_{i=1}^{n}\bb{s}(x_i, \alpha_0, \sigma_0)=0$.

Next, one determines the Fisher information matrix $I(\alpha, \sigma)=\EE\big\{\bb{s}(X, \alpha, \sigma) \bb{s}(X, \alpha, \sigma)^{\top}\big\}$.
Note that if $X\sim \mathrm{Lomax}(\alpha, \sigma)$, one can show that
\[
\EE(X^k) = \frac{k! \, \sigma^k}{(\alpha - 1) \ldots (\alpha - k)}, \quad \alpha\in (k, \infty), ~k\in \N_0 = \{0, 1, 2, \ldots\}.
\]
Therefore, for $k\in \N_0$ and $\alpha, \sigma\in (0, \infty)$, one has
\begin{align*}
\EE\left\{X^k(1 + X/\sigma)^{-k}\right\}
&= \int_0^{\infty}x^k(1 + x/\sigma)^{-k} (\alpha/\sigma) (1 + x/\sigma)^{-(\alpha + 1)} \rd x \\
&= \frac{\alpha}{\alpha + k} \int_0^{\infty}x^k((\alpha + k)/\sigma) (1 + x/\sigma)^{-(\alpha + k + 1)} \rd x \\
&= \frac{\alpha}{\alpha + k} \times \frac{k!\sigma^k}{(\alpha + k-1)\ldots \alpha} = \frac{k!\sigma^k}{(\alpha + k)\ldots (\alpha + 1)}.
\end{align*}

One has
\[
- \EE\left\{\partial_{\alpha} s_1(X, \alpha, \sigma)\right\} = \frac{1}{\alpha^2},
\]
\[
- \EE\left\{\partial_{\sigma} s_1(X, \alpha, \sigma)\right\}
= -\frac{1}{\sigma^2} \EE\left\{X (1 + X / \sigma)^{-1}\right\}
= -\frac{1}{\sigma^2} \times \frac{\sigma}{\alpha + 1}
= \frac{-1}{(\alpha + 1)\sigma}
\]
and
\begin{align*}
- \EE\left\{\partial_{\sigma} s_2(X, \alpha, \sigma)\right\}
&= \frac{-1}{\sigma^2} - (\alpha + 1) \EE\left\{-\frac{2X}{\sigma^3(1 + X/\sigma)} + \frac{X^2}{\sigma^4(1 + X/\sigma)^2}\right\} \\
&= \frac{-1}{\sigma^2} - (\alpha + 1) \left\{-\frac{2}{\sigma^3} \times \frac{\sigma}{\alpha + 1} + \frac{1}{\sigma^4} \times \frac{2\sigma^2}{(\alpha + 2) (\alpha + 1)}\right\} \\
&= \frac{-1}{\sigma^2} + \frac{2}{\sigma^2} - \frac{2}{(\alpha + 2)\sigma^2} = \frac{\alpha}{(\alpha + 2)\sigma^2}.
\end{align*}
Given that the usual regularity conditions are satisfied, the Fisher information matrix is given by
\[
I(\alpha, \sigma)
=
\begin{bmatrix}
\frac{1}{\alpha^2} &\frac{-1}{(\alpha + 1)\sigma} \\[1mm]
\frac{-1}{(\alpha + 1)\sigma} & \frac{\alpha}{(\alpha + 2)\sigma^2}
\end{bmatrix}.
\]

Next, one determines the matrix $G(\alpha, \sigma)=\EE\left\{\bb{\tau}(X, \alpha, \sigma) \bb{s}(X, \alpha, \sigma)^{\top}\right\}$.
\[
\bb{\tau}(x, \alpha, \sigma)
\equiv
\begin{bmatrix}
\tau_1(x, \alpha, \sigma) \\[1mm]
\tau_2(x, \alpha, \sigma)
\end{bmatrix}
=
\begin{bmatrix}
\cos\left\{2\pi F(x \nvert \alpha, \sigma)\right\} \\[1mm]
\sin\left\{2\pi F(x \nvert \alpha, \sigma)\right\}
\end{bmatrix}
=
\begin{bmatrix}
\cos\left[2\pi \left\{1 - (1 + x/\sigma)^{-\alpha}\right\}\right] \\[1mm]
\sin\left[2\pi \left\{1 - (1 + x/\sigma)^{-\alpha}\right\}\right]
\end{bmatrix}.
\]

Given that $\EE\left[\cos\left\{2\pi F(X \nvert \alpha, \sigma)\right\}\right] = 0$ and using the change of variables $v = -\ln\{(1 + x/\sigma)^{-\alpha}\}$, one has
\begin{align*}
\EE\left\{\tau_1(X, \alpha, \sigma) s_1(X, \alpha, \sigma)\right\}
&= \EE\left[\cos\left\{2\pi F(X \nvert \alpha, \sigma)\right\} \left\{\alpha^{-1} - \ln(1 + X / \sigma)\right\}\right] \\
&= - \alpha\sigma^{-1} \int_0^{\infty} \cos\left[2\pi \left\{1 - (1 + x/\sigma)^{-\alpha}\right\}\right] \ln(1 + x / \sigma) (1 + x/\sigma)^{-(\alpha + 1)} \rd x \\
&= \frac{-1}{\alpha} \int_0^\infty \cos\left\{2\pi (1 - e^{-v})\right\} v e^{-v} \, \rd v \\
&= \frac{-1}{\alpha} \int_0^\infty \cos\left\{2\pi \Gamma_{1,1}(v)\right\} f_{\mathrm{ga}}(v \nvert 2, 1) \, \rd v
= \frac{-h_{6}(1,2,1)}{\alpha}.
\end{align*}
Similarly, one has
\begin{align*}
&\EE\left\{\tau_2(X, \alpha, \sigma) s_1(X, \alpha, \sigma)\right\}
= \EE\left[\sin\left\{2\pi F(X \nvert \alpha, \sigma)\right\} \left\{\alpha^{-1} - \ln(1 + X / \sigma)\right\}\right] \\
&\quad= \frac{-1}{\alpha} \int_0^\infty \sin\left\{2\pi \Gamma_{1,1}(v)\right\} f_{\mathrm{ga}}(v \nvert 2, 1) \, \rd v = \frac{-h_{7}(1,2,1)}{\alpha}.
\end{align*}

Given that $\EE\left[\cos\left\{2\pi F(X \nvert \alpha, \sigma)\right\}\right] = 0$ and using the change of variables $v = \alpha\ln(1 + x/\sigma)\Leftrightarrow x=\sigma(e^{v/\alpha}-1)$ with $\rd v=\frac{\alpha}{\sigma(1+x/\sigma)}\rd x\Leftrightarrow \rd x = \frac{\sigma}{\alpha}e^{v/\alpha} \rd v$, one has
\begin{align*}
\EE\left\{\tau_1(X, \alpha, \sigma) s_2(X, \alpha, \sigma)\right\}
&= \EE\left[\cos\left\{2\pi F(X \nvert \alpha, \sigma)\right\} \left\{-\frac{1}{\sigma} + \frac{(\alpha + 1)X}{\sigma^2(1 + X/\sigma)}\right\}\right] \\
&= \alpha (\alpha + 1) \sigma^{-3} \int_0^{\infty} \cos\left[2\pi \left\{1 - (1 + x/\sigma)^{-\alpha}\right\}\right] x (1 + x/\sigma)^{-(\alpha + 2)} \rd x \\
&= \alpha (\alpha + 1) \sigma^{-3} \int_0^{\infty} \cos\left[2\pi(1 - e^{-v})\right] \sigma(e^{v/\alpha}-1)e^{-\frac{v(\alpha+2)}{\alpha}}\frac{\sigma}{\alpha}e^{v/\alpha}\rd v \\
&= \frac{\alpha}{\sigma} \frac{\alpha + 1}{\alpha} \int_0^{\infty} \cos\left[2\pi(1 - e^{-v})\right](e^{-v}-e^{-\frac{v(\alpha+1)}{\alpha}})\rd v \\
&= -\frac{\alpha}{\sigma} \int_0^{\infty} \cos\left[2\pi(1 - e^{-v})\right]\frac{\alpha + 1}{\alpha}e^{-\frac{v(\alpha+1)}{\alpha}}\rd v \\
&= -\frac{\alpha}{\sigma} \int_0^{\infty} \cos\left[2\pi \Gamma_{1,1}(v)\right]f_{\mathrm{ga}}(v \nvert 1,\alpha/(\alpha + 1))\rd v \\
&=\frac{-\alpha h_{6}(1,1,\alpha/(\alpha + 1))}{\sigma}.
\end{align*}
Similarly, one has
\begin{align*}
\EE\left\{\tau_2(X, \alpha, \sigma) s_2(X, \alpha, \sigma)\right\}
&= \EE\left[\sin\left\{2\pi F(X \nvert \alpha, \sigma)\right\} \left\{-\frac{1}{\sigma} + \frac{(\alpha + 1)X}{\sigma^2(1 + X/\sigma)}\right\}\right] \\
&= -\frac{\alpha}{\sigma} \int_0^{\infty} \sin\left[2\pi \Gamma_{1,1}(v)\right]f_{\mathrm{ga}}(v \nvert 1,\alpha/(\alpha + 1))\rd v
=\frac{-\alpha h_{7}(1,1,\alpha/(\alpha + 1))}{\sigma}.
\end{align*}

Therefore, one has
\[
G(\alpha, \sigma)
=
\begin{bmatrix}
\frac{-h_{6}(1,2,1)}{\alpha} & \frac{-\alpha h_{6}(1,1,\alpha/(\alpha + 1))}{\sigma} \\[1mm]
\frac{-h_{7}(1,2,1)}{\alpha}& \frac{-\alpha h_{7}(1,1,\alpha/(\alpha + 1))}{\sigma}
\end{bmatrix}.
\]

\section{Proof for the \texorpdfstring{$\mathrm{inverse\mhyphen Gaussian}(\mu, \lambda)$}{inverse-Gaussian} test}\label{proof:inverse.Gaussian}

Straightforward calculations yield
\[
\bb{s}(x, \mu, \lambda)
=
\begin{bmatrix}
s_1(x, \mu, \lambda) \\[1mm]
s_2(x, \mu, \lambda)
\end{bmatrix}
=
\begin{bmatrix}
\partial_{\mu} \ln f(x \nvert \mu, \lambda) \\[1mm]
\partial_{\lambda} \ln f(x \nvert \mu, \lambda)
\end{bmatrix}
=
\begin{bmatrix}
\frac{\lambda(x - \mu)}{\mu^3} \\[1mm]
\frac{1}{2\lambda} - \frac{(x - \mu)^2}{2 \mu^2 x}
\end{bmatrix}
=
\begin{bmatrix}
\frac{\lambda(x - \mu)}{\mu^3} \\[1mm]
\frac{1}{2\lambda} - \frac{x}{2\mu^2} + \frac{1}{\mu} - \frac{1}{2x}
\end{bmatrix}.
\]
The ML estimators for $\mu_0$ and/or $\lambda_0$ come directly from the equation $\sum_{i=1}^{n}\bb{s}(x_i, \mu_0, \lambda_0)=0$.

Next, one determines the Fisher information matrix $I(\mu, \lambda)=\EE\big\{\bb{s}(X, \mu, \lambda) \bb{s}(X, \mu, \lambda)^{\top}\big\}$.
Given that $\EE(X) = \mu$, one has
\[
- \EE\left\{\partial_{\mu} s_1(X, \mu, \lambda)\right\}
= \frac{3\lambda\EE(X)}{\mu^4} - \frac{2\lambda}{\mu^3}
= \frac{\lambda}{\mu^3}, \quad
- \EE\left\{\partial_{\lambda} s_1(X, \mu, \lambda)\right\}
= -\frac{\EE(X) - \mu}{\mu^3}
= 0
\]
and
\[
- \EE\left\{\partial_{\lambda} s_2(X, \mu, \lambda)\right\} = \frac{1}{2\lambda^2}.
\]

Given that the usual regularity conditions are satisfied, the Fisher information matrix is given by
\[
I(\mu, \lambda)
=
\begin{bmatrix}
\frac{\lambda}{\mu^3} & 0 \\[1mm]
0 & \frac{1}{2\lambda^2}
\end{bmatrix}.
\]

Next, one determines the matrix $G(\mu, \lambda)=\EE\left\{\bb{\tau}(X, \mu, \lambda) \bb{s}(X, \mu, \lambda)^{\top}\right\}$.
\[
\bb{\tau}(x, \mu, \lambda)
=
\begin{bmatrix}
\tau_1(x, \mu, \lambda) \\[1mm]
\tau_2(x, \mu, \lambda)
\end{bmatrix}
=
\begin{bmatrix}
\cos\left\{2\pi F(x \nvert \mu, \lambda)\right\} \\[1mm]
\sin\left\{2\pi F(x \nvert \mu, \lambda)\right\}
\end{bmatrix}.
\]

Given that $\EE\left[\cos\left\{2\pi F(X \nvert \mu, \lambda)\right\}\right] = 0$, one has
\begin{align*}
\EE\left\{\tau_1(X, \mu, \lambda) s_1(X, \mu, \lambda)\right\}
&= \frac{\lambda}{\mu^3} \EE\left[\cos\left\{2\pi F(X \nvert \mu, \lambda)\right\} (X - \mu)\right]
= \frac{\lambda}{\mu^3} \EE\left[X \cos\left\{2\pi F(X \nvert \mu, \lambda)\right\}\right] \\
&= \frac{\lambda}{\mu^3} \int_0^{\infty} v \cos\left\{2\pi F(v \nvert \mu, \lambda)\right\} f(v \nvert \mu, \lambda) \rd v
= \frac{\lambda}{\mu^3} h_{27}(\mu, \lambda).
\end{align*}
Similarly, one has
\begin{align*}
\EE\left\{\tau_2(X, \mu, \lambda) s_1(X, \mu, \lambda)\right\}
&= \frac{\lambda}{\mu^3} \int_0^{\infty} v \sin\left\{2\pi F(v \nvert \mu, \lambda)\right\} f(v \nvert \mu, \lambda) \rd v
= \frac{\lambda}{\mu^3} h_{28}(\mu, \lambda).
\end{align*}
Given that $\EE\left[\cos\left\{2\pi F(X \nvert \mu, \lambda)\right\}\right] = 0$, one has
\begin{align*}
\EE\left\{\tau_1(X, \mu, \lambda) s_2(X, \mu, \lambda)\right\}
&= \EE\left[\cos\left\{2\pi F(X \nvert \mu, \lambda)\right\} \left(\frac{1}{2\lambda} - \frac{X}{2\mu^2} + \frac{1}{\mu} - \frac{1}{2X}\right)\right] \\
&= -\frac{1}{2\mu^2} \EE\left[\cos\left\{2\pi F(X \nvert \mu, \lambda)\right\} \left(\frac{X^2 + \mu^2}{X}\right)\right] \\
&= -\frac{1}{2\mu^2} \int_0^{\infty} v^{-1} (v^2 + \mu^2) \cos\left\{2\pi F(v \nvert \mu, \lambda)\right\} f(v \nvert \mu, \lambda) \rd v
= - \frac{1}{2\mu^2} h_{29}(\mu, \lambda).
\end{align*}
Similarly, one has
\begin{align*}
\EE\left\{\tau_2(X, \mu, \lambda) s_2(X, \mu, \lambda)\right\}
&= -\frac{1}{2\mu^2} \int_0^{\infty} v^{-1} (v^2 + \mu^2) \sin\left\{2\pi F(v \nvert \mu, \lambda)\right\} f(v \nvert \mu, \lambda) \rd v
= - \frac{1}{2\mu^2} h_{30}(\mu, \lambda).
\end{align*}
Therefore, one has
\[
G(\mu, \lambda)
=
\begin{bmatrix}
\frac{\lambda}{\mu^3} h_{27}(\mu, \lambda) & - \frac{1}{2\mu^2} h_{29}(\mu, \lambda) \\[1mm]
\frac{\lambda}{\mu^3} h_{28}(\mu, \lambda) & - \frac{1}{2\mu^2} h_{30}(\mu, \lambda)
\end{bmatrix}.
\]

\section{Proof for the \texorpdfstring{$\mathrm{Beta}(\alpha, \beta)$}{Beta} test}\label{proof:beta}

Straightforward calculations yield
\[
\bb{s}(x, \alpha, \beta)
=
\begin{bmatrix}
s_1(x, \alpha, \beta) \\[1mm]
s_2(x, \alpha, \beta)
\end{bmatrix}
=
\begin{bmatrix}
\partial_{\alpha} \ln f(x \nvert \alpha, \beta) \\[1mm]
\partial_{\beta} \ln f(x \nvert \alpha, \beta)
\end{bmatrix}
=
\begin{bmatrix}
\psi(\alpha + \beta) - \psi(\alpha) + \ln x \\[1mm]
\psi(\alpha + \beta) - \psi(\beta) + \ln (1 - x)
\end{bmatrix}.
\]
The ML estimators for $\alpha_0$ and/or $\beta_0$ come directly from the equation $\sum_{i=1}^{n}\bb{s}(x_i, \alpha_0, \beta_0)=0$.

Next, one determines the Fisher information matrix $I(\alpha, \beta)=\EE\big\{\bb{s}(X, \alpha, \beta) \bb{s}(X, \alpha, \beta)^{\top}\big\}$.
Given that
\[
- \EE\left\{\partial_{\alpha} s_1(X, \alpha, \beta)\right\} = \psi_1(\alpha) - \psi_1(\alpha + \beta), \quad
- \EE\left\{\partial_{\beta} s_1(X, \alpha, \beta)\right\} = - \psi_1(\alpha + \beta)
\]
and
\[
- \EE\left\{\partial_{\beta} s_2(X, \alpha, \beta)\right\} = \psi_1(\beta) - \psi_1(\alpha + \beta),
\]
and that the usual regularity conditions are satisfied, the Fisher information matrix is given by
\[
I(\alpha, \beta) =
\begin{bmatrix}
\psi_1(\alpha) - \psi_1(\alpha + \beta) & - \psi_1(\alpha + \beta) \\[1mm]
- \psi_1(\alpha + \beta) & \psi_1(\beta) - \psi_1(\alpha + \beta)
\end{bmatrix}.
\]

Next, one determines the matrix $G(\alpha, \beta)=\EE\left\{\bb{\tau}(X, \alpha, \beta) \bb{s}(X, \alpha, \beta)^{\top}\right\}$.
\[
\bb{\tau}(x, \alpha, \beta)
\equiv
\begin{bmatrix}
\tau_1(x, \alpha, \beta) \\[1mm]
\tau_2(x, \alpha, \beta)
\end{bmatrix}
=
\begin{bmatrix}
\cos\left\{2\pi F_{\mathrm{be}}(x \nvert \alpha, \beta)\right\} \\[1mm]
\sin\left\{2\pi F_{\mathrm{be}}(x \nvert \alpha, \beta)\right\}
\end{bmatrix}.
\]

Given that $\EE[\cos\{2\pi F_{\mathrm{be}}(X \nvert \alpha, \beta)\}] = 0$, one has
\begin{align*}
\EE\left\{\tau_1(X, \alpha, \beta) s_1(X, \alpha, \beta)\right\}
&= \EE\left[\cos\left\{2\pi F_{\mathrm{be}}(X \nvert \alpha, \beta)\right\} (\psi(\alpha + \beta) - \psi(\alpha) + \ln X)\right] \\
&= \int_0^1 (\ln x)\cos\left\{2\pi F_{\mathrm{be}}(x \nvert \alpha, \beta)\right\} f_{\mathrm{be}}(x \nvert \alpha, \beta) \rd x =h_{23}(\alpha, \beta).
\end{align*}
Similarly, one has
\begin{align*}
\EE\left\{\tau_2(X, \alpha, \beta) s_1(X, \alpha, \beta)\right\}
&= \EE\left[\sin\left\{2\pi F_{\mathrm{be}}(X \nvert \alpha, \beta)\right\} (\psi(\alpha + \beta) - \psi(\alpha) + \ln X)\right] \\
&= \int_0^1 (\ln x)\sin\left\{2\pi F_{\mathrm{be}}(x \nvert \alpha, \beta)\right\} f_{\mathrm{be}}(x \nvert \alpha, \beta) \rd x =h_{24}(\alpha, \beta).
\end{align*}
Given that $\EE[\cos\{2\pi F_{\mathrm{be}}(X \nvert \alpha, \beta)\}] = 0$, one has
\begin{align*}
\EE\left\{\tau_1(X, \alpha, \beta) s_2(X, \alpha, \beta)\right\}
&= \EE\left[\cos\left\{2\pi F_{\mathrm{be}}(X \nvert \alpha, \beta)\right\} \left\{\psi(\alpha + \beta) - \psi(\beta) + \ln (1 - X)\right\}\right] \\
&= \int_0^1 \ln (1 - x)\cos\left\{2\pi F_{\mathrm{be}}(x \nvert \alpha, \beta)\right\} f_{\mathrm{be}}(x \nvert \alpha, \beta) \rd x= h_{25}(\alpha, \beta).
\end{align*}
Similarly, one has
\begin{align*}
\EE\left\{\tau_2(X, \alpha, \beta) s_2(X, \alpha, \beta)\right\}
&= \EE\left[\sin\left\{2\pi F_{\mathrm{be}}(X \nvert \alpha, \beta)\right\} \left\{\psi(\alpha + \beta) - \psi(\beta) + \ln (1 - X)\right\}\right] \\
&= \int_0^1 \ln (1 - x)\sin\left\{2\pi F_{\mathrm{be}}(x \nvert \alpha, \beta)\right\} f_{\mathrm{be}}(x \nvert \alpha, \beta) \rd x= h_{26}(\alpha, \beta).
\end{align*}
Therefore, one has
\[
G(\alpha, \beta)
=
\begin{bmatrix}
h_{23}(\alpha, \beta) & h_{25}(\alpha, \beta) \\[1mm]
h_{24}(\alpha, \beta) & h_{26}(\alpha, \beta)
\end{bmatrix}.
\]

\section{Proof for the \texorpdfstring{$\mathrm{Kumaraswamy}(\alpha, \beta)$}{Kumaraswamy} test}\label{proof:Kumaraswamy}

Straightforward calculations yield
\[
\bb{s}(x, \alpha, \beta)
=
\begin{bmatrix}
s_1(x, \alpha, \beta) \\[1mm]
s_2(x, \alpha, \beta)
\end{bmatrix}
=
\begin{bmatrix}
\partial_{\alpha} \ln f(x \nvert \alpha, \beta) \\[1mm]
\partial_{\beta} \ln f(x \nvert \alpha, \beta)
\end{bmatrix}
=
\begin{bmatrix}
\alpha^{-1} + \ln x - (\beta - 1) x^{\alpha} (1 - x^{\alpha})^{-1} \ln x \\[1mm]
\beta^{-1} + \ln(1 - x^{\alpha})
\end{bmatrix}.
\]
The ML estimators for $\alpha_0$ and/or $\beta_0$ come directly from the equation $\sum_{i=1}^{n}\bb{s}(x_i, \alpha_0, \beta_0)=0$ since
\[
 n\alpha^{-1} + \sum_{i=1}^{n}\ln x_i - (\beta - 1) \sum_{i=1}^{n} x_i^{\alpha} (1 - x_i^{\alpha})^{-1} \ln x_i =0\Leftrightarrow
\frac{n\alpha^{-1}+ \sum_{i=1}^{n}\ln x_i}{\sum_{i=1}^{n} x_i^{\alpha} (1 - x_i^{\alpha})^{-1} \ln x_i} + 1 - \beta = 0.
\]

Next, one determines the Fisher information matrix $I(\alpha, \beta)=\EE\big\{\bb{s}(X, \alpha, \beta) \bb{s}(X, \alpha, \beta)^{\top}\big\}$.
Consider that $\beta\neq 2$. Using the change of variables $v = x^{\alpha}$ in the fifth equality and \texttt{Mathematica} in the sixth equality, one has
\begin{align*}
- \EE\left\{\partial_{\alpha} s_1(X, \alpha, \beta)\right\}
&= -\EE\left[\partial_{\alpha} \left\{\alpha^{-1} + \ln X - (\beta - 1) X^{\alpha} (1 - X^{\alpha})^{-1} \ln X\right\}\right] \\
&= \alpha^{-2} + (\beta - 1) \EE\left\{X^{\alpha}(1 - X^{\alpha})^{-2} (\ln X)^2\right\} \\
&= \alpha^{-2} + (\beta - 1) \alpha \beta \int_0^1 x^{\alpha} (1 - x^{\alpha})^{-2} (\ln x)^2 x^{\alpha - 1} (1 - x^{\alpha})^{\beta - 1} \rd x \\
&= \alpha^{-2} + (\beta - 1) \alpha \beta \int_0^1 (\ln x)^2 x^{2\alpha - 1} (1 - x^{\alpha})^{\beta - 3} \rd x \\
&= \alpha^{-2} + \alpha^{-2} \beta \int_0^1 (\beta - 1) (\ln v)^2 v (1 - v)^{\beta - 3} \rd v \\
&= \frac{1}{\alpha^2} + \frac{\beta}{\alpha^2(\beta - 2)} \left\{\psi^2(\beta) + 2 \psi(\beta) (\gamma-1) - \psi_1(\beta) + \pi^2/6 + (\gamma - 2)\gamma\right\}\\
& = \frac{1}{\alpha^2} + \frac{\beta}{\alpha^2(\beta - 2)} \left\{(\psi(\beta) + \gamma-1)^2 - \psi_1(\beta) + \pi^2/6 -1\right\},
\end{align*}
where $\gamma$ is the Euler-Mascheroni constant defined as $\gamma = - \psi(1) = 0.57721566490153\ldots$.

If $\beta = 2$, one has, using \texttt{Mathematica},
\[
- \EE\left\{\partial_{\alpha} s_1(X, \alpha, \beta)\right\}
= \frac{1}{\alpha^2} + \frac{2}{\alpha^2} \int_0^1 (\ln v)^2 v (1 - v)^{-1} \rd v
= \frac{1}{\alpha^2} + \frac{4}{\alpha^2} \left\{\zeta(3) - 1\right\}
= \frac{1}{\alpha^2} \left\{4\zeta(3) - 3\right\},
\]
where $\zeta(3) = \sum_{n=1}^{\infty} n^{-3} = 1.20205690315959\ldots$

Note that $-\EE\left(\partial_{\alpha} s_1(X, \alpha, \beta)\right)$ is continuous in $\beta = 2$ since, using L'H\^opital's rule,
\begin{align*}
\lim_{\beta\rightarrow 2} - \EE\left\{\partial_{\alpha} s_1(X, \alpha, \beta)\right\}
&= \frac{1}{\alpha^2} + \frac{2}{\alpha^2} \lim_{\beta\rightarrow 2} \frac{\psi^2(\beta) + 2\psi(\beta) (\gamma - 1) - \psi_1(\beta) + \pi^2/6 + (\gamma - 2)\gamma}{\beta - 2} \\
&= \frac{1}{\alpha^2} + \frac{2}{\alpha^2} \times \left\{2 \psi(2) \psi_1(2) + 2 \psi_1(2) (\gamma - 1) - \psi_2(2)\right\} \\
&= \frac{1}{\alpha^2} + \frac{2}{\alpha^2} \times \left[2 (1 - \gamma) (\pi^2/6 - 1) + 2 (\pi^2/6 - 1) (\gamma - 1) + 2 \left\{\zeta(3) - 1\right\}\right] \\
&= \frac{1}{\alpha^2} + \frac{2}{\alpha^2} \times 2 \left\{\zeta(3) - 1\right\}
= \frac{1}{\alpha^2} + \frac{4}{\alpha^2} \left\{\zeta(3) - 1\right\}
= \frac{1}{\alpha^2} \left\{4 \zeta(3) - 3\right\},
\end{align*}
since $\psi(2) = 1 - \gamma$, $\psi_1(2) = \pi^2/6-1$ and $\psi_2(2) = \frac{\rd}{\rd \beta} \psi_1(\beta)|_{\beta = 2} = - 2 \left\{\zeta(3) - 1\right\}$.

Therefore, one has
\[
- \EE\left\{\partial_{\alpha} s_1(X, \alpha, \beta)\right\}
=
\begin{cases}
\frac{1}{\alpha^2} + \frac{\beta}{\alpha^2(\beta - 2)} \left\{(\psi(\beta) + \gamma-1)^2 - \psi_1(\beta) + \pi^2/6 -1\right\}, & \mbox{if } \beta\neq 2, \\[1mm]
\frac{1}{\alpha^2} \left\{4\zeta(3) - 3\right\}, & \mbox{if } \beta = 2.
\end{cases}
\]

Consider that $\beta\neq 1$. Using the change of variables $v = x^{\alpha}$ in the fifth equality and \texttt{Mathematica} in the sixth equality, one has
\begin{align*}
- \EE\left\{\partial_{\alpha} s_2(X, \alpha, \beta)\right\}
&= -\EE\left[\partial_{\alpha} \left\{\beta^{-1} + \ln(1 - X^{\alpha})\right\}\right]
= \EE\left\{X^{\alpha} (1 - X^{\alpha})^{-1} \ln X\right\} \\
&= \alpha \beta \int_0^1 x^{\alpha} (1 - x^{\alpha})^{-1} (\ln x) x^{\alpha - 1} (1 - x^{\alpha})^{\beta - 1} \rd x \\
&= \alpha \beta \int_0^1 x^{2\alpha - 1} (1 - x^{\alpha})^{\beta - 2} (\ln x) \rd x
= \alpha^{-1} \beta \int_0^1 v (1 - v)^{\beta - 2} (\ln v) \rd v \\
&= \alpha^{-1} \beta \, \frac{1 - \gamma - \psi(\beta + 1)}{\beta(\beta - 1)}
= \frac{1 - \gamma - \psi(\beta + 1)}{\alpha(\beta - 1)}= \frac{\psi(\beta) + \gamma - 1 + 1 / \beta}{\alpha(1 - \beta)}.
\end{align*}
If $\beta = 1$, one has, using \texttt{Mathematica},
\[
- \EE\left\{\partial_{\alpha} s_2(X, \alpha, \beta)\right\}
= \alpha^{-1} \int_0^1 v (1 - v)^{-1} (\ln v) \rd v
= \frac{1 - \pi^2/6}{\alpha}.
\]
Note that $-\EE\left\{\partial_{\alpha} s_2(X, \alpha, \beta)\right\}$ is continuous in $\beta = 1$ since, using L'H\^opital's rule,
\[
\lim_{\beta\rightarrow 1} \frac{\psi(\beta) + \gamma - 1 + 1 / \beta}{\alpha(1 - \beta)}=\lim_{\beta\rightarrow 1} \frac{\psi_1(\beta)-1/\beta^2}{-\alpha} = \frac{1 - \pi^2/6}{\alpha}.
\]
Therefore, one has
\[
- \EE\left\{\partial_{\alpha} s_2(X, \alpha, \beta)\right\}
=
\begin{cases}
\frac{\psi(\beta) + \gamma - 1 + 1 / \beta}{\alpha(1 - \beta)}, & \mbox{if } \beta\neq 1, \\[1mm]
\frac{1 - \pi^2/6}{\alpha}, & \mbox{if } \beta = 1.
\end{cases}
\]
One has
\[
- \EE\left\{\partial_{\beta} s_2(X, \alpha, \beta)\right\}
= -\EE\left[\partial_{\beta} \left\{\beta^{-1} + \ln(1 - X^{\alpha})\right\}\right]
= \beta^{-2}.
\]
Given that the usual regularity conditions are satisfied, the Fisher information matrix is given by
\[
I(\alpha, \beta)
=
\begin{bmatrix}
\frac{1}{\alpha^2} + \frac{\beta}{\alpha^2(\beta - 2)} \left\{(\psi(\beta) + \gamma-1)^2 - \psi_1(\beta) + \pi^2/6 -1\right\} & \frac{\psi(\beta) + \gamma - 1 + 1 / \beta}{\alpha(1 - \beta)} \\[1mm]
\frac{\psi(\beta) + \gamma - 1 + 1 / \beta}{\alpha(1 - \beta)} & \frac{1}{\beta^2}
\end{bmatrix},
\]
with
\[
[I(\alpha, \beta)]_{1, 2} = \frac{1 - \pi^2/6}{\alpha} ~\text{if } \beta = 1
\quad \text{and} \quad
[I(\alpha, \beta)]_{1, 1} = \frac{1}{\alpha^2} \left\{4\zeta(3) - 3\right\}=\frac{1.80822761263836}{\alpha^2} ~\text{if } \beta = 2
\]
and $\zeta(3) = \sum_{n=1}^{\infty} n^{-3} = 1.20205690315959\ldots$

Next, one determines the matrix $G(\alpha, \beta)=\EE\left\{\bb{\tau}(X, \alpha, \beta) \bb{s}(X, \alpha, \beta)^{\top}\right\}$.
One has
\[
\bb{\tau}(x, \alpha, \beta)
=
\begin{bmatrix}
\tau_1(x, \alpha, \beta) \\[1mm]
\tau_2(x, \alpha, \beta)
\end{bmatrix}
=
\begin{bmatrix}
\cos\left\{2\pi F(x \nvert \alpha, \beta)\right\} \\[1mm]
\sin\left\{2\pi F(x \nvert \alpha, \beta)\right\}
\end{bmatrix}
=
\begin{bmatrix}
\cos\left[2\pi \left\{1 - (1 - x^{\alpha})^{\beta}\right\}\right] \\[1mm]
\sin\left[2\pi \left\{1 - (1 - x^{\alpha})^{\beta}\right\}\right]
\end{bmatrix}.
\]

Given that $\EE\left[\cos\left\{2\pi F(X \nvert \alpha, \beta)\right\}\right] = 0$ and using the change of variables $v = x^{\alpha}$, one has
\begin{align*}
&\EE\left\{\tau_1(X, \alpha, \beta) s_1(X, \alpha, \beta)\right\} \\[1mm]
&\quad= \EE\left[\cos\left\{2\pi F(X \nvert \alpha, \beta)\right\} \left\{\alpha^{-1} + \ln X - (\beta - 1) X^{\alpha} (1 - X^{\alpha})^{-1} \ln X\right\}\right] \\
&\quad= \alpha \beta \int_0^1 \cos\left[2\pi \left\{1 - (1 - x^{\alpha})^{\beta}\right\}\right] (\ln x) \left\{1 - (\beta - 1)x^{\alpha}(1 - x^{\alpha})^{-1}\right\} x^{\alpha - 1} (1 - x^{\alpha})^{\beta - 1} \rd x \\
&\quad= \alpha^{-1} \beta \int_0^1 \cos\left[2\pi \left\{1 - (1 - v)^{\beta}\right\}\right] (\ln v) (1 - v)^{\beta - 1} \left\{1 - (\beta - 1) v (1 - v)^{-1}\right\} \rd v \\
&\quad= \alpha^{-1} \beta \int_0^1 \cos\left[2\pi \left\{1 - (1 - v)^{\beta}\right\}\right] (\ln v)(1 - v)^{\beta-2}(1-\beta v) \rd v
= \alpha^{-1} \beta h_{31}(\beta).
\end{align*}
Similarly, one has
\begin{align*}
&\EE\left\{\tau_2(X, \alpha, \beta) s_1(X, \alpha, \beta)\right\} \\[1mm]
&\quad= \EE\left[\sin\left\{2\pi F(X \nvert \alpha, \beta)\right\} \left\{\alpha^{-1} + \ln X - (\beta - 1) X^{\alpha} (1 - X^{\alpha})^{-1} \ln X\right\}\right] \\
&\quad= \alpha^{-1} \beta \int_0^1 \sin\left[2\pi \left\{1 - (1 - v)^{\beta}\right\}\right] (\ln v)(1 - v)^{\beta-2}(1-\beta v) \rd v
= \alpha^{-1} \beta h_{32}(\beta).
\end{align*}

Given that $\EE\left[\cos\left\{2\pi F(X \nvert \alpha, \beta)\right\}\right] = 0$ and using the change of variables $v = x^{\alpha}$, one has
\begin{align*}
\EE\left\{\tau_1(X, \alpha, \beta) s_2(X, \alpha, \beta)\right\}
&= \EE\left[\cos\left\{2\pi F(X \nvert \alpha, \beta)\right\} \left\{\beta^{-1} + \ln(1 - X^{\alpha})\right\}\right] \\
&= \alpha \beta \int_0^1 \cos\left[2\pi \left\{1 - (1 - x^{\alpha})^{\beta}\right\}\right] \ln(1 - x^{\alpha}) x^{\alpha - 1} (1 - x^{\alpha})^{\beta - 1} \rd x \\
&= \beta \int_0^1 \cos\left[2\pi \left\{1 - (1 - v)^{\beta}\right\}\right] \ln(1 - v) (1 - v)^{\beta - 1} \rd v
= \beta h_{33}(\beta).
\end{align*}
Similarly, one has
\begin{align*}
\EE\left\{\tau_2(X, \alpha, \beta) s_2(X, \alpha, \beta)\right\}
&= \EE\left[\sin\left\{2\pi F(X \nvert \alpha, \beta)\right\} \left\{\beta^{-1} + \ln(1 - X^{\alpha})\right\}\right] \\
&= \beta \int_0^1 \sin\left[2\pi \left\{1 - (1 - v)^{\beta}\right\}\right] \ln(1 - v) (1 - v)^{\beta - 1} \rd v
= \beta h_{34}(\beta).
\end{align*}

Therefore, one has
\[
G(\alpha, \beta)
=
\beta
\begin{bmatrix}
\alpha^{-1} h_{31}(\beta) & h_{33}(\beta) \\[1mm]
\alpha^{-1} h_{32}(\beta)& h_{34}(\beta)
\end{bmatrix}.
\]

\section{Proof for the \texorpdfstring{$\mathrm{uniform}(a, b)$}{uniform} test as a limiting case}\label{proof:uniform}

The proof consists of considering the continuous $\mathrm{uniform}(a, b)$ as the limiting case of the $\mathrm{EPD}_{\lambda}(\mu,\sigma)$ as $\lambda$ tends to infinity,
with $a=\mu-\sigma$ and $b=\mu+\sigma$.

First, it is shown that the density of an $\mathrm{EPD}_{\lambda}(\mu,\sigma)$ converges to that of a $\mathrm{uniform}(\mu-\sigma,\mu+\sigma)$ distribution as $\lambda\rightarrow\infty$. Consider the following results. One has
\begin{equation*}
\lim_{\lambda\rightarrow\infty}\Gamma(1 + 1/\lambda)=\Gamma(1)=1~~\text{and}~~
\lim_{\lambda\rightarrow\infty}\lambda^{1/\lambda}=\exp\Big(\lim_{\lambda\rightarrow\infty}\frac{\ln\lambda}{\lambda}\Big)
=\exp\Big(\lim_{\lambda\rightarrow\infty}\frac{1/\lambda}{1}\Big)=\exp(0)=1,
\end{equation*}
using L'Hospital's rule. Furthermore, for $y>0$,
\begin{equation*}
\lim_{\lambda\rightarrow\infty}\frac{y^\lambda}{\lambda}
=
\begin{cases}
0, & \hbox{if }0<y\le 1;\\
\lim_{\lambda\rightarrow\infty}\frac{y^{\lambda}\ln(y)}{1}=\infty, & \hbox{if } y> 1,
\end{cases}
\end{equation*}
using again L'Hospital's rule. Therefore, using $\Gamma(1 + 1/\lambda)=\lambda^{-1}\Gamma(1/\lambda)$, one has
\begin{equation*}
\lim_{\lambda\rightarrow\infty}f_{\lambda}(x\nvert\mu,\sigma)= \lim_{\lambda\rightarrow\infty}\frac{1}{2\sigma \lambda^{1/\lambda} \Gamma(1 + 1/\lambda)}\exp\Big(\!-\frac{1}{\lambda}\Big|\frac{x-\mu}{\sigma}\Big|^\lambda\Big)
=
\begin{cases}
(2\sigma)^{-1} , & \hbox{if }|x-\mu|\le\sigma;\\
0, & \hbox{if } |x-\mu|>\sigma,
\end{cases}
\end{equation*}
or equivalently, $\lim_{\lambda\rightarrow\infty}f_{\lambda}(x\nvert\mu,\sigma)= (b-a)^{-1}$ if $a\le x \le b$, which is the density of the $\mathrm{uniform}(\mu-\sigma,\mu+\sigma)=\mathrm{uniform}(a,b)$.

The cdf of an $\mathrm{EPD}_{\lambda}(\mu,\sigma)$ also tends to that of a $\mathrm{uniform}(\mu-\sigma,\mu+\sigma)=\mathrm{uniform}(a,b)$, as $\lambda\rightarrow\infty$, since
\begin{align*}
\lim_{\lambda\rightarrow\infty}F_{\lambda}(x \nvert\mu, \sigma)& = \lim_{\lambda\rightarrow\infty}\frac{1}{2} \left[1 + \mathrm{sign}(x - \mu) \, \Gamma_{1/\lambda, 1}\!\left\{\frac{1}{\lambda} \left(|x - \mu|/\sigma\right)^{\lambda}\right\}\right]\\
&=\frac{1}{2} \left[1 + \mathrm{sign}(x - \mu) \,\frac{|x - \mu|}{\sigma}\right]=\frac{1}{2} \left[1 + \,\frac{x - \mu}{\sigma}\right]=\frac{x-(\mu-\sigma)}{2\sigma}=\frac{x-a}{b-a}.
\end{align*}
It suffices to show that, for $y>0$,
\begin{align*}
\lim_{\lambda\rightarrow\infty}\Gamma_{1/\lambda, 1}(y^\lambda/\lambda)
&=\lim_{\lambda\rightarrow\infty}\int_{0}^{y^\lambda/\lambda}\frac{w^{1/\lambda-1}\exp(-w)}{\Gamma(1/\lambda)}dw
=\lim_{\lambda\rightarrow\infty}\int_{0}^{y}\frac{\exp(-u^\lambda/\lambda)}{\lambda^{1/\lambda}\Gamma(1+1/\lambda)}du\\
&= \int_{0}^{y}\lim_{\lambda\rightarrow\infty}\frac{\exp(-u^\lambda/\lambda)}{\lambda^{1/\lambda}\Gamma(1+1/\lambda)}du
=\int_{0}^{y}du=y,
\end{align*}
using the change of variable $u = (\lambda w)^{1/\lambda}$, $w=u^\lambda/\lambda$, $dw=u^{\lambda-1}du$, $\Gamma(1/\lambda)=\lambda\Gamma(1+1/\lambda)$ and using Lebesgue's dominated convergence theorem.

Next, the limiting behavior of the ML estimators is studied for the $\mathrm{EPD}_{\lambda}(\mu,\sigma)$. The ML estimator of $\mu_0$ is the value of $\mu$ that minimizes $\sum_{i=1}^n |x_i - \mu|^{\lambda}$. If $\mu$ is set to $\mu=(x_{(1)}+x_{(n)})/2$, then the two largest terms that dominate the others are $(\mu-x_{(1)})^{\lambda}$ and $(x_{(n)} - \mu)^{\lambda}$, as $\lambda\rightarrow\infty$. Therefore, $\hat{\mu}_n$ is the value of $\mu$ that minimizes $(\mu-x_{(1)})^{\lambda}+(x_{(n)} - \mu)^{\lambda}$, as $\lambda\rightarrow\infty$, or equivalently, the solution in $\mu$ of the equation $\lambda(\mu-x_{(1)})^{\lambda-1}=\lambda(x_{(n)}-\mu)^{\lambda-1}$, which results in $\hat{\mu}_n\rightarrow (x_{(1)}+x_{(n)})/2$.

Then, using again the fact that the two largest terms that dominate the others are $(\hat{\mu}_n-x_{(1)})^{\lambda}$ and $(x_{(n)} - \hat{\mu}_n)^{\lambda}$, one has
\begin{align*}
\lim_{\lambda\rightarrow \infty}\hat{\sigma}_n &= \lim_{\lambda\rightarrow \infty}\left(\frac{1}{n} \sum_{i=1}^n |x_i - \hat{\mu}_n|^{\lambda}\right)^{1/\lambda}=\lim_{\lambda\rightarrow \infty}
\left\{\frac{1}{n} (\hat{\mu}_n-x_{(1)})^{\lambda}+\frac{1}{n}(x_{(n)} - \hat{\mu}_n)^{\lambda}\right\}^{1/\lambda}\\
&=\lim_{\lambda\rightarrow \infty}\left\{\frac{2}{n} \left(\frac{x_{(n)}-x_{(1)}}{2}\right)^{\lambda}\right\}^{1/\lambda}
=\lim_{\lambda\rightarrow \infty} \left(\frac{2}{n}\right)^{1/\lambda} \left(\frac{x_{(n)}-x_{(1)}}{2}\right)=\frac{x_{(n)}-x_{(1)}}{2}.
\end{align*}
Therefore, one has
\begin{equation*}
\hat{a}_n = \lim_{\lambda\rightarrow \infty}\hat{\mu}_n-\hat{\sigma}_n = \frac{x_{(1)}+x_{(n)}}{2} - \frac{x_{(n)}-x_{(1)}}{2} = x_{(1)}
\end{equation*}
and
\begin{equation*}
\hat{b}_n = \lim_{\lambda\rightarrow \infty}\hat{\mu}_n+\hat{\sigma}_n = \frac{x_{(1)}+x_{(n)}}{2} + \frac{x_{(n)}-x_{(1)}}{2} = x_{(n)},
\end{equation*}
which are, as expected, the ML estimators of a $\mathrm{uniform}(a_0, b_0)$.

Finally, the limit of the matrix $\Sigma(\mu,\sigma)$ is studied as $\lambda\rightarrow \infty$, for an $\mathrm{EPD}_{\lambda}(\mu,\sigma)$. Given that
\[
G(\mu, \sigma)=
 \frac{1}{\sigma}
 \begin{bmatrix}
0 ~&~ h_{1}(\lambda) \\[1mm]
\frac{h_{2}(\lambda)}{\lambda^{1/\lambda-1} \Gamma(1/\lambda)} ~&~ 0
\end{bmatrix},~~I(\mu,\sigma)=\frac{1}{\sigma^2}
\begin{bmatrix}
\frac{\lambda^{2 - 2/\lambda}\Gamma(2-1/\lambda)}{\Gamma(1/\lambda)} ~&~ 0 \\[1mm]
0 ~&~ \lambda
\end{bmatrix},
\]
one has
\[
\Sigma(\mu,\sigma)= \frac{1}{2} I_2 - G(\mu,\sigma) I(\mu,\sigma)^{-1} G(\mu,\sigma)^{\top}=
\begin{bmatrix}
1/2-\frac{h_{1}^2(\lambda)}{\lambda} ~&~ 0 \\[1mm]
0 ~&~ 1/2-\frac{h_{2}^2(\lambda)}{\Gamma(1/\lambda)\Gamma(2-1/\lambda)}
\end{bmatrix}.
\]

The limits of $h_{1}(\lambda)$ and $h_{2}(\lambda)$ (from the Table of constants in the main text) are studied as $\lambda\rightarrow\infty$. Given that $\max_{\lambda\geq 1,v\in [0,1]} v^{1/\lambda} \leq 1$ and the map $\lambda\mapsto v^{1/\lambda}$ is decreasing on $[1,\infty)$ for all $v > 1$, it is easy to see that, for all $\lambda\geq 1$ and all $v\in (0,\infty)$,
\[
\left|\cos[\pi \{1 + \Gamma_{1/\lambda,1}(v)\}]v^{1/\lambda}e^{-v}\right| \leq \max(1,v) \, e^{-v} \leq e^{-v/2}.
\]
Since $v\mapsto e^{-v/2}$ is integrable on $(0,\infty)$, Lebesgue's dominated convergence theorem yields
\begin{equation*}
\lim_{\lambda\rightarrow\infty} h_{1}(\lambda) = \frac{1}{\lim_{\lambda\rightarrow\infty}\Gamma(1+1/\lambda)} \int_{0}^{\infty}\lim_{\lambda\rightarrow\infty}\cos[\pi \{1 + \Gamma_{1/\lambda,1}(v)\}]v^{1/\lambda}e^{-v}dv
=\cos(2\pi)\int_{0}^{\infty}e^{-v}dv=1,
\end{equation*}
because $\Gamma_{1/\lambda,1}(v)$ and $v^{1/\lambda}$ converge pointwise to 1 as $\lambda\rightarrow\infty$, for every $v>0$. Similarly, for all $\lambda\geq 1$ and all $v\in (0,\infty)$,
\[
\left|\sin[\pi \{1 + \Gamma_{1/\lambda,1}(v)\}]e^{-v}\right| \leq e^{-v},
\]
and $v\mapsto e^{-v}$ is integrable on $(0,\infty)$, so Lebesgue's dominated convergence theorem yields
\begin{equation*}
\lim_{\lambda\rightarrow\infty}h_{2}(\lambda)=\int_{0}^{\infty}\lim_{\lambda\rightarrow\infty}\sin[\pi \{1 + \Gamma_{1/\lambda,1}(v)\}]e^{-v} dv
=\sin(2\pi)\int_{0}^{\infty}e^{-v}dv=\sin(2\pi)=0.
\end{equation*}
Therefore, one has
\begin{equation*}
\lim_{\lambda\rightarrow\infty}1/2-\frac{h_{1}^2(\lambda)}{\lambda} = 1/2, \qquad
\lim_{\lambda\rightarrow\infty}1/2-\frac{h_{2}^2(\lambda)}{\Gamma(1/\lambda)\Gamma(2-1/\lambda)} = 1/2,
\end{equation*}
which means that
\[
\lim_{\lambda\rightarrow\infty}\Sigma(\mu,\sigma)=
\lim_{\lambda\rightarrow\infty}\begin{bmatrix}
1/2-\frac{h_{1}^2(\lambda)}{\lambda} ~&~ 0 \\[1mm]
0 ~&~ 1/2-\frac{h_{2}^2(\lambda)}{\Gamma(1/\lambda)\Gamma(2-1/\lambda)}
\end{bmatrix} = \begin{bmatrix}
1/2 ~&~ 0 \\[1mm]
0 ~&~ 1/2
\end{bmatrix}=\frac{1}{2}I_2.
\]

\section{Proof of Proposition~4 for the asymptotics of the \texorpdfstring{$\mathrm{gamma}(\lambda, \beta)$}{Gamma} test under local alternatives}\label{proof:local.gamma}

Consider the gamma test under the following sequence of local alternatives:
\[
\mathcal{H}_0 : X_i \sim \mathrm{gamma}(\lambda_0, \beta_0) \equiv \mathrm{GG}(\lambda_0, \beta_0, \rho_0=1)\quad \text{ vs. }\quad
\mathcal{H}_{1,n}(\delta) : X_i \sim \mathrm{GG}\Big(\lambda_0, \beta_0, \rho_n=1 + \frac{\delta}{\sqrt{n}} \{1 + o(1)\}\Big),
\]
where $\delta\in\R\setminus\{0\}$ is fixed. Here $\bb{\theta}_0=[\lambda_0, \beta_0]^{\top}$ and $\bb{\phi}_0=\rho_0=1$ with $\dim(\bb{\phi}_0)=1$.

Note that the quantities $V(\bb{\phi}_0,\bb{\theta}_0)$ and $\Sigma_{\mathcal{R}}(\bb{\phi}_0,\bb{\theta}_0)$ depend only on the value of $\lambda_0$. Therefore, to simplify the expressions and the proof, we set $\beta_0=1$.

For the generalized gamma $\mathrm{GG}(\lambda_0,1,1)$ distribution, we have
\[
G(\lambda_0,1,1)
=
\begin{bmatrix}
h_{10}(\lambda_0) ~ & ~ \lambda_0 h_{6}(\lambda_0,\lambda_0+1,1) ~&~ -h_{8}(\lambda_0) \\[1mm]
h_{11}(\lambda_0) ~ & ~ \lambda_0 h_{7}(\lambda_0,\lambda_0+1,1) ~&~ -h_{9}(\lambda_0)
\end{bmatrix},
\]
which splits into
\[
G(\bb{\theta}_0)=
\begin{bmatrix}
h_{10}(\lambda_0) ~ & ~ \lambda_0 h_{6}(\lambda_0,\lambda_0+1,1) \\[1mm]
h_{11}(\lambda_0) ~ & ~ \lambda_0 h_{7}(\lambda_0,\lambda_0+1,1)
\end{bmatrix}, \qquad
G_{\bb{\phi}}(\bb{\phi_0},\bb{\theta}_0)
  =\begin{bmatrix}
 -h_{8}(\lambda_0) \\[1mm]
 -h_{9}(\lambda_0)
\end{bmatrix}.
\]

We also have, for the $\mathrm{GG}(\lambda_0,1,1)$ distribution,
\[
I(\lambda_0,1,1)
=
\begin{bmatrix}
 \psi_1(\lambda_0) ~ & ~ 1 ~ & ~ -\psi(\lambda_0) \\[1mm]
1 ~ & ~ \lambda_0 ~&~ -\lambda_0 \psi(\lambda_0) - 1 \\[1mm]
-\psi(\lambda_0) ~ & ~ -\lambda_0 \psi(\lambda_0) - 1 ~&~ \lambda_0 \psi^2(\lambda_0) + 2 \psi(\lambda_0) + \lambda_0 \psi_1(\lambda_0) + 1
\end{bmatrix},
\]
which splits into
\[
I(\bb{\theta}_0) =\begin{bmatrix}
 \psi_1(\lambda_0) ~ & ~ 1  \\[1mm]
1 ~ & ~ \lambda_0
\end{bmatrix},
\]
\[
I_{\phi,\phi}(\bb{\phi}_0,\bb{\theta}_0) =
\lambda_0 \psi^2(\lambda_0) + 2 \psi(\lambda_0) + \lambda_0 \psi_1(\lambda_0) + 1,
\]
\[
I_{\phi,\theta}(\bb{\phi}_0,\bb{\theta}_0)=\begin{bmatrix}
 -\psi(\lambda_0) ~ & ~ -\lambda_0 \psi(\lambda_0) - 1
\end{bmatrix}.
\]
Finally, one can verify that
\[
\Sigma_{\mathcal{R}}(\bb{\phi}_0,\bb{\theta}_0)=I_{\phi,\phi}(\bb{\phi}_0,\bb{\theta}_0)-I_{\phi,\theta}(\bb{\phi}_0,\bb{\theta}_0) I_{\theta,\theta}(\bb{\phi}_0,\bb{\theta}_0)^{-1}I_{\phi,\theta}(\bb{\phi}_0,\bb{\theta}_0)^{\top}= \frac{\lambda_0^{2}\,\psi_1(\lambda_0)^{2} - \psi_1(\lambda_0) - 1}{\lambda_0\,\psi_1(\lambda_0) - 1}.
\]

\section{Proof of Proposition~5 for the asymptotics of the \texorpdfstring{$\mathrm{EPD}_{\lambda}(\mu, \sigma)$}{EPD} test under local alternatives}\label{proof:local.EPD}

Recall the $\mathrm{APD}_{\lambda}(\alpha,\rho,\mu,\sigma)$ distribution and its density $f_{\lambda}(x \nvert  \alpha,\rho,\mu,\sigma)$.
One can show that
\begin{equation}\label{eq:sk}
\bb{s}_{\bb{\phi}}(x, \bb{\phi}_0,\bb{\theta}_0)\equiv \bb{s}_{\mathcal{K}}(x, \lambda_0,\mu_0,\sigma_0) =
\begin{bmatrix}
\partial_{\alpha} \ln \{f_{\lambda_0}(x \nvert \alpha=1/2,\rho=\lambda_0,\mu_0,\sigma_0)\} \\[1mm]
\partial_{\rho} \ln \{f_{\lambda_0}(x \nvert \alpha=1/2,\rho=\lambda_0,\mu_0,\sigma_0)\}
\end{bmatrix}=
\begin{bmatrix}
-2 |y|^{\lambda_0} \mathrm{sign}(y) \\[1mm]
-\frac{1}{\lambda_0} \Big[|y|^{\lambda_0} \ln |y| - \frac{C_{\lambda_0}}{\lambda_0}\Big]
\end{bmatrix},
\end{equation}
where $C_{\lambda_0}=\psi(1/\lambda_0+1)+\ln(\lambda_0)$, $y = (x - \mu_0) / \sigma_0$, $\lambda_0$ is known, and where the last equality in \eqref{eq:sk} is straightforward to verify using, e.g., \texttt{Mathematica}.

Note that the resulting matrices $V(\bb{\phi}_0,\bb{\theta}_0)$ (for both the ML and MM cases) and $\Sigma_{\mathcal{R}}(\bb{\phi}_0,\bb{\theta}_0)$ depend only on the value of $\lambda_0$. Therefore, to simplify the expressions and the proof, we set $\mu_0=0,\sigma_0=1$.

Next, one determines the matrix $G_{\bb{\phi}}(\bb{\phi_0},\bb{\theta}_0)
  = \EE\!\left\{\bb{\tau}(X,\bb{\theta}_0)\,\bb{s}_{\bb{\phi}}(X,\bb{\phi}_0,\bb{\theta}_0)^{\top}\right\}$.
One has, if $y = (x - \mu_0)/\sigma_0$,
\begin{align*}
\bb{\tau}(x,\bb{\theta}_0)\equiv \bb{\tau}(x, \lambda_0,\mu_0, \sigma_0)
&=
\begin{bmatrix}
\tau_1(x, \lambda_0,\mu_0, \sigma_0) \\[1mm]
\tau_2(x, \lambda_0,\mu_0, \sigma_0)
\end{bmatrix}
=
\begin{bmatrix}
\cos\left\{2\pi F_{\lambda_0}(x \nvert \mu_0, \sigma_0)\right\} \\[1mm]
\sin\left\{2\pi F_{\lambda_0}(x \nvert  \mu_0, \sigma_0)\right\}
\end{bmatrix}\\
&=
\begin{bmatrix}
\cos\left[\pi \left\{1 + \mathrm{sign}(y) \Gamma_{1/\lambda_0, 1}(|y|^{\lambda_0}/\lambda_0)\right\}\right] \\[1mm]
\sin\left[\pi \left\{1 + \mathrm{sign}(y) \Gamma_{1/\lambda_0, 1}(|y|^{\lambda_0}/\lambda_0)\right\}\right]
\end{bmatrix}\\
&=\begin{bmatrix}
\cos\left[\pi \left\{1 + \Gamma_{1/\lambda_0, 1}(|y|^{\lambda_0}/\lambda_0)\right\}\right] \\[1mm]
\mathrm{sign}(y)\sin\left[\pi \left\{1 + \Gamma_{1/\lambda_0, 1}(|y|^{\lambda_0}/\lambda_0)\right\}\right]
\end{bmatrix},
\end{align*}
since $\cos\{\pi(1-z)\}=\cos\{\pi(1+z)\}$ and $\sin\{\pi(1-z)\}=-\sin\{\pi(1+z)\}$ for all $0\le z \le 1$. One has, for $\lambda_0\in (0, \infty)$,
\begin{align*}
X\sim \mathrm{EPD}_{\lambda_0}(\mu_0, \sigma_0)
~~ & \Leftrightarrow ~~ Y = \frac{X - \mu_0}{\sigma_0} \sim \mathrm{EPD}_{\lambda_0}(0, 1)\\
\Leftrightarrow |Y| \sim \mathrm{half\mhyphen EPD}(\lambda_0, 1) ~~ & \Leftrightarrow ~~ V = \frac{1}{\lambda_0}|Y|^{\lambda_0} \sim \mathrm{gamma}(1/\lambda_0, 1).
\end{align*}
Note that $\bb{\tau}(x, \lambda_0,\mu_0, \sigma_0) = \bb{\tau}(y, \lambda_0,0, 1)$ and $\bb{s}_{\mathcal{K}}(x, \lambda_0,\mu_0,\sigma_0)=\bb{s}_{\mathcal{K}}(y, \lambda_0, 0, 1)$.
Using that $\EE\{\bb{\tau}(Y, \lambda_0,0, 1)\} = \bb{0}_2$, one has
\begin{align*}
\EE&\Big\{\tau_1(X, \lambda_0,\mu_0, \sigma_0) s_{\mathcal{K},2}(X,\lambda_0, \mu_0, \sigma_0)\Big\}
= -\frac{1}{\lambda_0}\EE\left\{\tau_1(Y, \lambda_0,0, 1) \Big[|Y|^{\lambda_0} \ln(|Y|) - \frac{1}{\lambda_0} \left\{\lambda_0 + \ln (\lambda_0) + \psi(1/\lambda_0)\right\}\Big]\right\}\\
&= -\frac{1}{\lambda_0}\EE\left\{\cos\left[\pi \left\{1 + \Gamma_{1/\lambda_0, 1}(|Y|^{\lambda_0}/\lambda_0)\right\}\right]|Y|^{\lambda_0} \ln(|Y|)\right\}
= -\frac{1}{\lambda_0}\EE\left\{\cos\left[\pi \left\{1 + \Gamma_{1/\lambda_0, 1}(V)\right\}\right]V \ln(\lambda_0 V)\right\}\\
&= -\frac{1}{\lambda_0}\int_0^{\infty} \cos\left[\pi \left\{1 + \Gamma_{1/\lambda_0, 1}(v)\right\}\right]v\ln(\lambda_0 v) f_{\mathrm{ga}}(v \nvert 1/\lambda_0, 1) \rd v\\
&= -\frac{1}{\lambda_0^2}\int_0^{\infty} \cos\left[\pi \left\{1 + \Gamma_{1/\lambda_0, 1}(v)\right\}\right]\ln(\lambda_0 v) f_{\mathrm{ga}}(v \nvert 1/\lambda_0+1, 1) \rd v
= -\frac{1}{\lambda_0^2}h_{3}(\lambda_0),
\end{align*}
\begin{align*}
\EE&\Big\{\tau_2(X, \lambda_0,\mu_0, \sigma_0) s_{\mathcal{K},1}(X, \lambda_0,\mu_0, \sigma_0)\Big\}
= -2\EE\left\{\mathrm{sign}(Y)\sin\left[\pi \left\{1 + \Gamma_{1/\lambda_0, 1}(|Y|^{\lambda_0}/\lambda_0)\right\}\right] |Y|^{\lambda_0} \mathrm{sign}(Y)\right\}\\
&= -2\lambda_0\EE\left\{\sin\left[\pi \left\{1 + \Gamma_{1/\lambda_0, 1}(V)\right\}\right] V \right\}
= -2\lambda_0\int_0^{\infty} \sin\left[\pi \left\{1 + \Gamma_{1/\lambda_0, 1}(v)\right\}\right]v f_{\mathrm{ga}}(v \nvert 1/\lambda_0, 1) \rd v\\
&= -2\int_0^{\infty} \sin\left[\pi \left\{1 + \Gamma_{1/\lambda_0, 1}(v)\right\}\right] f_{\mathrm{ga}}(v \nvert 1/\lambda_0 + 1, 1) \rd v= -2 h_{35}(\lambda_0).
\end{align*}

Given that $\tau_2(y, \lambda_0,0, 1)$ and $s_{\mathcal{K},1}(Y, \lambda_0,0, 1)$ are odd functions, that $\tau_1(y, \lambda_0,0, 1)$, $s_{\mathcal{K},2}(Y, \lambda_0,0, 1)$ and $f(y \nvert \lambda_0, 0, 1)$ are even functions, one has
\[
\EE\left\{\tau_1(X, \lambda_0,\mu_0, \sigma_0) s_{\mathcal{K},1}(X, \lambda_0,\mu_0, \sigma_0)\right\} = \EE\left\{\tau_1(Y, \lambda_0, 0, 1) s_{\mathcal{K},1}(Y, \lambda_0, 0, 1)\right\} = 0
\]
and
\[
\EE\left\{\tau_2(X, \lambda_0,\mu_0, \sigma_0) s_{\mathcal{K},2}(X, \lambda_0,\mu_0, \sigma_0)\right\} = \EE\left\{\tau_2(Y, \lambda_0,0, 1) s_{\mathcal{K},2}(Y, \lambda_0,0, 1)\right\} = 0
\]
and one obtains
\[
G_{\bb{\phi}}(\bb{\phi_0},\bb{\theta}_0)
=
\begin{bmatrix}
0 & -\frac{1}{\lambda_0^2}h_{3}(\lambda_0)\\[1mm]
-2 h_{35}(\lambda_0) & 0
\end{bmatrix}.
\]

For the ML estimator, one has $\bb{r}(x, \lambda_0,\mu_0, \sigma_0)=\bb{s}(x, \lambda_0,\mu_0, \sigma_0)$ and then
\[
S_{\phi,\theta}(\bb{\phi}_0,\bb{\theta}_0)=\EE\big\{\bb{s}_{\mathcal{K}}(X, \lambda_0,\mu_0,\sigma_0) \bb{r}(X, \lambda_0,\mu_0,\sigma_0)^{\top}\big\}=I_{\phi,\theta}(\bb{\phi}_0,\bb{\theta}_0)=\EE\big\{\bb{s}_{\mathcal{K}}(X, \lambda_0,\mu_0,\sigma_0) \bb{s}(X, \lambda_0,\mu_0,\sigma_0)^{\top}\big\},
\]
where
\[
\bb{s}(x, \lambda_0,\mu_0, \sigma_0)
=
\begin{bmatrix}
s_1(x, \lambda_0,\mu_0, \sigma_0) \\[1mm]
s_2(x, \lambda_0,\mu_0, \sigma_0)
\end{bmatrix}
=
\begin{bmatrix}
\partial_{\mu_0} \ln f(x \nvert \lambda_0,\mu_0, \sigma_0) \\[1mm]
\partial_{\sigma_0} \ln f(x \nvert \lambda_0,\mu_0, \sigma_0)
\end{bmatrix}
= \frac{1}{\sigma_0}
\begin{bmatrix}
|y|^{\lambda_0 - 1} \mathrm{sign}(y) \\[1mm]
|y|^{\lambda_0} - 1
\end{bmatrix}
\]
is given in Section~\ref{proof:EPD}, with $y = (x-\mu_0)/\sigma_0$.

Next, one determines the matrix $I_{\phi,\theta}(\bb{\phi}_0,\bb{\theta}_0)$ for the ML estimator. By invariance, we can set $\mu_0=0$ and $\sigma_0=1$. With similar arguments of odd and even functions, one shows that
\[
\EE\left\{s_{\mathcal{K},1}(X, \lambda_0,0, 1) s_2(X, \lambda_0,0,1) \right\} = \EE\left\{s_{\mathcal{K},1}(Y, \lambda_0,0, 1) s_2(Y, \lambda_0,0, 1)\right\} = 0
\]
and
\[
\EE\left\{s_{\mathcal{K},2}(X,\lambda_0, 0,1) s_1(X, \lambda_0,0,1)\right\} = \EE\left\{s_{\mathcal{K},2}(Y, \lambda_0,0, 1) s_1(Y, \lambda_0,0, 1)\right\} = 0.
\]

Now, using \eqref{eqn.EPD}, one has
\[
-\EE\Big\{s_{\mathcal{K},1}(Y, \lambda_0,0,1) s_1(Y, \lambda_0,0,1)\Big\}
= 2\EE\left\{|Y|^{\lambda_0}\mathrm{sign}(Y)|Y|^{\lambda_0-1}\mathrm{sign}(Y)\right\}
= 2\EE\left\{|Y|^{2\lambda_0-1}\right\} = \frac{2\lambda_0^{2-1/\lambda_0}}{\Gamma(1/\lambda_0)},
\]
\begin{align*}
-\lambda_0&\EE\Big\{s_{\mathcal{K},2}(Y, \lambda_0,0,1) s_2(Y, \lambda_0,0,1)\Big\}
= \EE\left\{(|Y|^{\lambda_0} - 1)|Y|^{\lambda_0} \ln |Y|\right\}= \EE\left\{|Y|^{2\lambda_0} \ln |Y|\right\}
 -\EE\left\{|Y|^{\lambda_0} \ln |Y|\right\}\\
 &=\frac{\lambda_0\Gamma(1/\lambda_0+2)\{\psi(1/\lambda_0+2)+\ln(\lambda_0)\}}{\Gamma(1/\lambda_0)}-
 \frac{\Gamma(1/\lambda_0+1)\{\psi(1/\lambda_0+1)+\ln(\lambda_0)\}}{\Gamma(1/\lambda_0)}\\
 &=(1/\lambda_0+1)\{\psi(1/\lambda_0+1)+(1/\lambda_0+1)^{-1}+\ln(\lambda_0)\}-(1/\lambda_0)\{\psi(1/\lambda_0+1)+\ln(\lambda_0)\}\\
 &=\psi(1/\lambda_0+1)+\ln(\lambda_0)+1 = C_{\lambda_0}+1,
\end{align*}
and one obtains, for the ML estimator,
\[
I_{\phi,\theta}(\bb{\phi}_0,\bb{\theta}_0)
=
\begin{bmatrix}
\frac{-2\lambda_0^{2-1/\lambda_0}}{\Gamma(1/\lambda_0)} & 0\\[1mm]
0 & -\frac{1}{\lambda_0}\{C_{\lambda_0}+1\}
\end{bmatrix}.
\]

Recall from the EPD distribution that
\[
G(\bb{\theta}_0)\equiv G(\lambda_0,0, 1)=
\begin{bmatrix}
0 ~&~ h_{1}(\lambda_0) \\[1mm]
\frac{h_{2}(\lambda_0)}{\lambda_0^{1/\lambda_0-1} \Gamma(1/\lambda_0)} ~&~ 0
\end{bmatrix},\quad I(\bb{\theta}_0)\equiv I(\lambda_0,0,1)=
\begin{bmatrix}
\frac{\lambda_0^{2 - 2/\lambda_0}\Gamma(2-1/\lambda_0)}{\Gamma(1/\lambda_0)} ~&~ 0 \\[1mm]
0 ~&~ \lambda_0
\end{bmatrix}.
\]
Next, one determines the matrix $S_{\phi,\theta}(\bb{\phi}_0,\bb{\theta}_0)$ for MM estimator. Recall from \eqref{eq.r.EPD}   that
\begin{equation*}
\bb{r}(x,\lambda_0,\mu_0,\sigma_0)
= R(\lambda_0,\mu_0,\sigma_0)\bb{g}(x, \lambda_0,\mu_0,\sigma_0)=\frac{1}{\sigma_0}
\begin{bmatrix}
C_{2,\lambda_0}y \\[1mm]
2D_{\lambda_0}\left(C_{2,\lambda_0}y^2 - 1\right)
\end{bmatrix},
\end{equation*}
where $y=(x-\mu_0)/\sigma_0$ and
\[C_{2,\lambda_0}=\frac{\Gamma(1/\lambda_0)}{\lambda_0^{2/\lambda_0}\Gamma(3/\lambda_0)},\quad D_{\lambda_0}=\frac{\Gamma^2(3/\lambda_0)}{\Gamma(1/\lambda_0)\Gamma(5/\lambda_0)-\Gamma^2(3/\lambda_0)}.
\]
Recall also that
\[
R(\bb{\theta}_0)\equiv R(\lambda_0,0,1)=\begin{bmatrix}
C_{2,\lambda_0} & 0 \\[1mm]
0 & 4D_{\lambda_0}
\end{bmatrix}.
\]

With similar arguments of odd and even functions, one shows that
\[
\EE\left\{s_{\mathcal{K},1}(X, \lambda_0,\mu_0, \sigma_0) r_2(X, \lambda_0,\mu_0, \sigma_0) \right\} = \sigma_0^{-1}\EE\left\{s_{\mathcal{K},1}(Y, \lambda_0,0, 1) r_2(Y, \lambda_0,0, 1)\right\} = 0
\]
and
\[
\EE\left\{s_{\mathcal{K},2}(X, \lambda_0,\mu_0, \sigma_0) r_1(X, \lambda_0,\mu_0, \sigma_0)\right\} = \sigma_0^{-1}\EE\left\{s_{\mathcal{K},2}(Y, \lambda_0,0, 1) r_1(Y, \lambda_0,0, 1)\right\} = 0.
\]
Now, using \eqref{eqn.EPD}, one has

\begin{align*}
-\sigma_0\EE\Big\{s_{\mathcal{K},1}(X, \lambda_0,\mu_0, \sigma_0) r_1(X, \lambda_0,\mu_0, \sigma_0)\Big\}
&= 2C_{2,\lambda_0}\EE\left\{|Y|^{\lambda_0}\mathrm{sign}(Y)Y\right\}= 2C_{2,\lambda_0}\EE\left\{|Y|^{\lambda_0+1}\right\}\\
&=2\frac{\Gamma(1/\lambda_0)}{\lambda_0^{2/\lambda_0}\Gamma(3/\lambda_0)} \frac{\lambda_0^{1+1/\lambda_0}\Gamma(1+2/\lambda_0)}{\Gamma(1/\lambda_0)}= \frac{4\Gamma(2/\lambda_0)}{\lambda_0^{1/\lambda_0}\Gamma(3/\lambda_0)},
\end{align*}

\begin{align*}
-\frac{\lambda_0\sigma_0}{2D_{\lambda_0}}&\EE\Big\{s_{\mathcal{K},2}(X, \lambda_0,\mu_0, \sigma_0) r_2(X, \lambda_0,\mu_0, \sigma_0)\Big\}
= \EE\left\{|Y|^{\lambda_0} \ln |Y|\left(C_{2,\lambda_0}Y^2 - 1\right)\right\}\\
&=C_{2,\lambda_0}\EE\left\{|Y|^{\lambda_0+2} \ln |Y| \right\} -
\EE\left\{|Y|^{\lambda_0} \ln |Y|\right\}\\
&=\frac{\Gamma(1/\lambda_0)}{\lambda_0^{2/\lambda_0}\Gamma(3/\lambda_0)}\frac{\lambda_0^{2/\lambda_0}\Gamma(1+3/\lambda_0)\{\psi(1+3/\lambda_0)+\ln(\lambda_0)\}}{\Gamma(1/\lambda_0)}
 - \frac{\Gamma(1+1/\lambda_0)\{\psi(1+1/\lambda_0)+\ln(\lambda_0)\}}{\Gamma(1/\lambda_0)}\\
&=3\lambda_0^{-1}\{\psi(1+3/\lambda_0)+\ln(\lambda_0)\} - \lambda_0^{-1}\{\psi(1+1/\lambda_0)+\ln(\lambda_0)\}\\
&=\lambda_0^{-1}\left\{2\ln(\lambda_0)+ 3\psi(3/\lambda_0) - \psi(1/\lambda_0)\right\}
\end{align*}
and one obtains, for the MM estimator, setting $\sigma_0=1$ by invariance,
\[
S_{\phi,\theta}(\bb{\phi}_0,\bb{\theta}_0)
=
\begin{bmatrix}
\frac{-4\Gamma(2/\lambda_0)}{\lambda_0^{1/\lambda_0}\Gamma(3/\lambda_0)} & 0\\[1mm]
0 & -\frac{2D_{\lambda_0}}{\lambda_0^2}\left\{2\ln(\lambda_0)+ 3\psi(3/\lambda_0) - \psi(1/\lambda_0)\right\}
\end{bmatrix}.
\]

Finally, one determines  $I_{\phi,\phi}(\bb{\phi}_0,\bb{\theta}_0)$.

\[
\EE\left\{\bb{s}_{\bb{\phi},1}(X, \bb{\phi}_0,\bb{\theta}_0)^2 \right\}=\EE\left[\left\{-2 |Y|^{\lambda_0} \mathrm{sign}(Y)\right\}^2\right]=4\EE\left\{|Y|^{2\lambda_0} \right\}=4(\lambda_0+1)
\]
using \eqref{eqn.EPD}.

\begin{align*}
&\lambda_0^2\,\EE\left\{\bb{s}_{\bb{\phi},2}(X, \bb{\phi}_0,\bb{\theta}_0)^2 \right\}=\EE\left[\left\{|Y|^{\lambda_0} \ln |Y| - \frac{C_{\lambda_0}}{\lambda_0}\right\}^2\right]=\EE\left[|Y|^{2\lambda_0} (\ln |Y|)^2\right]
+\frac{C_{\lambda_0}^2}{\lambda_0^2}-\frac{2 C_{\lambda_0}}{\lambda_0}\EE\left[|Y|^{\lambda_0} \ln |Y| \right]\\
&= \frac{\Gamma(1/\lambda_0+2)}{\Gamma(1/\lambda_0)}\big[\psi_1(1/\lambda_0+2)+\{\psi(1/\lambda_0+2)+\ln(\lambda_0)\} ^2\big] +\frac{C_{\lambda_0}^2}{\lambda_0^2}
-\frac{2 C_{\lambda_0}^2\Gamma(1/\lambda_0+1)}{\lambda_0\Gamma(1/\lambda_0)}\\
&= \frac{1/\lambda_0+1}{\lambda_0}\left[\psi_1(1/\lambda_0+1)-\frac{1}{(1/\lambda_0+1)^2}+\left\{\psi(1/\lambda_0+1)+\frac{1}{1/\lambda_0+1}+\ln(\lambda_0)\right\} ^2\right] +\frac{C_{\lambda_0}^2}{\lambda_0^2}
-\frac{2 C_{\lambda_0}^2}{\lambda_0^2}\\
&= \frac{1/\lambda_0+1}{\lambda_0}\left[\psi_1(1/\lambda_0+1)-\frac{1}{(1/\lambda_0+1)^2}+\left\{\frac{1}{1/\lambda_0+1}+C_{\lambda_0}\right\} ^2\right] -\frac{C_{\lambda_0}^2}{\lambda_0^2}\\
&= \frac{1}{\lambda_0}\left\{\Big(\frac{1}{\lambda_0}+1\Big) \psi_1\Big(\frac{1}{\lambda_0}+1\Big)  + (1 + C_{\lambda_0})^ 2- 1\right\}.
\end{align*}
With arguments of odd and even functions, one shows that
\[
\EE\left\{\bb{s}_{\bb{\phi},1}(X, \bb{\phi}_0,\bb{\theta}_0)\bb{s}_{\bb{\phi},2}(X, \bb{\phi}_0,\bb{\theta}_0) \right\}=0,
\]
and one obtains
\[
I_{\phi,\phi}(\bb{\phi}_0,\bb{\theta}_0) =\begin{bmatrix}
4 (\lambda_0 + 1) ~&~ 0 \\[1mm]
0 ~&~ \frac{1}{\lambda_0^3}\left\{\Big(\frac{1}{\lambda_0}+1\Big) \psi_1\Big(\frac{1}{\lambda_0}+1\Big) - 1 + (C_{\lambda_0} + 1)^ 2\right\}
\end{bmatrix}.
\]

One can also verify that
\begin{align*}
\Sigma_{\mathcal{R}}(\bb{\phi}_0,\bb{\theta}_0)&=I_{\phi,\phi}(\bb{\phi}_0,\bb{\theta}_0)-I_{\phi,\theta}(\bb{\phi}_0,\bb{\theta}_0) I_{\theta,\theta}(\bb{\phi}_0,\bb{\theta}_0)^{-1}I_{\phi,\theta}(\bb{\phi}_0,\bb{\theta}_0)^{\top}\\
&=
\begin{bmatrix}
4(\lambda_0 + 1) \;-\; \displaystyle \frac{4\lambda_0^{2}}{\Gamma(2 - 1/\lambda_0)\,\Gamma(1/\lambda_0)} & 0 \\[1mm]
0 &  \frac{\big(\frac{1}{\lambda_0}+1\big)\psi_{1}\big(\tfrac{1}{\lambda_0}+1\big) - 1} {\lambda_0^{3}}
\end{bmatrix}.
\end{align*}
This concludes the proof.

\section{Power curves for the empirical power study}

\begin{figure}[!htbp]
\centering
\includegraphics[width=0.32\textwidth]{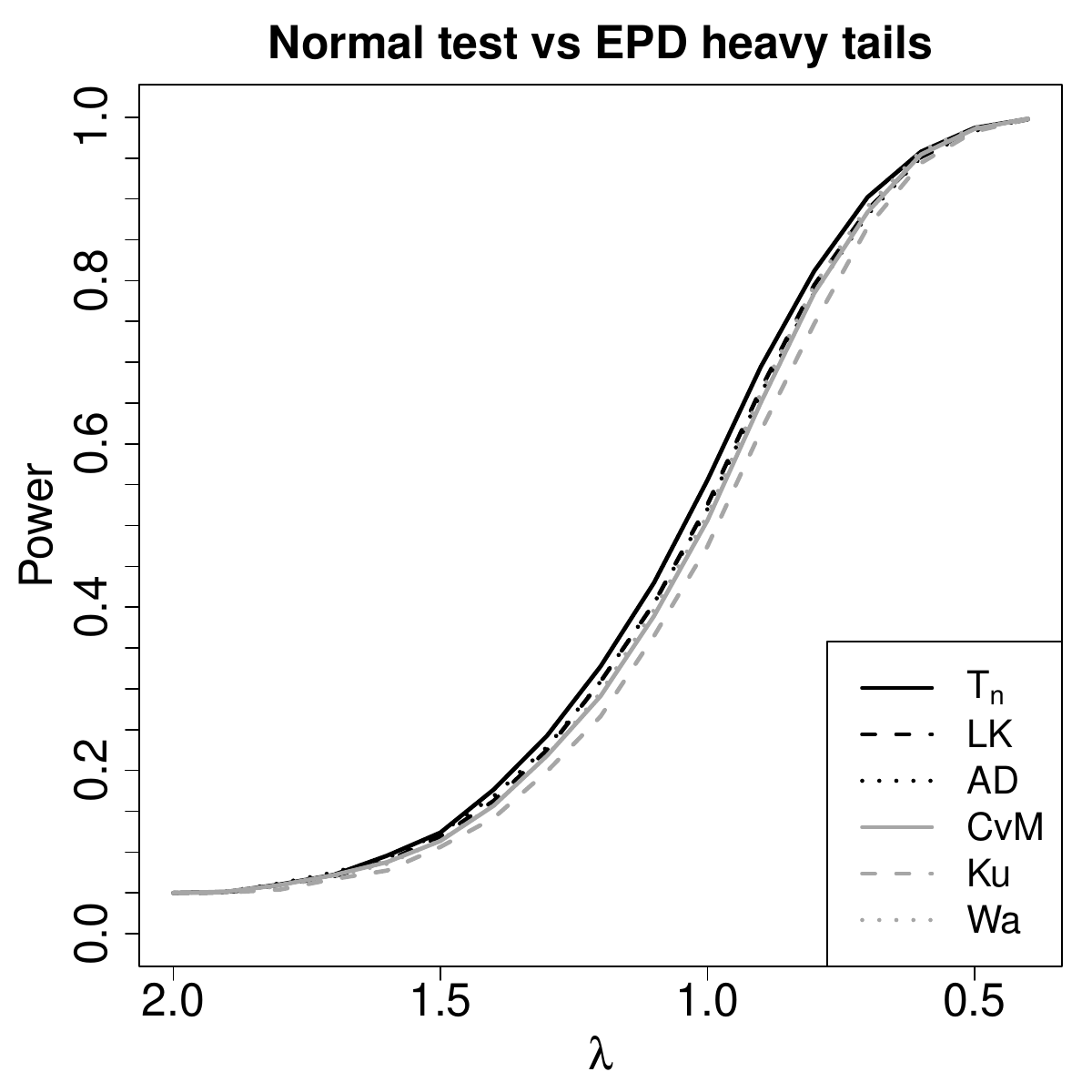}
\includegraphics[width=0.32\textwidth]{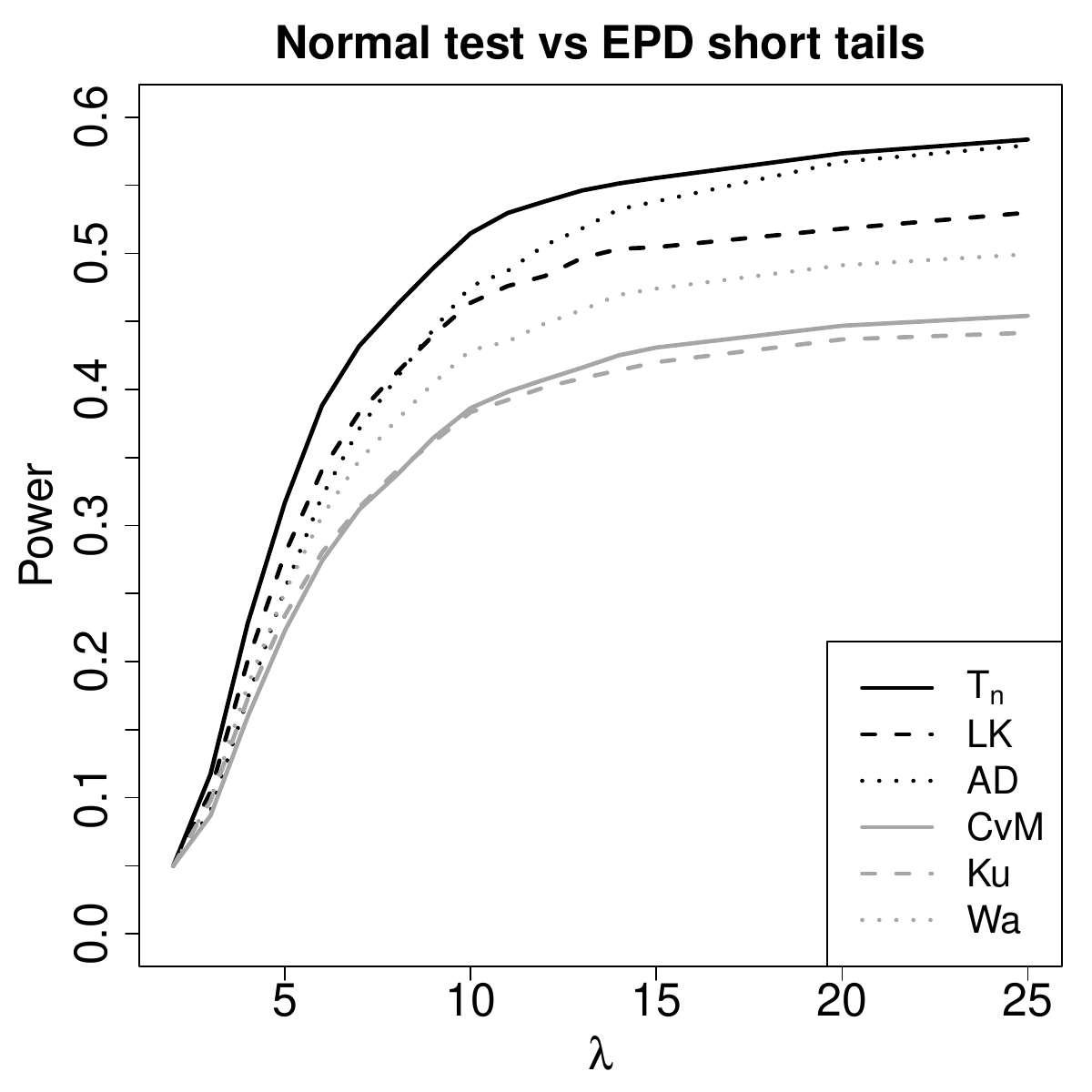}
\includegraphics[width=0.32\textwidth]{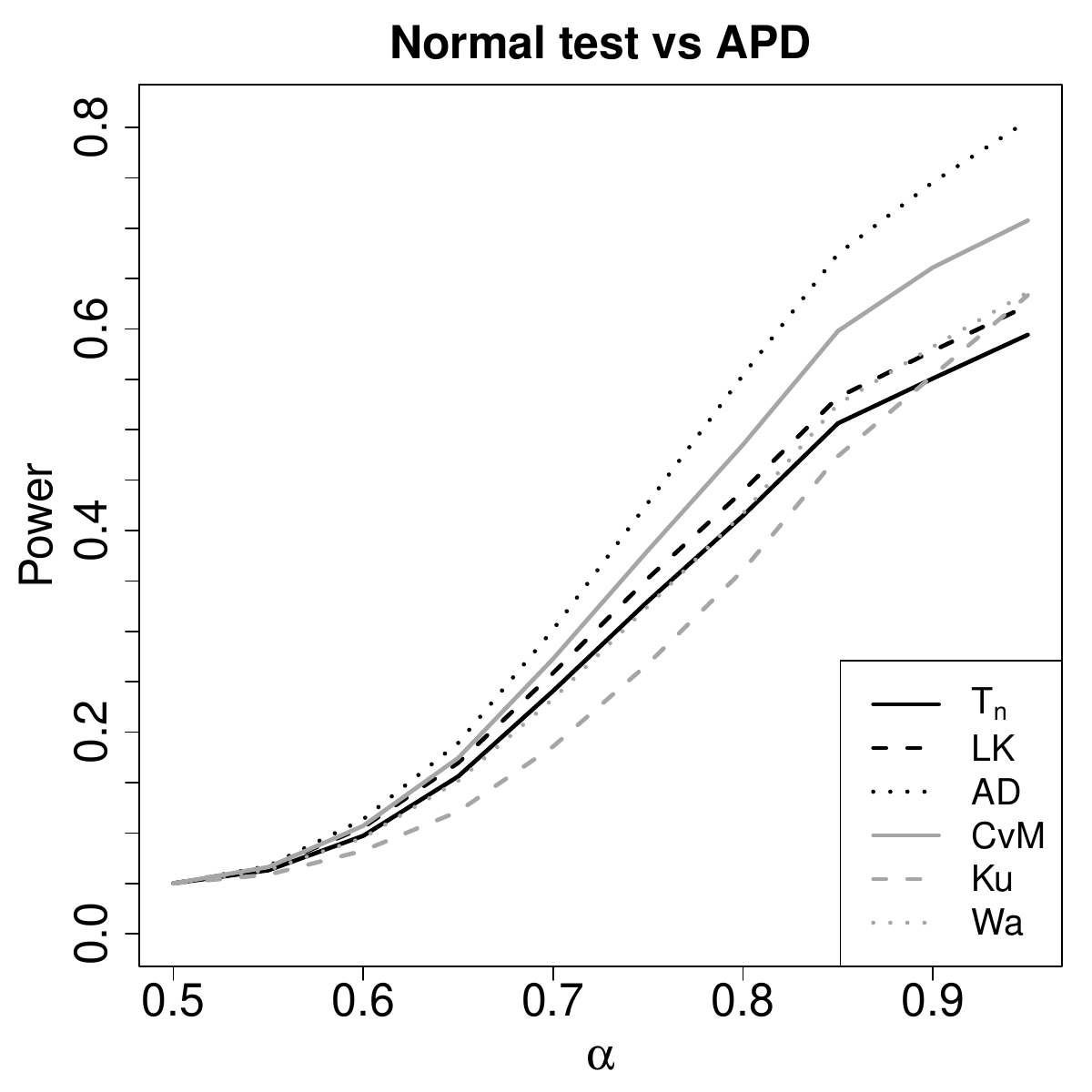}

\caption{Power curves of the $T_n$, LK, AD, CvM, Ku, and Wa tests for assessing
normality (with $\mu_0$ and $\sigma_0$ unknown) under $\mathrm{EPD}(0.4\leq \lambda \leq 2, \mu, \sigma)$,
$\mathrm{EPD}(2 \leq \lambda \leq 25, \mu, \sigma)$, and $\mathrm{APD}_{\lambda=2}(0.5 \leq \alpha < 1, \rho = 2, \mu, \sigma)$
alternatives. The nominal significance level is $0.05$, and the sample size is $n = 50$.}
\label{fig:norm.curve}
\end{figure}

%\begin{figure}[!htbp]
%\centering
%\includegraphics[width=0.32\textwidth]{Power_Curves_t4_EPD_heavy}
%\includegraphics[width=0.32\textwidth]{Power_Curves_t4_EPD_short}
%\includegraphics[width=0.32\textwidth]{Power_Curves_t4_SN}

%\caption{Power curves of the $T_n$, LK, AD, CvM, Ku, and Wa tests for assessing goodness-of-fit to the Student's $t_4$ distribution (with $\lambda_0 = 4$ known and
%$\mu_0$, $\sigma_0$ unknown) under $\mathrm{EPD}(0.3 \leq \lambda \leq 1.2, \mu, \sigma)$, $\mathrm{EPD}(1.2 \leq \lambda \leq 15, \mu, \sigma)$,
%and $\mathrm{SN}(0 \leq \lambda \leq 16, \mu, \sigma)$ alternatives. The nominal significance level is $0.05$, and the sample size is
%$n = 50$.}
%\label{fig:t4.curve}
%\end{figure}

\begin{figure}[!htbp]
\centering
\includegraphics[width=0.32\textwidth]{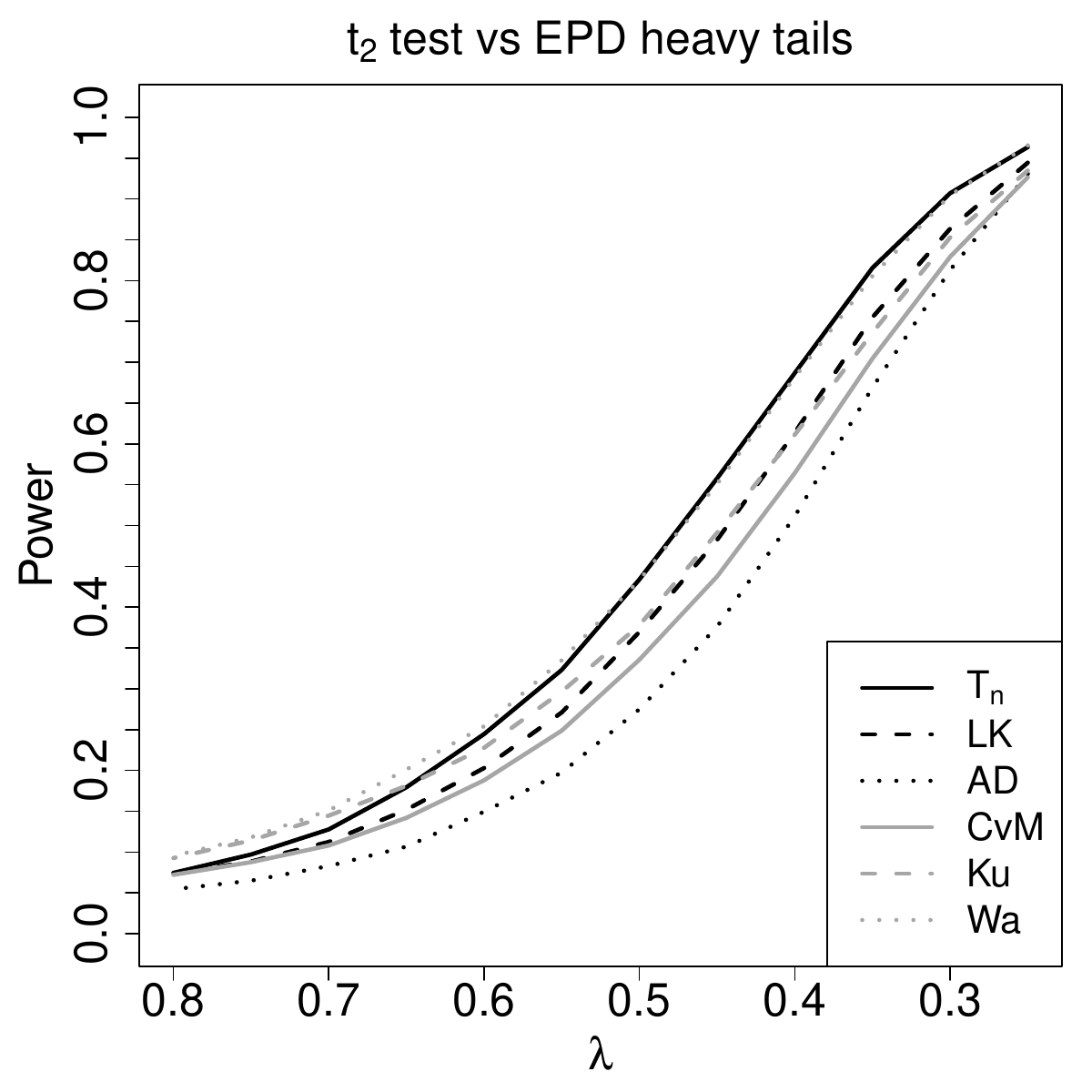}
\includegraphics[width=0.32\textwidth]{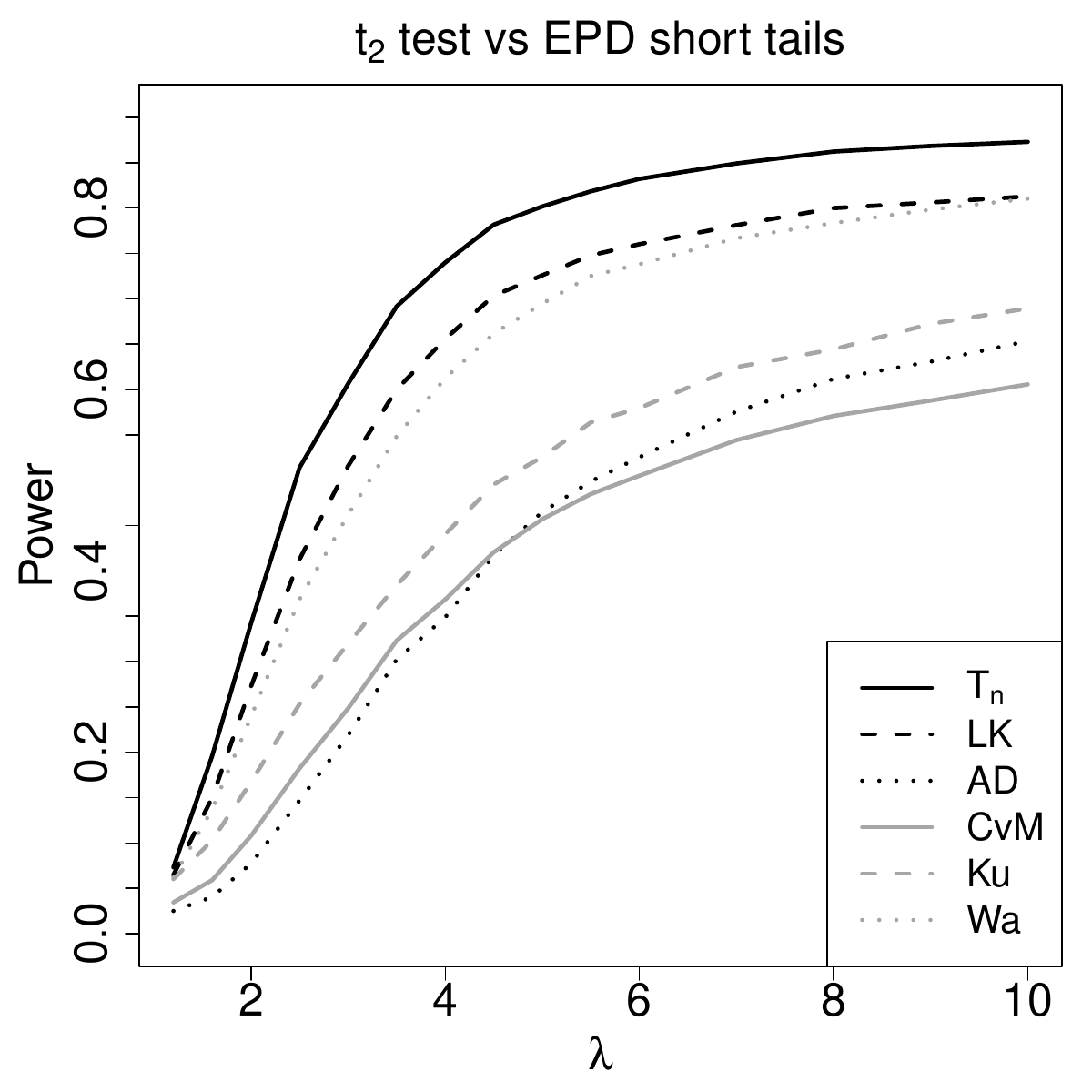}
\includegraphics[width=0.32\textwidth]{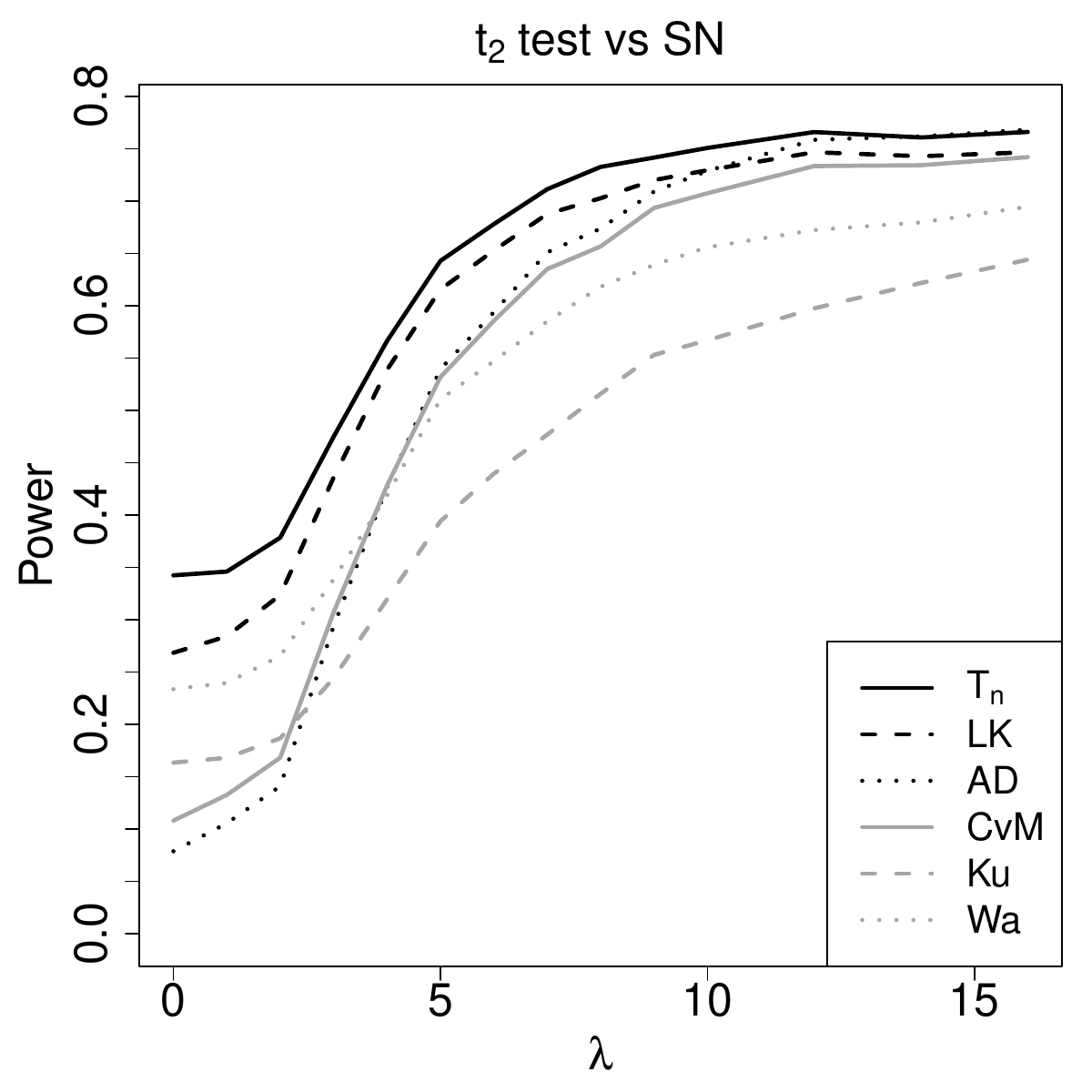}

\caption{Power curves of the $T_n$, LK, AD, CvM, Ku, and Wa tests for assessing goodness-of-fit to the Student's $t_2$ distribution (with $\lambda_0 = 2$ known and
$\mu_0$, $\sigma_0$ unknown) under $\mathrm{EPD}(0.25 \leq \lambda \leq 0.8, \mu, \sigma)$, $\mathrm{EPD}(1.2 \leq \lambda \leq 10, \mu, \sigma)$,
and $\mathrm{SN}(0 \leq \lambda \leq 16, \mu, \sigma)$ alternatives. The nominal significance level is $0.05$, and the sample size is
$n = 50$.}
\label{fig:t2.curve}
\end{figure}

\begin{figure}[!htbp]
\centering
\includegraphics[width=0.24\textwidth]{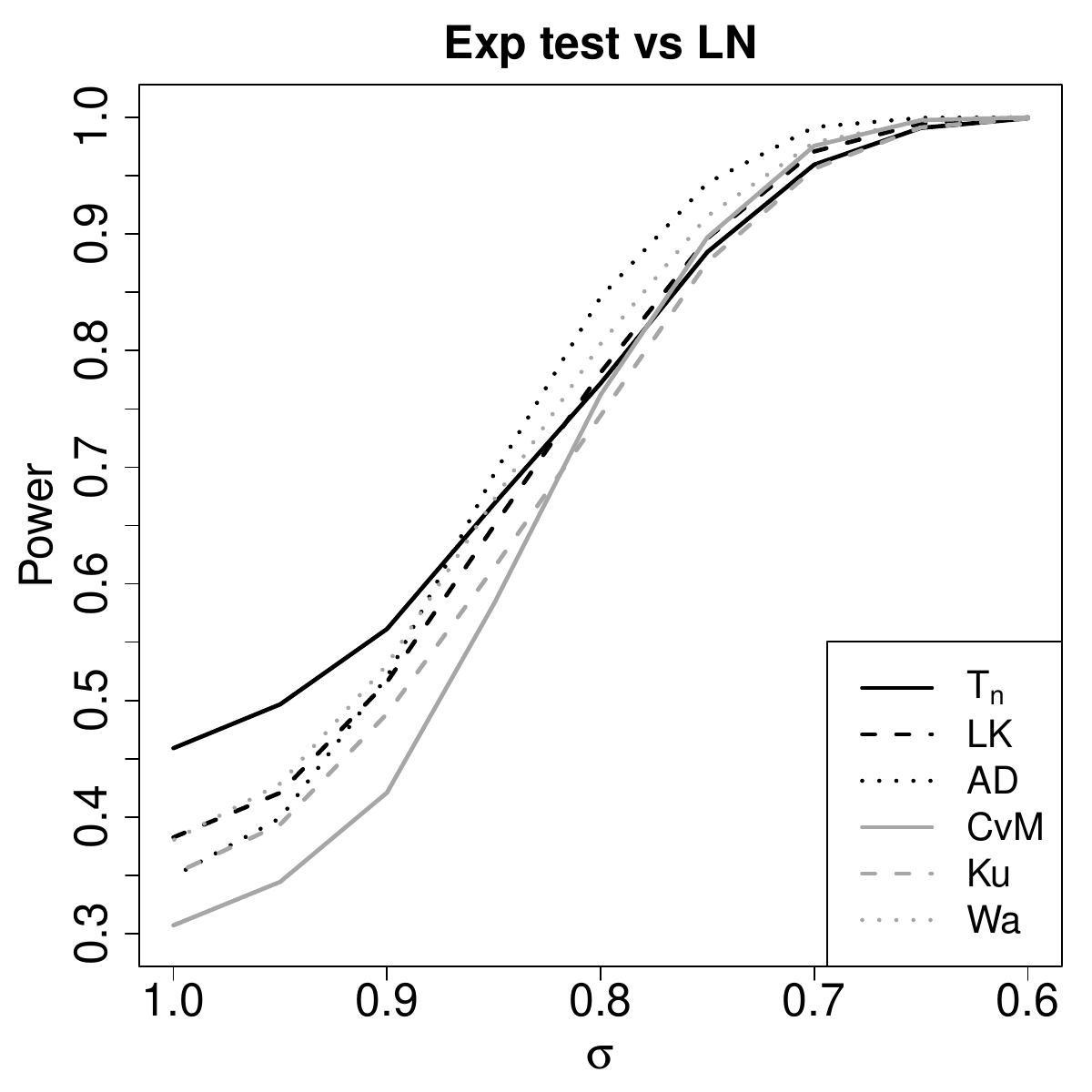}
\includegraphics[width=0.24\textwidth]{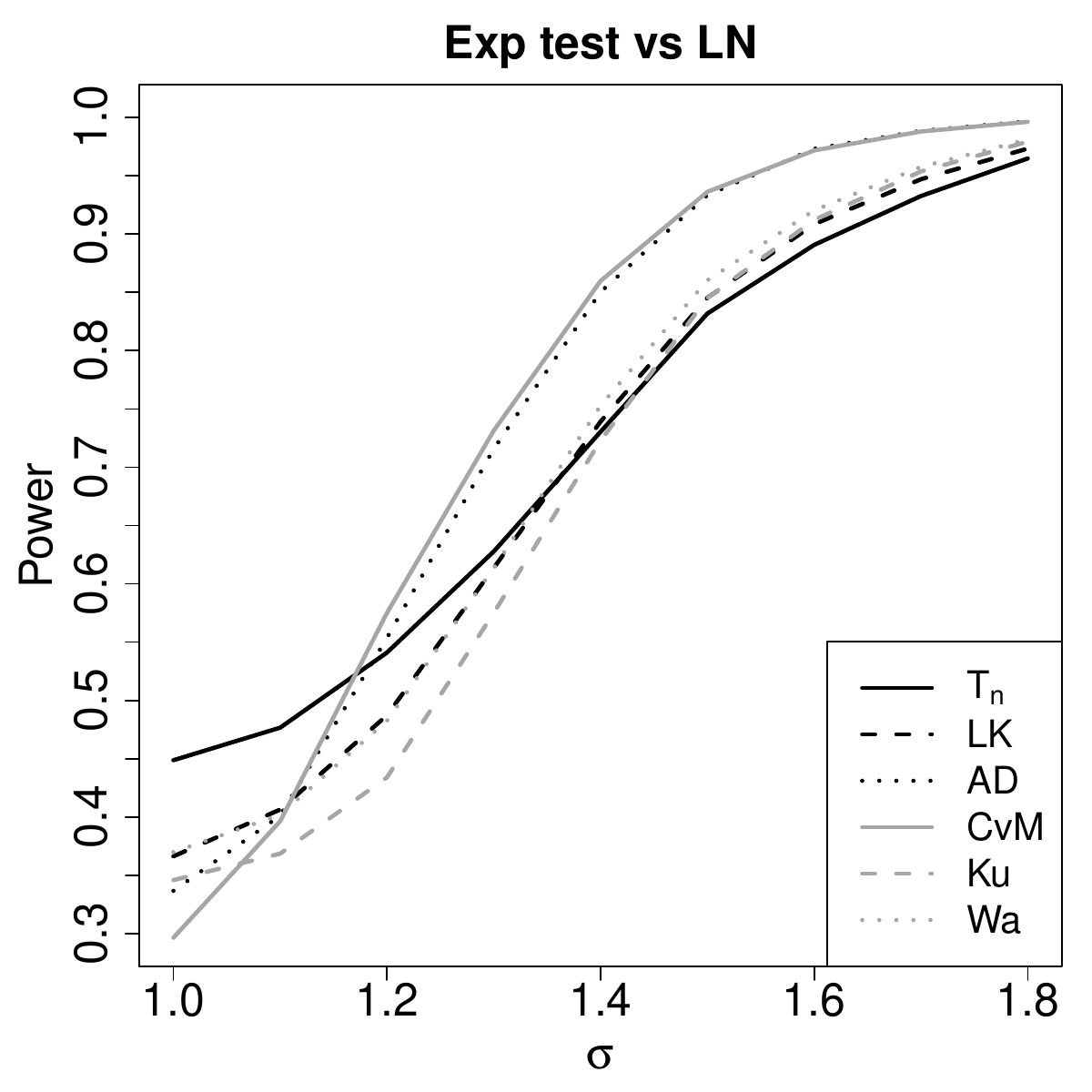}
\includegraphics[width=0.24\textwidth]{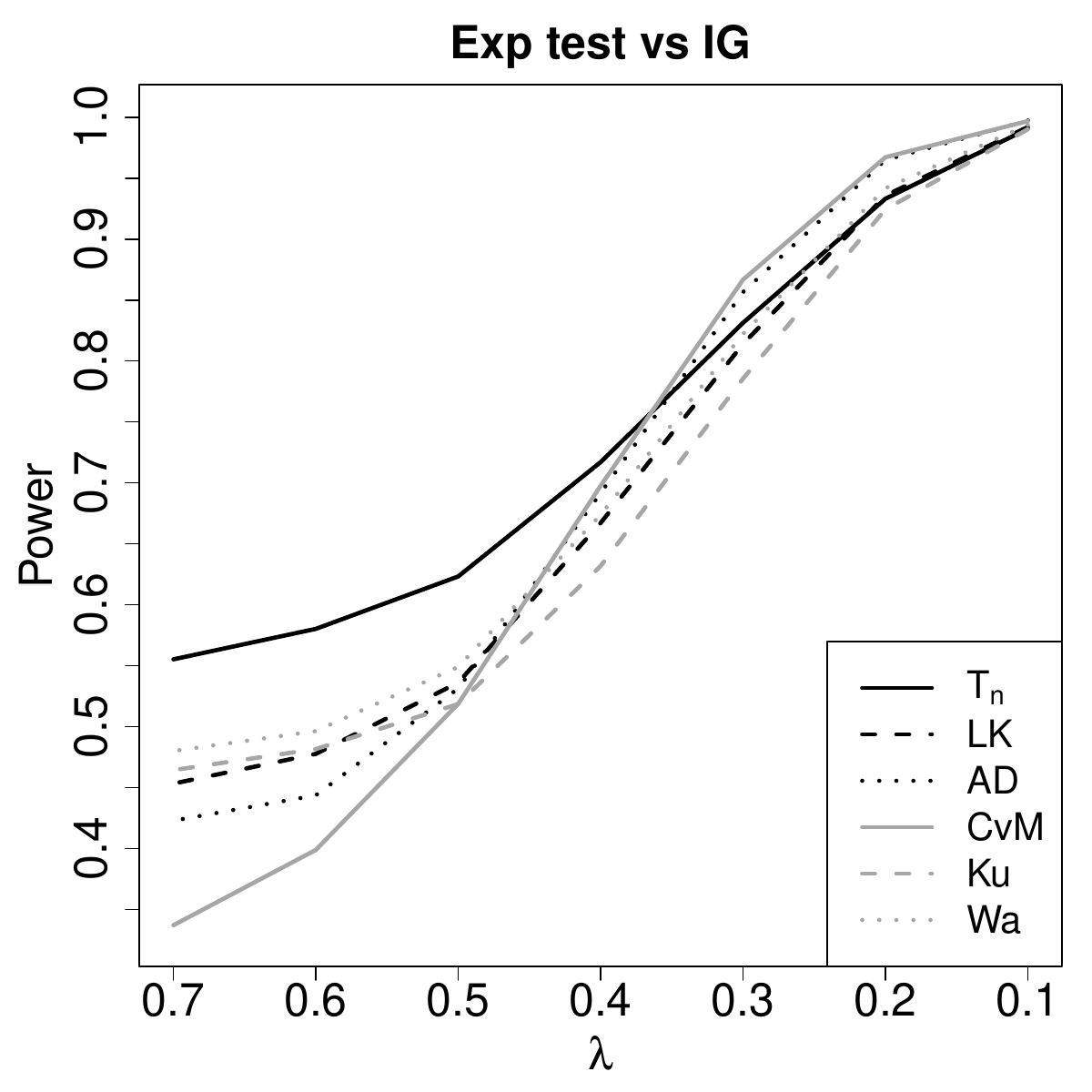}
\includegraphics[width=0.24\textwidth]{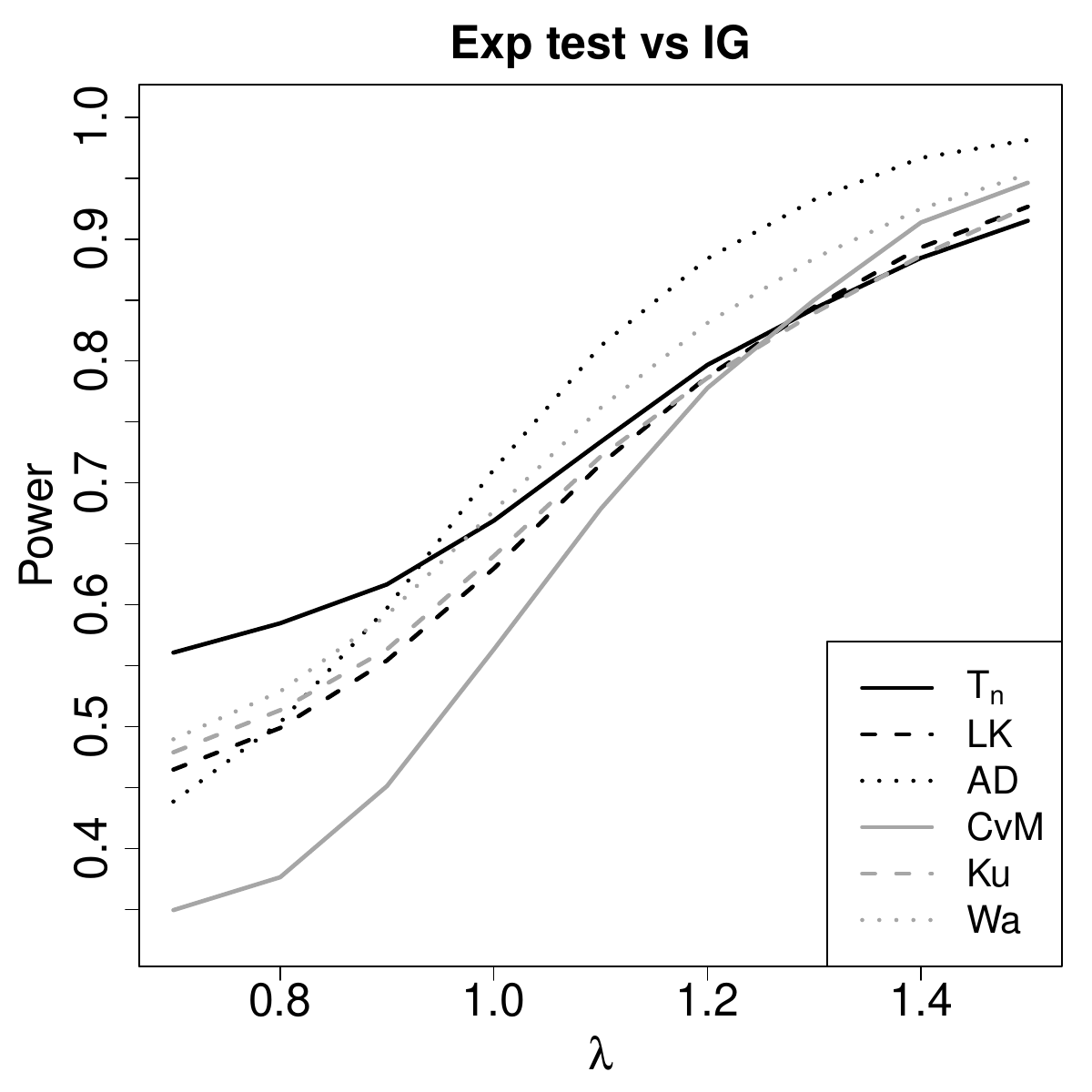}

\caption{Power curves of the $T_n$, LK, AD, CvM, Ku, and Wa tests for assessing goodness-of-fit to the Exponential distribution (with $\beta_0$ unknown) under $\mathrm{log}\mhyphen\mathrm{normal}(\mu, 0.6 \leq \sigma \leq 1)$, $\mathrm{log}\mhyphen\mathrm{normal}(\mu, 1 \leq \sigma \leq 1.8)$, $\mathrm{inverse}\mhyphen\mathrm{Gaussian}(\mu=1,0.1 \leq \lambda \leq 0.7)$ and $\mathrm{inverse}\mhyphen\mathrm{Gaussian}(\mu=1,0.7 \leq \lambda \leq 1.5)$ alternatives. The nominal significance level is $0.05$, and the sample size is
$n = 50$.}
\label{fig:exp.curve}
\end{figure}


\begin{thebibliography}{25}
\providecommand{\natexlab}[1]{#1}
\providecommand{\url}[1]{\texttt{#1}}
\expandafter\ifx\csname urlstyle\endcsname\relax
  \providecommand{\doi}[1]{doi: #1}\else
  \providecommand{\doi}{doi: \begingroup \urlstyle{rm}\Url}\fi

\bibitem[Ayebo and Kozubowski(2003)]{AyeboKozubowski2003}
A.~Ayebo and T.~J. Kozubowski.
\newblock An asymmetric generalization of {G}aussian and {L}aplace laws.
\newblock \emph{Journal of Probability and Statistical Science}, 1:\penalty0
  187--210, 2003.

\bibitem[Bera and Ghosh(2002)]{MR1893336}
A.~K. Bera and A.~Ghosh.
\newblock Neyman's {S}mooth {T}est and {I}ts {A}pplications in {E}conometrics.
\newblock In \emph{Handbook of applied econometrics and statistical inference},
  volume 165 of \emph{Statist. Textbooks Monogr.}, pages 177--230. Dekker, New
  York, 2002.
\newblock \href{http://www.ams.org/mathscinet-getitem?mr=MR1893336}{MR1893336}.

\bibitem[Beutner(2025)]{MR4874943}
E.~Beutner.
\newblock Two-sample smooth test for the equality of distributions for
  dependent data and its bootstrap consistency.
\newblock \emph{Electron. J. Stat.}, 19\penalty0 (1):\penalty0 982--1033, 2025.
\newblock \href{http://www.ams.org/mathscinet-getitem?mr=MR4874943}{MR4874943}.

\bibitem[Cupari\'{c} et~al.(2022)Cupari\'{c}, Milo\v{s}evi\'{c}, and
  Obradovi\'{c}]{MR4411854}
M.~Cupari\'{c}, B.~Milo\v{s}evi\'{c}, and M.~Obradovi\'{c}.
\newblock Asymptotic distribution of certain degenerate {V}- and {U}-statistics
  with estimated parameters.
\newblock \emph{Math. Commun.}, 27\penalty0 (1):\penalty0 77--100, 2022.

\bibitem[D'Agostino and Stephens(1986)]{doi:10.1201/9780203753064}
R.~B. D'Agostino and M.~A. Stephens.
\newblock \emph{Goodness-of-{F}it {T}echniques}.
\newblock Marcel Dekker, Inc., 1986.
\newblock ISBN 0824774876.

\bibitem[de~Wet and Randles(1987)]{MR885745}
T.~de~Wet and R.~H. Randles.
\newblock On the effect of substituting parameter estimators in limiting
  {$\chi^2\;U$} and {$V$} statistics.
\newblock \emph{Ann. Statist.}, 15\penalty0 (1):\penalty0 398--412, 1987.
\newblock \href{http://www.ams.org/mathscinet-getitem?mr=MR885745}{MR885745}.

\bibitem[Desgagn\'e and Ouimet(2026{\natexlab{a}})]{DesgagneOuimet2026github}
A.~Desgagn\'e and F.~Ouimet.
\newblock {T}est{T}rigonometric{M}oments, 2026{\natexlab{a}}.
\newblock GitHub repository at \href{https://github.com/desgagna}{https://github.com/desgagna}.

\bibitem[Desgagn\'e and Ouimet(2026{\natexlab{b}})]{supp}
A.~Desgagn\'e and F.~Ouimet.
\newblock Supplementary material for ``{O}mnibus goodness-of-fit tests for univariate continuous distributions based on trigonometric moments'', 2026{\natexlab{b}}.

\bibitem[Desgagn\'{e} et~al.(2022)Desgagn\'{e}, Lafaye~de Micheaux, and
  Ouimet]{MR4512291}
A.~Desgagn\'{e}, P.~Lafaye~de Micheaux, and F.~Ouimet.
\newblock A comprehensive empirical power comparison of univariate
  goodness-of-fit tests for the {L}aplace distribution.
\newblock \emph{J. Stat. Comput. Simul.}, 92\penalty0 (18):\penalty0
  3743--3788, 2022.
\newblock \href{http://www.ams.org/mathscinet-getitem?mr=MR4512291}{MR4512291}.

\bibitem[Desgagn\'{e} et~al.(2023)Desgagn\'{e}, Lafaye~de Micheaux, and
  Ouimet]{MR4547729}
A.~Desgagn\'{e}, P.~Lafaye~de Micheaux, and F.~Ouimet.
\newblock Goodness-of-fit tests for {L}aplace, {G}aussian and exponential power
  distributions based on {$\lambda $}-th power skewness and kurtosis.
\newblock \emph{Statistics}, 57\penalty0 (1):\penalty0 94--122, 2023.
\newblock \href{http://www.ams.org/mathscinet-getitem?mr=MR4547729}{MR4547729}.

\bibitem[Desgagn\'{e} et~al.(2025)Desgagn\'{e}, Genest, and Ouimet]{MR4906044}
A.~Desgagn\'{e}, C.~Genest, and F.~Ouimet.
\newblock Asymptotics for non-degenerate multivariate {$U$}-statistics with
  estimated nuisance parameters under the null and local alternative
  hypotheses.
\newblock \emph{J. Multivariate Anal.}, 208:\penalty0 Paper No. 105398, 18 pp.,
  2025.
\newblock \href{http://www.ams.org/mathscinet-getitem?mr=MR4906044}{MR4906044}.

\bibitem[Duchesne et~al.(2016)Duchesne, Lafaye De~Micheaux, and
  Tagne~Tatsinkou]{MR3536197}
P.~Duchesne, P.~Lafaye De~Micheaux, and J.~Tagne~Tatsinkou.
\newblock Estimating the mean and its effects on {N}eyman smooth tests of
  normality for {ARMA} models.
\newblock \emph{Canad. J. Statist.}, 44\penalty0 (3):\penalty0 241--270, 2016.
\newblock \href{http://www.ams.org/mathscinet-getitem?mr=MR3536197}{MR3536197}.

\bibitem[Fern\'{a}ndez et~al.(1995)Fern\'{a}ndez, Osiewalski, and
  Steel]{MR1379475}
C.~Fern\'{a}ndez, J.~Osiewalski, and M.~F.~J. Steel.
\newblock Modeling and inference with {$v$}-spherical distributions.
\newblock \emph{J. Amer. Statist. Assoc.}, 90\penalty0 (432):\penalty0
  1331--1340, 1995.
\newblock \href{http://www.ams.org/mathscinet-getitem?mr=MR1379475}{MR1379475}.

\bibitem[Hui et~al.(2008)Hui, Gel, and Gastwirth]{doi:10.18637/jss.v028.i03}
W.~L. Hui, Y.~R. Gel, and J.~L. Gastwirth.
\newblock lawstat: An \textsf{R} package for law, public policy and
  biostatistics.
\newblock \emph{Journal of Statistical Software}, 28\penalty0 (3):\penalty0
  1--26, October 2008.

\bibitem[Kallenberg and Ledwina(1997)]{MR1482140}
W.~C.~M. Kallenberg and T.~Ledwina.
\newblock Data-driven smooth tests when the hypothesis is composite.
\newblock \emph{J. Amer. Statist. Assoc.}, 92\penalty0 (439):\penalty0
  1094--1104, 1997.
\newblock \href{http://www.ams.org/mathscinet-getitem?mr=MR1482140}{MR1482140}.

\bibitem[Komunjer(2007)]{MR2395888}
I.~Komunjer.
\newblock Asymmetric power distribution: theory and applications to risk
  measurement.
\newblock \emph{J. Appl. Econometrics}, 22\penalty0 (5):\penalty0 891--921,
  2007.
\newblock \href{http://www.ams.org/mathscinet-getitem?mr=MR2395888}{MR2395888}.

\bibitem[Langholz and Kronmal(1991)]{MR1146353}
B.~Langholz and R.~A. Kronmal.
\newblock Tests of distributional hypotheses with nuisance parameters using
  {F}ourier series methods.
\newblock \emph{J. Amer. Statist. Assoc.}, 86\penalty0 (416):\penalty0
  1077--1084, 1991.
\newblock \href{http://www.ams.org/mathscinet-getitem?mr=MR1146353}{MR1146353}.

\bibitem[Ledwina(1994)]{MR1294744}
T.~Ledwina.
\newblock Data-driven version of {N}eyman's smooth test of fit.
\newblock \emph{J. Amer. Statist. Assoc.}, 89\penalty0 (427):\penalty0
  1000--1005, 1994.
\newblock \href{http://www.ams.org/mathscinet-getitem?mr=MR1294744}{MR1294744}.

\bibitem[McCulloch and Percy(2013)]{MR3010617}
J.~H. McCulloch and E.~R. Percy, Jr.
\newblock Extended {N}eyman smooth goodness-of-fit tests, applied to competing
  heavy-tailed distributions.
\newblock \emph{J. Econometrics}, 172\penalty0 (2):\penalty0 275--282, 2013.
\newblock \href{http://www.ams.org/mathscinet-getitem?mr=MR3010617}{MR3010617}.

\bibitem[Moore(1977)]{MR451521}
D.~S. Moore.
\newblock Generalized inverses, {W}ald's method, and the construction of
  chi-squared tests of fit.
\newblock \emph{J. Amer. Statist. Assoc.}, 72\penalty0 (357):\penalty0
  131--137, 1977.
\newblock \href{http://www.ams.org/mathscinet-getitem?mr=MR451521}{MR451521}.

\bibitem[Ngounou~Bakam and Pommeret(2024)]{MR4718464}
Y.~I. Ngounou~Bakam and D.~Pommeret.
\newblock Smooth test for equality of copulas.
\newblock \emph{Electron. J. Stat.}, 18\penalty0 (1):\penalty0 895--941, 2024.
\newblock \href{http://www.ams.org/mathscinet-getitem?mr=MR4718464}{MR4718464}.

\bibitem[Pierce(1982)]{MR653522}
D.~A. Pierce.
\newblock The asymptotic effect of substituting estimators for parameters in
  certain types of statistics.
\newblock \emph{Ann. Statist.}, 10\penalty0 (2):\penalty0 475--478, 1982.
\newblock \href{http://www.ams.org/mathscinet-getitem?mr=MR653522}{MR653522}.

\bibitem[Randles(1982)]{MR653521}
R.~H. Randles.
\newblock On the asymptotic normality of statistics with estimated parameters.
\newblock \emph{Ann. Statist.}, 10\penalty0 (2):\penalty0 462--474, 1982.
\newblock \href{http://www.ams.org/mathscinet-getitem?mr=MR653521}{MR653521}.

\bibitem[Rayner and Best(1989)]{MR1029526}
J.~C.~W. Rayner and D.~J. Best.
\newblock \emph{Smooth {T}ests of {G}oodness of {F}it}.
\newblock The Clarendon Press, Oxford University Press, New York, 1989.
\newblock ISBN 0-19-505610-8.
\newblock \href{http://www.ams.org/mathscinet-getitem?mr=MR1029526}{MR1029526}.

\bibitem[van~der Vaart(1998)]{MR1652247}
A.~W. van~der Vaart.
\newblock \emph{Asymptotic {S}tatistics}, volume~3 of \emph{Cambridge Series in
  Statistical and Probabilistic Mathematics}.
\newblock Cambridge University Press, Cambridge, 1998.
\newblock ISBN 0-521-49603-9; 0-521-78450-6.
\newblock \href{http://www.ams.org/mathscinet-getitem?mr=MR1652247}{MR1652247}.

\end{thebibliography}
\end{document}